\documentclass[11pt]{article}
\usepackage[utf8]{inputenc}
\usepackage[T1]{fontenc}
\usepackage{lmodern}
\usepackage{amsmath,amssymb,amsfonts,amsthm,mathtools,bm}
\usepackage{bbm}
\usepackage{enumitem}
\usepackage{booktabs}
\usepackage{algorithm}
\usepackage{algpseudocode}
\usepackage{natbib}
\usepackage[colorlinks=true,linkcolor=blue,citecolor=blue,urlcolor=blue]{hyperref}
\usepackage{xcolor}
\usepackage{graphicx}
\usepackage{float}
\graphicspath{{./}}
\usepackage{adjustbox}
\usepackage{multirow}
\usepackage[margin=0.8in]{geometry}
\newtheorem{theorem}{Theorem}
\newtheorem{lemma}{Lemma}
\newtheorem{proposition}{Proposition}
\newtheorem{corollary}{Corollary}
\newtheorem{assumption}{Assumption}
\newtheorem{definition}{Definition}
\newtheorem{remark}{Remark}

\newcommand{\E}{\mathbb E}
\newcommand{\Pp}{\mathbb P}
\newcommand{\R}{\mathbb R}

\newcommand{\G}{\mathbb G}
\newcommand{\ind}[1]{\mathbbm 1\{#1\}}
\newcommand{\Var}{\mathbb V\mathrm{ar}}
\newcommand{\Cov}{\mathbb C\mathrm{ov}}
\newcommand{\opro}{o_{\Pp}}

\newcommand{\Op}{O_{\Pp}}
\newcommand{\dto}{\rightsquigarrow}
\newcommand{\dtow}{\overset{\mathrm w}{\rightsquigarrow}}
\newcommand{\pto}{\overset{p}{\to}}
\newcommand{\Law}{\mathrm{Law}}
\DeclareMathOperator*{\argmin}{arg\,min}

\newcommand{\dist}{\operatorname{dist}}

\newcommand{\supp}{\operatorname{supp}}
\newcommand{\diag}{\operatorname{diag}}
\newcommand{\tr}{\operatorname{tr}}
\newcommand{\diam}{\operatorname{diam}}
\newcommand{\calA}{\mathcal A}
\newcommand{\calB}{\mathcal B}
\newcommand{\calC}{\mathcal C}
\newcommand{\calD}{\mathcal D}

\newcommand{\calF}{\mathcal F}
\newcommand{\calG}{\mathcal G}

\newcommand{\calM}{\mathcal M}

\newcommand{\calS}{\mathcal S}
\newcommand{\calT}{\mathcal T}
\newcommand{\calU}{\mathcal U}
\newcommand{\calX}{\mathcal X}
\newcommand{\bmu}{\bm\mu}

\newcommand{\hateta}{\widehat\eta}

\newcommand{\hU}{\widehat U}
\newcommand{\tU}{\widetilde U}
\newcommand{\hR}{\widehat R}
\newcommand{\hPsi}{\widehat\Psi}
\newcommand{\hphi}{\widehat\phi}
\newcommand{\eps}{\varepsilon}
\newcommand{\norm}[1]{\lVert #1\rVert}
\newcommand{\given}{\,\big|\,}
\newcommand{\BL}{\mathrm{BL}_1}
\newcommand{\am}{\alpha_{\mathrm M}}

\usepackage{setspace}
\usepackage{bibspacing}
\newcommand{\blind}{0}

\begin{document}

\def\spacingset#1{\renewcommand{\baselinestretch}%
{#1}\small\normalsize} \spacingset{1}


\if0\blind
{
  \title{The Resolution of Causal Heterogeneity} 
  \author{Yuki Ohnishi$^{*}$ and Fan Li$^{\dagger}$\\
    Department of Biostatistics,  
    Yale School of Public Health\\
    $^*$yuki.ohnishi@yale.edu, 
    $^\dagger$fan.f.li@yale.edu}
  \maketitle
} \fi

\if1\blind
{
  \bigskip
  \bigskip
  \bigskip
  \begin{center}
    {\LARGE\bf }
\end{center}
  \medskip
} \fi


\begin{abstract}
Causal subgroup analyses often report a small number of groups summarizing treatment effect heterogeneity, as if that number were a well-defined estimand. Outside genuinely latent class populations, however, a ``true'' subgroup count is model dependent rather than a population functional. We replace it with a new population estimand, the resolution profile, a functional of the causal feature law giving the fewest groups explaining a prescribed fraction of causal heterogeneity, defined for every population without latent structure. Inference is organized around one cross-fitted Bayesian-bootstrap posterior for a single structured moment process, its scores corrected with influence functions, so that paths, profiles, fixed-resolution summaries, and subgroup effects follow by composition. A uniform conditional Bernstein--von Mises theorem over a loss class containing the nonsmooth quantization losses shows this posterior merges with the efficient Gaussian limit under stated nuisance-rate and margin conditions. Subgroup-number uncertainty is not model selection but threshold nonregularity, the profile being an integer-valued threshold of a continuous path, discontinuous in the law at each knot. At these knots no single-valued selector is locally uniformly consistent over root-$n$ neighborhoods, and the set-valued report obtained by inverting a simultaneous band retains locally uniform validity over exactly the same perturbations. Simulations support the approximations, and an analysis of the MineThatData e-mail experiment illustrates the resolution-indexed report, in which two to three groups summarize the visit response while finer structure falls below a noise-floor diagnostic.
\end{abstract}

\noindent%

\spacingset{1.1} 

\section{Introduction}
\label{sec:intro}

Heterogeneity of treatment effects has evolved from a nuisance concept to the object of causal inference. Modern randomized trials and observational studies are increasingly analyzed with flexible data-adaptive estimators precisely because causal effects differ across patients, customers, and program participants. The operational question is whether individuals organize into a small number of causally interpretable subgroups with similar counterfactual response profiles. Physicians ask whether a therapy has distinct responder classes, whereas policy analysts ask whether a program helps one population segment while leaving another unmoved. Subgroup discovery is where estimated heterogeneity meets decisions, but it demands more than point estimates. Because individual characteristics interact in complex ways, the relevant subgroups are typically latent, and even their number is contested. Imposing a single, overconfident subgroup structure therefore risks reading into the data patterns that are artifacts of the analysis rather than features of the population.

Standard clustering methods, however, were not designed for this question, and pressing them into service raises two difficulties. The first is conceptual. An intuitive approach is to first construct a vector of causal features for each covariate profile, denoted $U(X)$, such as conditional treatment effects $\mu_a(x)-\mu_0(x)$ with $\mu_a(x)=\E(Y^a\mid X=x)$, and then cluster the induced feature values by $k$-means or a finite mixture model. The second step conceals a foundational difficulty. The feature law $P_U$, the distribution of $U(X)$ induced by the covariate distribution, is a well-defined nonparametric object. Outside genuinely latent class data-generating processes, however, it is not exactly a finite mixture of any particular order. It may be supported on a curve or manifold, be skewed or heavy-tailed, or be purely atomic, a case excluded as an exact truth by standard continuous-kernel mixture models. Thus, ``the true number of causal subgroups $K_0$'' is in general a property of the model rather than of the population, and a posterior over $K$ can provide calibrated-looking uncertainty about a quantity that, absent additional modeling assumptions, is not a population functional.

The second challenge is inferential. The causal features are not observed but instead functionals of the outcome regression, which must be estimated by flexible data-adaptive methods. At the level of population-level risks, semiparametric efficiency theory provides the standard remedy. For example, one can replace the plug-in empirical loss $\widehat L_n$ with a cross-fitted criterion corrected using influence functions, so that loss values are estimable at parametric rates up to second-order bias terms \citep{chernozhukov2018,kennedy2023review}. However, the common two-stage default, estimating individual effects and then clustering them as if they were observed data, may carry first-stage error into the clustering step while quantifying neither feature uncertainty nor partition uncertainty. The inferential target is therefore not simply a collection of point estimates but joint uncertainty quantification for several latent quantities, including how many groups each resolution supports, which units belong together, and how large their causal effects are. A posterior distribution is a natural reporting device for this task, but only after the population object being updated has been defined, whereas a reliable likelihood for the unobserved causal features may not be easily obtained.

\subsection{Our proposal and contributions}\label{sec:intro:proposal}

We contend that the solution to the first challenge is not better inference about $K$ but a better question, namely how many groups are needed to explain a given fraction of the heterogeneity, the question practitioners already ask informally. The proposal is estimand-first. We define the resolution profile of causal heterogeneity, the map $\gamma\mapsto K^\star(\gamma)$ returning the smallest number of groups whose best $K$-group summary explains at least a fraction $\gamma$ of the variance of the causal features. It is a new population target, a nonparametric functional of $P_U$ defined for every population without latent structure, and our question is inference for the functional itself rather than evaluation at several prechosen resolutions. The number of subgroups is thereby demoted from a parameter to be discovered to a coordinate on a resolution path, and one reports the whole profile with uncertainty rather than defending any single $K$. The conceptual shift is from subgroup discovery to causal heterogeneity summarization, the description of a feature distribution at chosen resolutions.

The profile has properties that a single number of clusters lacks. It is free of mixture-model assumptions, being a functional of $P_U$ alone, though it depends on declared ingredients, namely the feature map, the metric, the covariate population, and the ceiling $\overline K$. It is pathwise, reporting the entire step function with uncertainty, for example that two groups suffice for 60\% and five for 90\% of the heterogeneity. It is also backward compatible, reproducing the classical answer $K_0$ when $P_U$ consists of $K_0$ well-separated tight groups.

Inference is then organized around a single corrected moment process. Section~\ref{sec:method} handles mixture losses, quantization losses, membership scores, and subgroup-effect ratios through one structured moment process, with uniform efficient influence function (EIF) corrections and second-order bias control, including the margin term for quantization. We define a calibrated Bayesian-bootstrap feature-law posterior for it by cross-fitting the nuisances, forming corrected score evaluations, and reweighting the score array by Dirichlet weights \citep{rubin1981}. Here ``posterior'' refers to this conditional weighted law, not a generative Bayesian posterior over outcomes, partitions, or model order. The main guarantee is a uniform conditional Bernstein--von Mises theorem over a loss class that contains the nonsmooth quantization losses, showing that the weighted law of the corrected process merges with its efficient Gaussian limit. Every downstream report, the paths, profiles, projections, subgroup effects, and joint statements across $K$, then follows by composition of this one process limit, so partition uncertainty propagates automatically and fixed-$K$ marginal results cannot be reassembled into these joint and threshold statements.

The subgroup count enters this construction not as a model index but as a threshold. Because $K^\star(\gamma)$ is an integer-valued threshold functional of a continuous path, its value is discontinuous in the underlying law at each knot, a threshold nonregularity rather than model-selection uncertainty. At a threshold the posterior profile splits across the admissible neighboring counts, and Theorem~\ref{thm:impossibility} shows that no single-valued rule can select the count with local uniform consistency over root-$n$ neighborhoods, an instance of impossibility for nondifferentiable functionals. The matched response is a set-valued report obtained by inverting the simultaneous band, and Theorem~\ref{thm:honest} shows that it retains its nominal simultaneous guarantee over exactly the root-$n$ perturbations on which single-valued selection fails. The pair of theorems delimits the inference attainable about the profile at the root-$n$ scale, identifying the set-valued report as honest rather than conservative. Simulations and a data application illustrate these operating characteristics.

\subsection{Related work and positioning}\label{sec:intro:related}

Classical rules for choosing the number of clusters, such as the gap statistic \citep{tibshirani2001}, silhouette scores, and information criteria, turn an empirical clustering into a point choice of $K$, with neither uncertainty nor a population estimand. A more principled population formulation comes from quantization theory, in which each $K$ defines a best $K$-point approximation to the feature law $P_U$ with its approximation loss \citep{pollard1981,graf2000,bartlett1998}, moving from an empirical count to how population approximation error changes with resolution. Our contribution is inference for this quantization path when the feature law is causal, unobserved, and accessible only through cross-fitted corrected scores. Data-adaptive heterogeneity discovery is dominated by recursive partitioning and forests \citep{athey2016,wager2018}, which target the conditional effect surface rather than a population summary of its law. The closest causal inference work is Kim et al.\ \cite{kim2026causal}, who introduced causal $k$-means and provide bias-corrected codebook estimation of $P_U$ with asymptotic normality, root-$n$ inference at a single fixed and prechosen resolution. What is new is the resolution profile as the estimand, the subgroup count as a nonregular threshold functional, a set-valued report matched to an impossibility result, and one corrected moment process for all resolutions jointly. Relative to that fixed-$K$ analysis, Theorem~\ref{thm:eif}(ii) makes the second-order bias control uniform over codebooks and over $K\le\overline K$, Theorem~\ref{thm:bvm} gives the joint process limit across resolutions, Theorems~\ref{thm:impossibility} and~\ref{thm:profile} pair the impossibility at the knots where the profile jumps with the matched set-valued inference, and Theorem~\ref{thm:effects} propagates partition uncertainty into subgroup effects. Density-based formulations of causal subgroup structure \citep{kim2024hierarchical} instead pursue level-set and hierarchical targets, irregular functionals outside our root-$n$ scope, with no threshold theory and no joint guarantee across resolutions, so we regard them as complementary. Margin conditions for empirical quantization go back to Levrard \cite{levrard2015} and Biau et al. \cite{biau2008}, whose conditions are local at the optimal codebook, whereas we require a margin holding uniformly over hyperplanes, a strictly stronger hypothesis that buys bias control uniformly over codebooks.

Our set-valued report also connects to two further literatures. Confidence sets for a discrete model index go back to the model confidence set of Hansen et al. \cite{hansen2011}, and the impossibility of consistent selection under contiguous alternatives is central to the post model selection literature \citep{leeb2005,leeb2006} and to impossibility theory for nondifferentiable functionals \citep{hirano2012}. Theorem~\ref{thm:impossibility} instantiates this phenomenon at the knots of the quantization path of a causal feature law accessible only through corrected scores, and the matched band inversion delivers a set-valued guarantee that Theorem~\ref{thm:honest} shows is honest, in the locally uniform sense of the honest-inference literature \citep{li1989}, over the same root-$n$ perturbations on which selection fails. A related concern animates mixture order estimation, where posteriors and information criteria for the number of components are inconsistent or fragile under misspecification \citep{miller2014,rousseau2011}, reinforcing the case for resolution-indexed rather than order-based reports. A complementary line conditions on the clustering event itself, yielding selective tests for contrasts within a chosen partition \citep{gao2024}, whereas our target is the resolution profile of the population feature law, with partition uncertainty propagated rather than conditioned away. A separate line summarizes treatment effect heterogeneity through prespecified functionals, including the variance of the conditional treatment effect and its variable-importance extensions \citep{levy2021,hines2022}, group average effects sorted by predicted benefit \citep{chernozhukov2025generic}, and sorted effect curves \citep{chernozhukov2018sorted}. The causal heterogeneity $R^2$ curve of Section~\ref{sec:resolution} extends variance-explained summaries from prespecified to optimized groupings.

Finally, Bayesian approaches provide posterior distributions over parameters, partitions, or model order under a sampling model or loss \citep{miller2018,wade2018,rigon2023}, asking which order or partition the posterior favors, whereas the resolution profile is a functional of the causal feature law $P_U$ itself.
Our uncertainty device is instead a calibrated Bayesian-bootstrap posterior for corrected feature-law moments. The weighting scheme descends from Rubin \cite{rubin1981} and its $M$-estimation interpretation \citep{newton1994}, related to predictive-resampling and martingale-posterior constructions \citep{fong2023}. Closest to us, Yiu et al. \cite{yiu2025} debias Bayesian-bootstrap posteriors for scalar smooth functionals via one-step corrections, and a uniform Bernstein--von Mises theorem for the Dirichlet process over function classes is itself available \citep{rayvdv2021}. The specific content of Theorem~\ref{thm:bvm} is the process-level statement for cross-fitted, estimated influence-function scores, including the nonsmooth quantization class, consumed by the downstream path, profile, and threshold analyses. Semiparametric Bernstein--von Mises theory for model-based posteriors \citep{castillo2015,ray2020} gives context for when full Bayes attains efficient limits, whereas our construction is modular rather than generative, designed to attain the efficient process limit directly.

\subsection{Preliminary notation}\label{sec:intro:notation}

For a probability measure $Q$ and $Q$-integrable $f$ we write $Qf=\int f\,dQ$. $P_n$ is the empirical measure of $O_1,\dots,O_n$ and $\G_n=\sqrt n(P_n-P_0)$. For a class $\calF$, $\ell^\infty(\calF)$ is the space of bounded real functions on $\calF$ with the supremum norm. Weak convergence $\dto$ is in the Hoffmann--J{\o}rgensen sense \citep{vdvwellner1996}. 
For posterior-draw quantities, $\dtow$ denotes conditional weak convergence in probability given the data, in the bounded-Lipschitz sense. In cross-fitting, $I_b$ denotes fold $b$, $b(i)$ the fold containing observation $i$, and $\hat\mu^{(-b)}_a$, $\hat\pi^{(-b)}_a$ the nuisance estimators trained without fold $b$. 
We write $\hateta^{(-b)}=(\hat\bmu^{(-b)},\hat\pi^{(-b)})$. 
$\norm{\cdot}_{P,2}$ and $\norm{\cdot}_\infty$ are the $L_2(P)$ norm and the uniform norm. Constants $C$ may change between displays and depend only on quantities declared in the assumptions. 
We abbreviate $a\wedge b=\min(a,b)$, $a\vee b=\max(a,b)$, and write $[K]=\{1,\dots,K\}$. Population quantities evaluated at $P_0$ carry a subscript zero when the dependence on the truth matters, as in $\rho_0$, $W_0$, and $K^\star_0$.

\section{Estimands for summarizing causal heterogeneity}\label{sec:estimands}

\subsection{Causal feature law}\label{sec:setup}
We observe $n$ independent copies of $O=(X,A,Y)\sim P_0$, with covariates $X\in\calX\subseteq\R^d$, treatment $A\in\calA=\{0,1,\dots,p\}$, and outcome $Y\in\R$. Let $Y^a$ denote the potential outcome under treatment $a$, $\pi_a(x)=P_0(A=a\mid X=x)$ the generalized propensity score, and $\mu_a(x)=\E_0(Y\mid A=a,X=x)$ the outcome regression. Let $\bmu=(\mu_0,\dots,\mu_p)^\top$, $\pi=(\pi_0,\dots,\pi_p)^\top$ and $\eta=(\bmu,\pi)$ for the nuisance pair. The following standard identification conditions are assumed.

\begin{assumption}[Identification]\label{ass:ident}
(i) Consistency, $Y=Y^A$ almost surely. (ii) No unmeasured confounding, $A\perp\!\!\!\perp(Y^0,\dots,Y^p)\mid X$. (iii) Positivity, $\pi_a(X)\ge\eps_\pi>0$ almost surely for every $a\in\calA$.
\end{assumption}

Under Assumption~\ref{ass:ident}, $\mu_a(x)=\E_0(Y^a\mid X=x)$, so the vector $\bmu(x)$ is the identified counterfactual mean response profile at $x$. The analyst selects a fixed matrix $H\in\R^{q\times(p+1)}$, thereby defining the
feature map $U(x)=H\bmu(x)\in\R^q$.
Canonical choices include the vector of contrasts relative to control, $U=(\mu_1-\mu_0,\dots,\mu_p-\mu_0)^\top$, the control-anchored profile $U=(\mu_0,\mu_1-\mu_0,\dots,\mu_p-\mu_0)^\top$, and the full profile $U=\bmu$ with $H=I$. Linearity of $H$ is assumed to simplify the second-order analysis. A fixed nonlinear $C^2$ map $T$ can replace $H$ throughout, with $H$ replaced by the Jacobian $\dot T_{\mu(x)}$, at the cost of heavier notation.

\begin{definition}[Causal feature law]\label{def:featurelaw}
The causal feature law is $P_U=P_0\circ U^{-1}$, the distribution of $U(X)$ when $X\sim P_{0,X}$.
\end{definition}

The feature law is the primitive estimand of this paper. It is a nonparametric functional of $P_0$, well-defined under Assumption~\ref{ass:ident} alone, requiring no latent class or mixture structure. Each scientific quantity we consider in what follows is a functional of $P_U$, or of the joint law of $(U(X),\bmu(X),X)$.

\subsection{The resolution profile as a functional of the causal feature law}\label{sec:resolution}

In our setup, group structure is quantified by how well the feature law is summarized by a small set of candidate response profiles. Let $\calC\subset\R^q$ be a fixed compact convex set containing the feature support. For a codebook of $K$ such profiles, $c=(c_1,\dots,c_K)\in\calC^K$ define the quantization loss
\begin{equation}\label{eq:quantloss}
g_c(u)=\min_{h\in[K]}\norm{u-c_h}^2 ,
\end{equation}
the squared distance from the response profile $u$ to its nearest summary. Let $\overline K$ be a fixed, analyst-chosen bound on the complexity of summaries entertained. For $K\in[\overline K]$, the $K$-point quantization risk of the feature law and its optimizers are
\begin{equation}\label{eq:WK}
W(K)\;=\;\inf_{c\in\calC^K}P_U(g_c),
\qquad
\calC^\star(K)\;=\;\argmin_{c\in\calC^K}P_U(g_c),
\end{equation}
the population $k$-means objective on $P_U$ \citep{pollard1981,graf2000}. Since $g_c$ depends on $c$ only through the set $\{c_1,\dots,c_K\}$, the population values are label invariant.  For empirical-process arguments, however, we index the loss class by ordered tuples in $\calC^K$, allow repeated centers, and break ties deterministically, say by the lowest label.  For population optimizers we pass back to unordered sets after deleting labels. Assumption~\ref{ass:quant}\textup{(i)} of Section~\ref{sec:assumptions} is the special case in which the optimal representative has exactly $K$ distinct centers. At $K=1$, the unique minimizer is $c^\star(1)=\E U$, and $W(1)=\E\norm{U-\E U}^2=\tr\{\Var(U)\}$, which is the total variance of the causal features. We assume throughout that $W(1)>0$ and define the \emph{causal heterogeneity $R^2$} curve and the \emph{resolution profile}, for $\gamma\in[0,\rho(\overline K)]$,
\begin{equation}\label{eq:profile}
\rho(K)\;=\;1-\frac{W(K)}{W(1)}\in[0,1],
\quad
K^\star(\gamma)\;=\;\min\bigl\{K\in[\overline K]:\rho(K)\ge\gamma\bigr\}.
\end{equation}
The quantity $W(\cdot)$ is nonincreasing, hence $\rho(\cdot)$ nondecreasing with $\rho(1)=0$, and $K^\star(\cdot)$ is a nondecreasing step function with jumps at the knots $\mathcal R_0=\{\rho(K):K\in[\overline K]\}$. The profile is defined for $\gamma\in[0,\rho(\overline K)]$, and all inferential statements are made for $\gamma$ in the open interval $(0,\rho(\overline K))$. For $\gamma>\rho(\overline K)$, no $K\le\overline K$ attains $\gamma$, so $K^\star(\gamma)$ is undefined. The ceiling $\overline K$ is fixed throughout the main text, while Supplementary Remark~\ref{rem:growingK} discusses growing ceilings.

The entire development is carried out under $W(1)>0$. At the degenerate law the normalization defining $\rho$ is undefined and the root-$n$ analysis of Section~\ref{sec:theory} does not apply, so reported intervals for $W(1)$ above zero are statements within the maintained positive-heterogeneity regime rather than a formal test of $W(1)=0$, which is outside this paper's scope. Supplementary Section~\ref{sec:app:star_null} illustrates the recommended behavior when $W(1)$ is statistically indistinguishable from zero.

Several further remarks follow. First, $\rho(K)$ is the explained portion of heterogeneity and has an exact analysis-of-variance reading. Writing $\Pi_c(u)=c_{h_c(u)}$ for the codebook projection, $\rho(K)=\sup_{c\in\calC^K}\tr\{\Var(\Pi_c(U))\}/\tr\{\Var(U)\}$ when centroids are cell means, the fraction of causal-feature variance captured by the best $K$-group summary. Second, the profile is backward compatible with a mixture truth. If $P_U=\sum_{h\le K_0}\omega_h Q_h$ with component means separated by at least $\Delta$ and within-component spread at most $\sigma^2$, then $1-\rho(K_0)\le\sigma^2/W(1)$ and $1-\rho(K)\ge\omega_{\min}(\Delta^2/8-\sigma^2)/W(1)$ for $K<K_0$, so whenever $\sigma^2<\omega_{\min}\Delta^2/16$ one has $K^\star(\gamma)=K_0$ on a nonempty interval of resolutions, and the profile reports the classical $K_0$ with its supporting resolutions (Supplementary Lemma~\ref{lem:nearmixture}). Third, the profile is invariant to relabeling and to the choice of optimal codebook when $\calC^\star(K)$ is not a singleton, being defined through the values $W(K)$ rather than the minimizers, so it is a well-defined target without any uniqueness condition, although the band-based inference of Section~\ref{sec:path_theory} still uses Assumption~\ref{ass:quant}\textup{(i)} and Supplementary Remark~\ref{rem:nonunique} records what fails without it. Fourth, for multivariate features the Euclidean metric in \eqref{eq:quantloss} is itself part of the estimand, and a fixed positive-definite weighting can replace it without new theory (Supplementary Remark~\ref{rem:metric}). Finally, the resolution profile fixes the desired approximation quality, while Supplementary Section~\ref{sec:geometry_penalized_profile_supp} develops a complementary formulation assigning a linear incremental cost to each additional subgroup, with the corresponding inference.

\subsection{Soft summaries at a working resolution}\label{sec:soft}

Once a working resolution is chosen, the analyst wants interpretable group descriptions and group-specific effects, which we provide as projections with no truth claim attached. Fix a location--scale family $k(\cdot;\theta)$ on $\R^q$ (e.g.\ Gaussian) and let
\begin{equation}\label{eq:proj}
\beta^\star(K)\;\in\;\argmin_{\beta\in\calB_K}P_U(\ell_\beta),
\qquad \ell_\beta(u)=-\log m_\beta(u),\quad m_\beta(u)=\sum_{h\in[K]}\omega_h\,k(u;\theta_h),
\end{equation}
the Kullback--Leibler (KL) projection of $P_U$ onto $K$-component mixtures over a compact set $\calB_K$. The induced soft memberships are the component posterior weights under the projected mixture, $r_h(u;\beta)=\omega_h\,k(u;\theta_h) / m_\beta(u)$ for $h\in[K]$,
and the subgroup-specific mean response of arm $a$ in group $h$ is the membership-weighted mean
\begin{equation}\label{eq:effects}
\psi_{h,a}(K)\;=\;\frac{\E_0\{r_h(U;\beta^\star(K))\,\mu_a(X)\}}{\E_0\{r_h(U;\beta^\star(K))\}} ,
\end{equation}
with contrasts $\psi_{h,a}-\psi_{h,a'}$ the subgroup treatment effects. Hard-cell analogues replacing $r_h$ by Voronoi indicators require the margin condition rather than the projection regularity below, so we focus on the soft versions for smoothness. On the interpretive stance, $\beta^\star(K)$ is the best $K$-component description of $P_U$ in KL divergence, a well-defined functional under misspecification, and nothing in the sequel asserts that $P_U$ is a mixture. 
Table~\ref{tab:cast} collects the representative population objects of the analysis together with the result that delivers each one's posterior inference, 
and Supplementary Section~\ref{sec:report} states the reporting protocol.

\begin{table}[ht!]
\centering
\caption{Representative population objects of the analysis. Each is a functional of the causal feature law $P_U$ (or of the joint law of $\{U(X),\bmu(X)\}$), estimated through the moment process of Definition~\ref{def:process}. The last column gives the result delivering its inference, and the optional penalized dual is summarized in Supplementary Section~\ref{sec:geometry_penalized_profile_supp}.}
\label{tab:cast}
\small
\begin{tabular}{llll}
\toprule
Object & Defined & Meaning & Inference\\
\midrule
$W(K)$ & \eqref{eq:WK} & best within-group dispersion with $K$ groups & Cor.~\ref{cor:path}\\
$\rho(K)$ & \eqref{eq:profile} & fraction of heterogeneity explained by $K$ groups & Cor.~\ref{cor:path}\\
$K^\star(\gamma)$ & \eqref{eq:profile} & smallest $K$ achieving resolution $\gamma$ & Thms.~\ref{thm:profile}, \ref{thm:honest}\\
knots $\mathcal R_0$ & \S\ref{sec:resolution} & resolutions at which $K^\star$ jumps & Thms.~\ref{thm:impossibility} and \ref{thm:profile}(ii)\\
$\psi_{h,a}(K)$ & \eqref{eq:effects} & mean response of arm $a$ in subgroup $h$ & Thm.~\ref{thm:effects}\\
\bottomrule
\end{tabular}
\end{table}

\section{Introducing the feature-law posterior}\label{sec:method}

\subsection{One moment process for the entire analysis}\label{sec:moment}

Every estimand in Sections~\ref{sec:resolution}--\ref{sec:soft} is an explicit functional of the feature law, or of the joint law of $(U(X),\bmu(X))$. The path and profile are built from the values $P_U(g_c)$, the projection parameters from $P_U(\ell_\beta)$, and the subgroup means are ratios of joint moments involving memberships and arm-specific responses. A cluster analysis thus involves many functionals at once, and estimating them separately, each with its own correction, would forfeit the joint uniform control needed for simultaneous reporting, so we organize the full catalog through a single indexed family.

\begin{definition}[Structured moment process]\label{def:process}
Let $\calF$ be a class of measurable functions $f:\R^q\times\R^{p+1}\to\R$. The structured moment process is
\begin{equation}\label{eq:process}
\Psi_f(P)\;=\;\E_P\bigl[f\{U_P(X),\bmu_P(X)\}\bigr],\qquad f\in\calF,
\end{equation}
where $U_P=H\bmu_P$ and $\bmu_P$ is the outcome-regression vector under $P$. We write $\Psi_0=\Psi_\cdot(P_0)$ and regard $\Psi(P)=\{\Psi_f(P):f\in\calF\}$ as an element of $\ell^\infty(\calF)$.
\end{definition}

Functionals depending only on $u$ correspond to $f(u,m)=g(u)$, written $\calG\subset\calF$, so $\Psi_g(P)=\int g\,dP_U$. The dependence on $m=\bmu(x)$ beyond $u$ lets subgroup-specific effects (means of $\mu_a$ within feature-defined groups) live in the same process. The class $\calF$ is the union of three blocks, formalized in Assumption~\ref{ass:class} below:
\begin{itemize}[leftmargin=2em]
\item[$\calF_{\mathrm{sm}}$:] smooth losses $f(u,m)$ that are twice continuously differentiable with uniformly bounded first and second derivatives on the relevant compact space, e.g.\ the projection losses $\ell_\beta$ of \eqref{eq:proj} and soft-membership scores;
\item[$\calF_{\mathrm{qt}}$:] the quantization losses $g_c$ of \eqref{eq:quantloss}, indexed by codebooks $c\in\calC^K$, $K\in[\overline K]$;
\item[$\calF_{\mathrm{str}}$:] structured effect scores. For a mixture parameter $\beta\in\calB_K$, component $h$, and arm $a$,
\begin{equation}\label{eq:effectscores}
f^{N}_{h,a;\beta}(u,m)=r_h(u;\beta)\,m_a,
\qquad
f^{D}_{h;\beta}(u,m)=r_h(u;\beta).
\end{equation}
\end{itemize}
In this notation, the estimand catalogue of Sections~\ref{sec:resolution}--\ref{sec:soft} is a short list of process functionals. Since $\Psi_0(g)=P_U(g)$ for $g\in\calG$, the definitions \eqref{eq:WK} and \eqref{eq:proj} read verbatim in process notation, i.e., $W(K)=\inf_{c\in\calC^K}\Psi_0(g_c)$
and $\beta^\star(K)\in\argmin_{\beta\in\calB_K}\Psi_0(\ell_\beta)$, while the subgroup means \eqref{eq:effects} are ratios of structured moments, $\psi_{h,a}(K) = \Psi_0\bigl(f^N_{h,a;\beta^\star(K)}\bigr) \:/\: \Psi_0\bigl(f^D_{h;\beta^\star(K)}\bigr)$.
A single uniform inferential statement about $\Psi$ therefore delivers joint inference for every object in the catalogue by composition. We next introduce the regularity conditions.

\begin{assumption}[Boundedness and truncation]\label{ass:bound}
We assume
$|Y|\le B_Y$ almost surely, $\max_a\norm{\mu_a}_\infty\le B_\mu$, the feature support $\calU=\supp(P_U)$ is contained in the compact convex set $\calC$ of \eqref{eq:quantloss}, and $W(1)=\tr\Var(U)>0$. The nuisance estimators are truncated so that, deterministically, $\max_a\norm{\hat\mu^{(-b)}_a}_\infty\le B_\mu$, $\hat\pi^{(-b)}_a\ge\eps_\pi/2$, and $H\hat\bmu^{(-b)}(x)\in\calC$ for all $x,b$.
\end{assumption} 

\begin{assumption}[Function class]
\label{ass:class}
Let $\calF=\calF_{\mathrm{sm}}\cup\calF_{\mathrm{qt}}\cup\calF_{\mathrm{str}}$.
\begin{enumerate}[label=(\roman*),leftmargin=2em]
\item $\calF_{\mathrm{sm}}=\{f_t:t\in T_{\mathrm{sm}}\}$, where
$T_{\mathrm{sm}}\subset\R^{d_{\mathrm{sm}}}$ is compact, each $f_t$ is $C^2$
on an open neighborhood of $\calC\times[-B_\mu,B_\mu]^{p+1}$, and $\sup_t\{\norm{f_t}_\infty+\norm{\nabla f_t}_\infty
+\norm{\nabla^2 f_t}_\infty\}\le B_F$,
with $t\mapsto(f_t,\nabla f_t)$ Lipschitz in supremum norm.
\item $\calF_{\mathrm{qt}}=\{g_c:c\in\calC^K,\ K\in[\overline K]\}$, with
$\overline K$ fixed, $g_c$ as in \eqref{eq:quantloss}, ordered codebook
tuples, and the deterministic tie rule above.
\item $\calF_{\mathrm{str}}$
contains the scores in \eqref{eq:effectscores} and the coordinate functions of
$\nabla_\beta\ell_\beta$ for $\beta$ in compact sets $\calB_K$ satisfying
$\inf_{u\in\calC,\beta}m_\beta(u)\ge\underline m>0$.  The maps $(u,\beta)\mapsto
(\ell_\beta,\nabla_u\ell_\beta,\nabla_\beta\ell_\beta,\nabla^2\ell_\beta)$
are continuous and obey the same uniform boundedness and Lipschitz conditions
as in \textup{(i)}, including the mixed derivatives
$\nabla_u(\partial_{\beta_j}\ell_\beta)$ for every coordinate $j$.
\end{enumerate}
Vector scores are always interpreted coordinatewise inside $\Psi$ or $\phi$, with
finite-dimensional vectors reassembled afterward.
\end{assumption}

The first theorem gives the EIF of the process and the key second-order bias bound uniformly over $\calF$. For a candidate nuisance write $\eta=(\bmu_\eta,\pi_\eta)$ and $U_\eta(x)=H\bmu_\eta(x)$, and define the inverse probability weighted residual vector $R(O;\eta)\in\R^{p+1}$ by
\begin{equation}\label{eq:residual}
R_a(O;\eta)\;=\;\frac{\ind{A=a}}{\pi_{\eta,a}(X)}\bigl\{Y-\mu_{\eta,a}(X)\bigr\},\qquad a\in\calA ,
\end{equation}
and, for $f\in\calF$ differentiable in its arguments, the corrected score
\begin{equation}\label{eq:score}
\phi_f(O;\eta)\;=\;f\{U_\eta(X),\bmu_\eta(X)\}
+\Bigl[\nabla_u f\{U_\eta(X),\bmu_\eta(X)\}^\top H+\nabla_m f\{U_\eta(X),\bmu_\eta(X)\}^\top\Bigr]R(O;\eta).
\end{equation}
For the quantization losses \eqref{eq:quantloss}, which are not everywhere differentiable, $\nabla_u g_c(u)=2\{u-c_{h_c(u)}\}$ with $h_c(u)=\argmin_{h}\norm{u-c_h}$ is defined off the Lebesgue-null set of Voronoi boundaries, where the convention is immaterial under the margin condition below. For $g\in\calG$ the score reduces to $\phi_g(O;\eta)=g\{U_\eta(X)\}+\nabla g\{U_\eta(X)\}^\top HR(O;\eta)$, and since $\E_0\{R(O;\eta_0)\mid X\}=0$ we have $\E_0\{\phi_f(O;\eta_0)\}=\Psi_f(P_0)$, so the correction is mean-zero.

\begin{theorem}[Efficient influence process and second-order bias]\label{thm:eif}
Let Assumptions~\ref{ass:ident}--\ref{ass:class} hold. For quantization scores $g_c\in\calF_{\mathrm{qt}}$, part \textup{(i)} is asserted only for codebooks whose Voronoi boundaries are $P_U$-null. For the uniform bias bound in part \textup{(ii)}, assume the margin condition, Assumption~\ref{ass:margin} of Section~\ref{sec:assumptions}, whenever the supremum includes quantization scores. Then the following hold.
\begin{enumerate}[label=(\roman*),leftmargin=2em]
\item For any such $f\in\calF$, the map $P\mapsto\Psi_f(P)$ is pathwise differentiable at $P_0$ in the nonparametric model, with efficient influence function $\phi_f(O;\eta_0)-\Psi_f(P_0)$.

\item For any candidate nuisance pair $\bar\eta=(\bar\bmu,\bar\pi)$ satisfying Assumption~\ref{ass:bound}, the one-step bias obeys
\begin{equation}\label{eq:bias}
\sup_{f\in\calF}\ \Bigl|\E_0\{\phi_f(O;\bar\eta)\}-\Psi_f(P_0)\Bigr|
\;\le\; C\,Rem_2(\bar\eta),
\end{equation}
where, with 
$r_\mu=\max_a\norm{\bar\mu_a-\mu_a}_{P_0,2}$ and 
$r_\pi=\max_a\norm{\bar\pi_a-\pi_a}_{P_0,2}$,
\begin{equation}\label{eq:R2}
Rem_2(\bar\eta) = r_\mu^2+r_\mu r_\pi +\ind{\calF\supseteq\calF_{\mathrm{qt}}} r_\mu^{\,2(1+\am)/(2+\am)} ,
\end{equation}
and $\am$ is the margin exponent in Assumption~\ref{ass:margin}.
\end{enumerate}
\end{theorem}

The boundary-null condition above is implied by Assumption~\ref{ass:margin}, and the proof is in Supplementary Section~\ref{app:eif}. We use \eqref{eq:bias} in two ways. First, for the smooth and structured classes the bias has the product and squares form $r_\mu^2+r_\mu r_\pi$, rate robustness rather than exact symmetric double robustness. When $\bar\bmu=\bmu$ the centering is exactly unbiased under any bounded $\bar\pi$, whereas no propensity estimator can repair a regression estimator that converges to the wrong limit, since the feature law is itself a functional of $\bmu$. The quantization term differs again, depending on $r_\mu$ alone, and Supplementary Remark~\ref{rem:ratedr} expands on this robustness structure. Second, the bound is uniform over the loss class, so a single nuisance fit controls the bias of every risk value, codebook objective, and effect score simultaneously. The last term of \eqref{eq:R2} and the margin exponent are absent if attention is restricted to $\calF_{\mathrm{sm}}\cup\calF_{\mathrm{str}}$.

\begin{remark}[A noise floor for the path level]\label{rem:noisefloor}
The bias bound \eqref{eq:R2} also has a constructive reading at the level of the path. When the design propensities are known and used, so that $\hat\pi=\pi$ as in a randomized or stratified trial, the corrected quantization score for every codebook $c$ equals the population loss $g_c\{U(X)\}$ minus the squared feature-estimation error $\|\hU(X)-U(X)\|^2$, up to margin-controlled boundary terms, and Supplementary Proposition~\ref{prop:levelshift} states this level-shift identity precisely with all its qualifications. The squared error is a common downward level shift of the whole path $K\mapsto W(K)$, so path contrasts $W(K)-W(K')$ and the merge scales are insensitive to it while the level $W(1)$ and the normalization $\rho(K)=1-W(K)/W(1)$ absorb it in full. Two consequences follow. First, a corrected $\widehat W(1)$ at or below zero signals that the feature-estimation error $\E\|\hU-U\|^2$ is at least comparable to the true heterogeneity $W(1)$, so the nuisance-rate condition of Assumption~\ref{ass:nuisance} fails in the sample at hand rather than being evidence that $W(1)=0$, complementing the caveat on $W(1)>0$ in Section~\ref{sec:resolution}. Second, the level shift is estimable. Refitting the outcome regressions on arm-stratified halves $S_A,S_B$ with the same learner stack to obtain $\hat\bmu_A,\hat\bmu_B$, the split-difference diagnostic $\widehat\Delta$ is the sample average of $\tfrac12\|H\{\hat\bmu_A(X_i)-\hat\bmu_B(X_i)\}\|^2$, and under a mild variance-scaling assumption on the learner $[\widehat\Delta/2,\widehat\Delta]$ brackets the variance part of the full-sample floor (Supplementary Section~\ref{sec:noisefloor_supp}). The quantity $\widehat\Delta$ is a diagnostic, with no coverage claim attached to $\widehat W(1)+\widehat\Delta$. We recommend reporting $\widehat\Delta$ alongside $\widehat W(1)$ and reading the two through a single two-tier gate, applied identically in every analysis below.

The first tier is detection. A simultaneous band for $W(1)$ that excludes zero supports $W(1)>0$, and because the level shift is a downward bias this reading is one-sided conservative, valid even when the floor is large. A band that does not clear zero instead signals that the data cannot support resolution analysis at the attempted feature dimension and sample size, in the sense of Assumption~\ref{ass:nuisance}, and the analysis stops there.

The second tier is the reliability of the $\rho$ scale. By the multiplicative form of the level-shift identity (Supplementary Corollary~\ref{cor:levelshift_rho}), a common shift $\delta$ inflates the normalized curve, $\rho_\delta(K)=\rho(K)\,W(1)/\{W(1)-\delta\}$, so the $\rho$ path and the set-valued report $\widehat C(\gamma)$ tilt toward coarser counts and nothing cancels. Reliability therefore requires the floor to be small relative to the level, measured by the reliability ratio $\hat r=[\widehat\Delta/2,\widehat\Delta]/\widehat W(1)$. When $\hat r$ is small the $\rho$-scale reports carry their nominal reading. When it is not, every $\rho$-scale statement must be accompanied by the shift-sensitivity reading obtained by re-inverting the profile at the bracket endpoints, and any set-valued report at a resolution inside the shift-sensitivity range of a knot is reported as sensitivity-qualified rather than as a nominal confidence statement.
\end{remark}

\subsection{Construction}\label{sec:construction}

Split $[n]$ into $B$ folds $I_1,\dots,I_B$ of comparable size ($B$ fixed, e.g.\ $B=5$).  For each fold $b$, estimate $\hateta^{(-b)}=(\hat\bmu^{(-b)},\hat\pi^{(-b)})$ from the data outside $I_b$ by any supervised learners, truncated so that the bounds of Assumption~\ref{ass:bound} hold.  Write $b(i)$ for the fold containing observation $i$, define $\hU_i=H\hat\bmu^{(-b(i))}(X_i)$, $\hR_i=R\{O_i;\hateta^{(-b(i))}\}$,
and form the cross-fitted corrected evaluations $\hphi_{f,i}=\phi_f\bigl(O_i;\hateta^{(-b(i))}\bigr)$, for $i\in[n],\ f\in\calF$,
with $\phi_f$ as in \eqref{eq:score}.  The cross-fitted one-step process is $\hPsi(f)=\frac1n\sum_{i=1}^n \hphi_{f,i}$,
the uniform analogue of the standard double machine learning estimator \citep{chernozhukov2018}.

\begin{definition}[Feature-law posterior]\label{def:flposterior}
Let 
$w^{(s)}=(w^{(s)}_1,\dots,w^{(s)}_n)\sim n\cdot\mathrm{Dirichlet}(1,\dots,1),$ independently of the data, for $s=1,\dots,S$.  The feature-law posterior is the conditional law, given the data, of the random process
\begin{equation}\label{eq:posteriorprocess}
\Psi^{(s)}(f)
=
\frac1n\sum_{i=1}^n w^{(s)}_i\,\hphi_{f,i},
\qquad f\in\calF ,
\end{equation}
viewed as an element of $\ell^\infty(\calF)$.  For any functional $T$ defined on the relevant domain in $\ell^\infty(\calF)$, its posterior is the conditional law of $T(\Psi^{(s)})$.
\end{definition}

This conditional law is a posterior for the corrected moment process. It need not correspond to a probability measure on the feature space, since a weighted corrected functional can take negative values in finite samples, so the name feature-law posterior is shorthand for the Bayesian-bootstrap posterior of the corrected feature-law moments.

The full procedure is summarized in Supplementary Algorithm~\ref{alg:main}. In outline, one cross-fits the nuisances, forms the corrected evaluations $\hphi_{f,i}$ and the point process $\hPsi(f)=n^{-1}\sum_i\hphi_{f,i}$, and applies each estimand functional to $\hPsi$ for the point estimates. One then reweights the corrected array by $S$ independent Dirichlet draws and within every draw recomputes each functional of interest, from the quantization values $W^{(s)}(K)$ and resolution curve $\rho^{(s)}(K)$ through the profile $K^{\star(s)}$ and the selected fixed-resolution and subgroup summaries. The draws then deliver the scale estimates, the simultaneous quantile, and the confidence band, inverted to the set-valued profile report $\widehat C(\gamma)$. Table~\ref{tab:cast} collects the estimands so sampled.

For squared-distance quantization, the corrected criterion has a useful pseudo-feature form. With $\tU_i=\hU_i+H\hR_i$, for any codebook $c\in\calC^K$, $\hphi_{g_c,i} = \bigl\|\tU_i-c_{h_c(\hU_i)}\bigr\|^2 - \bigl\|H\hR_i\bigr\|^2$, where $h_c(\hU_i)$ is the Voronoi cell of the estimated feature $\hU_i$, not of the pseudo-feature $\tU_i$. Centers may therefore be updated by averaging pseudo-features, but assignments must be computed from $\hU_i$, and plain $k$-means on $\{\tU_i\}$ optimizes a different, noise-convolved objective. The formal weighted identity and implementation details are in Supplementary Section~\ref{sec:pseudo_supp}. Because \eqref{eq:posteriorprocess} is linear in the corrected score array, any component common to two corrected risks cancels in every draw of a path contrast, the main reason to base inference on a weighted corrected moment process rather than an exponentiated order posterior (Supplementary Remark~\ref{rem:whynot} and Supplementary Section~\ref{sec:app:energy_amplification}).

The word ``posterior'' is used in a calibrated sense. The construction is a Bayesian-bootstrap posterior for the corrected moment process, not a generative Bayesian posterior for outcomes, causal features, partitions, or model order, and randomness enters only through exchangeable data weights, so no Markov chain over partitions or mixture parameters is run and no model order is sampled. Every reported quantity is recomputed as a deterministic functional of each weighted process draw. This per-draw recomputation matters for composite summaries. Recomputing $\beta^{(s)}(K)$ inside each draw, for instance, propagates uncertainty in the subgroup definition into the posterior of $\psi_{h,a}(K)$ rather than treating the partition as fixed (Supplementary Section~\ref{sec:bayesian}).
Two terms are used with care below. A posterior equal-tailed credible interval refers to the quantiles of the weighted conditional law of a smooth scalar summary, which by Theorem~\ref{thm:delta} is an asymptotically valid frequentist confidence interval at the same level. The set-valued profile report $\widehat C(\gamma)$, by contrast, is proved as a frequentist confidence correspondence and called a confidence set throughout.

\section{Asymptotic theory}\label{sec:theory}

The corrected moment process underlying the feature-law posterior is the central object in the theoretical argument.  Once that process is shown to have the efficient Gaussian limit conditionally and unconditionally, the resolution path, resolution-profile report, and subgroup effects follow by finite-dimensional composition and band inversion.  We develop the desired results in the following order. The efficient process limit comes first, then the posterior functional delta method that transfers it to every derived report, then path inference, the impossibility result, the set-valued profile with its locally uniform validity guarantee, and subgroup effects.  Guarantees for atomic laws, together with consistency results and a caveat for mixed atomic-continuous laws, are collected in Supplementary Section~\ref{sec:app:additional_theory}, and inference for the optional penalized profile in Supplementary Section~\ref{sec:penalized_supp_inference}.

\subsection{Regularity assumptions}\label{sec:assumptions}

\begin{assumption}[Margin]\label{ass:margin}
There exist $\am\in(0,1]$, $C_M<\infty$, and $t_0>0$ such that for every hyperplane $B\subset\R^q$ and every $t\in(0,t_0]$,
$P_0\bigl(\dist\{U(X),B\}\le t\bigr)\;\le\;C_M\,t^{\am}$.
\end{assumption}

\begin{assumption}[Quantization]\label{ass:quant}
For each $K\in[\overline K]$, we assume \textup{(i)} the optimal codebook is unique as a set, $\calC^\star(K)=\{c^\star(K)\}$ with $K$ distinct centers in the interior of $\calC$, and \textup{(ii)} $W(K)<W(K-1)$ for $K\ge2$, equivalently, $P_U$ is not supported on fewer than $\overline K$ points. 
\end{assumption}

\begin{assumption}[Projection]\label{ass:proj}
For each $K$ under consideration, $\beta^\star(K)$ is the unique minimizer of $\beta\mapsto\Psi_0(\ell_\beta)$ over $\calB_K$ up to label permutation, lies in the interior. The Hessian $V_\beta=\nabla^2_\beta\Psi_0(\ell_\beta)\big|_{\beta^\star(K)}$ is nonsingular, and moreover $\Psi_0(f^D_{h;\beta^\star(K)})>0$ for each $h$.
\end{assumption}

\begin{assumption}[Nuisance rates]
\label{ass:nuisance}
With $r_\mu,r_\pi$ the $L_2(P_0)$ rates of $\hateta^{(-b)}$ as in Theorem~\ref{thm:eif}, we assume \textup{(i)} $\sqrt n\,Rem_2(\hateta^{(-b)})=\opro(1)$ for each $b$, with $Rem_2$ as in \eqref{eq:R2}, and \textup{(ii)} $\delta_n=\opro(1)$, where $\delta_n=\max_b\,\sup_{f\in\calF}\,\bigl\Vert\phi_f(\cdot;\hateta^{(-b)})-\phi_f(\cdot;\eta_0)\bigr\Vert_{P_0,2}$.
\end{assumption}

Assumption~\ref{ass:margin} is the process-level margin condition needed for quantization losses uniformly over codebooks.  It holds with $\am=1$ when $P_U$ has a bounded Lebesgue density on $\calC$ and fails for atoms. Supplementary Proposition~\ref{prop:atomic} characterizes which quantization results continue to hold in the purely atomic setting and how the corresponding conclusions must be modified. Assumption~\ref{ass:quant}\textup{(i)} rules out exactly symmetric optimal codebooks. The profile itself is value-based and remains the recommended target when codebook labels or minimizers are not unique.  
Assumption \ref{ass:proj} is only for soft summaries. The nuisance-rate condition in Assumption~\ref{ass:nuisance}\textup{(i)} is the product-rate requirement from Theorem~\ref{thm:eif}. For quantization, the margin term strengthens the regression-rate requirement to $r_\mu=\opro(n^{-(2+\am)/(4(1+\am))})$, which is $\opro(n^{-3/8})$ when $\am=1$.
The increment condition \textup{(ii)} imposes no additional rate restriction beyond nuisance consistency. Supplementary Lemma~\ref{lem:deltaqt} shows that on the quantization class it holds automatically at the rate $r_\mu^{\am/(2+\am)}+r_\mu+r_\pi$, so the binding requirement throughout is the bias condition \textup{(i)}.
For theorem statements involving only a subclass $\calF_0\subseteq\calF$, the nuisance-rate
condition is understood with $Rem_2$ and the supremum defining $\delta_n$ restricted
to $\calF_0$.  Thus, for
$\calF_{\mathrm{eff}}:=\calF_{\mathrm{sm}}\cup\calF_{\mathrm{str}}$, used in the
projection and subgroup-effect results, $R_{2,\mathrm{eff}}=r_\mu^2+r_\mu r_\pi$.
The quantization boundary term and its margin-rate condition are required only
when scores from $\calF_{\mathrm{qt}}$ are included.

\subsection{The efficient feature-law posterior}\label{sec:bvm}

The primary result below is stated in the mode of conditional weak convergence used throughout for the reweighted process.

\begin{definition}[Conditional weak convergence in probability]\label{def:condweak}
Let $\{Z_n^{(s)}\}$ be random elements of a metric space $(\mathbb D,d)$ depending on the data $O_{1:n}$ and on weights $w^{(s)}$ independent of the data, and let $Z$ be a tight Borel element of $\mathbb D$. We write $Z_n^{(s)}\dtow Z$ if
\begin{equation*}
\sup_{\varphi\in\BL(\mathbb D)}\Bigl|\E_w\,\varphi\bigl(Z_n^{(s)}\bigr)-\E\,\varphi(Z)\Bigr|\;\longrightarrow\;0
\quad\text{in outer probability},
\end{equation*}
where $\BL(\mathbb D)$ is the set of $1$-Lipschitz functions bounded by $1$ and $\E_w$ integrates over the weights with the data fixed \citep[Section~3.6]{vdvwellner1996} and \citep[Section~2.2.3]{kosorok2008}.
\end{definition}

\begin{theorem}[Efficient feature-law posterior]\label{thm:bvm}
Under Assumptions~\ref{ass:ident}--\ref{ass:margin} and \ref{ass:nuisance}, the cross-fitted one-step process and the feature-law posterior satisfy the following.
\begin{enumerate}[label=(\roman*),leftmargin=2em]
\item \textup{(Uniform asymptotic linearity.)}
\begin{equation}\label{eq:linearity}
\sup_{f\in\calF}\Bigl|\hPsi(f)-P_n\phi_f(\cdot;\eta_0)\Bigr|\;=\;\opro(n^{-1/2}),
\end{equation}
and consequently $\sqrt n\,(\hPsi-\Psi_0)\dto\G_0$ in $\ell^\infty(\calF)$, where $\G_0$ is the tight mean-zero Gaussian process with covariance $\Cov\{\phi_f(O;\eta_0),\phi_{f'}(O;\eta_0)\}$.
\item \textup{(Uniform conditional Bernstein--von Mises.)}
\begin{equation}\label{eq:bvm}
\sqrt n\,\bigl(\Psi^{(s)}-\hPsi\bigr)\ \dtow\ \G_0
\qquad\text{in } \ell^\infty(\calF).
\end{equation}
\end{enumerate}
Thus the posterior centered at $\hPsi$ consistently estimates the efficient sampling uncertainty of the entire corrected moment process, uniformly over the loss class.  The convergence in \eqref{eq:bvm} is conditional weak convergence in probability, as defined in Definition~\ref{def:condweak}.
\end{theorem}
The proof is in Supplementary Sections~\ref{app:onestep} and~\ref{app:bvm}. Part~(i), the uniform one-step expansion, follows because cross-fitting makes each evaluation fold independent of $\hateta^{(-b)}$, Theorem~\ref{thm:eif}(ii) controls uniform second-order bias, and the empirical-process increment vanishes under the entropy condition in Assumption~\ref{ass:class}, leaving $\G_0$. Part~(ii) decomposes $\sqrt n(\Psi^{(s)}-\hPsi)$ into an oracle exchangeable-multiplier process converging to $\G_0$ by the exchangeable-bootstrap central limit theorem and a weighted increment vanishing by Supplementary Lemma~\ref{lem:multmax} and Assumption~\ref{ass:nuisance}(ii). Because the Dirichlet weights satisfy $\sum_i(w_i^{(s)}-1)=0$, the posterior is centered at $\hPsi$. The argument applies verbatim to any fixed VC-type subclass. For $\calF_{\mathrm{eff}}$, it requires only the nuisance rate restricted to $\calF_{\mathrm{eff}}$, without a quantization-margin rate.
The posterior neither adds a first-order nuisance-uncertainty term, because the corrected score removes it, nor ignores first-order uncertainty, because the efficient influence process remains in the limit. Its uniformity over $\calF$ licenses simultaneous inference across paths, profiles, memberships, and effect contrasts. Efficiency is coordinatewise: every finite linear combination of process coordinates attains the semiparametric efficiency bound with the efficient influence function of Theorem~\ref{thm:eif}(i). Thus, the efficient Gaussian limit refers to this finite-dimensional efficiency, together with tightness of the limit process, without claiming process-level optimality (Supplementary Section~\ref{sec:app:bvm_details}).

\subsubsection{Posterior functional delta method}\label{sec:delta}

Theorem~\ref{thm:delta} below transfers the process-level result of Theorem~\ref{thm:bvm} to the smooth and argmin-type summaries, turning Hadamard differentiability of a target functional into matched sampling and posterior limits, so that the quantization values $W(K)$, the causal heterogeneity $R^2$ curve $\rho(K)$, the band it inverts, the mixture projection, and the subgroup effects all inherit valid inference by composition without a separate limit argument for each. The path corollary of Section~\ref{sec:path_theory} and the subgroup-effect limits of Section~\ref{sec:effectstheory} are instances of it, proved in Supplementary Section~\ref{sec:proof:thm:delta}. The threshold impossibility result of Section~\ref{sec:profiletheory} is of a different kind, a Le Cam two-point argument.

\begin{theorem}[Posterior delta method]\label{thm:delta}
Let $\mathbb D_0\subset\ell^\infty(\calF)$ contain $\Psi_0$ and let $T:\mathbb D_0\to\mathbb E$ be Hadamard differentiable at $\Psi_0$ tangentially to the set $C_\varsigma(\calF)$ of functions uniformly continuous with respect to the covariance semimetric $\varsigma$ of $\G_0$, with derivative $T'_{\Psi_0}$ admitting a continuous extension to $\ell^\infty(\calF)$. Under the conditions of Theorem~\ref{thm:bvm},
\begin{equation*}
\sqrt n\,\bigl\{T(\Psi^{(s)})-T(\hPsi)\bigr\}\ \dtow\ T'_{\Psi_0}(\G_0),
\qquad
\sqrt n\,\bigl\{T(\hPsi)-T(\Psi_0)\bigr\}\ \dto\ T'_{\Psi_0}(\G_0).
\end{equation*}
Consequently, if $T'_{\Psi_0}(\G_0)$ has a continuous distribution, posterior equal-tailed credible intervals for scalar $T$ built from the draws $\{T(\Psi^{(s)})\}$ are asymptotically valid frequentist confidence sets at the same level, and sup-$t$ bands for vector- and path-valued $T$ are likewise valid provided the distribution function of the associated maximum statistic is continuous and strictly increasing at the target quantile.
\end{theorem}

The proof is in Supplementary Section~\ref{sec:proof:thm:delta}. The quantization values require an envelope result of Danskin type for the infimum defining $W(K)$. The map $\iota_K(\nu)=\inf_{c\in\calC^K}\nu(g_c)$ on $\ell^\infty(\calF_{\mathrm{qt}})$ is concave and Hadamard directionally differentiable at $\Psi_0$ with derivative $\zeta\mapsto\inf_{c\in\calC^\star(K)}\zeta(g_c)$, and under Assumption~\ref{ass:quant}\textup{(i)} it is fully Hadamard differentiable with the linear derivative $\zeta\mapsto\zeta(g_{c^\star(K)})$. Supplementary Lemma~\ref{lem:envelope} gives the formal statement and proof, with classical antecedents in Shapiro \cite{shapiro1991} and D\"umbgen \cite{dumbgen1993}.
Without uniqueness the envelope map is only directionally differentiable and the posterior is generally inconsistent for the limit, a failure Supplementary Remark~\ref{rem:nonunique} records with its known remedies.

\subsection{Inference for the resolution path}\label{sec:path_theory}
The following corollary applies the posterior functional delta method of Theorem~\ref{thm:delta} to the quantization values.
\begin{corollary}[Joint inference for the heterogeneity path]\label{cor:path}
Under the conditions of Theorem~\ref{thm:bvm} and Assumption~\ref{ass:quant}, jointly over $K\in[\overline K]$,
\begin{equation*}
\sqrt n\bigl\{W^{(s)}(K)-\widehat W(K)\bigr\}_{K\le\overline K}\ \dtow\ \bigl\{\G_0(g_{c^\star(K)})\bigr\}_{K\le\overline K},
\end{equation*}
and the same limit holds for $\sqrt n\{\widehat W(K)-W(K)\}_{K\le\overline K}$ unconditionally.  Moreover,
\begin{equation*}
\sqrt n\bigl\{\rho^{(s)}(K)-\hat\rho(K)\bigr\}_{K\le\overline K}\ \dtow\ \Bigl\{-\tfrac1{W(1)}\Bigl[\G_0(g_{c^\star(K)})-\{1-\rho(K)\}\,\G_0(g_{c^\star(1)})\Bigr]\Bigr\}_{K\le\overline K},
\end{equation*}
a mean-zero Gaussian vector whose $K$th coordinate we denote $\mathbb H(K)$.  If in addition the variances $\Var\{\G_0(g_{c^\star(K)})\}$, $1\le K\le\overline K$, and $\sigma_K^2=\Var\{\mathbb H(K)\}$, $2\le K\le\overline K$, are positive, then the limit marginals are continuous, and posterior equal-tailed intervals for each $W(K)$ and $\rho(K)$, and posterior sup-$t$ bands for the whole path, are asymptotically valid.
\end{corollary}
The proof is in Supplementary Section~\ref{app:delta}. This corollary is the joint uncertainty statement for the quantization path. The posterior fluctuations of the residual dispersions $W(K)$ and the normalized causal heterogeneity $R^2$ curve $\rho(K)$ match their sampling fluctuations to first order, simultaneously over all reported $K$, so intervals and bands summarize uncertainty for the whole resolution path rather than for isolated choices of $K$.

The differentiability of the infimum map behind $W(K)$ follows from Supplementary Lemma~\ref{lem:envelope}. Under a unique optimal codebook, first-order perturbations of $W(K)$ depend only on the loss at $c^\star(K)$, giving the displayed limit. Without uniqueness the map is generally only directionally differentiable, the posterior is generally inconsistent for the limit, and the Gaussian delta method must be modified. Supplementary Remark~\ref{rem:nonunique} discusses what fails without uniqueness and possible remedies.

The theory is stated for exact minimizers, while the implementation of Section~\ref{sec:construction} returns approximate ones. Approximate minimizers with an $\opro(n^{-1/2})$ attained-value gap, together with a matching stationarity-gap condition for the mixture projection, leave every downstream conclusion unchanged, since each reported quantity depends on the computed codebooks only through the attained loss values, stated by Supplementary Proposition~\ref{prop:approxmin} precisely.

When literal finite response classes exist, the resolution profile is backward compatible with the classical target. If $P_U$ has $K_0$ separated atoms, then $K^\star(\gamma)=K_0$ for all sufficiently high resolutions below one, and the feature-law posterior recovers this behavior 
under nuisance consistency alone (Supplementary Section~\ref{sec:exact_classes}).

\subsubsection{Set-valued inference for the resolution profile}\label{sec:profiletheory}

The profile $K^\star(\cdot)$ is an integer-valued threshold functional of $\rho(\cdot)$.  Away from the knots $\mathcal R_0=\{\rho(K):K\in[\overline K]\}$ it is locally constant. At a knot, two adjacent answers are locally indistinguishable.  The next theorem states the impossibility result that motivates set-valued reporting.

\begin{theorem}[No locally uniformly consistent selection at a knot]\label{thm:impossibility}
Let Assumptions~\ref{ass:ident}--\ref{ass:quant} hold, and assume additionally that the compact feature set has a buffer around the true feature support, that is, $\{u:\dist(u,\calU)\le\eps_{\calC}\}\subset\calC$ for some $\eps_{\calC}>0$.  Let $\gamma=\rho_0(K_0)\in\mathcal R_0$ for some $2\le K_0<\overline K$ with $\rho_0(K_0-1)<\gamma<\rho_0(K_0+1)$ and $\sigma_{K_0}^2=\Var\{\mathbb H(K_0)\}>0$, and fix $h>0$ and $\eps\in(0,\sigma_{K_0})$.  Then there exist laws $P_n^{+}$ and $P_n^{-}$, each differing from $P_0$ by a density factor $1+O(n^{-1/2})$ in supremum norm and satisfying the identification, positivity, bounded-outcome, and buffered feature-support conditions uniformly, such that the following hold.
\begin{enumerate}[label=(\roman*),leftmargin=2em]
\item $\rho_{P_n^{\pm}}(K_0)=\gamma\pm h\,n^{-1/2}+o(n^{-1/2})$, so that $K^\star_{P_n^{+}}(\gamma)=K_0$ and $K^\star_{P_n^{-}}(\gamma)=K_0+1$ for all large $n$.
\item $(P_n^{+})^{\otimes n}$ and $(P_n^{-})^{\otimes n}$ are mutually contiguous.
\item Every possibly randomized selector $\widetilde K_n$ satisfies
\begin{equation*}
\liminf_{n\to\infty}\ \Bigl[\Pp^{+}_n\bigl(\widetilde K_n\ne K_0\bigr)+\Pp^{-}_n\bigl(\widetilde K_n\ne K_0+1\bigr)\Bigr]
\ \ge\ 2\,\Phi\!\Bigl(-\frac{h}{\sigma_{K_0}-\eps}\Bigr).
\end{equation*}
\end{enumerate}
\end{theorem}
The full proof is in Supplementary Section~\ref{app:impossibility}. Take a bounded, mean-zero submodel $dP_t=(1+t\tilde s)\,dP_0$ whose score is a truncated, renormalized copy of the influence function of $\rho(K_0)$, with drift coefficient $b=\E_0[\mathrm{IF}\,\tilde s]\in[\sigma_{K_0}-\eps,\sigma_{K_0}]$. An envelope (Danskin) expansion of the minimized values $W_{P_t}(K)$, with the quantization remainder controlled by the margin condition (Assumption~\ref{ass:margin}), yields the profile drift $\rho_{P_t}(K_0)=\gamma+t\,b+o(t)$, so $t=\pm h\,b^{-1}n^{-1/2}$ produces laws $P_n^{\pm}$ that straddle the knot. Differentiability in quadratic mean makes these mutually contiguous by Le Cam's first lemma, and a Neyman--Pearson two-point bound on the summed selection error contributes $2\Phi(-h/b)\ge2\Phi(-h/(\sigma_{K_0}-\eps))$, so data at the $n^{-1/2}$ scale cannot resolve which side of the knot the truth lies on. As $h\downarrow0$ the lower bound tends to one, so in shrinking neighborhoods of a knot any single-valued rule must fail on at least one of two statistically indistinguishable sequences.

\begin{remark}[What the theorem does not exclude]\label{rem:marginselector}
The theorem constrains uniform behavior, not pointwise consistency. A margin selector $\widetilde K_n(\gamma)=\min\{K:\hat\rho(K)\ge\gamma-a_n\}$ with $a_n\downarrow0$ and $\sqrt n\,a_n\to\infty$ is consistent at any fixed law, including one lying exactly at a knot. Under the conditions of Corollary~\ref{cor:path}, $\hat\rho(K)=\rho_0(K)+\Op(n^{-1/2})$ for each $K$, so the deterministic slack $a_n$ eventually dominates the $\Op(n^{-1/2})$ estimation error while vanishing in the limit. What no selector can achieve is uniform consistency over the root-$n$ neighborhoods of Theorem~\ref{thm:impossibility}. The mechanism is that the profile $\gamma\mapsto K^\star(\gamma)$ is a discontinuous functional of the path at a knot, so the theorem instantiates the Hirano and Porter phenomenon for nondifferentiable functionals \citep{hirano2012}, here for the quantization paths of causal feature laws accessed through corrected scores.
\end{remark}

The recommended report is therefore the following band inversion.  Let $\hat\sigma(K)$ be the posterior interquartile range of $\rho^{(s)}(K)$ divided by $2\Phi^{-1}(0.75)$, let $q_{1-\alpha}$ be the posterior $(1-\alpha)$ quantile of $\max_{2\le K\le\overline K}|\rho^{(s)}(K)-\hat\rho(K)|/\hat\sigma(K)$, and set
\begin{equation}\label{eq:band}
L_\rho(K)=\hat\rho(K)-q_{1-\alpha}\hat\sigma(K),\qquad
U_\rho(K)=\hat\rho(K)+q_{1-\alpha}\hat\sigma(K),\qquad 2\le K\le\overline K,
\end{equation}
with $L_\rho(1)=U_\rho(1)=0$.  The bounds may be monotonized without harming coverage.  Define
\begin{equation}\label{eq:profileset}
\widehat C(\gamma)\;=\;\bigl\{K\in[\overline K]:\ U_\rho(K)\ge\gamma\ \text{ and }\ L_\rho(K')<\gamma\ \forall K'<K\bigr\},
\qquad \gamma\in(0,\rho_0(\overline K)).
\end{equation}
The defining formula makes sense for any $\gamma\in(0,1)$, which matters when the target law drifts in Theorem~\ref{thm:honest} below. If no $K\le\overline K$ satisfies $U_\rho(K)\ge\gamma$, the report is the empty set, indicating that resolution $\gamma$ is not supported as attainable with $\overline K$ groups. On the coverage event of Theorem~\ref{thm:profile}(iii) below, this does not occur for $\gamma<\rho_0(\overline K)$.

\begin{theorem}[Resolution-profile inference]\label{thm:profile}
Let Assumptions~\ref{ass:ident}--\ref{ass:quant} and \ref{ass:nuisance} hold, and suppose $\sigma_K^2=\Var\{\mathbb H(K)\}>0$ for $2\le K\le\overline K$.
\begin{enumerate}[label=(\roman*),leftmargin=2em]
\item \textup{(Off-knot concentration.)} If $\gamma\in(0,\rho_0(\overline K))$ is not a knot, then $\Pp_w\{K^{\star(s)}(\gamma)=K^\star_0(\gamma)\}\to1$ in probability and $\Pp\{\widehat K^\star(\gamma)=K^\star_0(\gamma)\}\to1$.
\item \textup{(Knot behavior.)} If $\gamma=\rho_0(K_0)\in\mathcal R_0$ for some $2\le K_0<\overline K$ with $\rho_0(K_0-1)<\gamma<\rho_0(K_0+1)$, then $\Pp_w\{K^{\star(s)}(\gamma)\in\{K_0,K_0+1\}\}\to1$ in probability, and $\Pp_w\bigl\{K^{\star(s)}(\gamma)=K_0\bigr\}\ \dto\ \mathrm{Uniform}(0,1)$
over the sampling law.
\item \textup{(Simultaneous set-valued coverage.)} With $\widehat C$ as in \eqref{eq:profileset}, $\liminf_{n\to\infty}\ \Pp\Bigl(K^\star_0(\gamma)\in\widehat C(\gamma)\ \text{ for every }\gamma\in\bigl(0,\rho_0(\overline K)\bigr)\Bigr)\ \ge\ 1-\alpha .$
\end{enumerate}
\end{theorem}
The proof is in Supplementary Section~\ref{app:profile}. The posterior profile thus collapses to the population count away from knots and splits between admissible neighboring counts at a knot. The set $\widehat C(\gamma)$ is the main inferential object, a singleton when the path is separated enough and widening precisely where Theorem~\ref{thm:impossibility} says no single answer is uniformly reliable.

Two cautions calibrate the pairing of the two theorems. The objection that knots form a measure-zero set of resolutions misses the force of the impossibility, since contiguity is precisely the statement that the data cannot determine whether the truth is at a knot or $n^{-1/2}$-close to one, so an analyst scanning the profile can never confirm being in the easy regime. A second concern is that Theorem~\ref{thm:profile}(iii) is proved pointwise in $P_0$, whereas Theorem~\ref{thm:impossibility} operates over $n^{-1/2}$-neighborhoods, and a guarantee proved only at fixed $P_0$ would not answer the paper's own objection. The next theorem closes this gap, showing that the set-valued report is an honest confidence correspondence over the same local perturbation classes on which the impossibility operates, in the locally uniform sense of the post-selection literature \citep{leeb2005,li1989}.

\begin{theorem}[Locally uniform validity of the set-valued report]\label{thm:honest}
Let the assumptions of Theorem~\ref{thm:profile} hold, together with the buffered feature-support condition of Theorem~\ref{thm:impossibility}. Fix $\bar S<\infty$ and $B_u<\infty$, and let $\calS$ be a class of measurable functions of $O$ with $\E_0s=0$ and $\norm{s}_\infty\le\bar S$ for every $s\in\calS$, totally bounded in $L_2(P_0)$. For $|t|\le1/(2\bar S)$ and $s\in\calS$ let $dP_{t,s}=(1+ts)\,dP_0$, and for $n\ge4B_u^2\bar S^2$ write $P_{n,u,s}=P_{u/\sqrt n,\,s}$, so that every perturbation in the display is defined. Then
\begin{equation*}
\liminf_{n\to\infty}\ \inf_{s\in\calS,\,|u|\le B_u}\ \Pp_{n,u,s}\Bigl(K^\star_{P_{n,u,s}}(\gamma)\in\widehat C(\gamma)\ \text{ for every }\gamma\in\bigl(0,\rho_{P_{n,u,s}}(\overline K)\bigr)\Bigr)\ \ge\ 1-\alpha,
\end{equation*}
where $\Pp_{n,u,s}$ denotes probability under i.i.d.\ sampling from $P_{n,u,s}$.
\end{theorem}
The proof is in Supplementary Section~\ref{app:honest}. Its inputs are the uniform drift expansion from the impossibility proof, the local regularity of the corrected path estimator, under which the drift of the moving target and the shift of the estimator cancel exactly by Theorem~\ref{thm:bvm}(i) and Le Cam's third lemma, and a contiguity transfer showing the posterior band quantiles are unaffected by the perturbations. The perturbations of Theorem~\ref{thm:impossibility} are of exactly this form, so the two results hold over the same neighborhoods, and Supplementary Remark~\ref{rem:honestscope} notes that bounded tilts lose no force at the root-$n$ scale because the least favorable construction is itself one. Over every such neighborhood no single-valued selector is uniformly reliable at a knot while the set-valued report keeps its nominal guarantee, so set-valued reporting is the attainable summary rather than a conservative retreat, and the pair of theorems delimits what can be learned about the profile at the root-$n$ scale. Supplementary Remark~\ref{rem:cardinality} adds that along these sequences the report eventually contains at most the two knot-adjacent counts, so honesty is not purchased with an uninformative set.

\subsection{Subgroup effects with partition uncertainty}\label{sec:effectstheory}

Fix a working resolution $K$ and assume Assumption~\ref{ass:proj}. The subgroup mean $\psi_{h,a}(K)$ is a composite functional. The feature law determines the projection $\beta^\star(K)$, which defines the soft memberships, and the subgroup mean is then a ratio of two structured moments, so uncertainty in $\psi_{h,a}(K)$ comes both from estimating the within-subgroup mean for a fixed partition and from estimating the subgroup definition itself.

\begin{theorem}[Subgroup effects]\label{thm:effects}
Under Assumptions~\ref{ass:ident}--\ref{ass:class} and \ref{ass:proj}, with the nuisance-rate condition restricted to $\calF_{\mathrm{eff}}$, and after label alignment on the local chart around the selected representative of $\beta^\star(K)$, the map defining $\psi_{h,a}(K)$ is Hadamard differentiable at $\Psi_0$.  Jointly over the finite collection of reported pairs $(h,a)$,
\[
\sqrt n\{\hat\psi_{h,a}(K)-\psi_{h,a}(K)\}_{h,a}
\ \dto\
N(0,\Sigma_\psi),
\qquad
\sqrt n\{\psi^{(s)}_{h,a}(K)-\hat\psi_{h,a}(K)\}_{h,a}
\ \dtow\
N(0,\Sigma_\psi),
\]
where $\Sigma_\psi$ is the covariance matrix of the composite influence function given in Supplementary Section~\ref{sec:app:expanded_effect_scores}.  Treatment contrasts within a subgroup and subgroup contrasts within an arm follow by linearity.
\end{theorem}
The proof is in Supplementary Section~\ref{app:effects}. Recomputing $\beta^{(s)}(K)$ inside every draw targets the full composite limit law, whereas holding the partition fixed at $\hat\beta(K)$ targets only the fixed-partition component $\Phi^{\mathrm{fix}}_{h,a}$ of the composite influence function and is miscalibrated for the composite law whenever the projection term in Supplementary Section~\ref{sec:app:expanded_effect_scores} is nonzero.

The projection $\beta^\star(K)$ that defines the memberships has its own Bernstein--von Mises limit, supplying the partition-uncertainty component and making the subgroup-effect story self-contained. Supplementary Theorem~\ref{thm:projection} states the matched sampling and posterior limits for $\hat\beta(K)$ and $\beta^{(s)}(K)$ on the aligned label chart, with the sandwich covariance $V_\beta^{-1}\Sigma_\beta V_\beta^{-1}$, a $Z$-functional application of Theorem~\ref{thm:delta} in the tradition of weighted-bootstrap $M$-estimation \citep{ma2005,cheng2010}. The interpretive caveat bears repeating. The target $\beta^\star(K)$ is a KL projection, so the guarantee quantifies uncertainty about the best $K$-component description under misspecification, and the sandwich covariance uses no likelihood identity.

The subgroup effects above are stated at a fixed working resolution $K$, whereas in practice the resolution is often selected from the profile. The next corollary shows that selecting a prespecified nonknot resolution does not disturb the inference.

\begin{corollary}[Subgroup effects at a prespecified resolution]\label{cor:fixedgamma}
Fix $\gamma_0\in(0,\rho_0(\overline K))$ prespecified and not a knot, and let $K^\star_0=K^\star_0(\gamma_0)$. Let the hypotheses of Theorems~\ref{thm:profile} and~\ref{thm:effects} hold, with Assumption~\ref{ass:proj} holding at $K=K^\star_0$. Then $\widehat K^\star(\gamma_0)\pto K^\star_0$ by Theorem~\ref{thm:profile}\textup{(i)}. On the data event $\{\widehat K^\star(\gamma_0)=K^\star_0\}$, whose probability tends to one, the point subgroup-effect report at the selected resolution coincides with the fixed-resolution report at $K=K^\star_0$; each posterior draw coincides with its fixed-resolution counterpart on the per-draw event $\{K^{\star(s)}(\gamma_0)=K^\star_0\}$, whose conditional probability tends to one in probability. Consequently the sampling and conditional Gaussian limits of Theorem~\ref{thm:effects} at $K=K^\star_0$ apply verbatim to inference at the selected resolution.
\end{corollary}
The proof is in Supplementary Section~\ref{app:effects}. Because the resolution is selected from the profile, the full nuisance-rate condition of Assumption~\ref{ass:nuisance} is required here, not only its $\calF_{\mathrm{eff}}$ restriction. The argument is pointwise in $P_0$ and does not extend uniformly over root-$n$ neighborhoods of a knot, consistent with Theorem~\ref{thm:impossibility}, so when $\widehat C(\gamma_0)$ is not a singleton we recommend reporting subgroup effects at every supported $K$.

\section{Simulation studies}\label{sec:sims}

We report four simulation studies in the main text, one for each main claim family, together with a noise-floor study in Supplementary Section~\ref{sec:sims:study6_supp}. The four main-text studies examine resolution-profile inference near knots, the calibration and rate robustness of corrected path uncertainty, subgroup effects with propagated partition uncertainty, and the locally uniform validity of the set-valued report over the exact bounded density tilts of Theorem~\ref{thm:honest}, while the noise-floor study probes the corrected path in the regime of Remark~\ref{rem:noisefloor}, where the feature-estimation error rivals the heterogeneity. Supplementary Section~\ref{sec:app:additional_simulations} also reports optional penalized-profile diagnostics, atomic recovery with order-selection comparators, and an energy-scale sensitivity check for exponentiated order posteriors.

\subsection{Data-generating processes and implementation}\label{sec:sims:dgp}

The first design, DGP-A, is a one-dimensional causal-effect design in which the feature law $P_U$ is continuous, non-Gaussian, and deliberately non-mixture. It consists of three separated, skewed beta-shaped bumps, giving clear coarse structure although the population is not literally a finite Gaussian mixture. The feature is the CATE, $U(X)=\tau_0(X)$, and the outcome and propensity models are generated from smooth functions of $X\sim\mathrm{Unif}[0,1]^5$.  This design lets us study both easy regions of the path and difficult high-resolution regions close to knots.
The second design, DGP-B, keeps the bump shapes and weights fixed but varies the support gap between the middle bump $[0,0.8]$ and the right bump $[0.8+s,1.6+s]$.  The separation parameter $s$ controls how distinguishable the two right structures are, and as $s\downarrow0$ they meet at a common boundary. This controlled loss of separation moves $\rho_0(2)$ through the threshold $\gamma=0.9$, so $K^\star_0(0.9)$ changes from three to two groups through a knot, exactly the setting where a point choice of $K$ is intrinsically unstable and the set-valued report of Theorem~\ref{thm:profile} should widen.

Ground truth for the full quantization path is computed by dynamic programming on the population density, with closed-form checks where available. We use $S=1000$ posterior draws, $B=10$ cross-fitting folds, and ceiling $\overline K=8$ unless stated otherwise, with $R=500$ Monte Carlo replications for every study and their supplementary diagnostics, except the noise-floor study (Supplementary Section~\ref{sec:sims:study6_supp}), which uses $R=300$. Monte Carlo standard errors for all reported coverages are at most $0.029$.  The nuisance functions, arm-specific outcome regressions and propensity score, use a cross-fitted Super Learner ensemble \citep{vanderlaan2007super}.  Study~2 additionally varies the nuisance regime to illustrate the bias structure of Theorem~\ref{thm:eif}, especially the role of outcome-regression accuracy for feature-law and quantization-path inference. Full data-generating details, exact population paths, learner tuning, and implementation audits are given in Supplementary Section~\ref{sec:sims:dgp_supp}.

\subsection{Study 1. Resolution-profile inference}\label{sec:sims:profile}

Study~1 tests the set-valued report of Theorem~\ref{thm:profile}.  For DGP-A and the DGP-B separation path, we evaluate $\widehat C(\gamma)$ at $\gamma\in\{0.5,0.8,0.9,0.95,0.975,0.99\}$, covering easy resolutions, high-resolution near-knot regions, and the moving knot at $\gamma=0.9$.

\begin{figure}[t]
\centering
\includegraphics[width=0.82\textwidth]{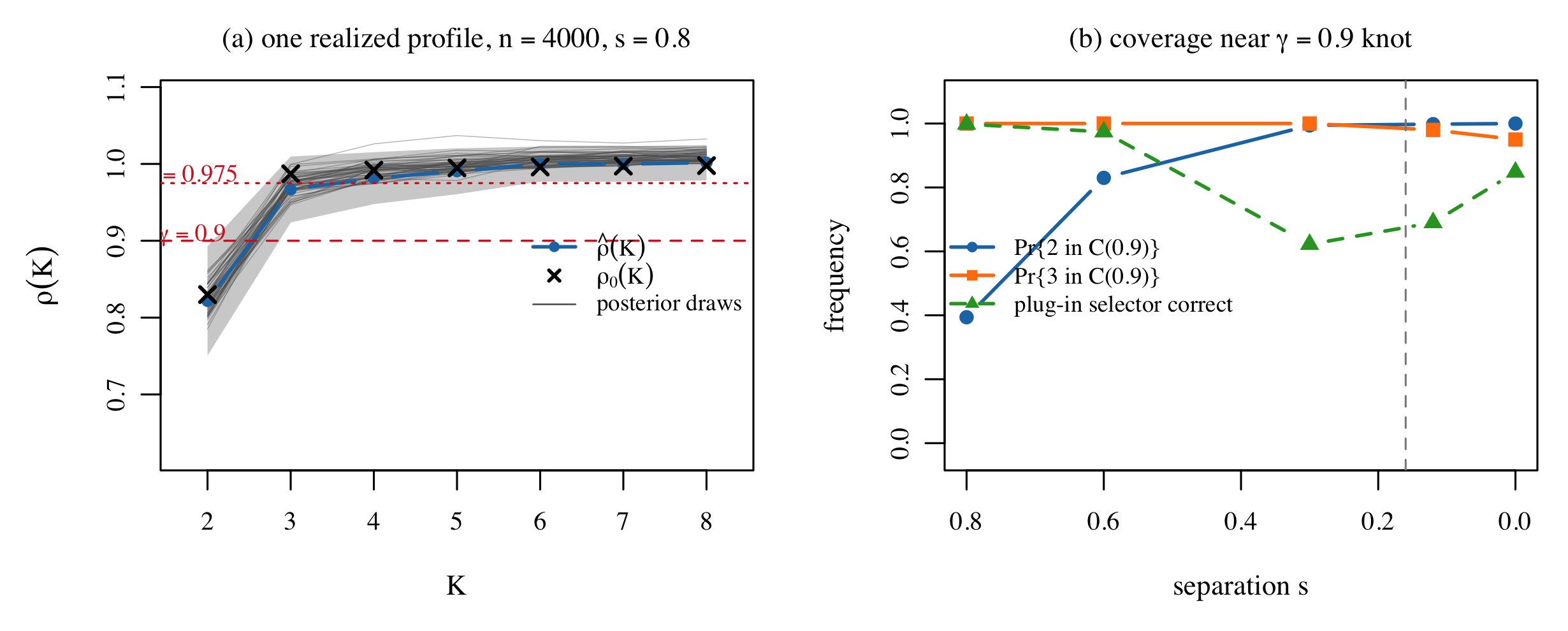}
\caption{Study~1.  (a) One replication at $n=4000$ in DGP-A, which is the DGP-B family at separation $s=0.8$. Point path $\hat\rho(K)$, simultaneous $95\%$ band, posterior draws, truth, and reading lines, with corrected posterior draws of $\rho$ occasionally exceeding one in finite samples.  (b) Along DGP-B, the population answer at $\gamma=0.9$ changes from $3$ to $2$ through a knot, the knot being approached as the separation $s$ decreases to the right.  The dashed curve is the frequency with which the plug-in selector $\widehat K^\star(0.9)$ equals $K^\star_0(0.9)$, and the two solid curves are the inclusion probabilities $\Pr\{2\in\widehat C(0.9)\}$ and $\Pr\{3\in\widehat C(0.9)\}$.}
\label{fig:sims:profile}
\end{figure}

The threshold report behaves as predicted. In the displayed DGP-A replication of Figure~\ref{fig:sims:profile}(a), the band at $\gamma=0.9$ returns the singleton $\{3\}$, in line with the average cardinality of $1.4$ at this threshold and separation, so away from the knot the report pins down the correct count and widens only as the knot approaches along the DGP-B path in panel (b). Supplementary Table~\ref{tab:sims:coverage_K_C} shows coverage of $K^\star_0(\gamma)$ essentially one through $\gamma=0.975$ and within Monte Carlo error of nominal at $\gamma=0.99$, with $0.96$ at the hardest cell, and the all-$\gamma$ event behaves similarly. The sets widen only where the path is locally ambiguous, mean cardinality one at $\gamma=0.5$, about two near the $\gamma=0.9$ knot, and five to six at $\gamma=0.99$. Coverage far above nominal at stable thresholds is expected, since away from knots the set collapses to the correct singleton by Theorem~\ref{thm:profile}(i), so the $1-\alpha$ guarantee binds only near knots.

Figure~\ref{fig:sims:profile}(b) isolates the moving-knot behavior along DGP-B, where the dashed curve is the correct-selection frequency of the single-valued rule, not the coverage of a set. Its drop near the crossing is expected because $\rho_0(2)$ crosses $\gamma=0.9$ between $s=0.30$ and $s=0.12$, so small fluctuations decide whether the point rule reports $2$ or $3$, as formalized by Theorem~\ref{thm:impossibility}. The inclusion probabilities are asymmetric as intended, $\Pr\{3\in\widehat C(0.9)\}$ staying high because $3$ is the true or adjacent order throughout the crossing while $\Pr\{2\in\widehat C(0.9)\}$ rises as the two right bumps become harder to distinguish and $2$ becomes nearly admissible, so the report keeps both neighboring orders near the knot and maintains coverage.

\subsection{Study 2. Calibration and rate robustness}\label{sec:sims:calibration}

Study~2 evaluates the corrected posterior in DGP-A for $n\in\{500,1000,2000,4000\}$ across oracle nuisances, both-flexible nuisances, flexible outcome with parametric or misspecified propensity, and misspecified outcome with flexible propensity.  Supplementary Figure~\ref{fig:sims:calibration} reports coverage for the protected coordinate $\rho(3)$ and the simultaneous band over $2\le K\le8$.

The coordinate $\rho(3)$ calibrates as the theory predicts. In every regime with a flexible outcome regression, its coverage in panel (a) reaches $0.90$ to $0.94$ by $n=1000$ and stays there, because when the outcome regression converges to the truth the centering is exactly unbiased for any bounded propensity limit, including a misspecified one. The single flexible-outcome regime that fails pairs a misspecified outcome regression with a flexible propensity and does not recover as the sample grows, the asymmetry the theory anticipates, since the estimand is a functional of the outcome regression and no propensity can repair a wrong regression limit. The apparent robustness to propensity misspecification is therefore one-sided, finite-sample support in a favorable regime rather than generic insensitivity. Under a persistently misspecified propensity, root-$n$ inference formally requires the regression error negligible at nearly parametric rates, namely $r_\mu=\opro(n^{-1/2})$ for the smooth classes, with the quantization boundary term additionally requiring $r_\mu=\opro(n^{-3/8})$ under $\am=1$. A separate rate-compliant learner restores $W(1)$ to nominal coverage while its honestly shorter intervals expose a distinct finite-sample empirical-minimum bias at fine resolution (Supplementary Table~\ref{tab:sims:study2cov}). The simultaneous band in panel (b) is well calibrated for flexible outcome models but its coverage drops sharply at large $n$ when the outcome model is not flexible enough.

\subsection{Study 3. Subgroup effects with uncertainty quantification}\label{sec:sims:subgroup_effects}

Study~3 supports the subgroup-effect reporting. In DGP-A at working order $K=3$ we estimate the subgroup treatment-effect contrast $\psi_{h,1}(3)-\psi_{h,0}(3)$ for each component $h$. The feature-law posterior recomputes the projection inside every draw, the naive two-stage comparator clusters $\widehat U$ and treats the resulting partition as fixed, and an oracle-cell AIPW procedure using the true population partition is a reference. 

Supplementary Figure \ref{fig:sims:subgroup_effects} reports the coverage of each procedure. The oracle-cell reference is close to nominal, so the failures of the naive two-stage method are not failures of AIPW effect estimation but arise from treating the learned partition as fixed. The feature-law posterior, which recomputes the subgroup definition inside each draw and so propagates the first-order projection uncertainty, is substantially better calibrated across all estimands, whereas the two-stage intervals undercover sharply. Full diagnostics of bias, RMSE, coverage, and interval length are in Supplementary Table~\ref{tab:sims:own_estimands}.

\subsection{Study 4. Locally uniform validity near knots}\label{sec:sims:honest}

Study~4 verifies Theorem~\ref{thm:honest} in its exact regime, the bounded density tilts over which Theorem~\ref{thm:impossibility} forces every single-valued rule to fail. The base law is DGP-B at the knot separation $s^\star$ solving $\rho_0(2;s^\star)=0.9$, paired with a smooth compactly supported outcome error so the tilted laws stay differentiable in quadratic mean, and the tilt score $s$ is the truncated, recentred, normalized influence function of $\rho(2)$, with achieved drift $b=0.920$ essentially equal to the efficiency bound $\sigma_{K_0}=0.924$ (Supplementary Section~\ref{sec:sims:study4b_supp}). For $n\in\{1000,2000,4000,8000\}$ and tilt magnitudes $u\in\{0,\pm0.5,\pm1,\pm2\}$ we draw exactly from $(1+u\,n^{-1/2}s)\,\mathrm dP_0$ by rejection sampling, under both flexible Super Learner and rate-clean oracle nuisances, with $R=500$ replications per cell, and evaluate $\widehat C(0.9)$ against the drifting truth. Coverage by $\widehat C(0.9)$ is at least $0.988$ in every one of the $56$ cells at nominal $0.95$, on both sides of the knot and in both nuisance arms, while the single-valued point selector degrades toward a coin toss along the tilts, correct in about $0.62$ of replications at the exact knot at $n=8000$ and about $0.50$ in the knot-adjacent tilt cells where the tilts render the two knot-adjacent counts indistinguishable, with mean report cardinality near two at the larger sample sizes (Supplementary Table~\ref{tab:sims:study4b}). An independently constructed structural perturbation family driving a separation parameter through the same knot gives the same behavior (Supplementary Section~\ref{sec:sims:honest_supp}).

A final design in Supplementary Section~\ref{sec:sims:study6_supp} matches the application regime directly, with reference-arm outcome $R^2$ near $0.05$ and total heterogeneity placed on a signal ladder around the estimated feature-error floor at samples up to $n=64000$, the regime of Remark~\ref{rem:noisefloor}. Below the floor there are no false detections at any $n$, the $W(1)$ band covers the truth honestly, and the $\rho$ scale is correctly never licensed. Detection then emerges with $n$ exactly as the floor logic predicts. The split-difference diagnostic $\widehat\Delta$ tracks the true half-sample feature error to within $0.82$ to $0.99$ across all cells, reading the floor without access to the truth. At large $n$ with strong signal the $W$-level coverage degrades exactly as the multiplicative identity $\rho_\delta(K)=\rho(K)\,W(1)/\{W(1)-\delta\}$ predicts, while evaluation resolutions placed away from the shifted knots keep their set-valued reports. The full design, implementation, diagnostics, and coverage table are reported there with Supplementary Table~\ref{tab:sims:study6}.

\section{Empirical application}\label{sec:application}

A recurring operational question for e-mail campaigns is whether they move a broad population uniformly or act on a few identifiable response segments. We analyze the MineThatData e-mail experiment \citep{hillstrom2008}, in which $N=64{,}000$ past customers were individually randomized in equal thirds to no e-mail (control, $n=21{,}306$), a men's merchandise e-mail ($n=21{,}307$), or a women's merchandise e-mail ($n=21{,}387$), with known and essentially flat arm shares near one third, so the feature correction uses the design propensities directly. The outcome is a site visit within two weeks of the campaign, with visit rates of $10.6\%$ under control, $18.3\%$ under the men's e-mail, and $15.1\%$ under the women's e-mail. We ask whether the visit response resolves into a few groups with distinct campaign profiles, and at what resolution the data support such a summary.

The causal feature is the $q=2$ vector of campaign benefits on the visit-rate scale, $U(x)=(\mu_{\mathrm{men}}(x)-\mu_{\mathrm{none}}(x),\,\mu_{\mathrm{women}}(x)-\mu_{\mathrm{none}}(x))$. Both coordinates are covariate-conditional visit-probability contrasts against control on a common scale, so the identity metric is the natural choice. The feature map itself is part of the estimand, and Supplementary Section~\ref{sec:app:hillstrom_anchored} reports a control-anchored variant that appends the baseline visit propensity to the two benefits, which yields a larger and clearly positive total heterogeneity and a finer resolution profile at the cost of mixing prognostic with causal variation. Arm-specific outcome regressions are estimated by a cross-fitted Super Learner ensemble \citep{vanderlaan2007super} over four base learners, a linear model, a ridge model, a random forest, and a gradient-boosted tree, formed within $5$ folds and averaged over $4$ repeats, with the correction applied at the known design propensities. All corrected summaries use ceiling $\overline K=6$ and $S=500$ posterior draws. The full specification and preprocessing are in Supplementary Section~\ref{sec:app:hillstrom_spec}.

Following the reporting protocol in Supplementary Section~\ref{sec:report}, we read the $W$-scale evidence first and apply the two-tier gate of Remark~\ref{rem:noisefloor}. The corrected total causal-feature heterogeneity is positive, $\widehat W(1)=9.69\times10^{-4}$ squared visit-rate points with $95\%$ simultaneous band $(1.57,\,17.80)\times10^{-4}$, and the parametric route agrees at $9.27\times10^{-4}$ with band $(2.01,\,16.53)\times10^{-4}$. The first tier, detection, is met since the band clears zero, and because the level shift biases downward the detection is one-sided conservative even though the lower band edge is small. The split-difference diagnostic returns $\widehat\Delta=5.85\times10^{-4}$ with bracket $[2.93,\,5.85]\times10^{-4}$, so the reliability ratio of the second tier is $\hat r=[\widehat\Delta/2,\widehat\Delta]/\widehat W(1)\approx[0.30,\,0.60]$, which is not small. Detection is therefore licensed within the maintained positive-heterogeneity regime of Assumption~\ref{ass:nuisance}, while the $\rho$ scale carries the shift-sensitivity qualification of the second tier, developed below.

\begin{figure}[!t]
\centering
\includegraphics[width=0.88\textwidth]{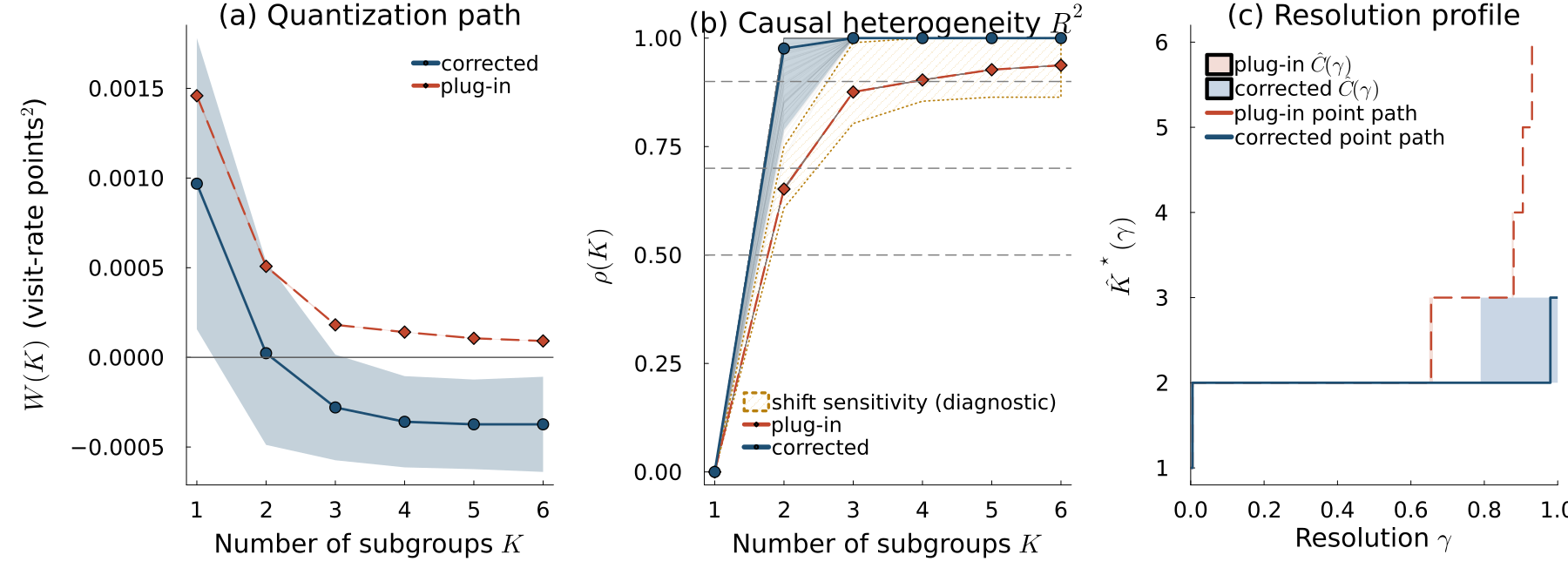}
\caption{MineThatData e-mail experiment corrected feature-law resolution summaries for the visit outcome. (a) Quantization path. Corrected causal-feature dispersion $\widehat W(K)$ with its $95\%$ simultaneous band, alongside the uncorrected plug-in path. (b) Causal heterogeneity $R^2$. Corrected $\hat\rho(K)$ with its $95\%$ band and the plug-in curve, with the gray fan showing the corrected and plug-in posterior draws at low opacity. The hatched overlay is the shift-sensitivity range $\rho_{\mathrm{adj}}(K;\delta)=1-\max\{\widehat W(K)+\delta,0\}/(\widehat W(1)+\delta)$ for $\delta\in[\widehat\Delta/2,\widehat\Delta]$, a diagnostic with no coverage claim attached rather than a confidence band. (c) Resolution profile. Point selector $\widehat K^\star(\gamma)$ with the set-valued report $\widehat C(\gamma)$ shaded, for the corrected and plug-in paths.}
\label{fig:hillstrom:resolution_summary}
\end{figure}

Figure~\ref{fig:hillstrom:resolution_summary} summarizes the corrected path. The raw corrected causal heterogeneity $R^2$ reads $\hat\rho(2)=0.976$ with band $(0.788,\,1.000)$ and $\hat\rho(3)=1.000$, which taken literally would place nearly all resolvable structure in two groups. Because the reliability ratio is not small, the second tier requires the shift-sensitivity reading, shown in panel (b) as a diagnostic overlay at the two floor levels of the bracket, the multiplicative inflation of Supplementary Corollary~\ref{cor:levelshift_rho}. On this range two groups explain about $0.61$ to $0.75$ of the heterogeneity, from $\widehat W(2)=2.3\times10^{-5}$, and three groups about $0.80$ to $0.99$, from $\widehat W(3)=-2.8\times10^{-4}$. No coverage claim attaches to it (Supplementary Section~\ref{sec:noisefloor_supp}). The set-valued report in panel (c) is the primary inferential summary and is valid under the maintained Assumption~\ref{ass:nuisance}. The corrected report reads $\widehat C(\gamma)=\{2\}$ at $\gamma=0.25$, $0.50$, and $0.70$, and $\{2,3\}$ at $\gamma=0.80$ and $0.90$, never touching the ceiling $\overline K=6$. Because $\rho_{\mathrm{adj}}(2)$ straddles $\gamma=0.70$, the two-group reading is not shift-robust there, so we do not advertise the singleton $\widehat C(0.70)=\{2\}$ as a standalone nominal statement. The defensible readings are the coarse-$\gamma$ confidence sets, the $\{2,3\}$ sets at fine $\gamma$, and their union, which together summarize the visit response at two to three groups. The uncorrected plug-in path overstates the base level, $\widehat W_{\mathrm{plug}}(1)=14.59\times10^{-4}$, and does not decline to the noise floor, so its point selector crosses at finer counts, $\widehat K^\star_{\mathrm{plug}}(0.70)=3$ and $\widehat K^\star_{\mathrm{plug}}(0.90)=4$, where the corrected report holds at $2$ and $\{2,3\}$.

\begin{figure}[!t]
\centering
\includegraphics[width=0.72\textwidth]{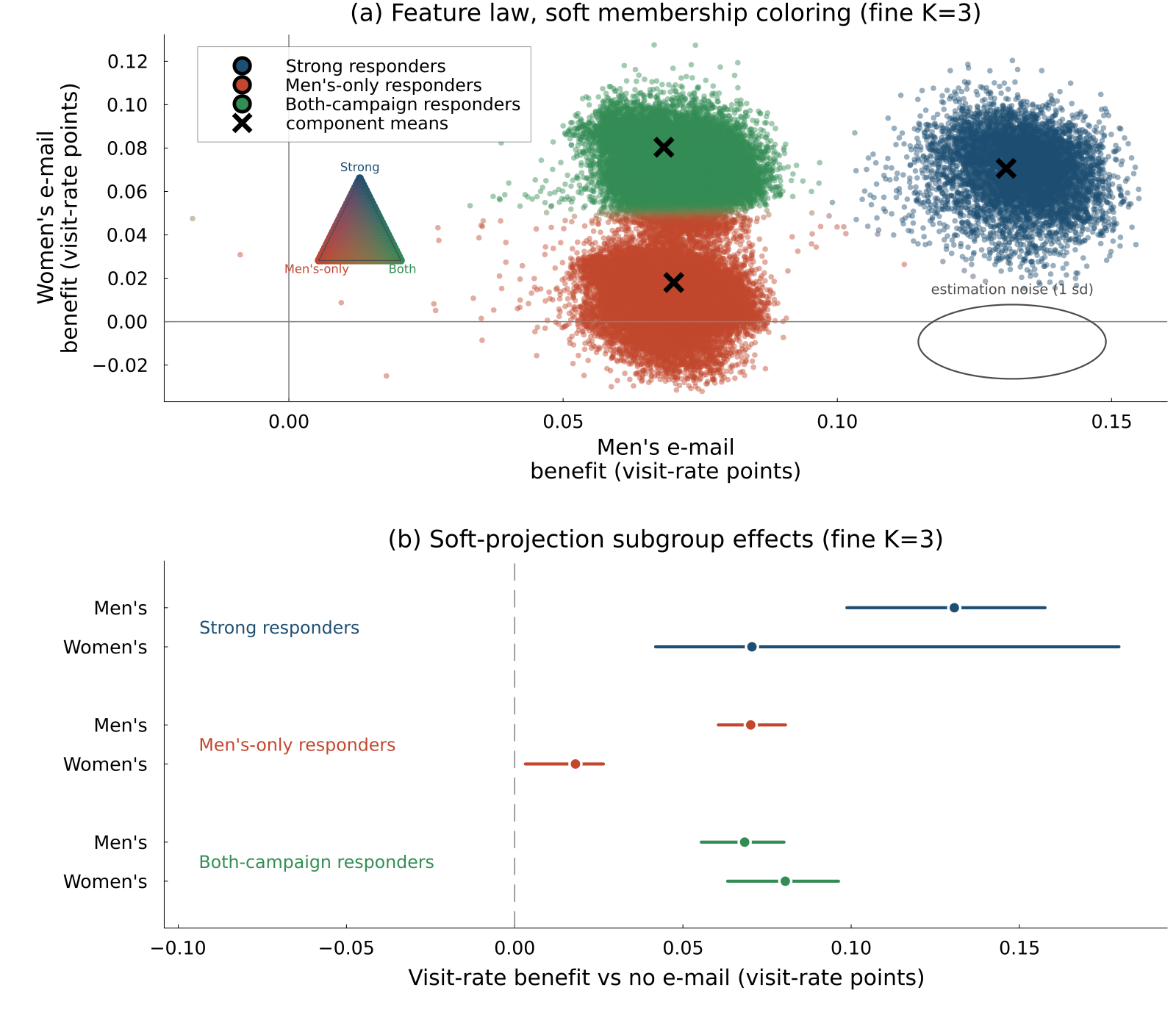}
\caption{MineThatData e-mail experiment fine $K=3$ soft-projection working summary for the visit outcome. (a) Estimated causal feature law with soft membership coloring. Each point is a cross-fitted campaign-benefit pair, the men's and women's e-mail visit benefits against control, colored by the membership-weighted blend of the three component colors, the weights being that unit's soft memberships $r_1,r_2,r_3$, so blended hues mark ambiguous membership. Blends are almost absent because the memberships are near binary, $98.9\%$ of units placing weight above $0.9$ on one component, which is the visual counterpart of that diagnostic. The three component means are marked by black crosses, the inset triangle is the membership simplex with vertices labeled by component, and the ellipse marks the per-coordinate estimation-noise scale implied by $\widehat\Delta$, radius $\sqrt{\widehat\Delta/q}$ per coordinate, so the within-component spread matches the estimation-noise scale while the between-component separations dwarf it. (b) Soft-projection subgroup effects, shown as a forest plot. Rows are the contrasts $\hat\psi_{h,a}(K)-\hat\psi_{h,0}(K)$ of \eqref{eq:effects} for arm $a$ the men's or women's campaign, the corrected visit-rate benefit against control within component $h$, with the rows grouped and colored by component and named by campaign on the vertical axis. The point marks the estimate and the single horizontal segment is the $95\%$ equal-tailed posterior interval from the per-draw recomputation of the projection. The vertical dashed line marks zero and the horizontal axis is in visit-rate points. The components are the strong responders, the men's-only responders, and the both-campaign responders, of soft shares $0.108$, $0.545$, and $0.347$.}
\label{fig:hillstrom:working_summary}
\end{figure}

The two subgroup summaries below, the coarse $K=2$ and fine $K=3$ reports, the latter at the fine end of $\widehat C(0.80)=\{2,3\}$, are soft mixture projections of the estimated feature law, the estimands of Section~\ref{sec:soft}. The fitted family is the $K$-component Gaussian mixtures on the benefit plane with a common spherical scale, mixing weights floored at $0.01$, and component means confined to a box around the estimated features. We declare the common scale at the feature-noise floor $\sigma=0.008$, a resolution constant of the description rather than a fitted quantity since the corrected criterion admits no interior scale, with the rationale and sensitivity profile in Supplementary Section~\ref{sec:app:hillstrom_k3}.

The fine $K=3$ projection in Figure~\ref{fig:hillstrom:working_summary} is the main interpretive summary. We report the fitted subgroup means $\hat\psi_{h,a}(K)$ of \eqref{eq:effects}, with arm $a=0$ the no e-mail control, $a=m$ the men's campaign, and $a=w$ the women's campaign, and read each component through its campaign contrasts against control. A small strong-responder component ($h=1$) raises visits under both campaigns, with $\hat\psi_{1,m}(3)-\hat\psi_{1,0}(3)=+0.131$, interval $(0.099,\,0.158)$, and $\hat\psi_{1,w}(3)-\hat\psi_{1,0}(3)=+0.071$, interval $(0.042,\,0.180)$. A large men's-only component ($h=2$) responds to the men's campaign alone, with $\hat\psi_{2,m}(3)-\hat\psi_{2,0}(3)=+0.070$, interval $(0.061,\,0.081)$, against $\hat\psi_{2,w}(3)-\hat\psi_{2,0}(3)=+0.018$, interval $(0.003,\,0.026)$. A both-responder component ($h=3$) responds to both, with $\hat\psi_{3,m}(3)-\hat\psi_{3,0}(3)=+0.068$, interval $(0.055,\,0.080)$, and $\hat\psi_{3,w}(3)-\hat\psi_{3,0}(3)=+0.080$, interval $(0.063,\,0.096)$. The soft shares of the three components are $0.108$, interval $(0.020,\,0.142)$, then $0.545$, interval $(0.383,\,0.668)$, then $0.347$, interval $(0.240,\,0.513)$. The men's-campaign effects of all three components and the women's-campaign effects of the men's-only and both-responder components are sharply determined, while the strong-responder component's women's-campaign effect and its share carry visibly wider intervals, reflecting draws in which that small component's weighted mass runs low. The mass split between the two adjacent large components remains the least determined quantity, their share intervals overlapping. The coarse $K=2$ projection spreads the strong-responder mass across the two remaining components, mostly into the both-responder one. It reads a men's-only component of share $0.579$, interval $(0.374,\,0.699)$, at $\hat\psi_{1,m}(2)-\hat\psi_{1,0}(2)=+0.076$ and $\hat\psi_{1,w}(2)-\hat\psi_{1,0}(2)=+0.019$, and a both-responder component of share $0.421$, interval $(0.301,\,0.626)$, at $+0.076$ and $+0.082$ (Supplementary Section~\ref{sec:app:hillstrom_k3}). The soft memberships are near binary, with $98.9\%$ of units placing weight above $0.9$ on one component at $K=3$ and $98.3\%$ at $K=2$, which Figure~\ref{fig:hillstrom:working_summary}(a) shows as the near-absence of blended color. Because these are the soft projections of Section~\ref{sec:soft}, the effect intervals are covered by Theorem~\ref{thm:effects} and the share intervals by Supplementary Theorem~\ref{thm:projection}. Corollary~\ref{cor:fixedgamma} licenses the coarse report at the prespecified thresholds, and because $\widehat C(0.80)$ is not a singleton we report both supported counts. Averaged over the population, the corrected arm contrasts against control are $+0.0761$ visit-rate points, interval $(0.0692,\,0.0823)$, for the men's e-mail and $+0.0454$, interval $(0.0392,\,0.0516)$, for the women's, both positive with the men's campaign the stronger.

The same experiment also records a two-week spending outcome in dollars, which behaves oppositely under the gate. On the spending scale the corrected total heterogeneity is $\widehat W(1)=-0.14$ with band $(-0.66,\,0.37)$ straddling zero, while the noise-floor diagnostic returns $\widehat\Delta=1.55$ with bracket $[0.77,\,1.55]$, so the band does not clear zero, the first tier of the gate fails, and no resolution beyond a single group is supported (Supplementary Section~\ref{sec:app:hillstrom_spend}). Within one trial the gate thus separates an informative outcome, the visit response, from an uninformative one, the noisy dollar response.

An informative resolution report is one in which the set-valued reports resolve and the exhaustion of resolvable structure is made explicit. Had the corrected $\widehat W$ path stayed positive through $K=\overline K$, the reading would be that six groups do not suffice, the set-valued reports would run to the ceiling, and no usable summary would emerge. The visit analysis delivers the complete arc instead. Detection clears the first tier, the reliability ratio keeps the $\rho$ scale under the second tier's qualification, and the confidence sets resolve at two to three groups without reaching the ceiling. The working structure is interpretable, a large group that responds to the men's campaign alone and the remainder that responds to both campaigns. The estimated feature law, the shift-sensitivity overlay, and the set-valued report tell one story, and the noise-floor diagnostic reconciles them. This is the reporting standard the paper argues for.


\section{Discussion}\label{sec:discussion}

We proposed to replace the question ``how many causal subgroups are there?'' by ``how many subgroups does each level of descriptive resolution require?'', a shift that makes the target a well-defined population functional and the feature-law posterior its inferential instrument, calibrated under the stated nuisance rates. The uniform conditional Bernstein--von Mises theorem converts posterior credible statements for the quantization path, the resolution profile, fixed-resolution summaries, and subgroup effects into asymptotically valid frequentist statements, and the set-valued profile report retains locally uniform validity over exactly the root-$n$ perturbation classes on which Theorem~\ref{thm:impossibility} shows single-valued selection must fail. The penalized report gives an optional price-indexed reading of the same path, but the primary estimand is the threshold resolution profile.

The practical reporting standard is therefore a resolution-profile plot, set-valued counts at scientifically chosen resolutions, and fixed-resolution summaries or subgroup effects only after the supported resolution has been stated, with a penalized report added when an effect-scale cost for one additional subgroup is interpretable. The MineThatData trial illustrates the standard. On the visit response the data indicate clear average campaign effects and a summary supported at two to three groups, a large group responding to the men's campaign alone, the remainder responding to both, and at fine resolution a small strong-responder group raising visits under both campaigns. Beyond three groups the split-difference diagnostic places any remaining structure below the noise floor, and no finer count is offered. The resolution-indexed report thus states which structure the data support and where that support ends.

Several limitations define the scope of the present theory. Study~2 separates two finite-sample mechanisms at high resolution, a total-heterogeneity bias from learners that miss the required regression rate, removed entirely by a rate-compliant learner, and the downward bias of empirical minima at fine resolution, which persists under every learner including oracle nuisances because the true risk decrements at overfitted resolutions are small relative to it. Inference should therefore lean on the set-valued profile, whose validity along drifting sequences Study~4 verifies directly. The main Gaussian path results assume the margin condition, and mixed laws with isolated atoms require localized versions, with root-$n$ path inference possibly depending on whether atoms lie on optimal cell boundaries. The fixed ceiling $\overline K$ could grow only with joint control of quantization entropy, margin constants, and shrinking knot gaps. Finally, distribution functions, quantiles, and threshold exceedance probabilities of treatment effects are irregular without smoothing, and extending resolution-profile inference to these mixed, adaptive, and irregular settings is a natural next step.
The level-shift structure of Remark~\ref{rem:noisefloor} points to a further direction. Estimators that recentre the quantization path, whether by the split-difference construction of Supplementary Section~\ref{sec:noisefloor_supp} or by higher-order U-statistic corrections in the manner of \cite{robins2008}, can remove the second-order level bias that a common feature-estimation error imposes on $W(\cdot)$. Such corrections do not touch the efficiency-bound sampling variance of the quadratic functionals, so their value is interpretive, sharpening the reading of the path level rather than providing a coverage guarantee at small samples. We leave a full development of recentred path inference to future work.

\section*{Acknowledgement}
Research in this article was supported by the United States National Institutes of Health (NIH), National Heart, Lung, and Blood Institute (NHLBI, grant number R01-HL168202). All statements in this report, including its findings and conclusions, are solely those of the authors and do not necessarily represent the views of the NIH.

\singlespacing
\bibliographystyle{Chicago}
\bibliography{literature}
\singlespacing

\newpage
\appendix
\setcounter{page}{1}
\setcounter{section}{0}
\renewcommand{\thesection}{S\arabic{section}}
\renewcommand{\appendixname}{Supplementary Material}

\section{Additional estimands and reporting protocol}
\subsection{Population geometry of the penalized profile}\label{sec:geometry_penalized_profile_supp}

The resolution profile reads the quantization path by fixing the desired approximation quality. The same path also admits an optional price-indexed profile.  
Here, we use the penalized profile only as a complementary interpretation of the same values $\{W(K)\}_{K\le\overline K}$. The threshold dial asks how many groups are needed to explain a prespecified fraction of causal heterogeneity, while the price dial asks whether the marginal reduction in residual dispersion is worth the scientific and inferential cost of another subgroup. 

For a price $\tau>0$ assigned to one additional reported subgroup, 
\begin{equation}\label{eq:Qdef}
    Q(\tau)=\min_{1\le K\le\overline K}\{W(K)+\tau K\},
    \quad
    \calS_0(\tau)=\argmin_{1\le K\le\overline K}\{W(K)+\tau K\},
    \quad
    K^\dagger(\tau)=\min\calS_0(\tau).
    \end{equation}
Here $\tau$ is best interpreted as a shadow price for complexity, not as the literal cost of adding another row to a trial report.  Although the mechanical cost of listing more clusters is often negligible, naming an additional subgroup carries scientific and inferential costs. The groups become smaller, subgroup effects are estimated with more uncertainty, interpretation becomes harder, and downstream users may treat the groups as clinically distinct.  Thus $\tau$ indexes how much reduction in expected within-group causal dispersion is required before an additional subgroup is treated as worth distinguishing.

\begin{definition}[Active orders, merge scales, persistence]\label{def:merge}
An order $K$ is active if $\calS_0(\tau)=\{K\}$ for some $\tau>0$.  Write the active set as $\calS=\{K_{(1)}<\cdots<K_{(J)}\}$.
For consecutive active orders, define the merge scale
\begin{equation*}
\tau_{(j)}
=
\frac{W(K_{(j)})-W(K_{(j+1)})}
{K_{(j+1)}-K_{(j)}}>0,
\qquad j=1,\dots,J-1.
\end{equation*}
This is the price at which the larger summary $K_{(j+1)}$ ceases to be worth its additional complexity and merges into the coarser summary $K_{(j)}$.  With the conventions $\tau_{(0)}=\infty$ and $\tau_{(J)}=0$, the persistence interval of $K_{(j)}$ is $(\tau_{(j)},\tau_{(j-1)})$, and its persistence is the length of this interval.  The adjacent decrements $d_K=W(K-1)-W(K)$, $2\le K\le\overline K$, are the usual elbow-plot ordinates; they coincide with merge scales only when consecutive integers are active.
\end{definition}

A concrete calibration comes from the outcome scale.  In a blood-pressure trial with scalar feature $U(X)$ equal to the treatment effect on systolic blood pressure, measured in mm Hg, $W(K)-W(K+1)$ is the population-averaged reduction in squared residual treatment-effect dispersion obtained by adding one subgroup.  If two subgroups with prevalences $\omega$ and $1-\omega$ differ in mean treatment effect by a clinically meaningful amount $\delta$, then merging them contributes approximately $\omega(1-\omega)\delta^2$ to the risk; for equally prevalent groups this is $\delta^2/4$.  Thus, when a minimal clinically important difference $\delta_{\mathrm{MCID}}$ is available, values of $\tau$ on the scale of $\omega(1-\omega)\delta_{\mathrm{MCID}}^2$ ask whether an additional subgroup explains heterogeneity comparable to a clinically meaningful between-group contrast, rather than merely improving fit by a statistically detectable but practically negligible amount.

\begin{remark}[Two dials on one path]\label{rem:twodials}
The threshold profile $K^\star(\gamma)$ and the penalized profile $K^\dagger(\tau)$ are complementary readings of the same quantization path $\{W(K)\}_{K\le\overline K}$.  The threshold dial is scale-free and asks how many subgroups are needed to explain fraction $\gamma$ of causal heterogeneity.  The price dial carries effect-size units and asks how many subgroups are worth reporting when one additional subgroup must reduce within-group dispersion by $\tau$.  The two readings differ in which orders they visit. The threshold profile visits every strictly improving order, whereas the penalized profile visits only envelope orders.  Both are deterministic functionals of the same path, so every posterior draw of the path delivers both at no additional computational cost.  Proposition~\ref{prop:envelope} shows that the penalized profile is the lower-envelope geometry of the path rather than an additional model-selection rule layered on top of it.
\end{remark}

The next proposition is elementary, but conceptually important.  It shows that the penalized profile is not an additional model-selection rule layered on top of the path.  It is the lower-envelope geometry of the same path $\{W(K)\}_{K\le\overline K}$.  Active orders are the slopes of this envelope, merge scales are its kinks, and persistence intervals are its facets.

\begin{proposition}[Population geometry of the penalized profile]\label{prop:envelope}
The function $Q$ is concave, nondecreasing, and piecewise linear on $(0,\infty)$, with at most $\overline K$ linear pieces.  The penalized profile $K^\dagger$ is a nonincreasing step function.  Moreover:
\begin{enumerate}[label=(\roman*),leftmargin=2em]
\item At every $\tau$ where $Q$ is differentiable, the optimal order is unique and equals the slope, $\calS_0(\tau)=\{Q'(\tau)\}$.
At a kink $\tau_{(j)}$, the superdifferential of $Q$ is the interval $[K_{(j)},K_{(j+1)}]$, whose integer endpoints are the two merging active orders.

\item An order $K$ is active if and only if $(K,W(K))$ is a strict vertex of the greatest convex minorant of the finite sequence $K\mapsto W(K)$ on $[\overline K]$.  Thus orders lying strictly above this minorant are optimal at no price.  In particular, $K=1$ is always active, and any order with $W(K)=W(K-1)$ is inactive.

\item The merge scales are strictly decreasing, $\tau_{(1)}>\cdots>\tau_{(J-1)}>0$,
and $\calS_0(\tau)=\{K_{(j)}\}$ if and only if $\tau\in(\tau_{(j)},\tau_{(j-1)})$.
Consequently, the penalized profile jumps exactly at the merge scales and skips inactive orders.  No convexity or strict-decrement condition on the raw sequence $W(1),\dots,W(\overline K)$ is required; flat or nonconvex stretches simply produce inactive orders.

\item Only the active orders determine the lower envelope, $Q(\tau)=\min_{K\in\calS}\{W(K)+\tau K\}.$
Conversely, on the active set, $W(K)=\sup_{\tau>0}\{Q(\tau)-\tau K\}$.
\end{enumerate}
\end{proposition}

Proposition~\ref{prop:envelope} gives a population version of the elbow heuristic.  A naive elbow plot displays the adjacent decrements $d_K$, but the price-indexed profile is governed by the slopes of the greatest convex minorant of $K\mapsto W(K)$.  Orders above that minorant are never optimal for any price and should not be treated as stable subgroup resolutions.  A pronounced elbow at an active order is instead a long persistence interval: the order remains uniquely best over a wide range of prices.  Corollary~\ref{cor:penalized} below turns this geometry into merge-scale inference and simultaneous set-valued coverage for the penalized profile.

\subsection{Inference for the optional penalized profile}\label{sec:penalized_supp_inference}

The penalized profile inherits inference from the same path; this is the formal result supporting the optional report.  Write $\Sigma_W=\bigl[\Cov\{\G_0(g_{c^\star(K)}),\G_0(g_{c^\star(L)})\}\bigr]_{K,L\le\overline K}$ and fix a reporting range $[\tau_{\min},\tau_{\max}]\subset(0,\infty)$ whose endpoints are not merge scales.  Let $\calS_{\mathrm{rep}}$ collect the active orders whose persistence intervals meet the range, and let the hatted and per-draw versions be computed from $\widehat W(\cdot)$ and $W^{(s)}(\cdot)$.

\begin{corollary}[Penalized-profile inference]\label{cor:penalized}
Let the conditions of Corollary~\ref{cor:path} hold.
\begin{enumerate}[label=(\alph*),leftmargin=2em]
\item \textup{(Merge scales.)} If every inactive order lies strictly above the greatest convex minorant of $K\mapsto W(K)$, then $\Pp(\widehat\calS_{\mathrm{rep}}=\calS_{\mathrm{rep}})\to1$ and $\Pp_w(\calS^{(s)}_{\mathrm{rep}}=\calS_{\mathrm{rep}})\pto1$.  Jointly over the merge scales in the reporting range,
\begin{equation*}
\sqrt n\bigl(\hat\tau_{(j)}-\tau_{(j)}\bigr)_j\ \dto\ N\bigl(0,\,D\Sigma_WD^\top\bigr),
\qquad
\sqrt n\bigl(\tau^{(s)}_{(j)}-\hat\tau_{(j)}\bigr)_j\ \dtow\ N\bigl(0,\,D\Sigma_WD^\top\bigr),
\end{equation*}
where $D$ maps the active path values to the consecutive slopes in Definition~\ref{def:merge}.
\item \textup{(Simultaneous set-valued coverage.)} Assume $\Var\{\G_0(g_{c^\star(K)})\}>0$ for all $K\le\overline K$.  Let $(L_W,U_W)$ be the simultaneous posterior band for $W(\cdot)$ constructed as in \eqref{eq:band}, and define
\begin{equation}\label{eq:penset}
\widehat C^\dagger(\tau)\;=\;\Bigl\{K\le\overline K:\ L_W(K)+\tau K\ \le\ \min_{K'\le\overline K}\bigl\{U_W(K')+\tau K'\bigr\}\Bigr\}.
\end{equation}
Then
\begin{equation*}
\liminf_{n\to\infty}\ \Pp\Bigl\{\calS_0(\tau)\subseteq\widehat C^\dagger(\tau)\ \text{ for every }\tau\in(0,\infty)\Bigr\}\ \ge\ 1-\alpha .
\end{equation*}
\end{enumerate}
\end{corollary}

The value process for $Q(\tau)$ is regular on compact sets avoiding merge scales, while it is not tight in neighborhoods of a merge scale; the formal statement is Corollary~\ref{cor:penalizedvalue} and Proposition~\ref{prop:nontight} below.  On the same band event used in Corollary~\ref{cor:penalized}\textup{(b)}, interval arithmetic gives simultaneous confidence intervals for all merge scales by bracketing the two path values that define each slope.  Pooling the max statistics for \eqref{eq:profileset} and \eqref{eq:penset} would make the threshold and penalized set-valued reports jointly valid; we report them separately for interpretability. We provide below further guarantees regarding the penalized profile.

\begin{corollary}[Penalized value process]\label{cor:penalizedvalue}
Let the conditions of Corollary~\ref{cor:path} hold.  For every compact $T\subset(0,\infty)$ containing no merge scale,
\begin{equation*}
\sqrt n\bigl\{\widehat Q(\tau)-Q(\tau)\bigr\}\ \dto\ \G_0\bigl(g_{c^\star(K^\dagger(\tau))}\bigr)
\quad\text{in }\ell^\infty(T),
\end{equation*}
and the conditional analogue holds for $\sqrt n\{Q^{(s)}(\tau)-\widehat Q(\tau)\}$.  At a merge scale $\tau_{(j)}$, under the nondegenerate envelope geometry of Corollary~\ref{cor:penalized}\textup{(a)},
\begin{equation*}
\sqrt n\{\widehat Q(\tau_{(j)})-Q(\tau_{(j)})\}\dto\min_{K\in\calS_0(\tau_{(j)})}\G_0(g_{c^\star(K)}),
\end{equation*}
where the tie set is the merging pair $\{K_{(j)},K_{(j+1)}\}$.  The limit at a merge scale is non-Gaussian and the posterior law is in general inconsistent for it \citep{fang2019}; this is why the main report uses the set-valued band inversion \eqref{eq:penset} rather than pointwise derivative inversion at kinks.
\end{corollary}

\begin{proposition}[Non-tightness at merge scales]\label{prop:nontight}
Let the conditions of Corollary~\ref{cor:path} hold together with the nondegenerate envelope geometry of Corollary~\ref{cor:penalized}\textup{(a)}, let $\tau_{(j)}$ be a merge scale with $\Var\{\G_0(g_{c^\star(K_{(j)})})-\G_0(g_{c^\star(K_{(j+1)})})\}>0$, and write $\Delta K=K_{(j+1)}-K_{(j)}$. Then, for each fixed $\gamma>0$, along the drifting prices $\tau_n=\tau_{(j)}+\gamma n^{-1/2}$,
\begin{equation*}
\sqrt n\bigl\{\widehat Q(\tau_n)-Q(\tau_n)\bigr\}\ \dto\ \min\bigl\{\G_0(g_{c^\star(K_{(j)})}),\ \G_0(g_{c^\star(K_{(j+1)})})+\gamma\,\Delta K\bigr\},
\end{equation*}
a limit that depends on $\gamma$; consequently $\sqrt n\{\widehat Q-Q\}$ is not asymptotically tight in $\ell^\infty$ of any neighborhood of $\tau_{(j)}$, and no weak limit in $\ell^\infty$ exists over any set containing a one-sided neighborhood of a merge scale.
\end{proposition}

\begin{remark}[The metric is part of the estimand]\label{rem:metric}
For multivariate causal features, the coordinates of $U$ may have different units or scientific meanings, so the metric in \eqref{eq:quantloss} is part of the estimand.  For a fixed symmetric positive-definite matrix $M$, one may replace $\norm{u-c_h}^2$ by $\norm{u-c_h}_M^2=(u-c_h)^\top M(u-c_h)$, so that $W(1)=\tr\{M\Var(U)\}$ and $\rho(K)$ is the fraction of $M$-weighted causal heterogeneity explained by the best $K$-group summary.  No new theory is needed. $M$-quantization of $P_U$ is ordinary Euclidean quantization of the law of $M^{1/2}U$, and the margin and posterior arguments carry over with constants depending on the eigenvalue bounds of $M$.  The choice of $M$ is a scientific normalization. Common choices include $M=\diag\{\Var(U)\}^{-1}$ for coordinatewise scale normalization and $M=\Var(U)^{-1}$ for affine equivariance, when these matrices are well defined.  For $q=1$, a positive scalar $M$ rescales $W(K)$ but cancels from causal heterogeneity $R^2$ and the resolution profile.  We treat $M$ as fixed and prespecified. Data-adaptive metrics would require augmenting the moment expansion with the estimation error of $M$, and we do not pursue that extension here.
\end{remark}

\subsection{Reporting protocol}\label{sec:report}

The output of Algorithm~\ref{alg:main} is a complete heterogeneity analysis. Table~\ref{tab:cast} of Section~\ref{sec:resolution} collects the population objects and the result delivering each one's inference. We recommend the following deliverables, the penalized report being optional, and each item carries the uncertainty statement that Section~\ref{sec:theory} establishes.

(i) The resolution-profile plot, the headline graphic: $\hat\rho(K)$ against $K$ with the simultaneous band \eqref{eq:band}, horizontal reading lines at the resolutions of scientific interest, and the induced sets $\widehat C(\gamma)$ of \eqref{eq:profileset} displayed as the highlighted runs of admissible $K$ at each such line.  Theorems state the population domain $\gamma\in(0,\rho_0(\overline K))$; in implementation the displayed reading lines are prespecified by the analyst and are truncated to the empirically attainable range, for example to $[0,\hat\rho(\overline K)]$ or to the upper simultaneous band limit. The plot subsumes the elbow heuristic (the elbow is the visible flattening of $\hat\rho$), replaces a point choice of $K$ by a set-valued report, and, because the band is simultaneous, supports free-form scanning across $\gamma$ without multiplicity corrections (Theorems~\ref{thm:profile} and~\ref{thm:honest}). This guarantee presumes the nuisance rates of Assumption~\ref{ass:nuisance}, including the margin-rate requirement for quantization scores. Posterior spaghetti of the draws $\rho^{(s)}(\cdot)$ overlaid at low opacity communicates the joint dependence across $K$ that the band alone hides.

(ii) Set-valued subgroup counts at the reported resolutions: $\widehat C(\gamma)$ collapses to a singleton away from resolution thresholds and widens to a pair at them, the widening being shown to be unavoidable rather than conservative (Theorems~\ref{thm:impossibility} and~\ref{thm:profile}).

(iii) The optional penalized sensitivity report: $\widehat Q(\tau)$ with its kinks, the merge scales with joint intervals, persistence intervals, and the set $\widehat C^\dagger(\tau)$ of \eqref{eq:penset} (Corollary~\ref{cor:penalized}). When a minimal clinically or economically important difference is available, the price grid is prespecifiable in its squared units (Section~\ref{sec:penalized_supp_inference}).

(iv) Subgroup effects with propagated partition uncertainty: the contrasts $\psi_{h,a}-\psi_{h,a'}$ with credible intervals in which the uncertainty of the partition itself propagates automatically through the per-draw recomputation (Theorem~\ref{thm:effects}), including the membership-gradient correction that the default ``cluster, then estimate within clusters'' pipeline omits (Supplementary Remark~\ref{rem:expandedscores}). When the working resolution is selected from the profile at a prespecified $\gamma$, inference is valid at the selected resolution provided $\gamma$ is not a knot (Corollary~\ref{cor:fixedgamma}); when $\widehat C(\gamma)$ is not a singleton, report the subgroup effects at every supported $K$ rather than at a single selected count.


\section{Feature-law posterior construction details}\label{app:posterior_details}

The full sampling procedure summarized in Section~\ref{sec:construction} is stated as Algorithm~\ref{alg:main}.

\begin{algorithm}[ht!]
\caption{Feature-law posterior for causal subgroup analysis}\label{alg:main}
\begin{algorithmic}[1]
\State \textbf{Input:} data $\{O_i\}_{i\le n}$; feature map $H$; ceiling $\overline K$; mixture family $k(\cdot;\theta)$; folds $B$; draws $S$.
\State Cross-fit $\hateta^{(-b)}$, form $\hU_i$, $\hR_i$, the corrected evaluations $\hphi_{f,i}$, and the point process $\hPsi(f)=n^{-1}\sum_i\hphi_{f,i}$; apply the estimand functionals to $\hPsi$ for all point estimates.
\For{$s=1,\dots,S$}
\State Draw $w^{(s)}\sim n\cdot\mathrm{Dirichlet}(1,\dots,1)$ and form $\Psi^{(s)}$ as in \eqref{eq:posteriorprocess}.
\State $W^{(s)}(K)\gets\inf_{c\in\calC^K}\Psi^{(s)}(g_c)$, $\rho^{(s)}(K)\gets 1-W^{(s)}(K)/W^{(s)}(1)$, $K^{\star(s)}(\gamma)\gets\min\{K:\rho^{(s)}(K)\ge\gamma\}$.
\State At selected working resolutions, $\beta^{(s)}(K)\gets\argmin_\beta\Psi^{(s)}(\ell_\beta)$ and $\psi^{(s)}_{h,a}(K)\gets\Psi^{(s)}\bigl(f^N_{h,a;\beta^{(s)}(K)}\bigr)\big/\Psi^{(s)}\bigl(f^D_{h;\beta^{(s)}(K)}\bigr)$.
\EndFor
\State From the draws, compute the scale estimates $\hat\sigma(K)$, the simultaneous quantile $q_{1-\alpha}$, the band \eqref{eq:band}, and the set-valued profile report $\widehat C(\gamma)$ by \eqref{eq:profileset}.
\State \textbf{Output:} simultaneous bands for $\rho(\cdot)$, the set-valued profile report $\widehat C(\gamma)$, and posterior intervals for selected fixed-resolution summaries and subgroup effects.
\end{algorithmic}
\end{algorithm}

\subsection{The sense in which this is Bayesian}\label{sec:bayesian}

The construction admits three complementary Bayesian-bootstrap readings, but the qualifier matters: it is not a full Bayesian posterior for the data-generating law.

\textbf{Dirichlet-process limit.} Conditionally on the corrected evaluations, \eqref{eq:posteriorprocess} is the posterior law of $\int\hphi_f\,dF$ when $F$, the distribution of the corrected score vector, carries a Dirichlet-process prior $\mathrm{DP}(a,F^*)$ in the noninformative limit $a\downarrow0$: the Bayesian bootstrap of \cite{rubin1981}, whose use for exactly identified moment functionals goes back to \cite{newton1994}. The posterior is over the empirical law of the corrected-score evaluations, not over an assumed mixture parameter, partition, or model order.

\textbf{Martingale posterior.} Equivalently, \eqref{eq:posteriorprocess} arises from predictive resampling with the empirical predictive: the simplest member of the martingale-posterior family of \cite{fong2023}, applied to the corrected scores. The general recipe, positing a predictive update and propagating it to functionals, licenses richer predictives; the theory of Section~\ref{sec:theory} covers the exchangeable-bootstrap weights used here.

\textbf{Corrected posterior.} Most directly, the scheme is the process-level extension of the semiparametric posterior corrections of \cite{yiu2025}: there, a Bayesian-bootstrap posterior for a scalar functional is recentered by a one-step influence-function correction; here, the correction is applied inside every evaluation $\hphi_{f,i}=\phi_f\bigl(O_i;\hateta^{(-b(i))}\bigr)$ and the resulting posterior is for an entire moment process, which is what cluster analysis consumes. The uniform Bernstein--von Mises theorem (Theorem~\ref{thm:bvm}) is the corresponding strengthening of their scalar matching result.

The scheme is not full Bayes for a generative model of $O$: no likelihood for $Y\mid A,X$ is specified, no prior is put on partitions or model order, and prior information enters only through $\hateta$ (or Remark~\ref{rem:informative}) and the choice of $\calF$. What is gained for that price is a calibrated posterior for the corrected moment process.  The object updated is that moment process itself, not a probability law on the feature space, and because a weighted corrected functional can take negative values in finite samples the draws need not be feature-space probability measures. The name feature-law posterior is thus a shorthand for the Bayesian-bootstrap posterior of the corrected feature-law moments.  Section~\ref{sec:app:energy_amplification} explains, more narrowly, why we do not replace this linear construction by an exponentiated order posterior when features are generated.

\begin{remark}[Nuisance uncertainty and informative priors]\label{rem:informative}
As stated, the nuisances are handled modularly. The fit $\hateta$ is computed once and held fixed across draws, a cut-model choice that the theory rewards, because Theorem~\ref{thm:eif} makes the posterior first-order insensitive to the nuisance fit. An analyst wishing to express prior information about $\eta$ may instead draw $\eta^{(s)}$ from a posterior for the regressions (e.g.\ a Gaussian-process or BART posterior) and form $\Psi^{(s)}(f)=n^{-1}\sum_i w^{(s)}_i\phi_f(O_i;\eta^{(s)})$: the correction term annihilates the first-order propagation of nuisance-prior bias, in the spirit of the prior corrections of \cite{ray2020} that depend on the propensity score. We record this as an extension. The theorems below are proved for the modular scheme, and the double-draw variant requires in addition that the nuisance posterior contract at the rates of Assumption~\ref{ass:nuisance}.
\end{remark}

\begin{remark}[Clustered sampling units]\label{rem:clustered}
When the data carry a cluster structure, for example families randomized jointly, the independent sampling unit is the cluster, and Definition~\ref{def:flposterior} should be applied at that level. Let $G_1,\dots,G_m$ partition $[n]$ into clusters with $g(i)$ the cluster of unit $i$, draw $v^{(s)}\sim m\cdot\mathrm{Dirichlet}(1,\dots,1)$ over clusters, and set $w^{(s)}_i=v^{(s)}_{g(i)}$, so that all members of a cluster share one weight. The point process $\hPsi$ is unchanged and each draw again recomputes every functional. With clusters as the independent sampling unit, the same arguments go through after replacing observations by cluster-level aggregates, the corrected scores summed within clusters, the number of clusters $m$ playing the role of $n$, and folds formed by cluster. This requires that cluster sizes be bounded, that the nuisance rates and entropy conditions be read at the cluster level, and that the target weighting, individual-weighted versus cluster-weighted feature law, be declared, since these differ when cluster sizes vary. A formal cluster-level development is beyond the scope of this paper. Individual-level weighting understates uncertainty when corrected scores are positively correlated within clusters, and the clustered weights restore the correct first-order variance.
\end{remark}

\subsection{Pseudo-feature identity and implementation details}\label{sec:pseudo_supp}

This section gives the formal identity behind the quantization computation in Section~\ref{sec:construction}.  Recall that
\begin{equation}\label{eq:pseudofeature}
\tU_i=\hU_i+H\hR_i .
\end{equation}

\begin{proposition}[Pseudo-feature identity]\label{prop:pseudo}
For every codebook $c\in\calC^K$ and every observation $i$,
\begin{equation}\label{eq:pseudoidentity}
\hphi_{g_c,i}
=
\bigl\|\tU_i-c_{h_c(\hU_i)}\bigr\|^2
-
\bigl\|H\hR_i\bigr\|^2,
\end{equation}
where $h_c(\hU_i)$ is the Voronoi cell of $\hU_i$, not of $\tU_i$.  Consequently, for any nonnegative weights $w=(w_1,\dots,w_n)$,
\begin{equation}\label{eq:pseudoweighted}
\frac1n\sum_{i=1}^n w_i\hphi_{g_c,i}
=
\frac1n\sum_{i=1}^n w_i
\bigl\|\tU_i-c_{h_c(\hU_i)}\bigr\|^2
-
\frac1n\sum_{i=1}^n w_i
\bigl\|H\hR_i\bigr\|^2 .
\end{equation}
The second term is free of $c$ and therefore can be dropped for optimization over codebooks, but it must be retained when reporting corrected risk values.
\end{proposition}

\begin{proof}
For the quantization loss $g_c(u)=\min_{h\le K}\norm{u-c_h}^2$, off Voronoi boundaries,
\[
\nabla g_c(u)=2\{u-c_{h_c(u)}\}.
\]
Thus the corrected evaluation at observation $i$ is
\[
\hphi_{g_c,i}
=
\norm{\hU_i-c_{h_c(\hU_i)}}^2
+
2\{\hU_i-c_{h_c(\hU_i)}\}^\top H\hR_i .
\]
Apply the identity
\[
\norm{a-c}^2+2(a-c)^\top v
=
\norm{a+v-c}^2-\norm{v}^2
\]
with $a=\hU_i$
, $v=H\hR_i$
, and $c=c_{h_c (\hU_i)}$.
This proves \eqref{eq:pseudoidentity}. Multiplying by weights and summing gives \eqref{eq:pseudoweighted}.
\end{proof}

\begin{remark}[Assignment and target]\label{rem:traps}
The assignment rule in Proposition~\ref{prop:pseudo} is essential.  Plain $k$-means on the pseudo-features $\{\tU_i\}$ assigns each point by the cell of $\tU_i$ and optimizes
\[
\sum_i w_i\min_h\norm{\tU_i-c_h}^2,
\]
which is not the corrected quantization criterion.  It targets the codebook of a noise-convolved pseudo-feature law rather than the codebook of $P_U$ when the residual correction is nonnegligible relative to cell separations.  The corrected objective instead assigns by $\hU_i$ and updates centers using $\tU_i$.
\end{remark}

\begin{remark}[Lloyd-type implementation]\label{rem:lloyd}
For fixed assignments $h_i$, the weighted objective in \eqref{eq:pseudoweighted} is minimized over centers in $\calC$ by
\[
c_h
=
\Pi_{\calC}\Bigl(\frac{\sum_{i:h_i=h}w_i\tU_i}
{\sum_{i:h_i=h}w_i}\Bigr),
\]
where $\Pi_{\calC}$ is the Euclidean projection onto the convex set $\calC$. The projection is needed because the pseudo-features $\tU_i=\hU_i+H\hR_i$ are not confined to $\calC$, so the weighted mean may leave it. In practice $\calC$ can be chosen large enough that the constraint never binds, in which case $\Pi_{\calC}$ acts as the identity and implementations may record whether it ever binds.
A natural iteration alternates this pseudo-feature mean update with reassignment
\[
h_i\gets h_c(\hU_i).
\]
Because assignments are driven by $\hU_i$ while centers average $\tU_i$, this iteration is a computational heuristic rather than classical Lloyd descent.  We therefore use it within a multi-start scheme, evaluate the exact corrected objective \eqref{eq:pseudoweighted} at every candidate, and retain the best visited value.
\end{remark}

\begin{remark}[No pseudo-feature shortcut for nonquadratic losses]\label{rem:nopseudo}
The identity is specific to squared-distance quantization losses.  For nonquadratic losses, including mixture projection losses,
\[
\E\{\ell_\beta(U+HR)\mid X\}
\ne
\ell_\beta(U)
\]
in general.  Thus fitting a mixture directly to pseudo-features is not a valid shortcut for the projection $\beta^\star(K)$.  The valid construction is to correct the loss evaluation $\hphi_{f,i}=\phi_f\bigl(O_i;\hateta^{(-b(i))}\bigr)$ and optimize the resulting corrected criterion.
\end{remark}

\subsubsection{Computation}\label{sec:computation_supp}

\paragraph{Cost structure}
Algorithm~\ref{alg:main} requires one cross-fitted nuisance fit, with cost $B\cdot\mathrm{fit}(\hateta)$.  Conditional on the fitted nuisances, the per-draw work is purely arithmetic: each $W^{(s)}(K)$ is a corrected weighted quantization problem on $\{(\hU_i,\tU_i)\}$, and each $\beta^{(s)}(K)$ is a low-dimensional minimization of the corrected projection criterion.  The total cost is approximately
\[
O\{B\cdot\mathrm{fit}(\hateta)+S\,\overline K^2\,n\,q\,T_{\mathrm{it}}\},
\]
where $T_{\mathrm{it}}$ is the number of quantization iterations.  Draws parallelize directly.

\paragraph{Monte Carlo error in the number of draws}
The band quantile $q_{1-\alpha}$, the scale estimates, and the equal-tailed intervals are computed from $S$ Dirichlet draws. Conditional on the data, the empirical $S$-draw quantiles converge to the exact conditional quantiles as $S\to\infty$, by the conditional Glivenko--Cantelli property together with the continuity of Supplementary Lemma~\ref{lem:condquant}, so the theorems, which concern the exact conditional weighted law, are matched as $S_n\to\infty$. The quantile Monte Carlo error is of order $S^{-1/2}$, about one percentage point at the nominal $95\%$ level for the draw counts used here. The simulation studies use $S=1000$ and the application uses $S=500$.

\paragraph{Quantization with the correct assignment rule}
For each draw and $K$, we minimize $\Psi^{(s)}(g_c)$ using a multi-start Lloyd-type scheme.  We initialize $c$ by weighted $k$-means++ on $\{\tU_i\}$; alternate
\[
h_i\gets \argmin_h\norm{\hU_i-c_h},
\qquad
c_h\gets \Pi_{\calC}\Bigl(\frac{\sum_{i:h_i=h}w_i\tU_i}{\sum_{i:h_i=h}w_i}\Bigr);
\]
evaluate the exact corrected objective \eqref{eq:pseudoidentity} at each candidate; and retain the best visited solution across restarts. The center update is followed by the Euclidean projection $\Pi_{\calC}$ of Remark~\ref{rem:lloyd}, since the weighted pseudo-feature mean can leave $\calC$.  Because assignments are driven by $\hU_i$ while centers average $\tU_i$, the iteration is not guaranteed to be monotone in the corrected objective.  Exact objective evaluation is therefore used as the descent record.

\paragraph{Projection fitting}
For $\beta^{(s)}(K)$, we minimize the corrected criterion $\Psi^{(s)}(\ell_\beta)$ directly.  Weighted EM fits on $\{\hU_i\}$ are useful only as initializers, because the influence-function gradient correction generally breaks the EM minorization.

\paragraph{Label alignment}
Codebooks and mixture components are identified only up to labels.  For componentwise summaries, each draw is aligned to the point estimate by minimum-cost bipartite matching, using Euclidean costs for codebooks and component-location costs for mixture projections.  Under the uniqueness assumption used for Theorem~\ref{thm:projection}, and the analogous codebook uniqueness of Assumption~\ref{ass:quant}\textup{(i)}, this alignment is eventually the identity with probability tending to one. The reason is elementary. Let $\Delta_{\min}$ be the smallest pairwise distance among the population component locations, which is positive under uniqueness with distinct components. Consistency places every fitted location, in the point estimate and in the draw being matched, within $\Delta_{\min}/8$ of its population counterpart on an event of probability tending to one. On that event, compare any candidate matching with the identity edge by edge. An edge pairing locations attached to the same population component costs at most $\Delta_{\min}/4$, while an edge pairing locations attached to different population components costs at least $\Delta_{\min}-\Delta_{\min}/4=3\Delta_{\min}/4$. A nonidentity matching contains at least one displaced edge, and every displaced edge replaces an identity edge, so its total cost strictly exceeds the identity cost. The minimum-cost matching is therefore the identity.  The path, the profiles, and the set-valued reports are invariant to labels and require no alignment.

\paragraph{Numerical safeguards}
We truncate $\hat\pi$ according to Assumption~\ref{ass:bound} and report the fraction of truncated observations.  We monitor the empirical distribution of $\norm{H\hR_i}$, whose scale relative to cell separations governs the practical difference between the correct assignment rule and plain $k$-means on $\tU_i$.  We also recompute $\widehat W(1)$ in closed form, as a weighted variance of $\tU_i$ minus the weighted mean of $\norm{H\hR_i}^2$, as a unit test of the quantization implementation.  As optimization diagnostics we report the best and second-best restart objectives, the restart-to-restart dispersion, the monotonicity of $K\mapsto\widehat W(K)$ enforced by taking running minima, and the frequency of empty or collapsed cells, together with near-nonuniqueness diagnostics, namely the objective gaps among distinct codebooks achieving near-optimal values and the label-matching failure rates across draws.

\subsection{The level-shift identity and the noise-floor diagnostic}\label{sec:noisefloor_supp}

This section makes precise the noise-floor reading of Remark~\ref{rem:noisefloor}. Throughout, the fitted outcome regressions $\hat\bmu$ are a fixed function, treated as an argument, and every conditional expectation $\E\{\cdot\mid X\}$ integrates over $(A,Y)$ at $P_0$ with $X$ and the fit held fixed. Write $\hU(x)=H\hat\bmu(x)$ and $U(x)=H\bmu(x)$.

\begin{proposition}[Level-shift identity]\label{prop:levelshift}
Suppose the design propensities are known and used, so that $\hat\pi_a(x)=\pi_a(x)$ for every arm $a$ and $P_{0,X}$-almost every $x$. Fix a codebook $c\in\calC^K$. Then, for $P_{0,X}$-almost every $x$, with $\Delta(x)=\hU(x)-U(x)$,
\begin{equation}\label{eq:levelshift_supp}
\E\{\phi_{g_c}(O;\hateta)\mid X=x\}
=g_c\{U(x)\}-\norm{\Delta(x)}^2+B_c(x),
\end{equation}
where the boundary term $B_c(x)$ equals zero whenever $U(x)$ and $\hU(x)$ lie in the same Voronoi cell of $c$, and otherwise obeys
\begin{equation*}
0\le B_c(x)\le\sum_{(h,j):c_h\ne c_j}2\,\norm{\Delta(x)}\,\norm{c_h-c_j}\,\ind{E_{hj}(x)},
\qquad
E_{hj}(x)\subseteq\bigl\{\dist(U(x),B_{hj})\le2\norm{\Delta(x)}\bigr\},
\end{equation*}
with $B_{hj}$ the bisector of $c_h,c_j$ and at most one indicator active, exactly as in the proof of Theorem~\ref{thm:eif}\textup{(ii)}. Consequently, taking expectations over $X$ and applying Assumption~\ref{ass:margin},
\begin{equation}\label{eq:levelshift_agg}
\E\{\phi_{g_c}(O;\hateta)\}
=P_U(g_c)-\E\norm{\hU-U}^2+O\bigl(r_\mu^{2(1+\am)/(2+\am)}\bigr),
\end{equation}
uniformly over $c\in\calC^K$ and $K\le\overline K$, where $r_\mu=\max_a\norm{\hat\mu_a-\mu_a}_{P_0,2}$.
\end{proposition}

\begin{proof}
By the corrected-score form \eqref{eq:score} for $g_c$, using $\nabla_mg_c=0$ and $U_{\hateta}(X)=H\hat\bmu(X)=\hU(X)$,
\[
\phi_{g_c}(O;\hateta)=g_c(\hU)+\nabla g_c(\hU)^\top H\,R(O;\hateta).
\]
At known design propensities the identity $\E\{R_a(O;\hateta)\mid X\}=(\pi_a/\hat\pi_a)(\mu_a-\hat\mu_a)(X)$ from the proof of Theorem~\ref{thm:eif}\textup{(ii)} reduces, under $\hat\pi=\pi$, to $\E\{R_a(O;\hateta)\mid X\}=(\mu_a-\hat\mu_a)(X)$, so $\E\{HR\mid X\}=H(\bmu-\hat\bmu)=U-\hU$ and
\[
\E\{\phi_{g_c}(O;\hateta)\mid X\}=g_c(\hU)+\nabla g_c(\hU)^\top(U-\hU).
\]
Apply Lemma~\ref{lem:geometry} with $\bar u=\hU$ and $u=U$, so that $\mathrm{rem}=g_c(U)-g_c(\hU)-\nabla g_c(\hU)^\top(U-\hU)$ and $\Delta=\hU-U$. The right-hand side above equals $g_c(U)-\mathrm{rem}$, and setting $B_c(x)=\norm{\Delta(x)}^2-\mathrm{rem}$ gives \eqref{eq:levelshift_supp}. Part~(a) of the lemma gives $\mathrm{rem}=\norm{\Delta}^2$, hence $B_c=0$, when $U$ and $\hU$ share a cell. Part~(b) gives $\mathrm{rem}=\norm{\Delta}^2+s(U)$ with $-2\norm{\Delta}\,\norm{c_h-c_{h^\star}}\le s(U)\le0$, so $B_c=-s(U)$ lies in the stated range, the containment $E_{hj}\subseteq\{\dist(U,B_{hj})\le2\norm{\Delta}\}$ and the at-most-one-active property being exactly those of the proof of Theorem~\ref{thm:eif}\textup{(ii)}. This is the same algebra as Proposition~\ref{prop:pseudo} carried to the population conditional mean rather than to a single observation. The aggregate bound \eqref{eq:levelshift_agg} then follows by taking expectations over $X$ and applying the margin truncation of that proof to $\E B_c$.
\end{proof}

The level-shift identity propagates to the normalized $\rho$ scale as a multiplicative inflation, which is the reliability content of the second tier of the gate in Remark~\ref{rem:noisefloor}.

\begin{corollary}[Multiplicative inflation of the resolution curve]\label{cor:levelshift_rho}
Write $W(K)=P_U(g_{c^\star(K)})$ for the population path and let $\delta=\E\norm{\hU-U}^2$ be the common floor of Proposition~\ref{prop:levelshift}. Ignoring the boundary terms, the corrected path reads $W_\delta(K)=W(K)-\delta$, and its normalized curve $\rho_\delta(K)=1-W_\delta(K)/W_\delta(1)$ satisfies
\begin{equation}\label{eq:levelshift_rho}
\rho_\delta(K)=\rho(K)\,\frac{W(1)}{W(1)-\delta},\qquad \rho(K)=1-\frac{W(K)}{W(1)} .
\end{equation}
\end{corollary}

\begin{proof}
Substitute $W(K)=W(1)\{1-\rho(K)\}$ into $\rho_\delta(K)=1-\{W(K)-\delta\}/\{W(1)-\delta\}$. The numerator becomes $\{W(1)-\delta\}-\{W(1)(1-\rho(K))-\delta\}=W(1)\,\rho(K)$, so $\rho_\delta(K)=W(1)\,\rho(K)/\{W(1)-\delta\}$.
\end{proof}

Because $W(1)/\{W(1)-\delta\}>1$, a downward level shift inflates every $\rho_\delta(K)$ toward one, so nothing cancels in the normalization and the inflation grows with $\delta/W(1)$, which is what the reliability ratio $\hat r=[\widehat\Delta/2,\widehat\Delta]/\widehat W(1)$ measures. Contrasts $W(K)-W(K')$, by contrast, are exactly floor-invariant up to the boundary terms, since $\delta$ cancels in the difference.

\paragraph{Split-difference estimate of the floor}
The floor $\E\norm{\hU-U}^2$ in \eqref{eq:levelshift_agg} is not directly observable, since $U$ is unknown. A split-difference construction estimates its variance-dominated part. Partition the sample into two arm-stratified halves $S_A$ and $S_B$ of equal size, and fit the outcome regressions separately on each half with the same learner stack, giving $\hat\bmu_A$ and $\hat\bmu_B$. Define
\begin{equation}\label{eq:splitdiff_supp}
\widehat\Delta=\frac12\,\frac1n\sum_{i=1}^n\bigl\|H\{\hat\bmu_A(X_i)-\hat\bmu_B(X_i)\}\bigr\|^2 .
\end{equation}
Condition on the covariates $\{X_i\}$ and on the two training halves, and decompose each half-sample fit into its conditional mean and a stochastic component,
\[
\hat\bmu_m(x)=\{\bmu(x)+b_m(x)\}+v_m(x),\qquad m\in\{A,B\},
\]
where $b_m(x)=\E\{\hat\bmu_m(x)\mid X\}-\bmu(x)$ is the conditional bias given the training half $S_m$ and $v_m(x)=\hat\bmu_m(x)-\E\{\hat\bmu_m(x)\mid X\}$ is its mean-zero fluctuation. Then $\hat\bmu_A-\hat\bmu_B=(v_A-v_B)+(b_A-b_B)$. Because $S_A$ and $S_B$ are disjoint, $v_A$ and $v_B$ are independent and mean zero, so the cross term $\E\{(Hv_A)^\top(Hv_B)\}$ vanishes and
\begin{equation}\label{eq:splitdiff_decomp}
\E\widehat\Delta
=\tfrac12\bigl\{\E\norm{Hv_A}^2+\E\norm{Hv_B}^2\bigr\}
+\tfrac12\,\E\norm{H(b_A-b_B)}^2 .
\end{equation}
Thus $\widehat\Delta$ estimates the variance-dominated part of the half-sample feature error $\E\norm{Hv_m}^2$, inflated by half the squared bias difference $\E\norm{H(b_A-b_B)}^2$, which is small relative to the variance whenever the two halves carry comparable bias.

The passage from the half-sample variance to the full-sample floor requires a scaling assumption on the learner rather than a theorem valid for all learners.

\begin{assumption}[Learner variance scaling]\label{ass:varscaling}
There is an exponent $\gamma\in(0,1]$ and a constant $c_v$, common to the halves and to the full sample, such that the stochastic component of the fitted regression obeys $\E\norm{v_m}^2=c_v\,m^{-\gamma}\{1+o(1)\}$ in the training size $m$.
\end{assumption}

Under Assumption~\ref{ass:varscaling} each half is trained on $n/2$ units while the full-sample fit uses $n$, so the half-sample variance exceeds the full-sample variance by the factor $2^\gamma\in[1,2]$. Ignoring the bias term in \eqref{eq:splitdiff_decomp}, the display \eqref{eq:splitdiff_decomp} has $\E\widehat\Delta$ equal to the half-sample variance $\E\norm{Hv_m}^2$, so the variance part of the full-sample floor equals $2^{-\gamma}\E\widehat\Delta$ in expectation. As $\gamma$ ranges over $(0,1]$ the factor $2^{-\gamma}$ ranges over $[\tfrac12,1)$, the endpoint $\widehat\Delta/2$ at $\gamma=1$ and the endpoint $\widehat\Delta$ as $\gamma\downarrow0$. Hence $[\widehat\Delta/2,\widehat\Delta]$ brackets the variance part of the full-sample noise floor whatever the unknown scaling exponent. Both the bias contribution in \eqref{eq:splitdiff_decomp} and the in-sample evaluation act in the conservative direction. Each half-sample fit is evaluated in \eqref{eq:splitdiff_supp} on its own training units, which rewards in-sample fit and, if anything, inflates $\widehat\Delta$, so the diagnostic overstates the floor rather than understating it.

\paragraph{Scope of the diagnostic}
The construction informs about the level bias of the path only. No coverage claim attaches to $\widehat W(1)+\widehat\Delta$, to the shift-adjusted range $\rho_\delta(\cdot)$ for $\delta\in[\widehat\Delta/2,\widehat\Delta]$ of Corollary~\ref{cor:levelshift_rho}, or to any recentred path. The recommended report applies the two-tier gate of Remark~\ref{rem:noisefloor}, reading detection from whether the band for $\widehat W(1)$ clears zero and reliability of the $\rho$ scale from the reliability ratio $\hat r=[\widehat\Delta/2,\widehat\Delta]/\widehat W(1)$, rather than adding the diagnostic to the path. The sampling variance of $\widehat W(1)$ is a separate quantity, fixed at the efficiency bound by the influence function of a quadratic functional of $P_U$, and no recentring of the path reduces it. Formal inference below the noise floor, where the level bias is comparable to the signal, would require higher-order corrections to that quadratic functional, which we do not pursue here.

\subsection{Linearity principle and energy-scale sensitivity}
\label{sec:app:energy_amplification}
The feature-law posterior uses linearity in the corrected score array as a design rationale.

\begin{remark}[Linearity for risk contrasts]\label{rem:whynot}
    The feature-law posterior \eqref{eq:posteriorprocess} is linear in the corrected-score array.  This matters because the main comparisons in the paper are risk contrasts along the quantization path.  For feature losses $g,g'\in\calG$, the influence function of $\Psi_0(g)-\Psi_0(g')$ is $\{g(U)-g'(U)\} + \{\nabla g(U)-\nabla g'(U)\}^{\top}HR - \{\Psi_0(g)-\Psi_0(g')\}$.
    Thus, the first-order law depends only on the difference between the losses and the difference between their feature-gradient fields. Any component common to the two corrected risks cancels in both $\hPsi(g)-\hPsi(g')$ and every posterior draw $\Psi^{(s)}(g)-\Psi^{(s)}(g')$.
\end{remark}

This subsection records the formal cancellation statement and the narrower energy-scale calculation for exponentiated order comparisons.  The purpose is not to make a general critique of generalized Bayes \cite{bissiri2016} or Bayesian clustering.  The point is specific: when losses are computed from generated causal features, exponentiating risk differences to compare complexities can amplify perturbations that are small on the risk scale.

\begin{proposition}[Exact common-shift cancellation]\label{prop:cancellation}
Let $g,g'\in\calG$ be feature losses.  Under the conditions of Theorem~\ref{thm:eif}, the influence function of the contrast $\Psi_0(g)-\Psi_0(g')$ is
\[
\{g(U)-g'(U)\}
+
\{\nabla g(U)-\nabla g'(U)\}^\top HR
-
\{\Psi_0(g)-\Psi_0(g')\}.
\]
Thus the first-order law of the contrast depends on $(g,g')$ only through their difference and the difference of their feature-gradient fields.  In particular, any generated-feature fluctuation common to the two corrected risks cancels identically in both the one-step contrast $\hPsi(g)-\hPsi(g')$ and every posterior draw $\Psi^{(s)}(g)-\Psi^{(s)}(g')$.
\end{proposition}

\begin{proof}
The result follows by applying Theorem~\ref{thm:eif} to the difference functional $\Psi_0(g)-\Psi_0(g')$.  Since the corrected score is linear in the loss and its gradient,
\[
\phi_g(O;\eta_0)-\phi_{g'}(O;\eta_0)
=
\{g(U)-g'(U)\}
+
\{\nabla g(U)-\nabla g'(U)\}^\top HR .
\]
Centering by $\Psi_0(g)-\Psi_0(g')$ gives the displayed influence function.  The same subtraction is exact in the empirical one-step process and in the weighted process \eqref{eq:posteriorprocess}, because both are linear in the corrected evaluations.
\end{proof}

The resolution path $W(\cdot)$, the causal heterogeneity $R^2$ curve $\rho(\cdot)$, and the resolution profile are built from such risk contrasts.  The feature-law posterior therefore measures uncertainty directly on the corrected risk scale.  By contrast, an exponentiated order posterior compares risks after multiplying them by $n$ and exponentiating.  The next calculation records the resulting scale issue.

Consider a generalized-Bayes posterior \cite{bissiri2016} over models $M\in\calM$ and parameters $\beta\in B_M$,
\begin{equation}\label{eq:gibbs}
\Pi_n(M,d\beta\mid O_{1:n})
\propto
q(M)\exp\{-\lambda_n n\widehat L_n(\beta)\}\Pi_M(d\beta).
\end{equation}
Let
\[
Z_M
=
\int_{B_M}\exp\{-\lambda_n n L_n^U(\beta)\}\Pi_M(d\beta),
\qquad
\widehat Z_M
=
\int_{B_M}\exp\{-\lambda_n n \widehat L_n(\beta)\}\Pi_M(d\beta),
\]
where $L_n^U$ is the same empirical loss computed with oracle features and $\widehat L_n$ is the loss computed with generated features or corrected generated-feature evaluations.

\begin{proposition}[Energy-scale sensitivity of exponentiated order comparisons]\label{prop:fragility}
Let
\[
\Delta_n=\sup_{\beta}|\widehat L_n(\beta)-L_n^U(\beta)|.
\]
Then, for any two models $M,M'$,
\[
\left|
\log
\frac{\widehat Z_M/\widehat Z_{M'}}
     {Z_M/Z_{M'}}
\right|
\le
2\lambda_n n\Delta_n .
\]
Consequently, a perturbation that is negligible on the risk scale need not be negligible on the posterior-odds scale.  In particular, if $\Delta_n\asymp n^{-1/2}$ and $\lambda_n\asymp1$, the bound permits order-$\sqrt n$ perturbations of log odds.

Generated features naturally create perturbations on this scale.  If
\[
\widehat L_n(\beta)=P_n\ell_\beta(\hU),
\qquad
\hU_i=U_i+b_n\xi_i,
\]
then, uniformly over regular finite-dimensional regions of $\beta$,
\[
\widehat L_n(\beta)-L_n^U(\beta)
=
b_nP_n\{\xi^\top\nabla_u\ell_\beta(U)\}+O_p(b_n^2).
\]
When the corresponding score fields differ across competing complexities, the induced log-odds shift is on the $\lambda_n n b_n$ scale.  Influence-function correction removes the deterministic mean shift, but if the corrected losses are subsequently exponentiated, the remaining mean-zero $O_p(n^{-1/2})$ score fluctuation is still multiplied by $n$ in order comparisons.
\end{proposition}

Proposition~\ref{prop:fragility} is a scale calculation for one possible construction, not a statement about all generalized-Bayes procedures and not a criticism of fixed-complexity generalized Bayes.  Temperature choices can be useful for calibrating within-model posterior spread, but they do not by themselves change the fact that order odds operate on an exponentiated energy scale.  The feature-law posterior avoids this issue by keeping uncertainty statements linear in corrected risk contrasts.  This is a supporting reason for the construction; the inferential target of the paper remains the resolution profile of the causal feature law, not an order posterior.

\begin{remark}[Sharpness of the scale bound]\label{rem:energy_sharp}
The bound in Proposition~\ref{prop:fragility} is sharp at the stated scale.  For any $c_n>0$, there exist perturbations with $\Delta_n=c_n$, supported on one model, for which
\[
\widehat Z_M/\widehat Z_{M'}
=
e^{\lambda_n n c_n}Z_M/Z_{M'}.
\]
The proof is given in Supplementary Section~\ref{app:props}.
\end{remark}

\section{Additional theoretical guarantees}\label{sec:app:additional_theory}

This section collects theoretical details that support the main report but are not needed to read the headline profile guarantees in Section~\ref{sec:theory}.  The notation and assumptions are those of the main text unless stated otherwise.

\subsection{Additional details on the feature-law posterior}\label{sec:app:bvm_details}

The proof of Theorem~\ref{thm:bvm}\textup{(ii)} uses the decomposition
\begin{equation}\label{eq:bvmdecomp}
\sqrt n\bigl(\Psi^{(s)}(f)-\hPsi(f)\bigr)
=\underbrace{\frac1{\sqrt n}\sum_i(w_i^{(s)}-1)\,\phi_f(O_i;\eta_0)}_{\text{oracle multiplier process}}
+\underbrace{\frac1{\sqrt n}\sum_i(w_i^{(s)}-1)\,\Delta_i(f)}_{\text{increment term}},
\end{equation}
with $\Delta_i(f)=\hphi_{f,i}-\phi_f(O_i;\eta_0)$.
The first term is a standard exchangeable-bootstrap process; the second is the data-dependent estimated-score increment that must be negligible uniformly over $\calF$. The mode of convergence in \eqref{eq:bvm} is conditional weak convergence in probability (Definition~\ref{def:condweak}).

\begin{remark}[What the theorem does and does not assert]\label{rem:bvmscope}
\eqref{eq:bvm} is conditional convergence in probability, the same mode in which bootstrap validity is ordinarily stated; we make no almost-sure claim. It asserts merging with the efficient limit: the posterior neither inflates for nuisance uncertainty (the correction has removed its first-order effect) nor ignores any first-order term. And it is genuinely uniform: the single statement \eqref{eq:bvm} is what licenses, through the delta method below, simultaneous inference across all $K$, all codebooks, all memberships, and all effect contrasts at once; the form of inference cluster analysis actually requires.
\end{remark}

\begin{remark}[Relation to scalar posterior corrections and the exchangeable bootstrap]\label{rem:bvmcredit}
It is worth recording exactly which ingredients of Theorem~\ref{thm:bvm} are classical and which are new. The limiting engine for the oracle term in \eqref{eq:bvmdecomp} is the exchangeable-bootstrap central limit theorem of \cite{praestgaard1993} over a fixed Donsker class; this part is borrowed, and it is the same engine behind the scalar corrected posteriors of \cite{yiu2025}. The content specific to the present setting is threefold. First, the uniform negligibility of the estimated-score increment, whereby the corrected evaluations $\hphi_{f,i}$ differ from their oracle versions by a data-dependent array, and Lemma~\ref{lem:multmax} with Assumption~\ref{ass:nuisance}(ii) shows the Dirichlet-weighted increment process vanishes uniformly over $\calF$; the step at which cross-fitting, the entropy structure of Assumption~\ref{ass:class}, and the sub-exponential representation of the weights interact. Second, the uniform second-order bias of Theorem~\ref{thm:eif}(ii), which must hold over a class containing the nonsmooth quantization losses; the margin-derived exponent $2(1+\am)/(2+\am)$ has no antecedent in the scalar posterior-correction literature, and relative to the fixed-$K$ margin analysis of corrected quantization risks in \cite{kim2026causal}, the bound here quantifies the boundary mechanism uniformly over codebooks and over $K\le\overline K$. Third, centering exactness, meaning that because $\sum_i(w^{(s)}_i-1)=0$, the posterior is centered at the one-step process $\hPsi$ identically, not merely asymptotically, so no recentering of the scalar kind is needed at any $f$. The theorem is stated for Dirichlet weights; we have not pursued the general exchangeable-weight family, for which the increment-term analysis would have to be redone.
\end{remark}

\subsection{Expanded subgroup-effect scores}\label{sec:app:expanded_effect_scores}

This section records the composite expansion behind Theorem~\ref{thm:effects}.  Fix a working resolution $K$ satisfying Assumption~\ref{ass:proj}, and write $\beta^\star=\beta^\star(K)$.  Let
\[
D_h=\Psi_0(f^D_{h;\beta^\star}),
\qquad
\psi_{h,a}
=
\frac{\Psi_0(f^N_{h,a;\beta^\star})}{D_h}.
\]
The projection influence function is
\[
\Phi_\beta(O)
=
-
V_\beta^{-1}\,
\phi_{\nabla_\beta\ell_{\beta^\star}}(O;\eta_0),
\]
where
\[
V_\beta
=
\nabla_\beta^2\Psi_0(\ell_\beta)\big|_{\beta=\beta^\star}.
\]
The composite influence function for the subgroup mean is
\begin{equation}\label{eq:compositeIF}
\Phi_{h,a}(O)
=
\underbrace{
\frac{1}{D_h}
\Bigl[
\{\phi_{f^N_{h,a;\beta^\star}}(O;\eta_0)-\Psi_0(f^N_{h,a;\beta^\star})\}
-
\psi_{h,a}
\{\phi_{f^D_{h;\beta^\star}}(O;\eta_0)-D_h\}
\Bigr]
}_{\displaystyle \Phi^{\mathrm{fix}}_{h,a}(O)}
+
\{\partial_\beta\psi_{h,a}(\beta^\star)\}^\top
\Phi_\beta(O).
\end{equation}
The first term, $\Phi^{\mathrm{fix}}_{h,a}$, is the influence function for the numerator--denominator ratio when the soft subgroup definition is held fixed at $\beta^\star$.  The second term is the additional contribution from estimating the projection that defines the memberships.  Thus recomputing $\beta^{(s)}(K)$ in every posterior draw is what carries partition uncertainty into the posterior law of $\psi_{h,a}(K)$; fixing $\beta$ at $\hat\beta(K)$ removes the second term and targets only the fixed-partition component.
\begin{remark}[Expanded scores]\label{rem:expandedscores}
Unpacking \eqref{eq:score} at $f^N$ and $f^D$ gives the explicit corrected scores
\begin{align*}
\phi_{f^N_{h,a;\beta}}(O;\eta)&=r_h\{U_\eta(X);\beta\}\,\mu_{\eta,a}(X)\;+\;r_h\{U_\eta(X);\beta\}\,R_a(O;\eta)\;+\;\mu_{\eta,a}(X)\,\nabla_ur_h\{U_\eta(X);\beta\}^\top HR(O;\eta),\\
\phi_{f^D_{h;\beta}}(O;\eta)&=r_h\{U_\eta(X);\beta\}\;+\;\nabla_ur_h\{U_\eta(X);\beta\}^\top HR(O;\eta),
\end{align*}
i.e.\ the within-group AIPW term plus a membership-gradient correction accounting for the group definition's dependence on the estimated features, the term a naive ``cluster, then run AIPW within clusters'' pipeline omits. Hard-cell analogues replace $r_h$ by Voronoi indicators. Unlike the quantization loss, whose two pieces agree to first order at a Voronoi boundary and so contribute the extra $\norm{\Delta}$ factor with bias exponent $2(1+\am)/(2+\am)$, the indicator's two pieces differ by $O(1)$ across the boundary, so its crossing contribution is first order in the feature error. Under the margin condition the resulting bias is of order $r_\mu^{2\am/(2+\am)}$, which is $r_\mu^{2/3}$ at $\am=1$, so root-$n$ control of a hard-cell contrast would require $r_\mu=\opro(n^{-3/4})$, far stronger than Assumption~\ref{ass:nuisance}. This is why the theory treats the soft projections and why the hard-cell displays of Supplementary Section~\ref{sec:app:hillstrom_k3} carry no nominal coverage.
\end{remark}

\subsection{Projection Bernstein--von Mises}\label{sec:app:projection_supp}

The following theorem, referenced from Section~\ref{sec:effectstheory} of the main text, supplies the partition-uncertainty component of the subgroup-effect limits.

\begin{theorem}[Projection Bernstein--von Mises]\label{thm:projection}
Fix $K$ and let Assumptions~\ref{ass:ident}--\ref{ass:class} and \ref{ass:proj} hold, together with the nuisance-rate condition restricted to $\calF_{\mathrm{eff}}$. Work on the local chart around the selected label representative of $\beta^\star(K)$, with all displayed vectors taken after the minimum-cost label alignment used in the algorithm. Let $\hat\beta(K)$ and $\beta^{(s)}(K)$ be measurable minimizers of $\hPsi(\ell_\cdot)$ and $\Psi^{(s)}(\ell_\cdot)$ over $\calB_K$, which exist by the image-admissible Suslin structure of the criterion classes \citep[Section~8.2]{kosorok2008}. Then
\begin{equation*}
\sqrt n\,\{\hat\beta(K)-\beta^\star(K)\}\ \dto\ N\bigl(0,\ V_\beta^{-1}\Sigma_\beta V_\beta^{-1}\bigr),
\quad
\sqrt n\,\{\beta^{(s)}(K)-\hat\beta(K)\}\ \dtow\ N\bigl(0,\ V_\beta^{-1}\Sigma_\beta V_\beta^{-1}\bigr),
\end{equation*}
with $\Sigma_\beta=\Var\{\phi_{\nabla_\beta\ell_{\beta^\star}}(O;\eta_0)\}$, where $\phi_{\nabla_\beta\ell_{\beta^\star}}$ denotes the finite vector obtained by applying \eqref{eq:score} to each coordinate of $\nabla_\beta\ell_{\beta^\star}$. If in addition $\Sigma_\beta$ is nonsingular, posterior credible sets for $\beta^\star(K)$, and for the membership surfaces $u\mapsto r_h(u;\beta^\star(K))$ through a further smooth map, are asymptotically valid on the aligned local chart, and the same holds jointly across the $K$ in any finite set.
\end{theorem}

The proof is provided in Supplementary Section~\ref{app:effects}.

\subsection{Exact finite response classes}\label{sec:exact_classes}

The resolution profile is descriptive rather than latent-class based, but it is backward compatible with genuine finite-class populations.  Suppose the causal feature law has exactly $K_0$ separated support points, $P_U=\sum_{h=1}^{K_0}\omega_h\delta_{u_h}$, $\omega_h\ge\omega_{\min}>0$, $\Delta=\min_{h\ne l}\norm{u_h-u_l}>0$.
This is the case in which individuals fall into exact finite response classes: every covariate profile has one of the $K_0$ causal response profiles $u_1,\ldots,u_{K_0}$.  Such a law is outside the margin-continuous regime used for the Gaussian path theory, but it is the regime in which a classical ``true number of response classes'' is meaningful.

\begin{proposition}[Atomic recovery]\label{prop:atomic}
Let $P_U$ be as displayed with $2\le K_0\le\overline K$, and set $\bar\tau=\omega_{\min}\Delta^2/(2K_0)$.
\begin{enumerate}[label=(\roman*),leftmargin=2em]
\item For every $\tau\in(0,\bar\tau)$ the penalized minimizer is unique up to labels and equals the atom set: $\calS_0(\tau)=\{K_0\}$ with optimal codebook $\{u_1,\dots,u_{K_0}\}$. Moreover $W(K)=0$ for all $K\ge K_0$, every order above $K_0$ is inactive, and $\calS\subseteq\{1,\dots,K_0\}$ with largest element $K_0$; for every cap $\overline K\ge K_0$, the order $K_0$ is optimal for all prices in $(0,\bar\tau)$.
\item $\rho(K)<1=\rho(K_0)$ for every $K<K_0$, so the threshold profile satisfies $K^\star(\gamma)=K_0$ for every $\gamma\in(\rho(K_0-1),1)$.
\item Under Assumptions~\ref{ass:ident}--\ref{ass:class} and $r_\mu\vee r_\pi=\opro(1)$ alone (no margin, no uniqueness, no Hessian, and no rate conditions)
\begin{equation*}
\sup_{K\le\overline K}\,\sup_{c\in\calC^K}\bigl|\hPsi(g_c)-\Psi_0(g_c)\bigr|=\opro(1),
\qquad
\Pp_w\Bigl(\,\sup_{K,c}\bigl|\Psi^{(s)}(g_c)-\Psi_0(g_c)\bigr|>\eps\Bigr)\pto0\quad\text{for every }\eps>0 .
\end{equation*}
Consequently $\Pp\{\widehat K^\dagger(\tau)=K_0\}\to1$ and $\Pp_w\{K^{\dagger(s)}(\tau)=K_0\}\pto1$ for every fixed $\tau\in(0,\bar\tau)$, and likewise $\widehat K^\star(\gamma)$ and $K^{\star(s)}(\gamma)$ recover $K_0$, in the same senses, for every fixed $\gamma\in(\rho(K_0-1),1)$.
\end{enumerate}
\end{proposition}

Thus, when exact response classes exist, they are recovered as the largest active order that remains optimal at sufficiently small positive prices, with centers equal to the class response profiles.  No latent class likelihood is required.  Conversely, when $P_U$ is not exactly finite, the profiles do not impose a spurious true class count; every reported order remains indexed by the resolution or price at which it is optimal.  An extension caveat for mixed atomic-continuous laws is given in Supplementary Section~\ref{sec:atomic_supp_theory}.

The next lemma quantifies the near-mixture backward compatibility asserted in Section~\ref{sec:resolution}.

\begin{lemma}[Near-mixture backward compatibility]\label{lem:nearmixture}
Let $P_U=\sum_{h=1}^{K_0}\omega_h Q_h$ with $2\le K_0\le\overline K$, component means $m_h=\E_{Q_h}U$, weights $\omega_h\ge\omega_{\min}>0$, separation $\Delta=\min_{h\ne l}\norm{m_h-m_l}>0$, and spreads $\E_{Q_h}\norm{U-m_h}^2\le\sigma^2$ for every $h$. Then
\begin{equation*}
W(K_0)\;\le\;\sigma^2,
\qquad
W(K)\;\ge\;\omega_{\min}\bigl(\Delta^2/8-\sigma^2\bigr)\quad\text{for every }K<K_0 .
\end{equation*}
Consequently $1-\rho(K_0)\le\sigma^2/W(1)$ and $1-\rho(K)\ge\omega_{\min}(\Delta^2/8-\sigma^2)/W(1)$ for $K<K_0$. If in addition $\sigma^2<\omega_{\min}\Delta^2/16$, then $K^\star(\gamma)=K_0$ for every
\begin{equation*}
\gamma\in\Bigl(1-\frac{\omega_{\min}\Delta^2}{16\,W(1)},\ 1-\frac{\sigma^2}{W(1)}\Bigr].
\end{equation*}
\end{lemma}

\begin{proof}
For the upper bound, the component means lie in $\calC$ because $\calC$ is convex and contains the feature support, so $c=(m_1,\dots,m_{K_0})$ is an admissible codebook and $W(K_0)\le P_U(g_c)\le\sum_h\omega_h\E_{Q_h}\norm{U-m_h}^2\le\sigma^2$. For the lower bound, fix $c\in\calC^K$ with $K<K_0$ and map each component $h$ to a nearest center index $j(h)\in\argmin_j\norm{m_h-c_j}$, writing $d_h=\norm{m_h-c_{j(h)}}$. By the pigeonhole principle there are $h\ne l$ with $j(h)=j(l)$, whence $d_h+d_l\ge\norm{m_h-m_l}\ge\Delta$, so $d_h\ge\Delta/2$ for at least one of the two. For $u$ in the support of $Q_h$, $\min_j\norm{u-c_j}\ge d_h-\norm{u-m_h}$, and the elementary inequality $(a-b)_+^2\ge a^2/2-b^2$ gives
\begin{equation*}
P_U(g_c)\ \ge\ \omega_h\,\E_{Q_h}\bigl(d_h-\norm{U-m_h}\bigr)_+^2\ \ge\ \omega_h\bigl(d_h^2/2-\sigma^2\bigr)\ \ge\ \omega_{\min}\bigl(\Delta^2/8-\sigma^2\bigr),
\end{equation*}
where the final step uses $\omega_h\ge\omega_{\min}$ and is stated for $\sigma^2\le\Delta^2/8$; when $\sigma^2>\Delta^2/8$ the claimed bound on $W(K)$ is nonpositive and holds trivially. Taking the infimum over $c$ bounds $W(K)$. The displays for $1-\rho$ follow from the definitions. For $\gamma$ in the stated interval, every $K<K_0$ has $\rho(K)\le1-\omega_{\min}(\Delta^2/8-\sigma^2)/W(1)\le1-\omega_{\min}\Delta^2/\{16W(1)\}<\gamma$, using $\Delta^2/8-\sigma^2\ge\Delta^2/16$, while $\rho(K_0)\ge1-\sigma^2/W(1)\ge\gamma$, so $K^\star(\gamma)=K_0$.
\end{proof}

\subsubsection{A caveat for exact finite response classes}
\label{sec:atomic_supp_theory}

Between margin-continuous laws with the full $\sqrt n$ theory and purely atomic laws with assumption-free recovery lies the important mixed case: for example, a zero-effect atom embedded in a continuum of responders. Any atom defeats Assumption~\ref{ass:margin} as stated, since a hyperplane through it carries mass at every distance. Three facts calibrate the damage. First, consistency survives: Proposition~\ref{prop:atomic}(iii) uses neither margin nor uniqueness and applies to every $P_U$, so estimated and per-draw paths, the resolution profile, optional penalized reports, and set-valued summaries remain consistent for any mixed law. Second, the obstruction is local. If no atom lies on an optimal Voronoi boundary, one can impose the hyperplane margin only away from atoms and localize the uniform arguments; we expect the path, profile, and band limits to carry over, with the atom contributing its cell mass, but leave the proof to future work. Third, if an atom lies exactly on an optimal boundary, $W(K)$ is only directionally differentiable at that order \citep{fang2019}, and set-valued reporting is the correct target rather than conservatism.

\subsection{Further theoretical remarks}\label{sec:furtherremarks}

\begin{remark}[In what sense doubly robust]\label{rem:ratedr}
For the smooth and structured classes the bias is the product and squares form $r_\mu^2+r_\mu r_\pi$.  Thus the correction is orthogonal to first-order regression error when the propensity is accurate, but a second-order regression term remains.  Conversely, when $\bar\bmu=\bmu$, the centering is exactly unbiased under any bounded $\bar\pi$, since the residual correction is conditionally mean zero.  The reverse exact robustness does not hold. With $\bar\bmu\ne\bmu$ fixed, no propensity estimator, however accurate, repairs the bias, because the estimand itself is a functional of $\bmu$.  For root-$n$ inference in the smooth and structured classes one needs $\sqrt n(r_\mu^2+r_\mu r_\pi)=o(1)$. For the quantization class the boundary term $r_\mu^{2(1+\am)/(2+\am)}$ is likewise irreparable by propensity accuracy and strengthens the regression-rate requirement to $r_\mu=\opro(n^{-(2+\am)/(4(1+\am))})$, i.e.\ $\opro(n^{-3/8})$ at $\am=1$, strictly stronger than the usual $n^{-1/4}$ (Assumption~\ref{ass:nuisance}).
\end{remark}

\begin{remark}[Growing ceiling]\label{rem:growingK}
All results fix $\overline K$. Allowing $\overline K_n\to\infty$ raises three separate issues, namely entropy of the quantization class growing with $K$ (manageable, as the half-space bound in Lemma~\ref{lem:donsker} is linear in $K$), uniformity of the margin and uniqueness assumptions over $K$, and the vanishing of knot gaps $\rho_0(K+1)-\rho_0(K)$, and it is left to future work. We regard fixed $\overline K$ as scientifically natural, since the ceiling expresses the maximal complexity of summaries the analyst is willing to interpret, not a belief about the population.
\end{remark}

\begin{remark}[Relation to fixed-$K$ causal $k$-means]\label{rem:kkk}
At a fixed and known $K$, Kim et al. \cite{kim2026causal} estimate the causal $k$-means codebook of $P_U$ and establish inference for its centers. That target is the single cross-section of our path at that order. The resolution profile adds the whole path with uniformity over $K$, the feature-law posterior, set-valued profile inference at the resolution thresholds with the matched impossibility result, and the composite subgroup-effect limits of Theorem~\ref{thm:effects}, none of which a single-$K$ codebook analysis provides.
\end{remark}

\section{Additional simulation studies}\label{sec:app:additional_simulations}

\subsection{Data-generating processes and implementation details}\label{sec:sims:dgp_supp}

This section gives the details suppressed from the main simulation narrative, including the bounded-tilt construction that verifies Theorem~\ref{thm:honest} in Study~4 (Supplementary Section~\ref{sec:sims:study4b_supp}) and the DGP-E noise-floor study (Supplementary Section~\ref{sec:sims:study6_supp}).  In DGP-A and DGP-B, covariates are generated as $X\sim\mathrm{Unif}[0,1]^5$ and treatment as
\[
A\mid X\sim\mathrm{Bern}\{e(X)\},
\qquad
e(x)=P(A=1\mid X=x)=\mathrm{expit}(0.6x_1-0.6x_2+0.3x_3-0.15).
\]
Outcomes satisfy
\[
Y=\mu_A(X)+\eps,\qquad \eps\sim N(0,1)\ \text{truncated to }[-3,3],
\]
with
\[
\mu_0(x)=0.5+0.5\sin(2\pi x_3)+0.25x_4-0.25x_5,
\qquad
\mu_1(x)=\mu_0(x)+\tau_0(x).
\]
The feature map is $H=(-1,1)$, so $U(X)=\tau_0(X)$.

The CATE is constructed so that its population law is exactly controlled.  Let
\[
V=\frac{x_1+x_2}{2}.
\]
Then $V$ has triangular distribution function
\[
F_{\mathrm{tri}}(v)
=
\begin{cases}
2v^2, & 0\le v\le 1/2,\\
1-2(1-v)^2, & 1/2<v\le1.
\end{cases}
\]
Thus $F_{\mathrm{tri}}(V)\sim\mathrm{Unif}(0,1)$ exactly.  We set
\[
\tau_0(x)=Q_U\{F_{\mathrm{tri}}(V)\},
\]
where $Q_U$ is the quantile function of the desired target feature law.  For DGP-A and DGP-B this target law is
\[
P_U
=
0.5\cdot\mathrm{Beta}(2,4)\vert_{[-2,-1]}
+
0.3\cdot\mathrm{Beta}(2,2)\vert_{[0,\,0.8]}
+
0.2\cdot\mathrm{Beta}(4,2)\vert_{[0.8+s,\,1.6+s]},
\]
where $\mathrm{Beta}(a,b)\vert_{[l,u]}$ denotes a beta distribution with shape parameters $a,b$, affinely rescaled to support $[l,u]$.  DGP-A uses $s=0.8$.  DGP-B uses $s\in\{0.6,0.3,0.12,0\}$.
Here $s$ is the support gap between the two right bumps; varying it changes their distinguishability while holding all component shapes and weights fixed.
At $s=0$, the two right bumps meet at a common endpoint with continuous density, since both beta densities vanish at the junction.  The law has bounded density, so Assumption~\ref{ass:margin} holds with $\am=1$, but it is not a finite Gaussian mixture because the component shapes are skewed and compactly supported.

Population quantization truth is computed by exact dynamic programming for one-dimensional quantization after discretizing the density on a $4001$-point grid.  Doubling the grid gives relative error below $1.2\times10^{-5}$ at all orders $K\le8$ and below $7\times10^{-6}$ at the reported orders.  Closed forms are used where available.  In DGP-A, the gaps imply that the optimal order-3 codebook is the vector of bump means,
\[
c^\star(3)=(-5/3,\,2/5,\,32/15),
\]
with
\[
W(3)=0.02954,
\qquad
W(1)=2.2945.
\]
The resulting causal heterogeneity $R^2$ values are
\[
\rho_0(2)=0.830,\qquad
\rho_0(3)=0.9871,\qquad
\rho_0(4)=0.9919,\qquad
\rho_0(5)=0.9948.
\]
Hence $K^\star_0(\gamma)=2$ for $\gamma\in\{0.5,0.8\}$, $K^\star_0(\gamma)=3$ for $\gamma\in\{0.9,0.95,0.975\}$, and $K^\star_0(0.99)=4$.  The readings $\gamma=0.975$ and $\gamma=0.99$ are deliberately difficult because they lie within $0.012$ and $0.002$ of knots.  Along DGP-B, $\rho_0(2)$ rises from $0.830$ to $0.917$ as $s$ decreases, crossing $\gamma=0.9$ between $s=0.3$ and $s=0.12$.

Nuisances are cross-fitted with $B=10$ folds.  The flexible outcome regression is a Super Learner ensemble \citep{vanderlaan2007super} fitted separately within each treatment arm on the raw covariates, with no basis expansion, so the arm-specific conditional means $\mu_a(x)=E(Y\mid A=a,X=x)$ are estimated by a per-arm T-learner.  The library holds four base learners, an ordinary least squares linear regression, a ridge regression with $\ell_2$ penalty $0.1$, a random forest of $300$ trees, and a gradient-boosted regression tree of depth two under squared-error loss with learning rate $0.001$, minimum child weight $35$, row subsampling $0.9$, column subsampling $0.85$, an $\ell_2$ leaf penalty $4$, an $\ell_1$ leaf penalty $1$, and minimum split gain $0.04$, whose number of boosting rounds is capped at $3000$ and chosen by early stopping with patience $50$ against a held-out validation fraction of $0.2$.  The base learners are aggregated by the classical Super Learner rule, a convex combiner that assigns nonnegative weights summing to one and minimizes the squared error of the stacked prediction over the probability simplex by projected gradient descent, trained on the held-out predictions of an internal three-fold cross-validation of the library.  The outcome fit uses a single cross-fitting repeat per Monte Carlo replication, and every fitted conditional mean is truncated to $|\hat\mu_a|\le10$.  The propensity score is a correctly specified logistic regression of the treatment indicator on $(1,X)$, fitted by iteratively reweighted least squares with a ridge term $10^{-8}$ for numerical conditioning and at most $50$ Newton iterations, with fitted values truncated to $\hat\pi\in[0.05,0.95]$.

Study~2 varies the nuisance regime.  The correctly specified parametric propensity model is logistic regression.  The misspecified outcome model is linear in $X$, and the misspecified propensity model is the constant $\bar A$.  These regimes distinguish the ordinary product-bias term for smooth scores from the stronger outcome-regression requirement induced by the quantization margin term.

Quantization uses the pseudo-feature identity of Proposition~\ref{prop:pseudo}.  We use weighted $k$-means$++$ initialization on $\tU$, assignments from $\hU$, exact corrected-objective evaluation at every iterate, $20$ restarts for point fits, and $3$ restarts per posterior draw.  The $K=1$ closed form is used as a unit test.  For $q=1$, we also audit against a segment dynamic-programming lower bound in $\hU$-order and double the number of restarts whenever the relative audit gap exceeds $10^{-7}$.  Bands use \eqref{eq:band} and profile sets use \eqref{eq:profileset}.  Projections use Nelder--Mead initialized from a weighted EM fit, and subgroup effects use the structured scores \eqref{eq:effectscores}.  Unless otherwise stated, simulation studies use $S=1000$ posterior draws and ceiling $\overline K=8$, with $R=500$ replications for every study and their supplementary diagnostics.

The atomic supplementary studies use a separate design, DGP-C, with $q=2$.  Four latent response classes are determined by rectangles of $(x_1,x_2)$ with masses
\[
\omega=(0.4,0.3,0.2,0.1).
\]
Within class $h$ the arm means are constants, so with
\[
H=
\begin{pmatrix}
-1 & 1\\
.5 & .5
\end{pmatrix},
\]
the feature, interpreted as effect and level, is purely atomic at
\[
u_1=(0,0.2),\qquad
u_2=(0.5,0.4),\qquad
u_3=(2,1.5),\qquad
u_4=(2.6,1.9).
\]
The treatment probability is $e(x)=P(A=1\mid X=x)=\mathrm{expit}(0.5x_1-0.5x_2)$ and $\eps\sim N(0,1)$ truncated to $[-3,3]$.  The geometry is hierarchical: a close pair $\{u_1,u_2\}$ and a far pair $\{u_3,u_4\}$, with merge scales
\[
\tau_{(3)}=0.0347,\qquad
\tau_{(2)}=0.0497,\qquad
\tau_{(1)}=1.209,
\]
computed by enumerating all $15$ set partitions of the atoms.  Here $\Delta^2=0.29$ and
\[
\bar\tau=\omega_{\min}\Delta^2/(2K_0)=0.003625.
\]
This design violates Assumption~\ref{ass:margin} and the Hessian condition above $K_0$, and is used only for the atomic-recovery and energy-scale diagnostics.

\begin{table}[t]
\centering
\caption{Study~1. Coverage of $K^\star_0(\gamma)$ by $\widehat C(\gamma)$, with mean cardinality in parentheses, $n=4000$ and $R=500$.  The $\gamma=0.9$ column crosses a knot along $s\downarrow0$.  ``Simul.'' is the all-$\gamma$ event.}
\label{tab:sims:coverage_K_C}
\begin{tabular}{l cccccc c}
\toprule
$s$ & $\gamma=.5$ & $.8$ & $.9$ & $.95$ & $.975$ & $.99$ & Simul.\\
\midrule
0.80 & 1.00 (1.0) & 1.00 (1.9) & 1.00 (1.4) & 1.00 (1.5) & 1.00 (3.3) & 1.00 (5.5) & 1.00\\
0.60 & 1.00 (1.0) & 1.00 (1.5) & 1.00 (1.8) & 1.00 (1.7) & 1.00 (3.5) & 1.00 (5.5) & 1.00\\
0.30 & 1.00 (1.0) & 1.00 (1.1) & 1.00 (2.0) & 1.00 (2.3) & 1.00 (3.8) & 0.96 (5.5) & 0.96\\
0.12 & 1.00 (1.0) & 1.00 (1.0) & 1.00 (2.0) & 1.00 (2.6) & 1.00 (4.3) & 0.97 (5.9) & 0.97\\
0.00 & 1.00 (1.0) & 1.00 (1.0) & 1.00 (1.9) & 1.00 (2.7) & 1.00 (4.6) & 0.98 (6.1) & 0.98\\
\bottomrule
\end{tabular}
\end{table}

\subsection{Study 2 coverage tables and per-coordinate audit}\label{sec:sims:study2_supp}

\begin{figure}[t]
\centering
\includegraphics[width=0.92\textwidth]{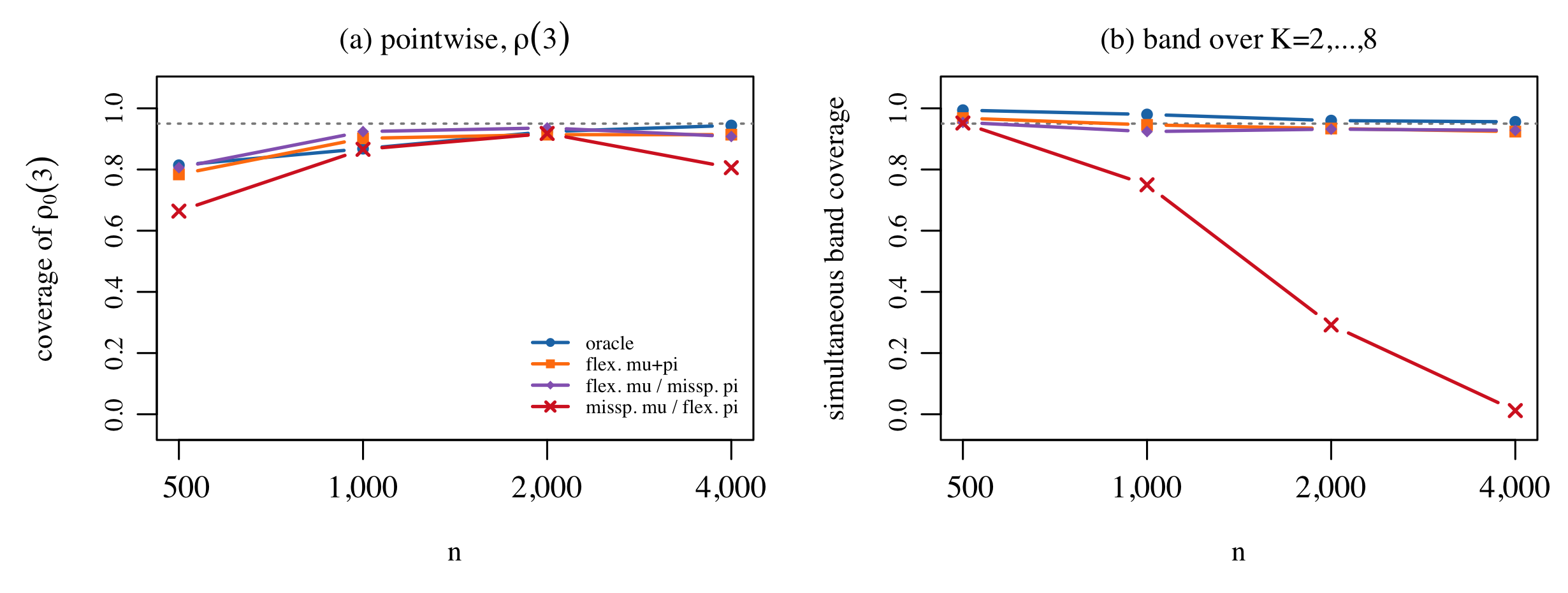}
\caption{Study~2. Empirical coverage of nominal $95\%$ procedures against $n$, by nuisance regime. The figure displays four of the six regimes of Table~\ref{tab:sims:study2cov}, omitting the flexible-outcome parametric-propensity and rate-compliant rows for legibility, with the full six-regime coverage in that table. The corrected posterior calibrates $\rho(3)$ under flexible outcome regression at moderate and large sample sizes. In panel (b) the simultaneous band over $2\le K\le8$ is more demanding, and its dominant feature is the first-order collapse of the misspecified-outcome regime to coverage near $0.01$, while the flexible regimes show only a mild second-order degradation at high resolutions.}
\label{fig:sims:calibration}
\end{figure}

Table~\ref{tab:sims:study2cov} reports the numerical coverage behind Figure~\ref{fig:sims:calibration}, giving, for each nuisance regime and sample size, the empirical coverage of the pointwise $95\%$ interval for the protected coordinate $\rho(3)$ and of the $95\%$ simultaneous band over $2\le K\le8$. A per-coordinate audit at $n=4000$, summarized in prose here, records pointwise interval coverage of $\rho(K)$ for every $K$, under oracle nuisances and under flexible Super Learner nuisances, together with the coverage of the total-heterogeneity coordinate $W(1)$. The oracle column isolates the finite-sample downward bias of empirical minima at fine resolution, which is present with no nuisance error at all. The rate-compliant row is the constructive counterpart described in the main text, a per-arm least squares on the exact five-term structural basis of the design, with the flexible propensity retained. Its outcome regression is correctly specified, so it attains the parametric rate and satisfies Assumption~\ref{ass:nuisance} with room to spare, with out-of-fold $L_2$ errors of $0.070$ at $n=2000$ and $0.030$ at $n=8000$, measured on independent draws of the design as a learner audit, the parametric halving under a quadrupled sample. The flexible nuisances leave $W(1)$ undercovering at $0.59$, the direct footprint of the slow flexible outcome-regression rate on the total-heterogeneity coordinate, even though the protected coordinates $\rho(2)$ and $\rho(3)$ stay near nominal at $0.92$ and $0.91$ because the shift largely offsets in the ratio when it is small relative to the level, the attenuation quantified by Supplementary Corollary~\ref{cor:levelshift_rho}. As Table~\ref{tab:sims:study2cov} and the audit show, the rate-compliant learner restores $W(1)$ to nominal coverage at $0.96$ and improves the accuracy of every coordinate, while pointwise coverage at $K\ge3$ remains low because the honestly shorter intervals expose the same empirical-minimum bias that the oracle row carries. The rate-compliant row sits somewhat below the oracle row at fine resolutions, so the parametric feature noise still contributes a residual boundary term to the empirical minimum, on top of the learner-free component the oracle row isolates. These tables are the audit that the reporting protocol of Supplementary Section~\ref{sec:report} instructs analysts to consult before quoting fine-resolution statements.

\begin{table}[t]
\centering
\caption{Study~2 numerical coverage. Empirical coverage of the nominal $95\%$ pointwise interval for $\rho(3)$ and of the nominal $95\%$ simultaneous band over $2\le K\le8$, by nuisance regime and sample size, DGP-A with $R=500$ replications. Monte Carlo standard errors are at most $0.023$.}
\label{tab:sims:study2cov}
\begin{adjustbox}{max width=\textwidth}
\begin{tabular}{l cc cc cc cc}
\toprule
& \multicolumn{2}{c}{$n=500$} & \multicolumn{2}{c}{$n=1000$} & \multicolumn{2}{c}{$n=2000$} & \multicolumn{2}{c}{$n=4000$}\\
\cmidrule(lr){2-3}\cmidrule(lr){4-5}\cmidrule(lr){6-7}\cmidrule(lr){8-9}
Nuisance regime & $\rho(3)$ & band & $\rho(3)$ & band & $\rho(3)$ & band & $\rho(3)$ & band\\
\midrule
Oracle nuisances & 0.81 & 0.99 & 0.87 & 0.98 & 0.92 & 0.96 & 0.94 & 0.96\\
Flexible outcome and propensity & 0.78 & 0.97 & 0.90 & 0.95 & 0.91 & 0.93 & 0.91 & 0.92\\
Flexible outcome, parametric propensity & 0.81 & 0.97 & 0.90 & 0.94 & 0.92 & 0.93 & 0.91 & 0.90\\
Flexible outcome, misspecified propensity & 0.81 & 0.95 & 0.92 & 0.92 & 0.94 & 0.93 & 0.91 & 0.93\\
Misspecified outcome, flexible propensity & 0.66 & 0.95 & 0.87 & 0.75 & 0.92 & 0.29 & 0.81 & 0.01\\
Rate-compliant outcome, flexible propensity & 0.47 & 0.95 & 0.62 & 0.89 & 0.71 & 0.84 & 0.81 & 0.85\\
\bottomrule
\end{tabular}
\end{adjustbox}
\end{table}

\subsection{Study 4 tilt construction and exact-tilt truth}\label{sec:sims:study4b_supp}

Study~4 samples directly from the bounded density tilts of Theorem~\ref{thm:honest}, using the least favorable construction of Supplementary Section~\ref{app:impossibility}. The base law $P_0$ is DGP-B at the knot separation $s^\star=0.1657$ solving $\rho_0(2;s^\star)=0.9$, found by the same bisection used for ground truth. Because the feature law does not depend on the outcome-error family, the entire feature-law truth $\{W(K),\rho_0(K),c^\star(K)\}$ is that of DGP-B at $s^\star$. The outcome error is the smooth compactly supported family $\eps=6B-3$ with $B\sim\mathrm{Beta}(4,4)$, so $\Var(\eps)=1$, the support is $[-3,3]$, and the density is twice continuously differentiable and vanishes together with its first two derivatives at $\pm3$. A bounded support with a smoothly vanishing boundary density means the tilted laws move no hard edge and remain differentiable in quadratic mean, inside the bounded-outcome assumption.

The tilt score is built from the influence function of $\rho(2)$ at $P_0$,
\[
\phi_{g_c}(o)=g_c\{U\}+g_c'\{U\}\,\mathrm{HR}(o),\qquad g_c(u)=\min_h(u-c_h)^2,
\]
\[
\mathrm{IF}(o)=-\frac1{W(1)}\Bigl[\bigl(\phi_{g_{c^\star(2)}}(o)-W(2)\bigr)-\bigl(1-\rho_0(2)\bigr)\bigl(\phi_{g_{c^\star(1)}}(o)-W(1)\bigr)\Bigr],
\]
where $\mathrm{HR}(o)=(A/\pi)(Y-\mu_1)-\{(1-A)/(1-\pi)\}(Y-\mu_0)$ is the AIPW feature correction evaluated at the population nuisances. The bounded, recentred, normalized score is $s_M=\mathrm{IF}\,\ind{|\mathrm{IF}|\le M}-\E_0[\mathrm{IF}\,\ind{|\mathrm{IF}|\le M}]$ and $s=s_M/\|s_M\|_{P_0,2}$, with truncation $M=5$, so $\E_0 s=0$ and $\E_0 s^2=1$. The recentering constant $c_1=\E_0[\mathrm{IF}\,\ind{|\mathrm{IF}|\le M}]=0.0011$, the norm $c_2=\|s_M\|_{P_0,2}=0.920$, and the achieved drift $b=\E_0[\mathrm{IF}\,s]=0.920$ are precomputed by a fixed-seed Monte Carlo of size $10^7$, and the sup-norm $\bar S=(M+|c_1|)/c_2=5.44$ is the constant $\bar S$ of the theorem. The drift $b$ is essentially the efficiency bound $\sigma_{K_0}=0.924$, the small residual being the truncation loss.

Data from $P_{n,u}$ are drawn by rejection. A proposal $O\sim P_0$ is accepted with probability $(1+u\,n^{-1/2}s(o))/(1+|u|\,n^{-1/2}\bar S)$, which is valid under the guard $|u|\bar S\,n^{-1/2}<1$, satisfied with room to spare in every cell, where the tighter theorem condition $|u|\bar S\,n^{-1/2}\le1/2$ also holds. A built-in unit test confirms $\E_{P_{n,u}}[s(O)]=u\,n^{-1/2}$ to Monte Carlo error at the tightest feasible cell. The drifting truth is reported to first order as $\rho_{P_{n,u}}(2)=0.9+u\,b\,n^{-1/2}$, exact to $o(n^{-1/2})$, the order at which the theorem operates. It is cross-checked by an exact both-channel computation in which the tilt both reweights the $X$ marginal and shifts the feature map to $U_t(x)=\tau_0(x)+t\,\E_0[\mathrm{HR}\,s\mid x]$, after which a large fixed-seed sample is passed to the same one-dimensional quantization dynamic program used for ground truth. The small-tilt central-difference slope of this exact resolution agrees with $b$, checked at setup, validating the whole score build, and each cell records both the first-order and the exact drifting resolution together with $K^\star_{P_{n,u}}(0.9)$. The grid uses $u\in\{-2,-1,-0.5,0,0.5,1,2\}$ and $n\in\{1000,2000,4000,8000\}$, with $R=500$ replications, $S=1000$ draws, ceiling $\overline K=8$, threshold $\gamma=0.9$, and both the flexible Super Learner and the oracle nuisance arms.

\begin{table}[ht!]
\centering
\caption{Study~4. Behavior along the exact bounded tilts $(1+u\,n^{-1/2}s)\,\mathrm dP_0$ at nominal $0.95$. Rows index the tilt magnitude $u$, with the knot law at $u=0$. The first two blocks give the coverage of the drifting truth $K^\star_{P_{n,u}}(0.9)$ by the set-valued report $\widehat C(0.9)$ under the flexible Super Learner and oracle nuisance arms. The last two blocks, under the flexible nuisances, give the correctness rate of the single-valued point selector $\min\{K:\hat\rho(K)\ge0.9\}$ and the mean cardinality of $\widehat C(0.9)$. Each cell uses $R=500$ replications.}
\label{tab:sims:study4b}
\begin{adjustbox}{max width=\textwidth}
\begin{tabular}{rrrrr}
\toprule
$u$ & $n=1000$ & $n=2000$ & $n=4000$ & $n=8000$\\
\midrule
\multicolumn{5}{l}{\textit{coverage of $K^\star_{P_{n,u}}(0.9)$, flex nuisances}}\\
$-2.0$ & 0.998 & 1.000 & 1.000 & 1.000\\
$-1.0$ & 0.994 & 0.996 & 1.000 & 1.000\\
$-0.5$ & 0.994 & 0.996 & 0.988 & 0.994\\
$+0.0$ & 1.000 & 1.000 & 0.998 & 0.998\\
$+0.5$ & 1.000 & 1.000 & 1.000 & 1.000\\
$+1.0$ & 1.000 & 1.000 & 1.000 & 1.000\\
$+2.0$ & 1.000 & 1.000 & 1.000 & 1.000\\
\midrule
\multicolumn{5}{l}{\textit{coverage of $K^\star_{P_{n,u}}(0.9)$, oracle nuisances}}\\
$-2.0$ & 1.000 & 1.000 & 1.000 & 1.000\\
$-1.0$ & 1.000 & 1.000 & 1.000 & 1.000\\
$-0.5$ & 1.000 & 0.998 & 1.000 & 1.000\\
$+0.0$ & 1.000 & 0.998 & 1.000 & 0.998\\
$+0.5$ & 1.000 & 1.000 & 0.998 & 1.000\\
$+1.0$ & 1.000 & 1.000 & 1.000 & 1.000\\
$+2.0$ & 1.000 & 1.000 & 1.000 & 1.000\\
\midrule
\multicolumn{5}{l}{\textit{point selector correct, flex nuisances}}\\
$-2.0$ & 0.612 & 0.734 & 0.852 & 0.850\\
$-1.0$ & 0.434 & 0.524 & 0.614 & 0.646\\
$-0.5$ & 0.308 & 0.452 & 0.502 & 0.524\\
$+0.0$ & 0.758 & 0.680 & 0.624 & 0.622\\
$+0.5$ & 0.848 & 0.792 & 0.776 & 0.766\\
$+1.0$ & 0.910 & 0.850 & 0.840 & 0.850\\
$+2.0$ & 0.962 & 0.946 & 0.944 & 0.942\\
\midrule
\multicolumn{5}{l}{\textit{mean cardinality of $\widehat C(0.9)$, flex nuisances}}\\
$-2.0$ & 2.99 & 2.29 & 2.00 & 1.95\\
$-1.0$ & 2.87 & 2.25 & 2.01 & 1.99\\
$-0.5$ & 2.82 & 2.21 & 1.99 & 1.99\\
$+0.0$ & 2.75 & 2.18 & 1.99 & 1.98\\
$+0.5$ & 2.63 & 2.11 & 1.98 & 1.97\\
$+1.0$ & 2.46 & 2.05 & 1.97 & 1.96\\
$+2.0$ & 2.33 & 1.92 & 1.85 & 1.80\\
\bottomrule
\end{tabular}
\end{adjustbox}
\end{table}

\subsection{Study 4 structural drifting-sequence companion}\label{sec:sims:honest_supp}

The structural companion to Study~4, reported here, drives the DGP-B separation path through the $K^\star(0.9)$ knot at the root-$n$ scale rather than tilting the density. For each cell we solve, by bisection against the same population dynamic program used for ground truth, the separation $s_n$ at which $\rho_{P_n}(2)=0.9+\mathrm{sign}\cdot h\,n^{-1/2}$, to tolerance $10^{-8}$. The map $s\mapsto\rho_0(2;s)$ is strictly decreasing on the bracket and $\rho_0(3;s)$ exceeds $0.98$ throughout, so the root is unique and the drifting truth satisfies $K^\star_{P_n}(0.9)\in\{2,3\}$ in every cell, equal to $2$ when $\rho_{P_n}(2)\ge0.9$ and $3$ otherwise. Positive drifts beyond the $s=0$ resolution ceiling are reached by allowing the two right bumps to overlap, which preserves a valid three-bump feature law. The drifting design matches the root-$n$ magnitude of the profile drift, $|\rho_{P_n}(2)-0.9|=h\,n^{-1/2}$, but it is a structural-parameter path rather than a bounded density tilt, and we do not claim that the laws are contiguous to the knot law. It is therefore an independently constructed illustration of the same near-knot operating characteristics, complementing the exact bounded-tilt verification reported as Study~4 in the main text (Supplementary Section~\ref{sec:sims:study4b_supp}), which exercises the regime of Theorem~\ref{thm:honest} directly. Each cell uses $R=500$ replications, $S=1000$ draws, ceiling $\overline K=8$, and the flexible Super Learner nuisances of the main studies.

Table~\ref{tab:sims:study5_addtional} reports, for every cell, the coverage of the drifting truth $K^\star_{P_n}(0.9)$ by the set-valued report $\widehat C(0.9)$, the correctness rate of the single-valued point selector $\min\{K:\hat\rho(K)\ge0.9\}$, and the mean cardinality of the report. Coverage is at least $0.98$ everywhere, above the nominal $0.95$ of Theorem~\ref{thm:honest}, and the excess is expected rather than conservative in the pejorative sense, because near the knot the report typically contains both admissible counts. The point selector is correct in only $60\%$ to $77\%$ of replications at the exact knot and drops below one half on the negative-drift side, the empirical footprint of Theorem~\ref{thm:impossibility} sharpened by the small feature-noise bias of the flexible learner, while the report concentrates on the two knot-adjacent counts as $n$ grows.

\begin{table}[t]
\centering
\caption{Study~4. Set-valued report along root-$n$ drifting sequences near the knot. Columns index the drift of $\rho_{P_n}(2)$ around the threshold $0.9$, namely the exact knot and $\pm h\,n^{-1/2}$ for $h\in\{0.5,1,2\}$. Blocks report coverage of the drifting truth $K^\star_{P_n}(0.9)$ by $\widehat C(0.9)$, the correctness rate of the single-valued point selector, and the mean cardinality of $\widehat C(0.9)$.}
\label{tab:sims:study5_addtional}
\begin{adjustbox}{max width=\textwidth}
\begin{tabular}{l ccccccc}
\toprule
& knot & $+0.5$ & $-0.5$ & $+1$ & $-1$ & $+2$ & $-2$\\
\midrule
\multicolumn{8}{l}{Coverage of $K^\star_{P_n}(0.9)$ by $\widehat C(0.9)$}\\
$n=1000$ & 1.000 & 1.000 & 0.994 & 1.000 & 0.990 & 1.000 & 0.998\\
$n=2000$ & 1.000 & 1.000 & 0.986 & 1.000 & 0.996 & 1.000 & 1.000\\
$n=4000$ & 0.998 & 1.000 & 0.990 & 1.000 & 1.000 & 1.000 & 1.000\\
$n=8000$ & 1.000 & 1.000 & 0.998 & 1.000 & 1.000 & 1.000 & 1.000\\
\midrule
\multicolumn{8}{l}{Point selector correct}\\
$n=1000$ & 0.77 & 0.87 & 0.33 & 0.91 & 0.44 & 0.97 & 0.69\\
$n=2000$ & 0.67 & 0.79 & 0.48 & 0.89 & 0.58 & 0.97 & 0.81\\
$n=4000$ & 0.62 & 0.75 & 0.54 & 0.88 & 0.63 & 0.96 & 0.86\\
$n=8000$ & 0.60 & 0.74 & 0.53 & 0.87 & 0.70 & 0.97 & 0.91\\
\midrule
\multicolumn{8}{l}{Mean cardinality of $\widehat C(0.9)$}\\
$n=1000$ & 2.71 & 2.62 & 2.74 & 2.79 & 2.68 & 2.59 & 2.54\\
$n=2000$ & 2.12 & 2.12 & 2.14 & 2.09 & 2.12 & 2.03 & 2.10\\
$n=4000$ & 1.99 & 1.98 & 1.99 & 1.96 & 1.99 & 1.83 & 1.95\\
$n=8000$ & 1.99 & 1.98 & 1.99 & 1.92 & 1.99 & 1.78 & 1.92\\
\bottomrule
\end{tabular}
\end{adjustbox}
\end{table}

\subsection{The noise-floor regime study}\label{sec:sims:study6_supp}

This study grounds the noise-floor reading of Remark~\ref{rem:noisefloor} and the level-shift structure of Proposition~\ref{prop:levelshift} in the application regime, where the outcome is weakly predictable and the feature-estimation error rivals the heterogeneity. It uses a three-arm design, DGP-E, calibrated to that regime. The arms $A\in\{0,1,2\}$ have a known flat design propensity $\pi_a=1/3$, so no propensity is estimated. Covariates are $X\sim\mathrm{Unif}[0,1]^5$ and the outcome is $Y=\mu_A(X)+\eps$ with $\eps=\sigma_\eps(6B-3)$, $B\sim\mathrm{Beta}(4,4)$, the smooth compactly supported error of Supplementary Section~\ref{sec:sims:study4b_supp}. The baseline surface is $\mu_0(x)=\sin(2\pi x_3)+0.7x_4-0.7x_5+0.5\sin(2\pi x_2)+0.4(x_1-0.5)$, and $\sigma_\eps=3.70$ is set so that the oracle arm-$0$ outcome $R^2$ equals $0.05$. The $q=2$ feature is the pair of active-arm contrasts $U(x)=(\mu_1-\mu_0,\,\mu_2-\mu_0)=\mathrm{scale}\cdot\upsilon(x)$, where $\upsilon$ is a fixed two-group unit shape. A smooth logistic gate $g(x_1)=\mathrm{expit}\{10(x_1-\tfrac12)\}$ interpolates between two well-separated centers $c_A=(-1.1,-0.7)$ at low $x_1$ and $c_B=(1.1,0.9)$ at high $x_1$, with a small within-group gradient of amplitude $0.45$ in $x_2,\dots,x_5$, so $U$ is a genuine continuous generated feature rather than an atomic law. Since every $W(K)$ scales as $\mathrm{scale}^2$, the resolution profile $\rho(K)=1-W(K)/W(1)$ is scale invariant, with $\rho_0(2)=0.79$ and $K^\star_0(0.5)=2$ fixed by the shape, while the amplitude slides $W(1)$ freely along the signal ladder.

The amplitude is set so that the true heterogeneity is a known multiple of the feature-error floor, $W(1)=\theta F_{n_0}$ for $\theta\in\{0.25,1,4\}$, giving $\mathrm{scale}=\sqrt{\theta F_{n_0}/W_{\mathrm{shape}}(1)}$ with $W_{\mathrm{shape}}$ the once-computed unit-shape path. The floor $F_{n_0}$ is the Monte Carlo feature error of this study's cross-fitted three-arm Super Learner at the reference size $n_0=4000$, estimated at zero effect amplitude. With the amplitude set to zero all three arms share the baseline, so the true feature is $U\equiv0$, and $F_{n_0}=\E_n\|\hU-0\|^2$ averaged over $R_0=20$ replicate datasets equals $0.426$. Zero amplitude isolates the irreducible feature-estimation variance from the outcome noise, which at $R^2\approx0.05$ dominates the amplitude-dependent smoothing bias, so $F_{n_0}$ is a clean amplitude-independent reference, exactly the quantity the ladder is measured against.

Each replication records the corrected $\widehat W(1)$ and its $95\%$ band with the coverage and width of the true $W(1)$, the plug-in $\widehat W(1)$ for contrast, an indicator that the corrected point sits at or below zero, read as a rate-condition failure, a detection indicator that the band excludes zero, the simultaneous coverage of the true $W$-path by its band, and the licensed $\rho$-scale report $\widehat C(0.5)$ with its cardinality and coverage of $K^\star_0(0.5)$, recorded only on the replications where the $W(1)$ band clears zero. The noise-floor diagnostic is the split difference
\[
\widehat\Delta=\tfrac12\,\frac1n\sum_{i=1}^n\bigl\|H\{\hat\bmu_A(X_i)-\hat\bmu_B(X_i)\}\bigr\|^2,
\]
computed from two disjoint arm-stratified half-sample outcome fits, and it is cross-read against the replication's own true feature error at both the full-sample and half-sample scale. The grid is $n\in\{1000,4000,16000,64000\}$ against $\theta\in\{0.25,1,4\}$, with $R=300$ replications, $S=500$ draws, ceiling $\overline K=6$, known $\pi=1/3$, and the AIPW $1/\pi$ residual.

\begin{table}[t]
\centering
\caption{The noise-floor protocol in DGP-E along the signal ladder $W(1)=\theta F_{n_0}$, with $\theta=0.25$ below the floor, $\theta=1$ at the floor, and $\theta=4$ above it. Blocks report coverage of the true $W(1)$ by its $95\%$ band, the detection rate at which the band clears zero, which is also the rate at which the $\rho$-scale report is licensed, and the simultaneous coverage of the true $W$-path by its band. Each cell uses $R=300$ replications.}
\label{tab:sims:study6}
\begin{adjustbox}{max width=\textwidth}
\begin{tabular}{rlrrrr}
\toprule
& $\theta$ & $n=1000$ & $n=4000$ & $n=16000$ & $n=64000$\\
\midrule
\multirow{3}{*}{\textit{W(1) band coverage}}
 & $0.25$ & 1.000 & 0.997 & 0.970 & 0.963\\
 & $1.00$ & 0.997 & 0.997 & 1.000 & 0.840\\
 & $4.00$ & 0.993 & 1.000 & 0.960 & 0.830\\
\midrule
\multirow{3}{*}{\textit{detection (band $>0$)}}
 & $0.25$ & 0.000 & 0.000 & 0.000 & 0.000\\
 & $1.00$ & 0.000 & 0.000 & 0.007 & 0.920\\
 & $4.00$ & 0.000 & 0.060 & 1.000 & 1.000\\
\midrule
\multirow{3}{*}{\textit{rho-report licensed}}
 & $0.25$ & 0.000 & 0.000 & 0.000 & 0.000\\
 & $1.00$ & 0.000 & 0.000 & 0.007 & 0.920\\
 & $4.00$ & 0.000 & 0.060 & 1.000 & 1.000\\
\midrule
\multirow{3}{*}{\textit{W-path simult. coverage}}
 & $0.25$ & 0.997 & 0.987 & 0.967 & 0.963\\
 & $1.00$ & 0.993 & 0.997 & 1.000 & 0.253\\
 & $4.00$ & 0.993 & 0.963 & 0.430 & 0.150\\
\bottomrule
\end{tabular}
\end{adjustbox}
\end{table}

The protocol behaves as designed across the ladder (Table~\ref{tab:sims:study6}). Below the floor, at $\theta=0.25$, there are no false detections at any $n$, the $W(1)$ band covers the truth, and the resolution report is correctly never licensed. Detection then emerges with $n$ as the floor logic predicts, $\theta=1$ becoming detectable only at $n=64000$ at rate $0.92$ and $\theta=4$ from $n=16000$ on. The split-difference diagnostic $\widehat\Delta$ tracks the true half-sample feature error to within $0.82$ to $0.99$ across all twelve cells, reading the floor without access to the truth. At large $n$ with strong signal the $W(1)$ band coverage falls to $0.83$ to $0.84$ at $n=64000$, and the path bands lower still, because the level-shift bias shrinks more slowly than $n^{-1/2}$, the regime Assumption~\ref{ass:nuisance} excludes and the diagnostic flags. In those same cells the licensed $\rho$-scale set report at $\gamma=0.5$ retains coverage $1.00$, but this reflects the favorable knot geometry rather than any cancellation. By the multiplicative identity $\rho_\delta(K)=\rho(K)\,W(1)/\{W(1)-\delta\}$ of Corollary~\ref{cor:levelshift_rho}, a common shift inflates the whole $\rho$ path toward coarser counts, and nothing cancels in the normalization. What survives is the shift-invariance of the $W$-contrasts $W(K)-W(K')$, together with the placement of $\gamma=0.5$ far below the shifted knots, since $\rho_0(2)=0.79$ with no intermediate count, so the inflation cannot evict the true count there. The general lesson is that level statements degrade as the identity predicts, $W$-contrasts are shift-insensitive, and $\rho$-scale sets survive only when the evaluation resolutions are separated from the shifted knots by more than the inflation, which the reliability ratio $\hat r$ measures. Placing $\gamma$ between the shifted and unshifted knots would instead produce undercoverage, and the diagnostic flags the rate violation in every such cell. The formal level-shift analysis is in Supplementary Section~\ref{sec:noisefloor_supp}.

\subsection{Optional penalized-profile diagnostics}\label{sec:sims:penalized_profile_supp}

The optional penalized report uses DGP-A and the DGP-B path described in Supplementary Section~\ref{sec:sims:dgp_supp}.  The two structural merge scales are $\tau_{(1)}=1.904$ and $\tau_{(2)}=0.361$ at $s=0.8$; the next merge scale is $0.011$, a thirty-fold drop.  Along the path, these shrink to $\tau_{(1)}=1.488$ and $\tau_{(2)}=0.105$ at $s=0$.  We evaluate $\widehat C^\dagger(\tau)$ on a twelve-point price grid and report merge-scale intervals on $[0.05,4]$, using both matched posterior-quantile and interval-arithmetic constructions.

The penalized set-valued report mirrors the resolution-profile pattern.  Across the twelve prices, $\widehat C^\dagger(\tau)$ attains simultaneous coverage of $1.000$ at every separation, widening at sub-structural prices and collapsing to a singleton at the largest structural price.  The optional merge-scale diagnostics, reported in Table~\ref{tab:sims:merge}, show the expected distinction between a sharper conditional interval and the unconditional fallback.  For the structural scale $\tau_{(2)}$, the matched posterior-quantile interval is sharp and near nominal across the path.  The fitted pair structure matches the population pair in $0.85$--$0.90$ of posterior draws on average, and matched posterior-quantile coverage ranges from $0.94$ to $0.98$.  The interval-arithmetic intervals require no matching and cover at $1.000$ throughout, at roughly $2.4$ to $2.7$ times the matched length.

\begin{table}[t]
\centering
\caption{Supplementary penalized-profile diagnostic: the structural merge scale $\tau_{(2)}$ along the degeneracy path, $n=4000$ and $R=500$.  Bias and RMSE are for the fitted slope at the population pair.  Coverage and mean length are shown for the matched posterior-quantile interval, with the frequency of the matching event, and for the unconditional band-induced interval-arithmetic interval.  True $\tau_{(2)}=0.361$, $0.183$, $0.105$ at $s=0.8$, $0.3$, $0$.}
\label{tab:sims:merge}
\begin{tabular}{l cc ccc cc}
\toprule
& & & \multicolumn{3}{c}{Posterior quantile} & \multicolumn{2}{c}{Interval arithmetic}\\
\cmidrule{4-6}\cmidrule{7-8}
$s$ & Bias & RMSE & Cov. & Len. & Match & Cov. & Len.\\
\midrule
0.80 & $-0.025$ & 0.054 & 0.94 & 0.18 & 0.89 & 1.00 & 0.43\\
0.30 & $-0.011$ & 0.034 & 0.96 & 0.14 & 0.90 & 1.00 & 0.36\\
0.00 & $-0.002$ & 0.029 & 0.98 & 0.10 & 0.85 & 1.00 & 0.28\\
\bottomrule
\end{tabular}
\end{table}

\subsection{Subgroup-effect numeric audit}\label{sec:sims:own_estimands_supp}

\begin{figure}[t]
\centering
\includegraphics[width=0.92\textwidth]{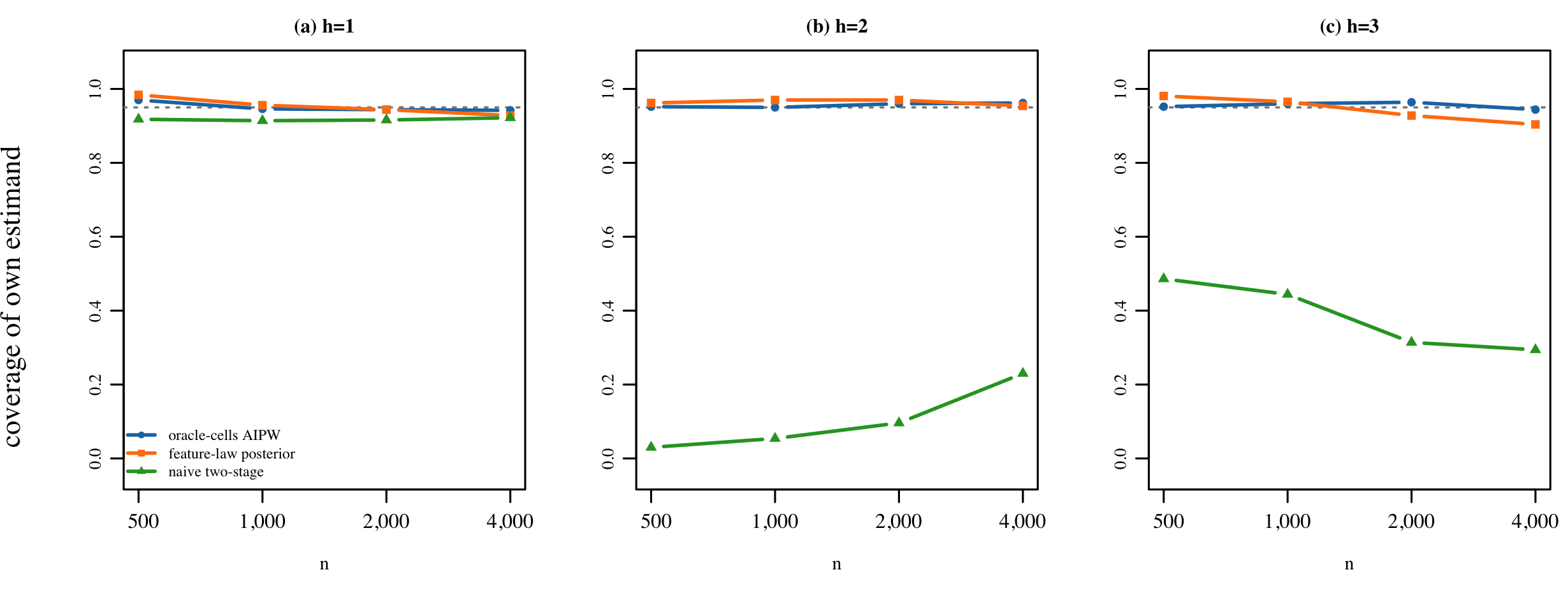}
\caption{Study~3. Coverage of each procedure's own population estimand for the subgroup treatment-effect contrast $\psi_{h,1}(3)-\psi_{h,0}(3)$, by component $h=1,2,3$, against $n$.  The oracle-cell reference is near nominal, isolating partition estimation as the failure channel.  The naive two-stage intervals under-cover sharply for some components, while the feature-law posterior remains substantially better calibrated.}
\label{fig:sims:subgroup_effects}
\end{figure}

Table~\ref{tab:sims:own_estimands} gives the numerical audit behind the middle panel of Figure~\ref{fig:sims:subgroup_effects}.  Each method is evaluated against its own population estimand: the feature-law posterior targets the soft-projection contrast, whereas the two-stage and oracle-cell procedures target hard-cell contrasts.  The naive two-stage intervals under-cover their own middle-component estimand even as $n$ increases, while the feature-law posterior remains substantially closer to nominal coverage in this rate-violating regime.  

\begin{table}[t]
\centering
\caption{Supplementary audit for Study~3 (DGP-A, $K=3$, Super Learner nuisances, $R=500$). Middle-component effect contrast, with each method evaluated against its own population estimand, namely the soft projection for the posterior and hard cells for the two-stage and oracle-cell procedures.}
\label{tab:sims:own_estimands}
\begin{adjustbox}{max width=\textwidth}
\begin{tabular}{l cccc c cccc}
\toprule
& \multicolumn{4}{c}{$n=1000$} & & \multicolumn{4}{c}{$n=4000$}\\
\cmidrule{2-5}\cmidrule{7-10}
Method & Bias & RMSE & Cov. & Len. & & Bias & RMSE & Cov. & Len.\\
\midrule
feature-law posterior (soft) & $-0.12$ & 0.50 & 0.970 & 1.67 & & $-0.11$ & 0.19 & 0.954 & 0.57\\
cluster then estimate (partition fixed) & $-0.51$ & 0.54 & 0.054 & 0.49 & & $-0.17$ & 0.19 & 0.230 & 0.24\\
oracle-cells AIPW & $+0.00$ & 0.13 & 0.950 & 0.47 & & $-0.002$ & 0.06 & 0.962 & 0.23\\
\bottomrule
\end{tabular}
\end{adjustbox}
\end{table}

\subsection{Atomic recovery and order-selection comparators}\label{sec:sims:atomic_supp}

This subsection reports the atomic-recovery diagnostics for Proposition~\ref{prop:atomic} and a head-to-head comparison with standard order selectors applied to the same estimated features. The design is DGP-C with $K_0=4$ atoms, Super Learner nuisances, $n\in\{1000,4000\}$, and $R=500$. Prices are read on the grid $\bar\tau\cdot\{0.1,0.25,0.5,0.75,1.5,3\}$ with $\bar\tau=\omega_{\min}\Delta^2/(2K_0)$.

Two findings support the theory. First, the set-valued penalized report covers the population optimum at every grid price with simultaneous frequency $1.00$ at both sample sizes, with mean set cardinality between six and seven of a possible eight, so honesty is achieved by widening rather than by selection. The point selector $\widehat K^\dagger(\tau)$ recovers $K_0$ in $7\%$ of replications at the smallest price and $86\%$ at the largest at $n=4000$, reflecting the small spurious decrements that feature noise leaves in the corrected path. Second, the wrong-assignment variant, plain $k$-means on the pseudo-features, recovers the population order in $0\%$ of replications at every price, which is the practical content of Remark~\ref{rem:traps}.

Table~\ref{tab:sims:comparators} reports order selectors on the same replications. Applied to the estimated features, the Dirichlet-process mixture posterior mode selects $K_0=4$ in $32\%$ of replications at $n=1000$ but only $11\%$ at $n=4000$, drifting toward six components as $n$ grows because feature-estimation noise scales into the order dimension. Gaussian-mixture BIC degrades from $31\%$ to $0\%$ with mode at eight, while the gap statistic improves from $58\%$ to $96\%$. Applied to lightly jittered oracle features, BIC is essentially exact and the DPM mode still concentrates on two components. No selector behaves stably across methods and sample sizes once the features are estimated, and posterior mass over $K$ is not a substitute for calibrated inference on a population functional. The set-valued report keeps its coverage guarantee in the same runs.

\begin{table}[t]
\centering
\caption{Supplementary order-selection comparison (DGP-C, $K_0=4$, Super Learner features, $R=500$): frequency of selecting $K=4$, with the modal selection in parentheses when it differs from $4$.}
\label{tab:sims:comparators}
\begin{tabular}{l cc}
\toprule
Selector & $n=1000$ & $n=4000$\\
\midrule
DPM posterior mode, estimated features & 0.32 & 0.11 (mode 6)\\
DPM posterior mode, oracle features & 0.00 (mode 2) & 0.00 (mode 2)\\
Gaussian-mixture BIC, estimated features & 0.31 & 0.00 (mode 8)\\
Gaussian-mixture BIC, oracle features & 0.96 & 1.00\\
Gap statistic, estimated features & 0.58 & 0.96\\
$\widehat C^\dagger$ covers $K_0$, simultaneous over prices & 1.00 & 1.00\\
\bottomrule
\end{tabular}
\end{table}

\subsection{Energy-scale sensitivity of an order posterior}\label{sec:sims:fragility_supp}

This supplementary energy-scale study illustrates Proposition~\ref{prop:fragility} in DGP-C, using $n\in\{1000,4000\}$ and $R=500$.  The exponentiated order posterior
\[
\pi_c(K)\propto\exp\{-c\sqrt n\,\widehat W(K)\},
\qquad c\in\{1,5,25\},
\]
is reported through the log odds of the atomic order $4$ against the overfitted order $5$.  Since $W(4)=W(5)=0$ in the population, this is the comparison most exposed to sampling and refit noise.  We also report the information-matched temperature
\[
\hat\lambda=
\frac{\tr(V_{K_0}^{-1})}
{\tr(V_{K_0}^{-1}\Sigma_{K_0}V_{K_0}^{-1})},
\]
which calibrates within-order Gibbs spread but is not designed to stabilize between-order odds.

The sensitivity remains visible against the calibrated comparator, as Figure~\ref{fig:sims:fragility} and Table~\ref{tab:sims:fragility_supp} report.  Under Super Learner nuisances, the across-replication SD of $\log\{\pi(4)/\pi(5)\}$ scales with the energy temperature at both sample sizes: $0.21$, $1.06$, and $5.28$ at $n=1000$, and $0.35$, $1.76$, and $8.78$ at $n=4000$ for $c=1,5,25$.  The fixed-temperature mode of the exponentiated order posterior is order $4$ in only $0.05$ and $0.02$ of replications.  

\begin{figure}[t]
\centering
\includegraphics[width=0.95\textwidth]{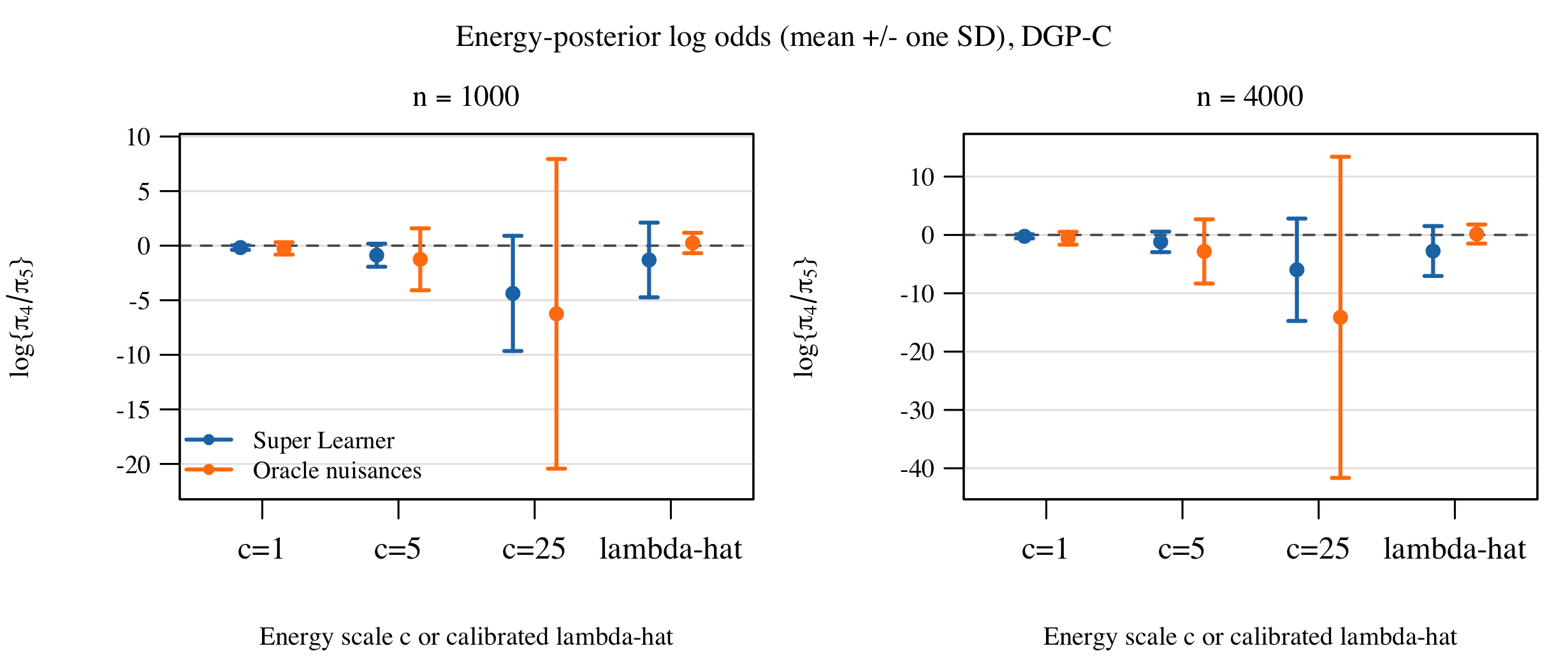}
\caption{Supplementary energy-scale study (DGP-C, $R=500$). Mean and one-SD spread of the log odds of the exponentiated order posterior for order $4$ against order $5$, at fixed temperatures $c\in\{1,5,25\}$ and at the information-matched calibrated temperature $\hat\lambda$.  The population risks satisfy $W(4)=W(5)=0$.  Under Super Learner nuisances the calibrated odds remain highly dispersed and rarely select order $4$, with SD $3.42$ and mode frequency $0.05$ at $n=1000$, and SD $4.28$ and mode frequency $0.02$ at $n=4000$.  Oracle nuisances remove the generated-feature distortion for the calibrated row, where the mode frequency is $1.00$ at both sample sizes.}
\label{fig:sims:fragility}
\end{figure}

\begin{table}[t]
\centering
\caption{Supplementary energy-scale study (DGP-C, $R=500$). Energy-scale sensitivity, reported as the across-replication SD of $\log\{\pi(4)/\pi(5)\}$ with the frequency of posterior mode $=4$ in parentheses.  Calibrated rows use the replications with nonempty reference cells.}
\label{tab:sims:fragility_supp}
\begin{adjustbox}{max width=\textwidth}
\begin{tabular}{l cccc}
\toprule
& \multicolumn{2}{c}{Super Learner} & \multicolumn{2}{c}{Oracle nuisances}\\
\cmidrule{2-3}\cmidrule{4-5}
Method & $n=1000$ & $n=4000$ & $n=1000$ & $n=4000$\\
\midrule
Gibbs $c=1$: SD log-odds (mode $=4$) & 0.21 (0.05) & 0.35 (0.02) & 0.57 (0.36) & 1.10 (0.32)\\
Gibbs $c=5$: SD log-odds (mode $=4$) & 1.06 (0.05) & 1.76 (0.02) & 2.84 (0.36) & 5.51 (0.32)\\
Gibbs $c=25$: SD log-odds (mode $=4$) & 5.28 (0.05) & 8.78 (0.02) & 14.18 (0.36) & 27.53 (0.32)\\
Gibbs calibrated $\hat\lambda$: SD log-odds (mode $=4$) & 3.42 (0.05) & 4.28 (0.02) & 0.93 (1.00) & 1.63 (1.00)\\
\bottomrule
\end{tabular}
\end{adjustbox}
\end{table}

\section{Additional empirical analyses}\label{sec:app:additional_analysis}

\subsection{MineThatData e-mail experiment specification and preprocessing}\label{sec:app:hillstrom_spec}

This section gives the complete empirical specification. The data are the public MineThatData e-mail experiment released by Hillstrom \cite{hillstrom2008} for an open analytics challenge and available without restriction. The $N=64{,}000$ customers were individually randomized in equal thirds to no e-mail (control, $n=21{,}306$), a men's merchandise e-mail ($n=21{,}307$), or a women's merchandise e-mail ($n=21{,}387$). The primary outcome is the indicator of a site visit within two weeks of the campaign. The secondary outcome, analyzed in Supplementary Section~\ref{sec:app:hillstrom_spend}, is two-week spending in dollars.

The causal feature used in the main text is the two-dimensional campaign-benefit profile
\[
U(x)=\bigl(\mu_{\mathrm{men}}(x)-\mu_{\mathrm{none}}(x),\ \mu_{\mathrm{women}}(x)-\mu_{\mathrm{none}}(x)\bigr),
\]
where each coordinate is the covariate-conditional visit-probability contrast of a campaign against control. Because both coordinates live on the common visit-rate scale, the metric $H$ is the identity.

The nuisance adjustment set is the pretreatment customer record. It includes recency in months since the last purchase, the dollar value of purchases over the prior year together with its ordinal history segment, indicators for prior men's and women's merchandise purchases, an indicator for a customer new in the past year, zip-code type entered as suburban and urban indicators with rural as the reference, and acquisition channel entered as phone and web indicators with multichannel as the reference.

The arm-specific outcome regressions are cross-fitted Super Learner ensembles over a library of four base learners, a linear model, a ridge model, a random forest, and a gradient-boosted tree with early stopping. The library predictions are combined by nonnegative least squares on the simplex, the classical Super Learner rule, with the combiner trained on an internal split. Cross-fitted predictions are truncated to the observed outcome range of the corresponding arm before use. The out-of-fold predictions are formed within $B=5$ cross-fitting folds and averaged over $4$ independent repeats. Out-of-fold $R^2$ values are modest, between $0.026$ and $0.029$ across arms, consistent with a weakly predictable individual outcome. The gradient-boosted learner uses early stopping, and across all $60$ fits, one per arm, fold, and repeat, none fell back to a constant predictor, so the learner fit nontrivially within the round cap in every fold. Because assignment was randomized with known flat shares, the correction uses the design propensities, taken as the empirical arm shares $(0.3329,\,0.3329,\,0.3342)$. A check of the estimated propensities against the design propensities returns a mean absolute difference below $2\times10^{-5}$ on assignment probabilities near one third, so the Super Learner propensity route and the design-propensity route coincide. A parametric route runs alongside the primary analysis as a low-complexity benchmark, replacing the Super Learner ensemble by arm-specific linear outcome regressions on the same adjustment set, with the same design propensities, cross-fitting, and truncation, and it is this route that supplies the corroborating $\widehat W(1)=9.27\times10^{-4}$ reported in the main text.

The split-difference floor diagnostic $\widehat\Delta$ of Remark~\ref{rem:noisefloor} is computed from a single arm-stratified half split of the sample at a fixed seed, refitting the same learner stack on each half and evaluating the squared feature difference over all units. The two half-sample fits are read in-sample, which inflates $\widehat\Delta$ in the conservative direction, so the diagnostic overstates rather than understates the floor. The noise-scale ellipse drawn in the feature-law panels is the circle of radius $\sqrt{\widehat\Delta/q}$ in the $q=2$ benefit coordinates, the conservative per-coordinate one-standard-deviation feature-estimation noise on the shared visit-rate scale. The Monte Carlo error of the $S=500$ draws behind every band and set-valued report is discussed in Supplementary Section~\ref{sec:pseudo_supp}. The experiment also records a conversion outcome, an indicator of a purchase, with a base rate near $0.9$ percent. At that base rate the feature-estimation floor dwarfs any plausible causal-feature signal, so the conversion outcome is not analyzed, and its gate reasoning parallels the spending outcome.

All corrected summaries use the feature-law criterion with ceiling $\overline K=6$ and $S=500$ posterior draws. The main-text working summary is the fine three-group projection at the upper end of the supported set $\widehat C(0.80)=\{2,3\}$, and the coarse two-group summary, the count supported at every displayed threshold, is reported alongside it, a display choice within the supported set rather than a selector crossing. For the per-draw subgroup summaries, each draw's centers are matched to the point-estimate centers by minimum-cost label alignment, and in this analysis all $500$ draws produced finite summaries for every group and contrast at both displayed resolutions.

The corrected path minimizations are initialized from the plug-in codebook, and the approximate-minimization diagnostics of Proposition~\ref{prop:approxmin} were recorded at $K=3$. A cold multistart attains a best corrected objective of $2.32\times10^{-5}$ there, a collapsed-center configuration that ties the $K=2$ optimum, whereas the plug-in-codebook warm start attains $-2.79\times10^{-4}$, the genuine three-center optimum. The cross-warm pass, in which each resolution is also restarted from the neighboring resolutions' solutions, is deterministic and only lowers the attained values, and the final path is monotone in $K$ by construction. This is precisely the attained-value gap that Proposition~\ref{prop:approxmin} conditions on, and the warm-start value is weakly smaller than the best cold value at every $K$. These are the corrected quantization-path minimizations behind the resolution profile. The soft mixture projection that supplies the reported subgroup summaries is a separate fit, specified in Supplementary Section~\ref{sec:app:hillstrom_k3}.

\subsection{Soft mixture projection and the coarse-to-fine relationship}\label{sec:app:hillstrom_k3}

The subgroup summaries reported in the main text are corrected soft mixture projections of the estimated feature law onto $K$-component Gaussian mixtures, the estimands of Section~\ref{sec:soft} whose inference is Theorem~\ref{thm:effects}. This section gives the full specification. The family is the mixtures $m_\beta(u)=\sum_{h}\omega_h\,\phi(u;\mu_h,\sigma^2 I_2)$ on the two-dimensional benefit plane, sharing one spherical scale $\sigma^2 I_2$, with mixing weights $\omega_h$ floored at $0.01$ so every component stays identifiable and means $\mu_h$ confined to the range of the estimated features inflated by $20\%$. The scale is a resolution constant of the description rather than a fitted quantity. The family is declared over $\sigma\in[0.008,\,0.05]$ and the reported fit sits at the lower endpoint $\sigma=0.008$, the feature-noise floor, near half the per-coordinate estimation-noise standard deviation $\sqrt{\widehat\Delta/q}=0.017$. The corrected projection criterion improves monotonically as $\sigma$ decreases across the declared range, with no interior minimizer, so the data offer no interior scale to select. The sigma profile, the criterion re-optimized at fixed $\sigma$ over the grid from $0.004$ to $0.030$, is monotone with no interior turn at both $K=2$ and $K=3$. Per-draw refits leave the scale at the floor across the central $95\%$ of draws at $K=3$, while at $K=2$ the draw-level scale moves slightly above the floor in a minority of draws, its posterior $97.5\%$ quantile reaching $0.011$.

The point projection minimizes the corrected criterion by a two-stage search. Ten weighted EM restarts on the estimated features supply candidate configurations, and each candidate, together with the plug-in codebook and boundary-hugging starts, is polished by Nelder-Mead applied directly to the corrected criterion. The corrected minimizer is compared against the plug-in Kullback--Leibler projection, the weighted-EM fit that omits the correction. The two differ materially here, the corrected scale binding at the declared floor and the corrected shares moving by about $0.09$, so the reported fit is the corrected minimizer. The $S=500$ posterior draws minimize the weighted corrected criterion by Nelder-Mead warm-started at the point fit, each draw run to the simplex tolerance under a ceiling of $3000$ function evaluations, with a single retry from the weighted EM initializer whenever a draw reaches the ceiling without meeting the tolerance or a component's weighted soft mass falls below $0.005$. A retried draw is retained and flagged rather than dropped, so no draw is discarded. At $K=2$ no draw was retried or flagged, and at $K=3$ $14$ of the $500$ draws were retried and $10$ remain flagged. The flagged draws are those in which the small strong-responder component's weighted soft mass nearly vanishes, so the ratio defining $\psi_{h,a}(K)$ in \eqref{eq:effects} becomes unstable, the empirical counterpart of the positivity requirement $\Psi_0(f^D_{h;\beta^\star(K)})>0$ of Assumption~\ref{ass:proj}. Each draw's components are aligned to the point fit by minimum-cost matching on the means, all $500$ draws produced finite summaries at both resolutions, and the equal-tailed quantiles of the aligned draws are the reported intervals. A confirmation audit re-optimizes $50$ of the posterior draws at a doubled $6000$-evaluation budget. At $K=2$ the reported and re-optimized draws agree to machine precision, with zero share, effect, and objective deviation. At $K=3$ the deviations are negligible, a maximum share deviation of $1.2\times10^{-4}$, a maximum effect deviation of $6.4\times10^{-5}$ visit-rate points, and a mean objective gain of $4\times10^{-9}$, so the reported ensemble is converged at both resolutions. The flagged low-mass draws are what widen the small strong-responder component's share and women's-campaign intervals, consistent with the confidence set placing the fine three-group reading at the supported upper end of $\widehat C(0.80)=\{2,3\}$ while the coarse two-component projection stays supported at every displayed threshold. Figure~\ref{fig:hillstrom:working_summary_K2} shows the coarse $K=2$ soft summary.

\begin{figure}[!t]
\centering
\includegraphics[width=\textwidth]{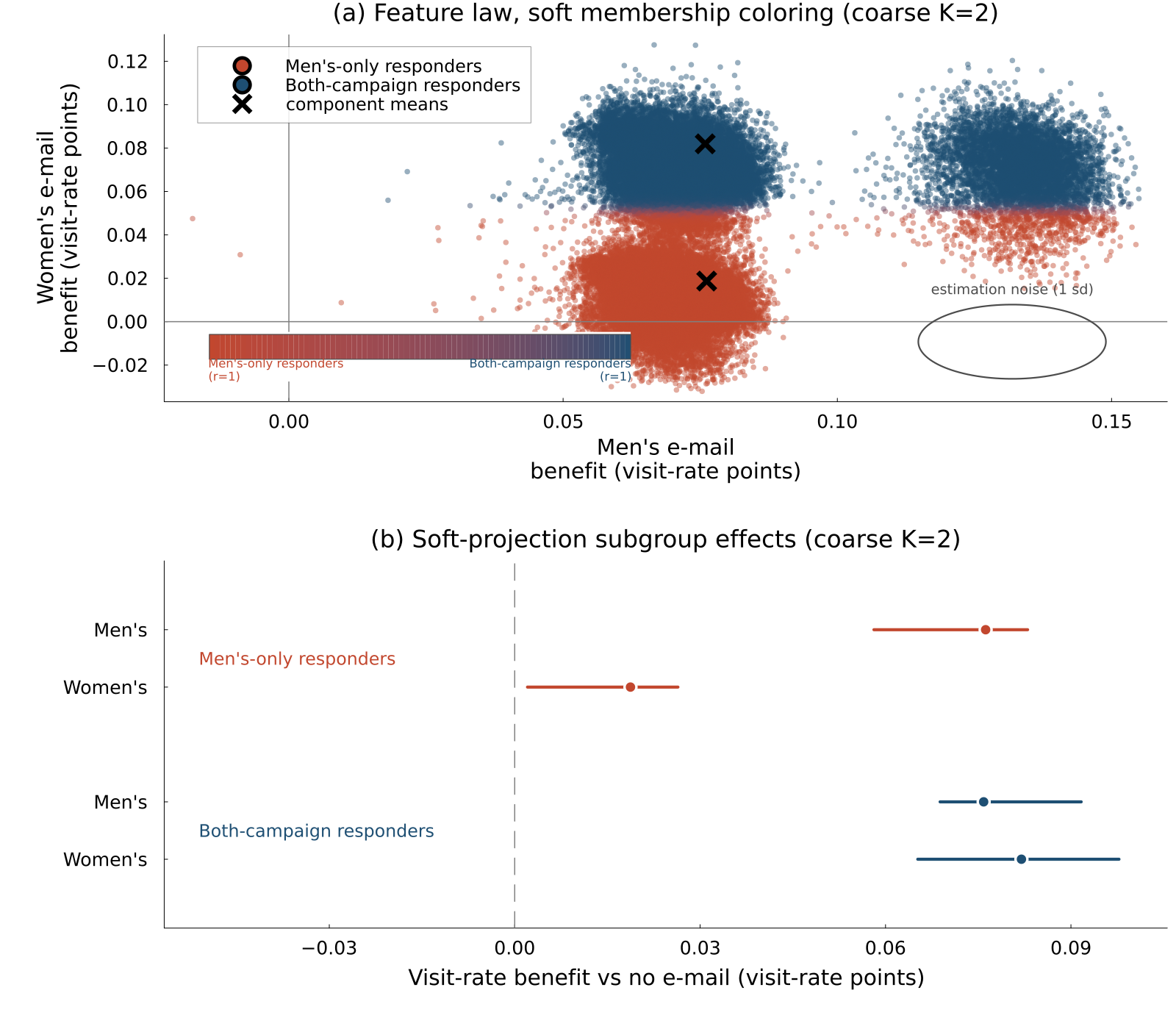}
\caption{MineThatData e-mail experiment coarse $K=2$ soft-projection working summary for the visit outcome. (a) Estimated causal feature law with soft membership coloring, each point a cross-fitted campaign-benefit pair colored by the membership-weighted blend of the two component colors, the weights being that unit's soft memberships $r_1,r_2$, so intermediate hues mark soft membership. Blends are almost absent because the memberships are near binary, $98.3\%$ of units placing weight above $0.9$ on one component. The two component means are marked by black crosses, the inset gradient bar is the membership scale running from one component to the other, and the ellipse, of radius $\sqrt{\widehat\Delta/q}$ per coordinate, marks the per-coordinate estimation-noise scale implied by $\widehat\Delta$. (b) Soft-projection subgroup effects, shown as a forest plot. Rows are the contrasts $\hat\psi_{h,a}(K)-\hat\psi_{h,0}(K)$ of \eqref{eq:effects} for arm $a$ the men's or women's campaign, the corrected visit-rate benefit against control within component $h$, with the rows grouped and colored by component and named by campaign on the vertical axis. The point marks the estimate and the single horizontal segment is the $95\%$ equal-tailed posterior interval from the per-draw recomputation of the projection. The vertical dashed line marks zero and the horizontal axis is in visit-rate points. The components are the men's-only responders, of soft share $0.579$, and the both-campaign responders, of soft share $0.421$.}
\label{fig:hillstrom:working_summary_K2}
\end{figure}

At $K=3$ the projection places component means at $(0.131,\,0.071)$ for the strong responders, $(0.070,\,0.018)$ for the men's-only responders, and $(0.068,\,0.080)$ for the both-campaign responders, in men's and women's benefit, with soft shares $0.108$, $0.545$, and $0.347$. The coarse $K=2$ projection places means at $(0.076,\,0.019)$ for the men's-only component and $(0.076,\,0.082)$ for the both-campaign component, with soft shares $0.579$ and $0.421$, spreading the strong-responder mass of the $K=3$ fit across the two components, mostly into the both-campaign one. The three $K=3$ components are well separated in Figure~\ref{fig:hillstrom:working_summary}(a), with between-component gaps several times the per-coordinate estimation-noise scale marked by the ellipse, which is why the confidence set $\widehat C(0.80)=\{2,3\}$ admits the third component at fine resolution while the coarse two-component summary stays supported at every displayed threshold. The soft memberships are near binary at both resolutions, $98.3\%$ of units above weight $0.9$ at $K=2$ and $98.9\%$ at $K=3$, so the membership coloring of the feature-law panels shows almost no blended hues.

A scale-sensitivity exhibit re-optimizes the criterion at the smaller scale $\sigma=0.004$, below the declared floor. On this outcome that fit is noise-floor degenerate. The feature-space residual correction is large relative to the feature spread, so at $\sigma=0.004$ the criterion rewards ever-smaller scales and its minimizer tracks estimation noise rather than the feature law, the masses drifting under a negligible objective change. Declaring $\sigma$ at the noise floor $0.008$ removes this degeneracy, and the monotone sigma profile confirms there is no interior scale to select.

\subsection{Spending outcome and the in-trial noise-floor contrast}\label{sec:app:hillstrom_spend}

The two-week spending outcome, in dollars, illustrates the noise-floor gate within the same trial. The arm mean spends are $0.65$ under control, $1.42$ under the men's e-mail, and $1.08$ under the women's e-mail, and the augmented inverse-probability arm contrasts against control are $+0.77$ dollars, interval $(0.50,\,1.04)$, for the men's e-mail and $+0.43$, interval $(0.18,\,0.66)$, for the women's e-mail, so the campaigns raise spending on average. The heterogeneity of that response, however, is not resolvable. On the spending scale the corrected total heterogeneity is $\widehat W(1)=-0.14$ with $95\%$ simultaneous band $(-0.66,\,0.37)$, which straddles zero, while the split-difference diagnostic returns $\widehat\Delta=1.55$ with bracket $[0.77,\,1.55]$. The band for $\widehat W(1)$ does not clear zero, so the first tier of the gate of Remark~\ref{rem:noisefloor} fails and the analysis stops there. The data cannot support resolution analysis for this outcome at this sample size, and the corrected resolution sets remain the full menu $\{2,3,4,5,6\}$ at every $\gamma$ with the $\rho$ scale not licensed. That the point $\widehat W(1)$ also sits below the bracket $[0.77,\,1.55]$ is secondary descriptive evidence to the same effect. The dollar outcome is far noisier per unit than the visit indicator, and its causal-feature error swamps whatever spending heterogeneity is present. The contrast with the visit analysis, drawn from the identical randomization and pipeline, shows the diagnostic separating an informative outcome from an uninformative one rather than reflecting any property of the design.

\subsection{Control-anchored feature variant}\label{sec:app:hillstrom_anchored}

The feature map is part of the estimand, as Section~\ref{sec:setup} declares, and a different but natural feature reads the same experiment at a different resolution. The control-anchored feature appends the baseline visit propensity to the two campaign benefits,
\[
U(x)=\bigl(\mu_{\mathrm{none}}(x),\ \mu_{\mathrm{men}}(x)-\mu_{\mathrm{none}}(x),\ \mu_{\mathrm{women}}(x)-\mu_{\mathrm{none}}(x)\bigr),
\]
a $q=3$ feature that grades customers by both their untreated visit rate and their campaign benefits. Under this feature the corrected total heterogeneity is larger and clearly positive, $\widehat W(1)=3.68\times10^{-3}$ with band $(2.63,\,4.73)\times10^{-3}$. The band clears zero, so the first tier of the gate of Remark~\ref{rem:noisefloor} is met, and with $\widehat\Delta=7.43\times10^{-4}$, bracket $[3.71,\,7.43]\times10^{-4}$, the reliability ratio $\hat r=[\widehat\Delta/2,\widehat\Delta]/\widehat W(1)\approx[0.10,\,0.20]$ is small, so the $\rho$ scale here carries its nominal reading. The corrected path stays positive across the whole profile, and the causal heterogeneity $R^2$ climbs gradually, $\hat\rho(2)=0.52$, $\hat\rho(3)=0.65$, $\hat\rho(4)=0.82$, $\hat\rho(5)=0.92$, and $\hat\rho(6)=0.94$, so the anchored feature supports finer resolution than the benefit-only feature. Figure~\ref{fig:app:hillstrom:anchored} shows the corrected quantization path, the causal heterogeneity $R^2$ with the shift-sensitivity overlay, and the resolution profile for this anchored feature. The set-valued reports are correspondingly finer and wider, $\widehat C(0.25)=\{2\}$, $\widehat C(0.50)=\{2,3,4\}$, $\widehat C(0.70)=\{3,4,5\}$, $\widehat C(0.80)=\{3,4,5,6\}$, and $\widehat C(0.90)=\{4,5,6\}$, the last two touching the ceiling. Because these fine sets run to the ceiling $\overline K=6$ and the anchored analysis is conducted only at this ceiling, by the main text's own guidance the fine anchored readings at $\gamma=0.80$ and $0.90$ should be treated as ceiling-limited, since a larger ceiling could extend them further. The added resolution has a clear source and a clear cost. The baseline coordinate mixes prognostic variation in the untreated visit rate with causal variation in the campaign benefits, so a portion of the anchored heterogeneity reflects who visits absent any e-mail rather than who responds to it. Because our target is the heterogeneity of the causal response, we keep the benefit-contrast feature as the primary analysis and report the anchored feature as a variant, which is precisely the sense in which the choice of feature map is a modeling decision internal to the estimand.

\begin{figure}[!t]
\centering
\includegraphics[width=0.92\textwidth]{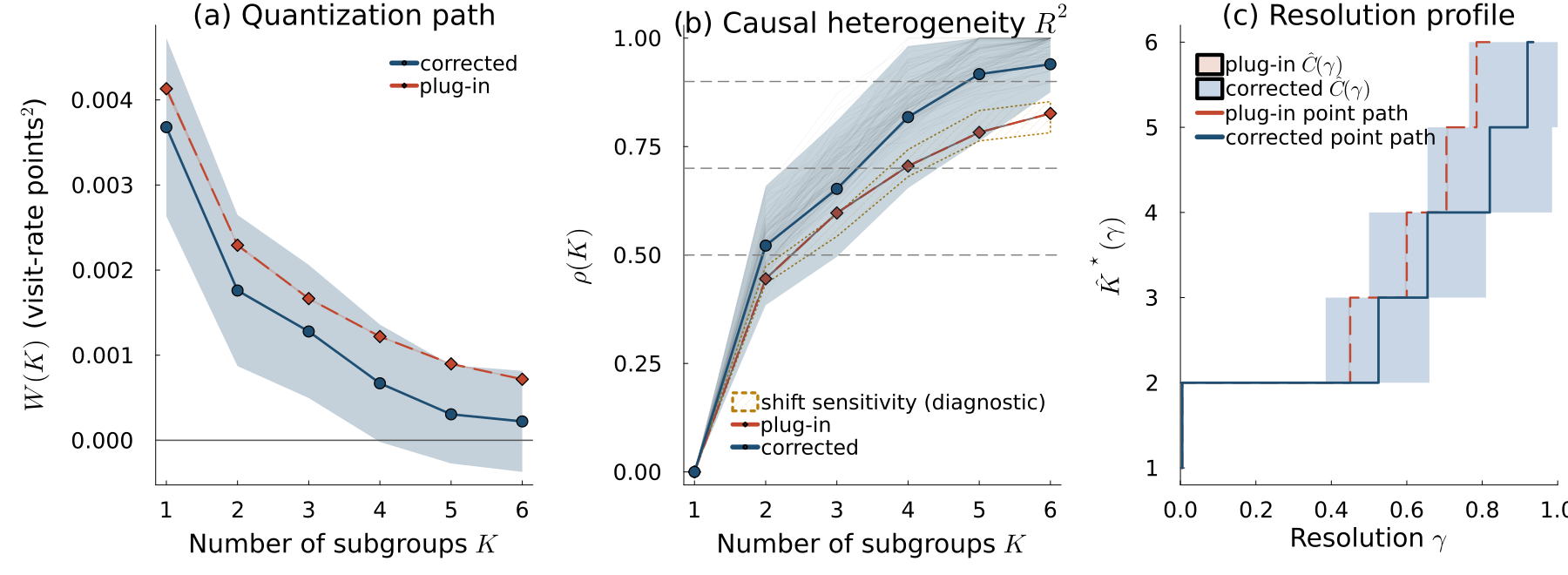}
\caption{MineThatData e-mail experiment control-anchored feature-law resolution summaries for the visit outcome, under the $q=3$ feature that appends the baseline visit propensity to the two campaign benefits. (a) Quantization path. Corrected causal-feature dispersion $\widehat W(K)$ with its $95\%$ simultaneous band, alongside the uncorrected plug-in path. (b) Causal heterogeneity $R^2$. Corrected $\hat\rho(K)$ with its $95\%$ band and the plug-in curve, with the gray fan showing the corrected and plug-in posterior draws at low opacity. The hatched overlay is the shift-sensitivity range $\rho_{\mathrm{adj}}(K;\delta)=1-\max\{\widehat W(K)+\delta,0\}/(\widehat W(1)+\delta)$ for $\delta\in[\widehat\Delta/2,\widehat\Delta]$, a diagnostic with no coverage claim attached rather than a confidence band. (c) Resolution profile. Point selector $\widehat K^\star(\gamma)$ with the set-valued report $\widehat C(\gamma)$ shaded, for the corrected and plug-in paths. The fine anchored sets run up to the ceiling $\overline K=6$.}
\label{fig:app:hillstrom:anchored}
\end{figure}

\subsection{Penalized profile}\label{sec:app:hillstrom_penalized}

The penalized profile gives a price-indexed reading of the same corrected quantization path used in the main application. It asks how many subgroups are worth reporting when one additional subgroup must reduce residual causal-feature dispersion by at least $\tau$ on the squared visit-rate scale, and unlike the causal heterogeneity $R^2$ curve it does not divide by $W(1)$.

\begin{figure}[!t]
\centering
\includegraphics[width=0.92\textwidth]{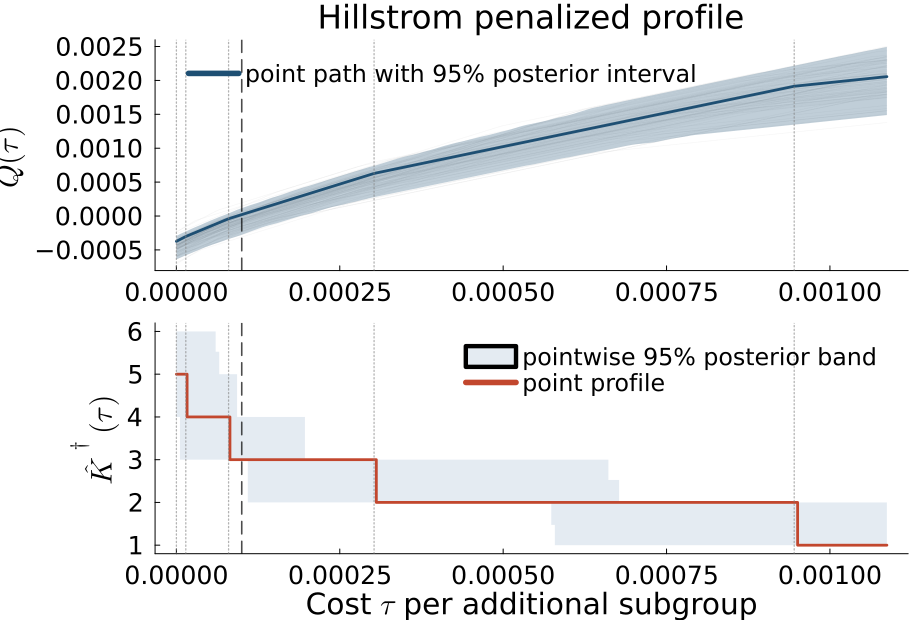}
\caption{MineThatData e-mail experiment penalized profile for the visit outcome under the corrected feature-law criterion. The top panel shows the point path of the penalized value $Q(\tau)$ with a pointwise $95\%$ posterior interval from the weighted path draws overlaid at low opacity. The bottom panel shows the point selector $\widehat K^\dagger(\tau)$ with a pointwise $95\%$ posterior band. Dotted vertical lines mark the fitted merge scales and the dashed line marks the reference price $\tau=\delta^2/4=1\times10^{-4}$ at $\delta=0.02$ visit-rate points.}
\label{fig:app:hillstrom:penalized}
\end{figure}
 Figure~\ref{fig:app:hillstrom:penalized} displays the price-indexed path. The fitted merge scales for the successive splits are about $9.5\times10^{-4}$ for one to two groups, $3.0\times10^{-4}$ for two to three, $8.0\times10^{-5}$ for three to four, and $1.4\times10^{-5}$ for four to five, so the first split dominates the path. At the reference price $\tau=\delta^2/4=1\times10^{-4}$, a between-group visit-rate gap of $\delta=0.02$ for two equally prevalent groups, the point selector is $\widehat K^\dagger(\tau)=3$ and the simultaneous set is $\{1,2,3,4,5,6\}$, so at that price the dial is compatible with anything from a single group to the ceiling. At smaller prices the selector rises to $\widehat K^\dagger=4$ and the simultaneous set retains at least two groups, $\widehat C^\dagger(\tau)=\{2,3,4,5,6\}$ at $\tau=2.5\times10^{-5}$ and $5\times10^{-5}$. The penalized dial thus keeps three to four groups at the small prices set by the finer merge scales and admits a single group from the price $1\times10^{-4}$ onward, already below the fitted merge scale for the split from two to three groups, and excludes it only at the two smallest displayed prices, a reading consistent with the coarse resolvable structure found on the $\rho$ scale.

\subsection{A below-floor benchmark. Project STAR}\label{sec:app:star_null}

This section complements the main application with a below-floor benchmark. It exercises the recommended behavior of the pipeline on a trial in which the total causal-feature heterogeneity $W(1)$ is statistically indistinguishable from zero, the regime flagged by the caveat in Section~\ref{sec:resolution}. The data are the kindergarten cohort of the Tennessee Student/Teacher Achievement Ratio (STAR) experiment \citep{finn1990}, in which entering students were randomized within $79$ schools to three class types, a regular class ($n=1{,}999$, taken as the reference), a small class ($n=1{,}731$), and a regular class with a full-time teacher's aide ($n=2{,}035$). Restricting to kindergarten entrants with complete kindergarten records leaves an analysis sample of $n=5{,}765$ children. The outcome is the log of the total kindergarten score, the sum of the scaled reading and mathematics scores, $\log(\text{read}+\text{math})$. The arm means of the log score are $6.819$, $6.834$, and $6.820$ for the regular, small, and aide classes.

The causal feature is the $q=2$ vector of class-type benefits $U(x)=(\mu_{\mathrm{small}}(x)-\mu_{\mathrm{reg}}(x),\,\mu_{\mathrm{aide}}(x)-\mu_{\mathrm{reg}}(x))$, the covariate-conditional log-score advantages of the small class and of the aide class over the regular class. The adjustment set holds gender, ethnicity indicators, birth year and quarter, free lunch eligibility, indicators for inner city, suburban, and urban schools with rural as the reference, and the school system identifier. The pipeline follows the main application, differing in that the outcome regressions are averaged over $8$ repeats rather than $4$ and the primary route uses an estimated Super Learner propensity rather than a design propensity. The arm-specific outcome regressions are cross-fitted Super Learner ensembles over a library of linear, ridge, random forest, and gradient-boosted learners with early stopping, combined by nonnegative least squares on the simplex, formed within $B=5$ folds and averaged over $8$ repeats. All corrected summaries use ceiling $\overline K=6$, $S=500$ posterior draws, and simultaneous bands. The primary route corrects with a cross-fitted Super Learner propensity fitted for each arm against the rest, and a parametric route runs alongside it, using arm-specific linear outcome models and the per-school empirical assignment shares as the design propensity. An estimated against design propensity diagnostic shows a mean absolute difference of about $0.06$ on assignment probabilities near one third.

The recommended protocol reports the $W$-scale evidence first. Table~\ref{tab:app:star_null} gives the corrected quantization path on the primary route. The total heterogeneity is $\widehat W(1)=2.40\times10^{-4}\,(\text{log points})^2$ with $95\%$ simultaneous band $(-0.54,\,5.35)\times10^{-4}$, which straddles zero. The parametric route gives $\widehat W(1)=-1.78\times10^{-4}$ with band $(-4.25,\,0.68)\times10^{-4}$, also straddling zero. On both routes the band for the total heterogeneity covers zero, so the first tier of the gate of Remark~\ref{rem:noisefloor} fails and the data cannot support resolution analysis at this feature dimension and sample size, consistent with the below-floor presentation. As secondary descriptive evidence, the split-difference diagnostic returns $\widehat\Delta=3.94\times10^{-4}$ with bracket $[1.97,\,3.94]\times10^{-4}$, and the corrected $\widehat W(1)=2.40\times10^{-4}$ falls inside this bracket, so even the point heterogeneity sits at the floor. No resolution beyond a single group is supported and reporting on the $\rho$ scale is not licensed. The small negative corrected values of $\widehat W(K)$ at the larger orders are the correction operating at the noise floor, where corrected risks can be negative in finite samples as noted in the main text, the regime characterized by the level-shift identity of Remark~\ref{rem:noisefloor} and Supplementary Section~\ref{sec:noisefloor_supp}.

\begin{table}[!t]
\centering
\caption{Project STAR below-floor benchmark. Corrected quantization path on the primary Super Learner propensity route, with $95\%$ simultaneous bands. Entries are in units of $10^{-4}\,(\text{log points})^2$. The band for the total heterogeneity $\widehat W(1)$ straddles zero, so no resolution beyond one group is supported and the $\rho$ scale is not reported.}
\label{tab:app:star_null}
\footnotesize
\setlength{\tabcolsep}{3pt}
\begin{tabular}{l cccccc}
\toprule
$K$ & 1 & 2 & 3 & 4 & 5 & 6\\
\midrule
$\widehat W(K)$ & $2.40$ & $0.22$ & $-0.34$ & $-0.98$ & $-1.68$ & $-1.87$\\
band & $(-0.54,\,5.35)$ & $(-1.62,\,2.06)$ & $(-2.08,\,1.41)$ & $(-2.62,\,0.65)$ & $(-3.06,\,-0.30)$ & $(-3.27,\,-0.48)$\\
\midrule
\multicolumn{7}{l}{The parametric route gives $\widehat W(1)=-1.78$ with band $(-4.25,\,0.68)$.}\\
\bottomrule
\end{tabular}
\end{table}

Ignoring the gate makes the failure concrete. Because $\widehat W(1)$ sits at the edge of zero, the projected per-draw normalization behind $\rho$ is unstable. On the primary route $7$ of the $500$ draws return a nonpositive $\widehat W(1)$, for which $\rho$ is forced to zero, the point $\hat\rho$ jumps to $0.908$ at $K=2$ and clamps to one for $K\ge3$, and the internal path diagnostic marks $K\ge3$ as not reportable, which is the near-zero denominator artifact the caveat anticipates. On the parametric route $\widehat W(1)$ is not positive, so the corrected resolution sets are empty at every $\gamma$, $\hat\rho$ is returned as zero throughout, and $\rho$-scale reporting is flagged unavailable. The pedagogically important contrast is the uncorrected plug-in path. Its causal heterogeneity $R^2$ climbs through $0.50$, $0.65$, $0.72$, $0.76$, and $0.80$ across $K=2,\dots,6$ on the primary route and to $0.81$ on the parametric route, and its resolution sets resolve to definite counts, for example $\widehat K^\star_{\mathrm{plug}}(0.70)=4$. An uncorrected pipeline would therefore manufacture subgroup structure out of feature-estimation noise where the corrected and gated analysis finds none. Figure~\ref{fig:app:star_null} displays the split.

\begin{figure}[!t]
\centering
\includegraphics[width=\textwidth]{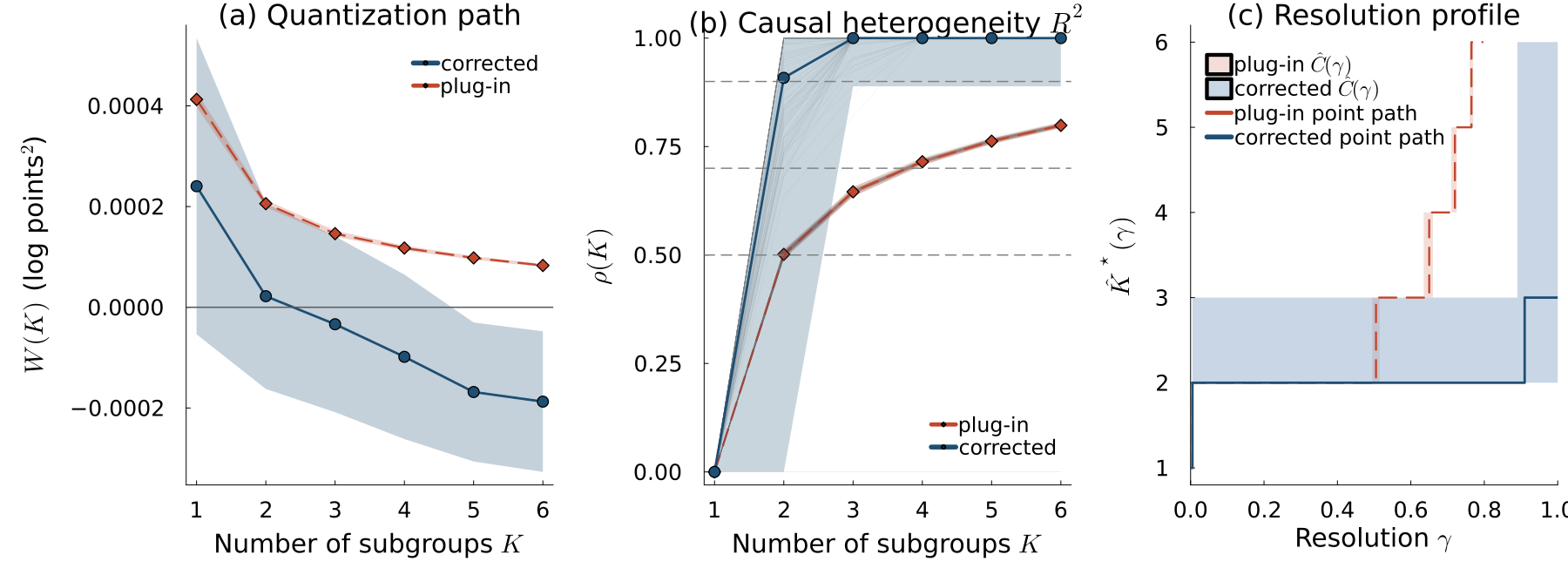}
\caption{Project STAR below-floor benchmark under the primary Super Learner propensity route. (a) Corrected quantization path $\widehat W(K)$ (solid) with its $95\%$ simultaneous band, and the uncorrected plug-in path (dashed) overlaid. The band for the total heterogeneity straddles zero. (b) Corrected and plug-in causal heterogeneity $R^2$. The corrected curve is erratic at the near-zero denominator, jumping to about $0.91$ and clamping at one, while the plug-in curve climbs smoothly to about $0.80$, with the gray fan showing the posterior draws at low opacity. (c) Resolution profile $\widehat K^\star(\gamma)$ with the set-valued report, corrected against plug-in. The $W$-scale band in panel (a) is the operative summary and reports no heterogeneity, whereas the plug-in path in panels (b) and (c) manufactures resolution.}
\label{fig:app:star_null}
\end{figure}

The penalized profile reads the same path on the unnormalized dial, which never divides by $W(1)$ and so stays well defined at the null. On the primary route the point selector $\widehat K^\dagger(\tau)$ equals one at the reference price $\tau=\delta^2/4=6.25\times10^{-4}$, a benefit gap of $\delta=0.05$ log points, about $5$ percent, between two equally prevalent groups, and at every larger price, with the simultaneous set $\widehat C^\dagger(\tau)$ collapsing to the singleton $\{1\}$ at twice the reference price, the next displayed point of the grid. On the parametric route $\widehat K^\dagger(\tau)=1$ at every displayed price. Both routes read as one group at any interpretable price.

The heterogeneity null is not an effect null. The augmented inverse-probability arm contrasts against the regular class are $+0.0154$ log points with interval $(0.0107,\,0.0197)$ for the small class, about a $1.5$ percent gain in scores, and $+0.0016$ with interval $(-0.0026,\,0.0059)$ for the aide class, indistinguishable from zero. The parametric route agrees, at $+0.0181$ with interval $(0.0134,\,0.0228)$ and $+0.0020$ with interval $(-0.0027,\,0.0062)$. Class type is delivered at the classroom level, so students in a class share teacher and classroom shocks, and the individual-level Dirichlet weights used here understate uncertainty in the manner of Supplementary Remark~\ref{rem:clustered}. Classroom-clustered weights would widen these intervals, which only reinforces the below-floor heterogeneity reading, since the $\widehat W(1)$ band already covers zero, while the average-effect intervals should be read as anti-conservative. The experiment thus carries a real small-class average effect while its heterogeneity across covariate profiles sits at the noise level, the below-floor regime for the resolution analysis. Effect heterogeneity across student subgroups reported in the broader STAR literature is thus evidently below this pipeline's floor at this feature dimension and $n=5{,}765$, rather than absent, which is why we read the benchmark as below-floor rather than as a demonstration that $W(1)=0$. The forced two-group working display underscores the point, with cells of shares $0.473$, interval $(0.455,\,0.497)$, and $0.527$, interval $(0.503,\,0.545)$, whose center coordinates all lie within $0.026$ of zero, consistent with an essentially unstructured feature law.

Under this null the pipeline reports one group at every interpretable price and heterogeneity indistinguishable from zero, rather than manufacturing subgroups from estimation noise, which is the intended honest behavior when $W(1)$ cannot be separated from the degenerate law.

\section{Proofs of the main results}\label{app:proofs}

Throughout the supplement, $C$ denotes a finite constant depending only on the constants declared in Section~\ref{sec:setup} and the assumptions, namely $B_Y$, $B_\mu$, $\eps_\pi$, $\norm{H}$, $\diam\calC$, $B_F$, $\underline m$, $\overline K$, $p$, $q$, $C_M$, $\am$, and $t_0$, and may change from line to line. Since any margin probability is at most one, the bound of Assumption~\ref{ass:margin} extends from $t\in(0,t_0]$ to all $t>0$ after enlarging $C_M$ to $C_M\vee t_0^{-\am}$, and this enlarged constant is used without further comment. 
Supporting empirical-process lemmas (entropy and Donsker properties, the multiplier maximal inequality, admissibility of the Dirichlet weights, and the automatic increment rate on the quantization class) are collected in Supplementary Section~\ref{app:lemmas}. 
We write $\bar F$ for the common envelope of the score classes: by Assumptions~\ref{ass:bound}--\ref{ass:class}, $\sup_{f\in\calF}\sup_\eta|\phi_f(o;\eta)|\le\bar F<\infty$ for all $o$ in the support and all $\eta$ in the truncated range, since $|R_a|\le2(B_Y+B_\mu)/\eps_\pi$ and all integrands and gradients are bounded.

\subsection{Proof of Theorem~\ref{thm:eif}}\label{app:eif}
The proof is a direct submodel-score computation for (i), and for (ii) a joint second-order Taylor expansion in $(u,m)$ whose remainder, for the kinked losses $g_c$, splits into an exact quadratic on the event that $H\bar\bmu(X)$ and $H\bmu(X)$ fall in the same Voronoi cell and a boundary-crossing event whose probability the margin condition controls; the exponent $2(1+\am)/(2+\am)$ arises from optimizing a truncation level.
\subsubsection{A geometric lemma}
We isolate the elementary geometry that powers the quantization-class bounds.

\begin{lemma}[Crossing geometry]\label{lem:geometry}
Let $c\in\calC^K$, $u,\bar u\in\calC$, $\Delta=\bar u-u$, and let $h=h_c(u)$, $h^\star=h_c(\bar u)$ be (any) nearest-center labels. Write $\mathrm{rem}=g_c(u)-g_c(\bar u)-\nabla g_c(\bar u)^\top(u-\bar u)$, where $\nabla g_c(\bar u)=2(\bar u-c_{h^\star})$.
\begin{enumerate}[label=(\alph*),leftmargin=2em]
\item If $c_h=c_{h^\star}$ as points of $\R^q$, then $\mathrm{rem}=\norm{\Delta}^2$ exactly.
\item If $c_h\ne c_{h^\star}$, let $d=\norm{c_h-c_{h^\star}}>0$ and let $B$ be the bisector hyperplane of $c_h$ and $c_{h^\star}$. Then
\begin{equation*}
|\mathrm{rem}|\le\norm{\Delta}^2+2\norm{\Delta}\,d,
\qquad
\dist(\bar u,B)\le\norm{\Delta},
\qquad
\dist(u,B)\le2\norm{\Delta} .
\end{equation*}
\end{enumerate}
\end{lemma}

\begin{proof}
(a) With the common center point $c$, $g_c$ agrees with $x\mapsto\norm{x-c}^2$ at both $u$ and $\bar u$ and $\nabla g_c(\bar u)=2(\bar u-c)$, so $\mathrm{rem}=\norm{u-c}^2-\norm{\bar u-c}^2-2(\bar u-c)^\top(u-\bar u)=\norm{u-\bar u}^2$.

(b) Expanding $\norm{u-c_{h^\star}}^2=\norm{\bar u-c_{h^\star}}^2+2(\bar u-c_{h^\star})^\top(u-\bar u)+\norm{\Delta}^2$ gives
$\mathrm{rem}=\bigl\{\norm{u-c_h}^2-\norm{u-c_{h^\star}}^2\bigr\}+\norm{\Delta}^2$.
Consider the affine function $s(x)=\norm{x-c_h}^2-\norm{x-c_{h^\star}}^2=-2x^\top(c_h-c_{h^\star})+\norm{c_h}^2-\norm{c_{h^\star}}^2$. Optimality of the labels gives $s(u)\le0\le s(\bar u)$, while $|s(\bar u)-s(u)|\le2\norm{\Delta} d$; hence $-2\norm{\Delta} d\le s(u)\le0$ and $0\le s(\bar u)\le2\norm{\Delta} d$, which yields the bound on $|\mathrm{rem}|$. Finally $B=\{s=0\}$ with $\norm{\nabla s}=2d$, so $\dist(x,B)=|s(x)|/(2d)$, giving $\dist(\bar u,B)\le\norm{\Delta}$ and $\dist(u,B)\le\dist(\bar u,B)+\norm{\Delta}\le2\norm{\Delta}$.
\end{proof}

Note also that under Assumption~\ref{ass:margin}, for every codebook $c$ the set of Voronoi-boundary points of $c$ charged by $P_U$ is null: boundaries between distinct center points lie in finitely many hyperplanes $B$, and $P_0\{\dist(U,B)\le t\}\le C_Mt^\am \downarrow0$.

We record the elementary regularity of the bounded linear submodels used to identify the gradient.

\begin{lemma}[Bounded linear submodels]\label{lem:submodel}
Let Assumption~\ref{ass:ident} hold. Fix $\bar S<\infty$ and let $S$ be measurable with $\E_0S=0$ and $\norm{S}_\infty\le\bar S$. For $|t|<1/(2\bar S)$ let $dP_t=(1+tS)\,dP_0$, and fix versions of the conditional expectations appearing below. Then for every $a\in\calA$ and $P_{0,X}$-almost every $x$,
\begin{equation*}
\mu_{a,t}(x)=\E_t(Y\mid A=a,X=x)=\frac{\E_0[Y(1+tS)\mid A=a,X=x]}{\E_0[(1+tS)\mid A=a,X=x]},
\end{equation*}
a ratio of affine functions of $t$ whose denominator lies in $[1/2,3/2]$. Hence $t\mapsto\mu_{a,t}(x)$ is infinitely differentiable on $\bigl(-1/(2\bar S),1/(2\bar S)\bigr)$ for $P_{0,X}$-almost every $x$, with
\begin{equation*}
\partial_t\mu_{a,t}(x)\big|_{t=0}=\E_0\bigl[\{Y-\mu_a(X)\}S\mid A=a,X=x\bigr],
\qquad
\bigl|\partial_t\mu_{a,t}(x)\bigr|\le C(B_Y+B_\mu)\bar S
\end{equation*}
uniformly in $(t,x)$. Consequently $U_t(x)=H\bmu_t(x)$ is differentiable in $t$ for $P_{0,X}$-almost every $x$, with $\sup_x\norm{\partial_tU_t(x)}\le C\norm{H}(B_Y+B_\mu)\bar S$ and $\sup_x\norm{U_t(x)-U(x)}\le C|t|$ holding for the fixed versions. Moreover, for such bounded mean-zero $S$ the path $t\mapsto P_t$ is differentiable in quadratic mean at $t=0$ with score $S$, because $t^{-1}\{(1+tS)^{1/2}-1\}\to S/2$ pointwise with the difference quotient bounded in absolute value by $|S|\le\bar S$, which lies in $L_2(P_0)$, so the $L_2(P_0)$ limit defining quadratic-mean differentiability holds by dominated convergence \citep[Chapter~25]{vandervaart1998}.
\end{lemma}

\begin{proof}
Under $P_t$ the conditional law of $(A,Y)$ given $X$ has density proportional to $1+tS$ against that under $P_0$, so Bayes' rule gives the displayed ratio, whose numerator and denominator are affine in $t$. The denominator $1+t\,\E_0[S\mid A=a,X=x]$ lies in $[1-\bar S|t|,\,1+\bar S|t|]\subseteq[1/2,3/2]$ for $|t|<1/(2\bar S)$, so the ratio is a smooth function of $t$ there. Differentiating at $t=0$ gives $\partial_t\mu_{a,t}(x)|_{0}=\E_0[YS\mid A=a,X=x]-\mu_a(x)\,\E_0[S\mid A=a,X=x]=\E_0[\{Y-\mu_a(X)\}S\mid A=a,X=x]$. The uniform bound follows from $|Y|\le B_Y$, $|\mu_a|\le B_\mu$, $\norm{S}_\infty\le\bar S$, and the denominator bound. The statements for $U_t=H\bmu_t$ follow by linearity of $H$ and the mean value theorem.
\end{proof}

\textbf{Part (i) (pathwise differentiability).}
It suffices to identify the gradient along the bounded linear submodels $dP_t=(1+tS)\,dP_0$ of Lemma~\ref{lem:submodel}, with $S$ bounded and mean zero, because their scores are dense in $L_2^0(P_0)$ as shown at the close of the argument. Write $S=S_X+S_{A\mid X}+S_{Y\mid A,X}$ along the factorization $P=P_X\,P_{A\mid X}\,P_{Y\mid A,X}$, each component having conditional mean zero given the preceding variables. By Lemma~\ref{lem:submodel} the path values obey $\sup_x\norm{U_t(x)-U(x)}\le C|t|$ and $\max_a\sup_x|\mu_{a,t}(x)-\mu_a(x)|\le C|t|$, so for all small $|t|$ the pair $(U_t(x),\bmu_t(x))$ lies in the compact unit neighborhood of $\calC\times[-B_\mu,B_\mu]^{p+1}$; there $g_c$, which is defined on all of $\R^q$, is Lipschitz uniformly over codebooks, and each $C^2$ loss of Assumption~\ref{ass:class} is the restriction of a $C^2$ function on an open neighborhood of $\calC\times[-B_\mu,B_\mu]^{p+1}$, so $\Psi_f(P_t)$ is well defined along the path. Then $\Psi_f(P_t)=\int f\{U_t(x),\bmu_t(x)\}\,dP_{t,X}(x)$ and, differentiating at $t=0$,
\begin{equation}\label{eq:pathderiv}
\frac{d}{dt}\Psi_f(P_t)\Big|_{t=0}
=\E_0\bigl[f(U,\bmu)\,S_X\bigr]
+\E_0\Bigl[\bigl\{\nabla_uf(U,\bmu)^\top H+\nabla_mf(U,\bmu)^\top\bigr\}\,\dot\bmu(X)\Bigr],
\end{equation}
where, by Lemma~\ref{lem:submodel}, $\dot\mu_a(x)=\partial_t\mu_{a,t}(x)|_{t=0}=\E_0[\{Y-\mu_a(X)\}\,S\mid A=a,X=x]=\E_0[R_a(O;\eta_0)\,S\mid X=x]$, the last equality by Assumption~\ref{ass:ident}(iii) and the definition \eqref{eq:residual}. Substituting into \eqref{eq:pathderiv} and moving the $X$-measurable gradient factor inside the conditional expectation, the second term equals $\E_0[\{\nabla_uf^\top H+\nabla_mf^\top\}R\,S]$. For the first term, since $f(U,\bmu)-\Psi_f(P_0)$ is $X$-measurable and the conditional scores integrate to zero given $X$,
$\E_0[fS_X]=\E_0[\{f-\Psi_f\}S_X]=\E_0[\{f-\Psi_f\}S]$. Hence $\frac{d}{dt}\Psi_f(P_t)|_0=\E_0[\{\phi_f(O;\eta_0)-\Psi_f(P_0)\}S]$. The candidate gradient $\phi_f(O;\eta_0)-\Psi_f(P_0)$ is bounded, hence in $L_2^0(P_0)$, using $\E_0\{R\mid X\}=0$. For $f=g_c\in\calF_{\mathrm{qt}}$ the chain-rule step in \eqref{eq:pathderiv} is justified as follows. By Lemma~\ref{lem:submodel} the curve $t\mapsto U_t(x)$ is differentiable at every $x$ with $\sup_x\norm{U_t(x)-U(x)}\le C|t|$. The loss $g_c$ is Lipschitz on $\calC$ and differentiable off the union of Voronoi boundaries, a set that is $P_U$-null by the boundary-null hypothesis of part (i), and automatically so under Assumption~\ref{ass:margin} as noted above, so for $P_0$-almost every $x$ the base point $U(x)$ is a point of differentiability of $g_c$. At such an $x$ the ordinary chain rule for a Lipschitz function composed with a curve differentiable in $t$ gives $t^{-1}\{g_c(U_t(x))-g_c(U(x))\}\to\nabla g_c\{U(x)\}^\top\partial_tU_t(x)|_{t=0}$, while the difference quotient is dominated by $\mathrm{Lip}(g_c)\,\sup_x\norm{U_t(x)-U(x)}/|t|\le C$ uniformly in $(t,x)$, so dominated convergence yields \eqref{eq:pathderiv} for $g_c$. This uses only the pointwise convergence $U_t(x)\to U(x)$ supplied by the lemma and never assumes pointwise differentiability of $\Psi_f$ along a general regular submodel. Finally, the bounded scores $S$ form a dense subset of $L_2^0(P_0)$, so the tangent set generated by the submodels of Lemma~\ref{lem:submodel} has closure equal to the full nonparametric tangent space $L_2^0(P_0)$, and the bounded candidate gradient is the unique gradient relative to this tangent set, hence the efficient influence function \citep[Section~25.3]{vandervaart1998}.

\textbf{Part (ii) (uniform second-order bias).}
Fix $\bar\eta$ in the truncated range and write $\bar U=H\bar\bmu(X)$, $\Delta=\bar U-U$, so $\norm{\Delta}\le\norm{H}\,\norm{\bar\bmu-\bmu}$ pointwise. Using $\E_0\{R_a(O;\bar\eta)\mid X\}=(\pi_a/\bar\pi_a)(\mu_a-\bar\mu_a)(X)$ and writing $(\pi_a/\bar\pi_a)(\mu_a-\bar\mu_a)=(\mu_a-\bar\mu_a)+e_a$ with $e_a=\{(\pi_a-\bar\pi_a)/\bar\pi_a\}(\mu_a-\bar\mu_a)$,
\begin{equation}\label{eq:biasdecomp}
\begin{split}
\E_0\{\phi_f(O;\bar\eta)\}-\Psi_f(P_0)
&=\E_0\bigl[f(\bar U,\bar\bmu)-f(U,\bmu)+\{\nabla_uf(\bar U,\bar\bmu)^\top H+\nabla_mf(\bar U,\bar\bmu)^\top\}(\bmu-\bar\bmu)\bigr]\\
&\qquad+\E_0\bigl[\{\nabla_uf^\top H+\nabla_mf^\top\}e\bigr].
\end{split}
\end{equation}
The last term is bounded, uniformly over $f$ with bounded gradients, by $C\sum_a\E_0|\,\pi_a-\bar\pi_a||\mu_a-\bar\mu_a|\cdot(2/\eps_\pi)\le C\,r_\mu r_\pi$ by Cauchy--Schwarz. For the first term: when $f\in\calF_{\mathrm{sm}}\cup\calF_{\mathrm{str}}$, a second-order Taylor expansion of $f$ at $(\bar U,\bar\bmu)$ in the direction $(U-\bar U,\bmu-\bar\bmu)=(-H(\bar\bmu-\bmu),-(\bar\bmu-\bmu))$ shows the bracketed integrand equals minus the Taylor remainder, bounded by $CB_F\norm{\bar\bmu-\bmu}^2$ pointwise, whence the contribution $Cr_\mu^2$. When $f=g_c\in\calF_{\mathrm{qt}}$, the bracketed integrand is $-\mathrm{rem}$ in the notation of Lemma~\ref{lem:geometry} (with $\nabla_mg_c=0$), so by the lemma, pointwise,
\begin{gather*}
|\mathrm{rem}|\le\norm{\Delta}^2+\sum_{(h,j):c_h\ne c_j}2\norm{\Delta}\,\norm{c_h-c_j}\,\ind{E_{hj}},\\
E_{hj}=\bigl\{c_{h_c(U)}=c_j,\ c_{h_c(\bar U)}=c_h\bigr\}\subseteq\bigl\{\dist(U,B_{hj})\le2\norm{\Delta}\bigr\},
\end{gather*}
where $B_{hj}$ is the bisector of $c_h,c_j$ and at most one indicator is active. For any pair and any $t\in(0,t_0]$, splitting on $\{\norm{\Delta}>t\}$,
\begin{equation*}
\E_0\bigl[\norm{\Delta}\,\ind{\dist(U,B_{hj})\le2\norm{\Delta}}\bigr]
\le t^{-1}\E_0\norm{\Delta}^2+t\,P_0\bigl(\dist(U,B_{hj})\le2t\bigr)
\le t^{-1}\,C r_\mu^2+2^\am  C_M\,t^{1+\am},
\end{equation*}
by Assumption~\ref{ass:margin}. Choosing $t=(Cr_\mu^2)^{1/(2+\am)}$ when this value is at most $t_0$ yields the bound $C\,r_\mu^{2(1+\am)/(2+\am)}$ for each of the at most $\overline K^2$ pairs, with $\norm{c_h-c_j}\le\diam\calC$ absorbed into $C$. When $(Cr_\mu^2)^{1/(2+\am)}>t_0$, the ratio $(Cr_\mu^2/t_0^{2+\am})^{(1+\am)/(2+\am)}$ exceeds one, so the left side of \eqref{eq:bias}, trivially bounded by $2\bar F$, is also bounded by $2\bar F\,t_0^{-(1+\am)}(Cr_\mu^2)^{(1+\am)/(2+\am)}$, and the same bound holds for every candidate pair after enlarging $C$, which may depend on $t_0$. Collecting terms gives \eqref{eq:bias}--\eqref{eq:R2}; every constant depends only on the declared quantities, so the bound is uniform over $\calF$. \hfill$\square$

\subsection{Proof of Theorem~\ref{thm:bvm}, part (i)}\label{app:onestep}

Write $\Delta_i(f)=\hphi_{f,i}-\phi_f(O_i;\eta_0)$ and, for fold $b$, $\Delta^{(b)}(f)(o)=\phi_f(o;\hateta^{(-b)})-\phi_f(o;\eta_0)$, so that
\begin{equation}\label{eq:foldsplit}
\hPsi(f)-P_n\phi_f(\cdot;\eta_0)
=P_n\Delta_\cdot(f)
=\sum_{b=1}^B\frac{n_b}{n}\Bigl[\underbrace{P_0\Delta^{(b)}(f)}_{\text{bias}}
+\underbrace{(P_{n_b}^{(b)}-P_0)\Delta^{(b)}(f)}_{\text{empirical process}}\Bigr],
\end{equation}
with $P_{n_b}^{(b)}$ the empirical measure of fold $b$. By Theorem~\ref{thm:eif}(ii), $\sup_f|P_0\Delta^{(b)}(f)|\le CRem_2(\hateta^{(-b)})=\opro(n^{-1/2})$ under Assumption~\ref{ass:nuisance}(i).

For the empirical-process term, condition on the training data $\calT_b$ of fold $b$: the fold-$b$ observations are i.i.d.\ $P_0$ and independent of $\calT_b$, and $\{\Delta^{(b)}(f):f\in\calF\}$ is then a fixed class with envelope $2\bar F$, $L_2(P_0)$ radius at most $\delta_n$, and uniform polynomial entropy by Lemma~\ref{lem:incrementclass}\textup{(a)}. The localized maximal inequality for uniformly bounded VC-type classes \citep{vdvwellner2011}, together with the unlocalized bound of \cite[Theorem~2.14.1]{vdvwellner1996}, then gives, on the event $\{\delta_n\le\delta\}$,
\begin{equation}\label{eq:epbound}
\E\Bigl[\sup_{f}\bigl|(P_{n_b}^{(b)}-P_0)\Delta^{(b)}(f)\bigr|\,\Big|\,\calT_b\Bigr]
\;\le\;\frac{C}{\sqrt{n_b}}\,\delta\sqrt{\log(A/\delta)}+\frac{C}{n_b}\log(A/\delta).
\end{equation}
Choosing $\delta=\delta_n\vee n^{-1/2}$ and using Assumption~\ref{ass:nuisance}(ii) (which implies $\delta_n\sqrt{\log(1/\delta_n)}=\opro(1)$ since $\delta_n\le1$ may be assumed), the right side is $\opro(n^{-1/2})$; Markov's inequality conditionally on $\calT_b$, then unconditionally, and a union over the $B$ folds give $\sup_f|P_n\Delta_\cdot(f)|=\opro(n^{-1/2})$, i.e.\ \eqref{eq:linearity}.

Finally, by Lemma~\ref{lem:donsker}(b) the class $\{\phi_f(\cdot;\eta_0):f\in\calF\}$ is $P_0$-Donsker with bounded envelope, so $\sqrt n(P_n\phi_\cdot(\eta_0)-\Psi_0)\dto\G_0$ in $\ell^\infty(\calF)$; combining with \eqref{eq:linearity} proves the weak convergence. Coordinatewise, the influence function equals the efficient influence function of Theorem~\ref{thm:eif}(i), and asymptotic linearity with the EIF implies regularity and efficiency \citep[Section~25.3]{vandervaart1998}. \hfill$\square$

\subsection{Proof of Theorem~\ref{thm:bvm}, part (ii)}\label{app:bvm}

For weight-and-data random variables $\zeta_n$ write $\zeta_n=\opro^{\mathrm w}(1)$
if $\E_w(|\zeta_n|\wedge1)\to0$ in outer probability; 
if $Z_n^{(s)}\dtow Z$ in $\ell^\infty(\calF)$ and $\sup_\calF|\tilde Z_n^{(s)}-Z_n^{(s)}|=\opro^{\mathrm w}(1)$, then $\tilde Z_n^{(s)}\dtow Z$, by the elementary bound $|\varphi(\tilde Z)-\varphi(Z)|\le\sup|\tilde Z-Z|\wedge2$ for $\varphi\in\BL$.

We also record a joint equicontinuity fact used below at weight-dependent random indices. Once the theorem is proved, the conditional weak convergence $\sqrt n(\Psi^{(s)}-\hPsi)\dtow\G_0$ to the tight limit $\G_0$, together with the unconditional convergence $\sqrt n(\hPsi-\Psi_0)\dto\G_0$ of part~\textup{(i)}, makes the weighted process $\sqrt n(\Psi^{(s)}-\Psi_0)$ jointly asymptotically equicontinuous with respect to the covariance semimetric $\varsigma$, because $\G_0$ has $\varsigma$-uniformly-continuous paths. Its increment $\sqrt n(\Psi^{(s)}-\Psi_0)(f_n-f_n')$ is therefore $\opro^{\mathrm w}(1)$ for any indices with $\varsigma(f_n,f_n')\pto0$, and in particular when evaluated at any jointly consistent random index. This is the standard equicontinuity step of weighted-bootstrap $Z$-estimation \citep{cheng2010}.

Decompose as in \eqref{eq:bvmdecomp}. \textbf{Oracle term.} The class $\Phi_0=\{\phi_f(\cdot;\eta_0):f\in\calF\}$ is $P_0$-Donsker with bounded envelope (Lemma~\ref{lem:donsker}(b)), and by Lemma~\ref{lem:pwcheck} the Dirichlet weights satisfy the conditions of the exchangeable-bootstrap central limit theorem \citep[Theorem~2.2]{praestgaard1993} and \citep[Theorem~3.6.13]{vdvwellner1996} with limiting multiplier variance $c^2=1$. Hence
\begin{equation}\label{eq:oracleconv}
\Bigl\{n^{-1/2}\textstyle\sum_i(w^{(s)}_i-1)\,\phi_f(O_i;\eta_0)\Bigr\}_{f\in\calF}\ \dtow\ \G_0
\qquad\text{in }\ell^\infty(\calF).
\end{equation}
(Because $\sum_i(w_i^{(s)}-1)=0$, the display is unchanged if $\phi_f$ is replaced by $\phi_f-Pf$ or any other recentering; this is why $\hPsi$ is the exact natural posterior center.)

\textbf{Increment term.} Using $w^{(s)}_i=nE_i/S_n$ with $E_1,\dots,E_n$ i.i.d.\ standard exponential, $S_n=\sum_jE_j$, $\bar E_n=S_n/n$, algebra gives
\begin{equation}\label{eq:increpr}
\frac1{\sqrt n}\sum_i(w^{(s)}_i-1)\Delta_i(f)
=\frac{n}{S_n}\cdot\frac1{\sqrt n}\sum_i(E_i-1)\Delta_i(f)
-\frac{n}{S_n}\,\sqrt n\,(\bar E_n-1)\cdot P_n\Delta_\cdot(f).
\end{equation}
For the second term, on $A_n=\{S_n\ge n/2\}$ (with $\Pp_w(A_n^c)\le e^{-cn}$),
$\E_w\bigl[\sup_f|\cdot|\wedge1\bigr]\le2\,\E_w|\sqrt n(\bar E_n-1)|\cdot\sup_f|P_n\Delta_\cdot(f)|+e^{-cn}
\le2\sup_f|P_n\Delta_\cdot(f)|+e^{-cn}$,
and $\sup_f|P_n\Delta_\cdot(f)|=\opro(1)$ by the proof of Theorem~\ref{thm:bvm}\textup{(i)}; hence this term is $\opro^{\mathrm w}(1)$.
For the first term, condition on the data and the training folds: $\xi_i=E_i-1$ are i.i.d.\ mean-zero with $\norm{\xi_1}_{\psi_1}\le\norm{E_1}_{\psi_1}+\norm{1}_{\psi_1}=2+1/\log2\le4$, independent of the fixed array $\{\Delta_i(f)\}$. By Lemma~\ref{lem:incrementclass}\textup{(b)--(c)}, the cross-fitted increment array class $\calD_n=\{(\Delta_1(f),\dots,\Delta_n(f)):f\in\calF\}$ has envelope $2\bar F$ and polynomial entropy under the empirical $L_2(P_n)$ norm, with constants independent of the realized nuisance fits. Lemma~\ref{lem:multmax} therefore gives
\begin{equation*}
\E_\xi\sup_{f\in\calF}\Bigl|\frac1{\sqrt n}\sum_i\xi_i\Delta_i(f)\Bigr|
\;\le\;C\Bigl\{\delta_n'\sqrt{\log(A/\delta_n')}+\frac{\log n}{\sqrt n}\Bigr\},
\qquad
\delta_n'=\sup_f\bigl\{P_n\Delta_\cdot(f)^2\bigr\}^{1/2}.
\end{equation*}
Lemma~\ref{lem:incrementclass}\textup{(d)} gives $\delta_n'\le\delta_n+\Op(n^{-1/4})$, where the proof uses the squared-increment entropy and the bounded Lipschitz transformation $|a^2-b^2|\le 2\bar F_D|a-b|$ with $\bar F_D=2\bar F$. Hence, by Assumption~\ref{ass:nuisance}\textup{(ii)}, $\delta_n'\sqrt{\log(A/\delta_n')}=\opro(1)$. Combining, the conditional expectation displayed above is $\opro(1)$; together with the $n/S_n$ factor handled on $A_n$ as before, the first term of \eqref{eq:increpr} is $\opro^{\mathrm w}(1)$ uniformly over $\calF$.

Adding \eqref{eq:oracleconv} and the two $\opro^{\mathrm w}(1)$ terms proves \eqref{eq:bvm}. \hfill$\square$

\subsection{Proof of Theorem \ref{thm:delta}}\label{sec:proof:thm:delta}
\begin{proof}
Theorem~3.9.11 of \cite{vdvwellner1996} (the delta method for the bootstrap) is stated for exchangeably weighted empirical processes. Beyond Hadamard differentiability tangential to a set supporting the limit, it requires the derivative to be defined and continuous on the whole space, which is exactly the extension hypothesis assumed here, and it consumes the two convergences supplied by Theorem~\ref{thm:bvm}, namely $\sqrt n(\hPsi-\Psi_0)\dto\G_0$ and $\sqrt n(\Psi^{(s)}-\hPsi)\dtow\G_0$. In every application in this paper the derivative is a continuous linear evaluation map, such as $\zeta\mapsto\zeta(g_{c^\star(K)})$ or $\zeta\mapsto-V_\beta^{-1}\zeta(\nabla_\beta\ell_{\beta^\star})$, so the extension hypothesis holds automatically. Since $\Phi_0$ is Donsker, $\G_0$ is tight with paths in $C_\varsigma(\calF)$ \citep[Section~1.5]{vdvwellner1996}, so the tangentiality hypothesis is met, and both conclusions follow. Validity of quantile-based credible sets follows from conditional weak convergence plus continuity of the limit distribution by the standard argument \citep[Lemma~23.3]{vandervaart1998}.
\end{proof}

\subsection{Proofs of Proposition~\ref{prop:cancellation} and the energy-scale sensitivity result}\label{app:props}

\subsubsection{Proposition~\ref{prop:cancellation}}
The influence-function formula is immediate from Theorem~\ref{thm:eif}\textup{(i)} and the linearity of $f\mapsto\phi_f$ in $(f,\nabla f)$: subtract the two corrected scores and center by $\Psi_0(g)-\Psi_0(g')$.  The same subtraction is exact in the finite sums defining $\hPsi$ and $\Psi^{(s)}$, because both processes are linear in the corrected evaluations $\hphi_{f,i}$.  Hence any generated-feature fluctuation that appears as a common additive term in the two corrected risks cancels term by term in the point contrast and in every weighted draw. \hfill$\square$

\subsubsection{Proposition~\ref{prop:fragility}}
For the odds-ratio bound, pointwise in $\beta$, $e^{-\lambda_nn\Delta_n}\le e^{-\lambda_nn\{\widehat L_n(\beta)-L_n^U(\beta)\}}\le e^{\lambda_nn\Delta_n}$, so integrating against $e^{-\lambda_nnL_n^U}d\Pi_M$ gives $\widehat Z_M\in[e^{-\lambda_nn\Delta_n}Z_M,\ e^{\lambda_nn\Delta_n}Z_M]$ for each model, and the odds-ratio bound follows. For attainment, take $\widehat L_n=L^U_n-c_n$ on $B_M$ and $\widehat L_n=L^U_n$ on $B_{M'}$: then $\Delta_n=c_n$, $\widehat Z_M=e^{\lambda_nnc_n}Z_M$, and $\widehat Z_{M'}=Z_{M'}$.

For the local generated-feature expansion, by Assumption~\ref{ass:class}(iii)-type smoothness ($\ell_\beta\in C^2$ in $u$, uniformly), for bounded $\xi_i$,
$\widehat L_n(\beta)-L^U_n(\beta)=b_n\,P_n\{\xi^\top\nabla_u\ell_\beta(U)\}+r_n(\beta)$ with $\sup_\beta|r_n(\beta)|\le Cb_n^2$. Suppose $\sup_{\beta\in B_M}|P_0\{\xi^\top\nabla_u\ell_\beta(U)\}-a|\le\eps$ and likewise with $(B_{M'},a')$, and set $\kappa=a-a'$. The class $\{\xi^\top\nabla_u\ell_\beta(U)\}$ is bounded and Lipschitz in the finite-dimensional index, hence Donsker, so $G_n:=\sup_\beta|(P_n-P_0)\{\xi^\top\nabla_u\ell_\beta(U)\}|=\Op(n^{-1/2})$. Therefore $\sup_{\beta\in B_M}|\widehat L_n-L^U_n-b_na|\le b_n(\eps+G_n)+Cb_n^2$, and the preceding odds-ratio sandwich applied modelwise with the model-specific shifts $b_na$, $b_na'$ gives
\begin{equation*}
\Bigl|\log\frac{\widehat Z_M/\widehat Z_{M'}}{Z_M/Z_{M'}}+\lambda_nnb_n\kappa\Bigr|
\;\le\;2\lambda_nn\bigl\{b_n(\eps+G_n)+Cb_n^2\bigr\}
\;=\;\lambda_nnb_n\bigl\{2\eps+\Op(n^{-1/2})+O(b_n)\bigr\},
\end{equation*}
so for $\eps$ small the log-odds shift is $-\lambda_nnb_n\kappa\{1+o_p(1)\}$, diverging at $b_n=n^{-1/2}$ with $\lambda_n\asymp1$ whenever $\kappa\ne0$. For the corrected loss, the decomposition \eqref{eq:foldsplit} together with Theorem~\ref{thm:eif}(ii) gives, uniformly in $\beta$,
$P_n\hphi_{\ell_\beta,\cdot}-L^U_n(\beta)=P_n\{\nabla_u\ell_\beta(U)^\top HR(\cdot;\eta_0)\}+\opro(n^{-1/2})$:
the deterministic shift is removed, but the displayed term is a mean-zero $\Op(n^{-1/2})$ random functional of the score field $\nabla_u\ell_\beta$, which differs across models with different fields and is amplified by the identical mechanism. \hfill$\square$

\subsection{Proofs for Section~\ref{sec:delta}}\label{app:delta}

The posterior delta method (Theorem~\ref{thm:delta}) is stated in the main text, Section~\ref{sec:delta}, and proved in Supplementary Section~\ref{sec:proof:thm:delta}. Here we prove the supporting quantization envelope lemma and the path corollary.

\subsubsection{The quantization envelope lemma}

\begin{lemma}[Differentiability of the quantization functional]\label{lem:envelope}
Fix $K\in[\overline K]$ and define $\iota_K:\ell^\infty(\calF_{\mathrm{qt}})\to\R$ by $\iota_K(\nu)=\inf_{c\in\calC^K}\nu(g_c)$. Then $\iota_K$ is concave and, at $\nu=\Psi_0$, Hadamard directionally differentiable tangentially to the directions $\zeta$ for which $c\mapsto\zeta(g_c)$ is continuous on the ordered tuples $(\calC^K,d_\infty)$, with derivative $\iota_K'(\zeta)\;=\;\inf_{c\in\calC^\star(K)}\zeta(g_c)$.
Under Assumption~\ref{ass:quant}\textup{(i)} the infimum is over a singleton and $\iota_K$ is fully Hadamard differentiable with linear derivative $\zeta\mapsto\zeta(g_{c^\star(K)})$.
\end{lemma}

\textbf{Proof.} Concavity is clear ($\iota_K$ is an infimum of linear functionals $\nu\mapsto\nu(g_c)$). Equip $\calC^K$ with the pseudometric $d(c,c')=\norm{g_c-g_{c'}}_\infty$, under which $c\mapsto\Psi_0(g_c)$ is $1$-Lipschitz, since $|\Psi_0(g_c)-\Psi_0(g_{c'})|\le\norm{g_c-g_{c'}}_\infty$, hence continuous. Note that $d(c,c')\le4\,\diam\calC\cdot d_\infty(c,c')$ for ordered tuples, so Euclidean convergence of tuples implies $d$-convergence, and continuity of a direction on $(\calC^K,d_\infty)$ is exactly what the compactness extraction below consumes. Let $t_m\downarrow0$ and $\zeta_m\to\zeta$ uniformly with $\zeta$ continuous on $(\calC^K,d_\infty)$, and write $\nu_m=\Psi_0+t_m\zeta_m$.

\textbf{Upper bound.} $\iota_K(\nu_m)\le\nu_m(g_c)=\Psi_0(g_c)+t_m\zeta_m(g_c)$ for every $c\in\calC^\star(K)$, so
\begin{equation*}
\limsup_m t_m^{-1}\{\iota_K(\nu_m)-W(K)\}\le\inf_{c\in\calC^\star(K)}\zeta(g_c).
\end{equation*}

\textbf{Lower bound.} Choose $c_m$ with $\nu_m(g_{c_m})\le\iota_K(\nu_m)+t_m^2$. Then
$\Psi_0(g_{c_m})\le\iota_K(\nu_m)+t_m^2+t_m\norm{\zeta_m}_\infty\le W(K)+O(t_m)$,
so $\Psi_0(g_{c_m})\to W(K)$; by Euclidean compactness of $\calC^K$, continuity of $c\mapsto\Psi_0(g_c)$, and the definition of $\calC^\star(K)$, every Euclidean subsequential limit point of $(c_m)$ lies in $\calC^\star(K)$ (argmin upper hemicontinuity). Hence
\begin{equation*}
t_m^{-1}\bigl\{\iota_K(\nu_m)-W(K)\bigr\}
\ \ge\ t_m^{-1}\bigl\{\Psi_0(g_{c_m})-W(K)\bigr\}+\zeta_m(g_{c_m})-t_m
\ \ge\ \zeta_m(g_{c_m})-t_m ,
\end{equation*}
and along any subsequence with $c_m\to\tilde c\in\calC^\star(K)$, $\zeta_m(g_{c_m})\to\zeta(g_{\tilde c})\ge\inf_{c\in\calC^\star(K)}\zeta(g_c)$ by continuity of $\zeta$ on $(\calC^K,d_\infty)$ and $\norm{\zeta_m-\zeta}_\infty\to0$. Thus $\liminf_m t_m^{-1}\{\iota_K(\nu_m)-W(K)\}\ge\inf_{\calC^\star(K)}\zeta(g_c)$, proving Hadamard directional differentiability with the stated derivative; under Assumption~\ref{ass:quant}(i) the infimum is over the singleton $\{c^\star(K)\}$, the derivative is linear and defined on all $d$-continuous $\zeta$, and full (tangential) Hadamard differentiability follows since the derivative formula is linear and continuous. \hfill$\square$

\begin{remark}[Nonunique codebooks]\label{rem:nonunique}
Without Assumption~\ref{ass:quant}\textup{(i)}, Lemma~\ref{lem:envelope} still gives Hadamard directional differentiability with the concave derivative $\zeta\mapsto\inf_{c\in\calC^\star(K)}\zeta(g_c)$, a phenomenon classical in stochastic programming \citep{shapiro1991,dumbgen1993}. The plug-in path estimators remain consistent and $\sqrt n$-tight, but the limit is non-Gaussian and the bootstrap and posterior are in general inconsistent for it \citep{fang2019}. Exactly symmetric feature laws are the canonical violation. We flag set-valued and directional extensions as future work and note again that the resolution-profile inference depends on the path only through the band of Corollary~\ref{cor:path}, whose validity under directional differentiability can be restored by the rescaled or numerical-derivative constructions of Fang and Santos \cite{fang2019} and Hong and Li \cite{hong2018}. We do not pursue this here.
\end{remark}

\subsubsection{Corollary~\ref{cor:path}}
By Lemma~\ref{lem:donsker}(c), after label matching $c\mapsto\phi_{g_c}(\cdot;\eta_0)$ is H\"older continuous from ordered tuples $(\calC^K,d_\infty)$ into $L_2(P_0)$, hence into the covariance semimetric $\varsigma$; since $\G_0$ has $\varsigma$-uniformly-continuous paths, $c\mapsto\G_0(g_c)$ is continuous, i.e.\ the paths of $\G_0$ restricted to $\calF_{\mathrm{qt}}$ are $d$-continuous tangent directions for Lemma~\ref{lem:envelope}. The map $\nu\mapsto(\iota_1(\nu),\dots,\iota_{\overline K}(\nu))$ is then Hadamard differentiable at $\Psi_0$ (coordinatewise differentiability with linear derivatives implies joint), with derivative $\zeta\mapsto\{\zeta(g_{c^\star(K)})\}_K$; Theorem~\ref{thm:delta} gives both displayed convergences for $W$. For $\rho(\cdot)$, the map $(x_1,\dots,x_{\overline K})\mapsto(1-x_K/x_1)_K$ is continuously differentiable at the point $(W(1),\dots,W(\overline K))$ with $W(1)>0$, so the chain rule for Hadamard derivatives \citep[Lemma~3.9.3]{vdvwellner1996} yields the stated derivative
$\zeta\mapsto-W(1)^{-1}[\zeta(g_{c^\star(K)})-\{1-\rho(K)\}\zeta(g_{c^\star(1)})]$
and the joint limits. Validity of the credible sets follows from Theorem~\ref{thm:delta} once the limit laws are continuous, which holds whenever the relevant $\sigma_K>0$. \hfill$\square$

\subsection{Proof of Theorem~\ref{thm:profile}}\label{app:profile}

Throughout, fix the joint convergences of Corollary~\ref{cor:path}:
$\sqrt n\{\hat\rho(\cdot)-\rho_0(\cdot)\}\dto\mathbb H(\cdot)$ in $\R^{\overline K-1}$ and, conditionally,
$\sqrt n\{\rho^{(s)}(\cdot)-\hat\rho(\cdot)\}\dtow\mathbb H(\cdot)$, with $\mathbb H$ mean-zero Gaussian, $\sigma_K^2=\Var\,\mathbb H(K)>0$.

We use repeatedly the following fact, also consumed by the proof of Theorem~\ref{thm:honest}.

\begin{lemma}[Conditional quantile convergence]\label{lem:condquant}
If real random variables $T_n^{(s)}$ satisfy $T_n^{(s)}\dtow T$ with $\Law(T)$ having a continuous distribution function $F$, then $D_n:=\sup_x|F_n^{\mathrm w}(x)-F(x)|\to0$ in probability, where $F_n^{\mathrm w}(x)=\Pp_w(T_n^{(s)}\le x)$. Consequently, for any $\tau\in(0,1)$ at which $F$ is strictly increasing, the conditional $\tau$-quantile of $T_n^{(s)}$ converges in probability to $F^{-1}(\tau)$.
\end{lemma}

\textbf{Proof.} Along any subsequence there is a further subsequence on which the bounded-Lipschitz distance to $\Law(T)$ converges almost surely to zero. On that event, weak convergence plus continuity of $F$ gives $\sup_x|F_n^{\mathrm w}(x)-F(x)|\to0$ by P\'olya's argument, and quantile convergence follows from strict increase at $F^{-1}(\tau)$. Since every subsequence has such a further subsequence, both convergences hold in probability. \hfill$\square$

\textbf{Part (i).} Let $K_0^\star=K^\star_0(\gamma)$, so $\rho_0(K)<\gamma$ for $K<K_0^\star$ and $\rho_0(K_0^\star)>\gamma$ ($\gamma\notin\mathcal R_0$). For each $K$, $\rho^{(s)}(K)-\rho_0(K)=\{\rho^{(s)}(K)-\hat\rho(K)\}+\{\hat\rho(K)-\rho_0(K)\}$; the first is $\Op^{\mathrm w}(n^{-1/2})$ by conditional tightness, the second $\opro(1)$. Hence for $\eps=\min_K|\rho_0(K)-\gamma|/2>0$, $\Pp_w\{|\rho^{(s)}(K)-\rho_0(K)|\le\eps\ \forall K\}\to1$ in probability, on which event $K^{\star(s)}(\gamma)=K_0^\star$ by definition of the threshold functional. The statement for $\widehat K^\star(\gamma)$ uses only $\hat\rho(K)\pto\rho_0(K)$.

\textbf{Part (ii).} For $K<K_0$, $\rho_0(K)<\gamma$, so $\Pp_w\{\rho^{(s)}(K)<\gamma\}\to1$ in probability as in part (i); likewise $\Pp_w\{\rho^{(s)}(K_0+1)\ge\gamma\}\to1$ since $\rho_0(K_0+1)>\gamma$. On the intersection, $K^{\star(s)}(\gamma)=K_0$ if $\rho^{(s)}(K_0)\ge\gamma$ and $=K_0+1$ otherwise; this proves the two-point concentration, and moreover
\begin{equation*}
\Pp_w\bigl(K^{\star(s)}(\gamma)=K_0\bigr)
=\Pp_w\bigl(\rho^{(s)}(K_0)\ge\gamma\bigr)+\opro(1)
=1-G_n^{\mathrm w}\bigl(Z_n^-\bigr)+\opro(1),
\end{equation*}
where $G_n^{\mathrm w}(x)=\Pp_w[\sqrt n\{\rho^{(s)}(K_0)-\hat\rho(K_0)\}\le x]$ and $Z_n=\sqrt n\{\gamma-\hat\rho(K_0)\}$. By Lemma~\ref{lem:condquant} with $T=N(0,\sigma_{K_0}^2)$, $\sup_x|G_n^{\mathrm w}(x)-\Phi(x/\sigma_{K_0})|\pto0$, and since $\Phi$ is continuous the left limit costs nothing:
$\Pp_w\{K^{\star(s)}(\gamma)=K_0\}=1-\Phi(Z_n/\sigma_{K_0})+\opro(1)$.
Because $\gamma=\rho_0(K_0)$, $Z_n=-\sqrt n\{\hat\rho(K_0)-\rho_0(K_0)\}\dto N(0,\sigma_{K_0}^2)$, so $\Phi(Z_n/\sigma_{K_0})\dto\Phi(N(0,1))=\mathrm{Uniform}(0,1)$, and $1-\mathrm{Uniform}(0,1)\overset d=\mathrm{Uniform}(0,1)$.

\textbf{Part (iii).} \textbf{Scale.} Conditionally, $\sqrt n\{\rho^{(s)}(K)-\hat\rho(K)\}\dtow N(0,\sigma_K^2)$, whose distribution function is continuous and strictly increasing; by Lemma~\ref{lem:condquant} applied at $\tau=0.25,0.75$, the conditional quartiles converge, so $\sqrt n\,\hat\sigma(K)\pto\sigma_K$ for each $K$.

\textbf{Posterior quantile.} Write the band statistic as
$M_n^{(s)}=\max_K\,\bigl|\sqrt n\{\rho^{(s)}(K)-\hat\rho(K)\}\bigr|\big/\{\sqrt n\,\hat\sigma(K)\}$.
Since $\sqrt n\hat\sigma(K)\pto\sigma_K>0$ and the vector $\sqrt n(\rho^{(s)}-\hat\rho)\dtow\mathbb H$, the conditional law of $M_n^{(s)}$ converges weakly in probability to that of $M=\max_K|\mathbb H(K)|/\sigma_K$ (replace $\hat\sigma$ by $\sigma$ at the cost of an $\opro^{\mathrm w}(1)$ perturbation, then apply the continuous mapping $x\mapsto\max_K|x_K|/\sigma_K$ to bounded-Lipschitz test functions). The distribution function $F_M$ is continuous and strictly increasing on $(0,\infty)$. To see this, write $B_x=\prod_K[-x\sigma_K,x\sigma_K]$, so $F_M(x)=\Pp(\mathbb H\in B_x)$, and let $L$ denote the linear span of the support of $\mathbb H$, on which the Gaussian law of $\mathbb H$ has a positive density relative to Lebesgue measure on $L$. Then $B_x\cap L=x(B_1\cap L)$, and $B_1\cap L$ is a compact convex neighborhood of the origin in $L$ because $B_1$ contains an open ball around the origin, as each $\sigma_K>0$. For $x'>x>0$ the set $x'(B_1\cap L)\setminus x(B_1\cap L)$ has positive Lebesgue measure in $L$, so $F_M(x')>F_M(x)$, and continuity holds because the boundary of $x(B_1\cap L)$ in $L$ is Lebesgue null for every $x>0$. Lemma~\ref{lem:condquant} gives $q_{1-\alpha}\pto m_{1-\alpha}:=F_M^{-1}(1-\alpha)>0$.

\textbf{Coverage of the band.} The event $B_n=\{L_\rho(K)\le\rho_0(K)\le U_\rho(K)\ \forall K\le\overline K\}$ equals $\{T_n\le q_{1-\alpha}\}$ with
$T_n=\max_K|\hat\rho(K)-\rho_0(K)|/\hat\sigma(K)$ (and $K=1$ contributing zero). By Corollary~\ref{cor:path} and $\sqrt n\hat\sigma\pto\sigma$, $T_n\dto M$ via the continuous-mapping theorem; with $q_{1-\alpha}\pto m_{1-\alpha}$ and $F_M$ continuous at $m_{1-\alpha}$, Slutsky gives
$\Pp(B_n)=\Pp(T_n-q_{1-\alpha}\le0)\to F_M(m_{1-\alpha})=1-\alpha$.

\textbf{From band to profile, simultaneously in $\gamma$.} On $B_n$, fix any $\gamma\in(0,\rho_0(\overline K))$ and let $K_0^\star=K_0^\star(\gamma)$. Then $U_\rho(K_0^\star)\ge\rho_0(K_0^\star)\ge\gamma$, and for every $K'<K_0^\star$, $L_\rho(K')\le\rho_0(K')<\gamma$; hence $K_0^\star\in\widehat C(\gamma)$ by \eqref{eq:profileset}. This holds for all $\gamma$ on the single event $B_n$, so
$\Pp\{K^\star_0(\gamma)\in\widehat C(\gamma)\ \forall\gamma\}\ge\Pp(B_n)\to1-\alpha$.
Monotonization replaces $L_\rho$ by $\tilde L_\rho(K)=\max_{K'\le K}L_\rho(K')$ and $U_\rho$ by $\tilde U_\rho(K)=\min_{K'\ge K}U_\rho(K')$; on $B_n$, monotonicity of $\rho_0$ gives $\tilde L_\rho(K)\le\rho_0(K)\le\tilde U_\rho(K)$ for all $K$, so the containment argument is unchanged. \hfill$\square$

\subsection{Approximate minimization}\label{app:approxmin}
The theory of Section~\ref{sec:path_theory} is stated for exact minimizers, while the implementation of Section~\ref{sec:construction} returns approximate ones. The following proposition shows that an optimization gap that is negligible at the inferential scale changes nothing, because every downstream report depends on the computed codebooks only through the attained loss values.

\begin{proposition}[Approximate minimization suffices]\label{prop:approxmin}
Assume the conditions of Corollary~\ref{cor:path}, Theorem~\ref{thm:profile}, and Theorem~\ref{thm:honest} for the value statements, and those of Supplementary Theorem~\ref{thm:projection} and Theorem~\ref{thm:effects} for the mixture-projection statement. Let $\tilde c_K$ be the codebook computed from $\hPsi$ and $\tilde c^{(s)}_K$ the codebook computed from $\Psi^{(s)}$, and write the attained values $\widetilde W(K)=\hPsi(g_{\tilde c_K})$ and $\widetilde W^{(s)}(K)=\Psi^{(s)}(g_{\tilde c^{(s)}_K})$. Suppose
\begin{equation*}
\max_{K\le\overline K}\Bigl\{\widetilde W(K)-\inf_{c\in\calC^K}\hPsi(g_c)\Bigr\}=\opro(n^{-1/2}),
\qquad
\max_{K\le\overline K}\Bigl\{\widetilde W^{(s)}(K)-\inf_{c\in\calC^K}\Psi^{(s)}(g_c)\Bigr\}=\opro^{\mathrm w}(n^{-1/2}),
\end{equation*}
the second in the conditional-in-probability mode $\opro^{\mathrm w}$, meaning $\E_w(|\cdot|\wedge1)\to0$ in outer probability, as defined at the start of the proof of Theorem~\ref{thm:bvm}\textup{(ii)} in Supplementary Section~\ref{app:bvm}. The unconditional hypotheses are stated under sampling from $P_0$. Then the conclusions of Corollary~\ref{cor:path}, Theorem~\ref{thm:profile}, and Theorem~\ref{thm:honest} hold verbatim with the exact minimized values replaced by the attained values $\widetilde W(K)$ and $\widetilde W^{(s)}(K)$ throughout the definitions of $\rho$, the band \eqref{eq:band}, and the inversion \eqref{eq:profileset}. The analogous statement holds for the mixture projection. Suppose the approximate minimizers are consistent, $\tilde\beta(K)\pto\beta^\star(K)$, and the per-draw minimizers conditionally consistent, $\Pp_w\{|\tilde\beta^{(s)}(K)-\beta^\star(K)|>\eps\}\pto0$ for every $\eps>0$, both after the label alignment of Supplementary Theorem~\ref{thm:projection}. A sufficient condition is that they nearly minimize the corrected criterion value, $\hPsi(\ell_{\tilde\beta(K)})\le\inf_\beta\hPsi(\ell_\beta)+\opro(1)$ and $\Psi^{(s)}(\ell_{\tilde\beta^{(s)}(K)})\le\inf_\beta\Psi^{(s)}(\ell_\beta)+\opro^{\mathrm w}(1)$, from which consistency follows because a near-minimizer of the criterion value up to $\opro(1)$ lies, with probability tending to one, in any fixed neighborhood of the unique aligned minimizer, by compactness of $\calB_K$ and continuity of $\beta\mapsto\Psi_0(\ell_\beta)$, the standard argmin consistency argument. The per-draw version holds in conditional probability by the same argument applied to the weighted criterion. If in addition they solve the corrected stationarity condition up to $\opro(n^{-1/2})$, respectively $\opro^{\mathrm w}(n^{-1/2})$, then the conclusions of Supplementary Theorem~\ref{thm:projection} and Theorem~\ref{thm:effects} hold with $\tilde\beta(K)$ and $\tilde\beta^{(s)}(K)$ in place of the exact minimizers.
\end{proposition}

\textbf{Proof of Proposition~\ref{prop:approxmin}.}
Write $\widehat W(K)=\inf_{c\in\calC^K}\hPsi(g_c)$ and $W^{(s)}(K)=\inf_{c\in\calC^K}\Psi^{(s)}(g_c)$ for the exact minimized values entering Corollary~\ref{cor:path}, and recall the attained values $\widetilde W(K)=\hPsi(g_{\tilde c_K})$, $\widetilde W^{(s)}(K)=\Psi^{(s)}(g_{\tilde c^{(s)}_K})$. Because an attained value never falls below the corresponding infimum, the hypotheses give
\begin{equation*}
0\le\max_{K\le\overline K}\{\widetilde W(K)-\widehat W(K)\}=\opro(n^{-1/2}),
\qquad
0\le\max_{K\le\overline K}\{\widetilde W^{(s)}(K)-W^{(s)}(K)\}=\opro^{\mathrm w}(n^{-1/2}).
\end{equation*}
Every reported quantization functional is computed from the path values alone. Write $\tilde\rho(K)=1-\widetilde W(K)/\widetilde W(1)$ and $\tilde\rho^{(s)}(K)=1-\widetilde W^{(s)}(K)/\widetilde W^{(s)}(1)$ for the attained-value curves actually reported by Algorithm~\ref{alg:main}. The heterogeneity map $(x_1,\dots,x_{\overline K})\mapsto(1-x_K/x_1)_K$ is Lipschitz on a neighborhood of $(W(1),\dots,W(\overline K))$ because $W(1)>0$, so the two displays above give, uniformly in $K$,
\begin{equation*}
\sqrt n\{\tilde\rho(K)-\hat\rho(K)\}=\opro(1),
\qquad
\sqrt n\{\tilde\rho^{(s)}(K)-\rho^{(s)}(K)\}=\opro^{\mathrm w}(1).
\end{equation*}
Hence the attained-value draws $\sqrt n\{\tilde\rho^{(s)}(\cdot)-\tilde\rho(\cdot)\}$ differ from the exact-value draws $\sqrt n\{\rho^{(s)}(\cdot)-\hat\rho(\cdot)\}$ by $\opro^{\mathrm w}(1)$ uniformly in $K$, so by the bounded-Lipschitz perturbation fact recorded in the proof of Theorem~\ref{thm:bvm}\textup{(ii)} their conditional weak limit is again $\mathbb H$, while $\sqrt n\{\tilde\rho(\cdot)-\rho_0(\cdot)\}\dto\mathbb H$ unconditionally by Slutsky's lemma and Corollary~\ref{cor:path}. The scale estimate and the simultaneous quantile are recovered exactly as in the proof of Theorem~\ref{thm:profile}: Lemma~\ref{lem:condquant} applied to the attained-value draws at $\tau=0.25,0.75$ gives $\sqrt n\,\hat\sigma(K)\pto\sigma_K$, and applied to the sup-$t$ statistic at $\tau=1-\alpha$ gives $q_{1-\alpha}\pto m_{1-\alpha}$. Replacing the exact draws by the attained draws perturbs $\sqrt n\,\hat\sigma(K)$ by $\opro(1)$, not $\opro^{\mathrm w}(1)$, because $\hat\sigma(K)$ is a data-measurable functional of the posterior draws, and this does not disturb $\sqrt n\,\hat\sigma(K)\pto\sigma_K>0$. Finally the band \eqref{eq:band} is Lipschitz in $(\hat\rho,\hat\sigma,q_{1-\alpha})$ and the inversion \eqref{eq:profileset} is monotone in the band endpoints, so by Slutsky's lemma and the continuous mapping theorem the band-coverage and set-valued containment steps in the proofs of Theorem~\ref{thm:profile} and Theorem~\ref{thm:honest} persist with the attained values in place of the exact minimized ones. Under the drifting laws $Q_n$ of Theorem~\ref{thm:honest}, the optimization-gap bounds, including the conditional one measured through $\E_w(|\cdot|\wedge1)$, are data-measurable and $\opro(1)$ under $P_0^{\otimes n}$; they remain $\opro(1)$ under $Q_n$ by Lemma~\ref{lem:transfer} together with the contiguity established in Step~2 of the proof of Theorem~\ref{thm:honest}, after which the remaining steps of that proof apply to the attained-value quantities unchanged.

For the mixture projection, let $\tilde\beta(K)$ and $\tilde\beta^{(s)}(K)$ be as in the statement, so that $\tilde\beta(K)\pto\beta^\star(K)$, $\Pp_w\{|\tilde\beta^{(s)}(K)-\beta^\star(K)|>\eps\}\pto0$, and
\begin{equation*}
\sqrt n\,\hPsi(\nabla_\beta\ell_{\tilde\beta(K)})=\opro(1),
\qquad
\sqrt n\,\Psi^{(s)}(\nabla_\beta\ell_{\tilde\beta^{(s)}(K)})=\opro^{\mathrm w}(1).
\end{equation*}
Abbreviate $\beta^\star=\beta^\star(K)$, $V=V_\beta$, and $G_0(\beta)=\Psi_0(\nabla_\beta\ell_\beta)$, so that $G_0(\beta^\star)=0$ and $G_0(\beta)=V(\beta-\beta^\star)+o(|\beta-\beta^\star|)$ with $V$ nonsingular by Assumption~\ref{ass:proj}. The pointer in the statement is to an explicit near-solution $Z$-estimation argument, which we now give, rather than to the Hadamard argument of the proof of Theorem~\ref{thm:projection}, whose functional $T_\beta$ of Lemma~\ref{lem:zfun} carries a tolerance $\norm{\nu-\Psi_0}_\infty^2$ tighter than $\opro(n^{-1/2})$, so that near-solutions are not admissible values of $T_\beta$. The gradient class $\calG_\nabla=\{\nabla_\beta\ell_\beta^{(j)}:\beta\in\calB_K,\,j\}\subset\calF_{\mathrm{str}}\subset\calF_{\mathrm{eff}}$ by Assumption~\ref{ass:class}(iii), and the corrected-score class $\{\phi_f(\cdot;\eta_0):f\in\calG_\nabla\}$ is $P_0$-Donsker by Lemma~\ref{lem:donsker}(b), with $\beta\mapsto\phi_{\nabla_\beta\ell_\beta}(\cdot;\eta_0)$ Lipschitz into $L_2(P_0)$.

\textbf{Rate.} By the uniform asymptotic linearity of Theorem~\ref{thm:bvm}\textup{(i)} restricted to $\calF_{\mathrm{eff}}$, $\sup_{f\in\calF_{\mathrm{eff}}}\sqrt n|(\hPsi-\Psi_0)(f)|=\Op(1)$, so
$\sqrt n|G_0(\tilde\beta(K))|\le\sqrt n|\hPsi(\nabla_\beta\ell_{\tilde\beta(K)})|+\sup_{f\in\calF_{\mathrm{eff}}}\sqrt n|(\hPsi-\Psi_0)(f)|=\Op(1)$. The local identifiability $|G_0(\beta)|\ge c|\beta-\beta^\star|$ for $\beta$ in a neighborhood of $\beta^\star$, which follows from the nonsingular $V$, together with $\tilde\beta(K)\pto\beta^\star$, gives $|\tilde\beta(K)-\beta^\star|=\Op(n^{-1/2})$.

\textbf{Expansion.} A first-order Taylor expansion of $G_0$ at $\beta^\star$ gives $G_0(\tilde\beta(K))=V(\tilde\beta(K)-\beta^\star)+\opro(n^{-1/2})$ at the rate just obtained. The uniform asymptotic linearity of Theorem~\ref{thm:bvm}\textup{(i)} reduces $\sqrt n(\hPsi-\Psi_0)$ on $\calF_{\mathrm{eff}}$ to the empirical process $\G_n\phi_\cdot(\cdot;\eta_0)$, whose Donsker property and the $L_2(P_0)$-continuity of $\beta\mapsto\phi_{\nabla_\beta\ell_\beta}$ give the stochastic-equicontinuity increment $\sqrt n(\hPsi-\Psi_0)(\nabla_\beta\ell_{\tilde\beta(K)}-\nabla_\beta\ell_{\beta^\star})=\opro(1)$. Combining with $\Psi_0(\nabla_\beta\ell_{\beta^\star})=0$ and the stationarity gap,
\begin{equation*}
\sqrt n\{\tilde\beta(K)-\beta^\star\}=-V^{-1}\sqrt n\,\hPsi(\nabla_\beta\ell_{\beta^\star})+\opro(1).
\end{equation*}

\textbf{Conditional analogue.} The two steps run for the draw with the conditional consistency hypothesis in place of $\tilde\beta(K)\pto\beta^\star$. For the rate, parts \textup{(i)} and \textup{(ii)} of Theorem~\ref{thm:bvm} are used jointly. Writing $\Psi^{(s)}-\Psi_0=(\Psi^{(s)}-\hPsi)+(\hPsi-\Psi_0)$ and applying the triangle inequality to the two tight limits gives the conditional tightness $\sup_{f\in\calF_{\mathrm{eff}}}\sqrt n|(\Psi^{(s)}-\Psi_0)(f)|=\Op^{\mathrm w}(1)$, whence $\sqrt n|G_0(\tilde\beta^{(s)}(K))|=\Op^{\mathrm w}(1)$ and, by the local identifiability of $\beta^\star$, $\norm{\tilde\beta^{(s)}(K)-\beta^\star}=\Op^{\mathrm w}(n^{-1/2})$. For the expansion, the weighted equicontinuity increment $\sqrt n(\Psi^{(s)}-\Psi_0)(\nabla_\beta\ell_{\tilde\beta^{(s)}(K)}-\nabla_\beta\ell_{\beta^\star})=\opro^{\mathrm w}(1)$ by the joint asymptotic equicontinuity fact recorded at the start of Supplementary Section~\ref{app:bvm}. Combining as in the point case,
$\sqrt n\{\tilde\beta^{(s)}(K)-\beta^\star\}=-V^{-1}\sqrt n\,\Psi^{(s)}(\nabla_\beta\ell_{\beta^\star})+\opro^{\mathrm w}(1)$.

\textbf{Conclusion.} The exact minimizers $\hat\beta(K)$ and $\beta^{(s)}(K)$ of Theorem~\ref{thm:projection} satisfy the identical linear expansions, being the special case of zero stationarity gap, so subtracting gives $\sqrt n\{\tilde\beta(K)-\hat\beta(K)\}=\opro(1)$ and $\sqrt n\{\tilde\beta^{(s)}(K)-\beta^{(s)}(K)\}=\opro^{\mathrm w}(1)$. For the numerator and denominator evaluations defining $\psi_{h,a}(K)$, decompose each difference into a centered process increment and a deterministic drift. For the point evaluation,
\begin{equation*}
\hPsi(f^N_{h,a;\tilde\beta(K)})-\hPsi(f^N_{h,a;\hat\beta(K)})
=(\hPsi-\Psi_0)\bigl(f^N_{h,a;\tilde\beta(K)}-f^N_{h,a;\hat\beta(K)}\bigr)
+\bigl\{\Psi_0(f^N_{h,a;\tilde\beta(K)})-\Psi_0(f^N_{h,a;\hat\beta(K)})\bigr\},
\end{equation*}
and likewise for $f^D_{h;\beta}$. The centered increment is $\opro(n^{-1/2})$ by the same equicontinuity step, because $\{f^N_{h,a;\beta}\}$ and $\{f^D_{h;\beta}\}$ lie in $\calF_{\mathrm{eff}}$ with $\beta$-index Lipschitz into $\varsigma$ and $\tilde\beta(K)-\hat\beta(K)\pto0$. The drift is bounded by the sup-norm Lipschitz continuity of $\beta\mapsto f_\beta$ from Assumption~\ref{ass:class}(iii) times $\norm{\tilde\beta(K)-\hat\beta(K)}=\opro(n^{-1/2})$, hence is itself $\opro(n^{-1/2})$. The conditional evaluations decompose the same way, with the weighted centered increment controlled by the conditional equicontinuity step recorded in the proof of Theorem~\ref{thm:bvm}\textup{(ii)} and the drift bounded by the same Lipschitz constant times $\norm{\tilde\beta^{(s)}(K)-\beta^{(s)}(K)}=\opro^{\mathrm w}(n^{-1/2})$, giving $\opro^{\mathrm w}(n^{-1/2})$. Thus the numerator and denominator evaluations are perturbed by $\opro(n^{-1/2})$, respectively $\opro^{\mathrm w}(n^{-1/2})$, when $\tilde\beta$ replaces $\hat\beta$. Hence the conclusions of Theorem~\ref{thm:projection} and Theorem~\ref{thm:effects} transfer to the approximate minimizers by Slutsky's lemma and the bounded-Lipschitz perturbation fact recorded in the proof of Theorem~\ref{thm:bvm}\textup{(ii)}. The argument requires the stationarity gap only at the $\opro(n^{-1/2})$ rate, rather than the tighter tolerance built into $T_\beta$ in Lemma~\ref{lem:zfun}. \hfill$\square$

\subsection{Proof of Theorem~\ref{thm:impossibility}}\label{app:impossibility}

\textbf{Step 1 (a bounded least-favorable path).}
Let $\mathrm{IF}$ denote the influence function of $\rho(K_0)$, i.e.\ the $K_0$ coordinate of the limit $\mathbb H$ in Corollary~\ref{cor:path}:
\begin{equation*}
\mathrm{IF}(o)
=-\frac1{W(1)}\Bigl[\bigl\{\phi_{g_{c^\star(K_0)}}(o;\eta_0)-W(K_0)\bigr\}
-\{1-\rho_0(K_0)\}\bigl\{\phi_{g_{c^\star(1)}}(o;\eta_0)-W(1)\bigr\}\Bigr],
\end{equation*}
mean zero with $\Var_0(\mathrm{IF})=\sigma_{K_0}^2>0$.  For a truncation level $M$, set $s_M=\mathrm{IF}\,\ind{|\mathrm{IF}|\le M}-\E_0[\mathrm{IF}\,\ind{|\mathrm{IF}|\le M}]$ and $\tilde s=s_M/\norm{s_M}_{P_0,2}$, so that $\tilde s$ is bounded, mean zero, with $\E_0\tilde s^2=1$ and
\begin{equation*}
b:=\E_0[\mathrm{IF}\,\tilde s]
=\frac{\E_0[\mathrm{IF}^2\ind{|\mathrm{IF}|\le M}]}{\norm{s_M}_{P_0,2}}
\ \longrightarrow\ \sigma_{K_0}
\quad(M\to\infty);
\end{equation*}
fix $M$ so large that $b\ge\sigma_{K_0}-\eps$ (Cauchy--Schwarz gives $b\le\sigma_{K_0}$ in any case).  Define $dP_t=(1+t\tilde s)\,dP_0$ for $|t|\le t_0:=1/(2\norm{\tilde s}_\infty)$.  Each $P_t$ is a probability law with density factor in $[1/2,3/2]$, hence with the same support in $(X,A,Y)$ as $P_0$; in particular $|Y|\le B_Y$ almost surely is preserved.  Its conditional treatment probabilities satisfy $\pi^{P_t}_a\ge\eps_\pi/3$, and $|\mu^{P_t}_a|\le B_Y$ because each $\mu^{P_t}_a(x)$ is a weighted average of outcomes in $[-B_Y,B_Y]$.  Since $\tilde s$ is bounded, the conditional mean formula gives $\sup_{a,x}|\mu^{P_t}_a(x)-\mu_a(x)|\le C|t|$, and hence $\sup_x\norm{U_t(x)-U(x)}\le C|t|$.  The buffer condition in the theorem therefore implies $U_t(X)\in\calC$ for all sufficiently small $|t|$.  Thus the local paths used below satisfy the identification, positivity, bounded-outcome, and feature-support conditions uniformly (with the positivity constant possibly reduced).  Although $P_t$ is constructed as an observed-data law, it is interpreted through the standard full-data embedding: draw $X$ from its marginal law, draw potential outcomes from the arm-specific conditional outcome laws, and then draw $A$ from the conditional treatment law, which preserves no unmeasured confounding and realizes the same observed distribution.  The path is differentiable in quadratic mean at $t=0$ with score $\tilde s$: pointwise, $t^{-1}\{(1+t\tilde s)^{1/2}-1\}\to\tilde s/2$ with the difference quotient dominated by $C\tilde s^2$, so the defining $L_2(P_0)$ limit holds by dominated convergence \citep[Section~7.2]{vandervaart1998}.

\textbf{Step 2 (drift of the profile along the path).}
We first claim, uniformly over $c\in\calC^K$ and $K\le\overline K$,
\begin{equation}\label{eq:driftclaim}
\Psi_{P_t}(g_c)
=\Psi_0(g_c)+t\,\E_0\bigl[\{\phi_{g_c}(O;\eta_0)-\Psi_0(g_c)\}\,\tilde s\bigr]
+O\bigl(t^{1+(\am\wedge1)}\bigr).
\end{equation}
Write $\mu^{P_t}_a(x)=\E_0[Y(1+t\tilde s)\mid A=a,X=x]/\E_0[(1+t\tilde s)\mid A=a,X=x]$; a geometric expansion with bounded numerator and denominator gives $\mu^{P_t}_a=\mu_a+t\dot\mu_a+O(t^2)$ uniformly, where
\begin{equation*}
\dot\mu_a(x)=\E_0\bigl[\{Y-\mu_a(X)\}\tilde s\given A=a,X=x\bigr]
=\E_0\bigl[R_a(O;\eta_0)\,\tilde s\given X=x\bigr],
\end{equation*}
the second equality by conditioning on $(A,X)$ and dividing by $\pi_a$.  Hence $U_t:=H\bmu^{P_t}=U+tH\dot\bmu+O(t^2)$ in supremum norm.  Decompose $\Psi_{P_t}(g_c)=\E_0[(1+t\tilde s)\,g_c(U_t)]=\E_0[g_c(U_t)]+t\E_0[\tilde s\,g_c(U_t)]$.  In the second term, $|\E_0[\tilde s\{g_c(U_t)-g_c(U)\}]|\le\norm{\tilde s}_\infty\cdot4\diam\calC\cdot\norm{U_t-U}_\infty=O(t)$, so it equals $t\E_0[\tilde s\,g_c(U)]+O(t^2)$.  For the first term, apply Lemma~\ref{lem:geometry} with $\bar u=U$, $u=U_t$, $\Delta_t=U_t-U$, $\norm{\Delta_t}_\infty\le Ct$:
\begin{gather*}
g_c(U_t)=g_c(U)+\nabla g_c(U)^\top\Delta_t+\mathrm{rem}_t,\\
|\mathrm{rem}_t|\le\norm{\Delta_t}^2+\sum_{(h,j):\,c_h\ne c_j}2\norm{\Delta_t}\,\norm{c_h-c_j}\,\ind{\dist(U,B_{hj})\le2\norm{\Delta_t}},
\end{gather*}
and Assumption~\ref{ass:margin} bounds each crossing probability by $C_M(2Ct)^\am$, so $\E_0|\mathrm{rem}_t|\le C(t^2+t^{1+\am})$, uniformly in $c$.  Finally, by the tower property,
$\E_0[\nabla g_c(U)^\top H\dot\bmu(X)]=\E_0[\nabla g_c(U)^\top H\,\E_0(R\tilde s\mid X)]=\E_0[\nabla g_c(U)^\top HR\,\tilde s]$,
so the linear terms assemble to $t\E_0[\{g_c(U)+\nabla g_c(U)^\top HR\}\tilde s]=t\E_0[\phi_{g_c}\tilde s]$, which equals the centered form in \eqref{eq:driftclaim} since $\E_0\tilde s=0$.  This proves \eqref{eq:driftclaim}.

Next, the minimized values.  Uniform convergence \eqref{eq:driftclaim} and uniqueness of the minimizers (Assumption~\ref{ass:quant}(i)) give, by the standard argmin argument, that any minimizer $c_t$ of $\Psi_{P_t}(g_\cdot)$ over $\calC^K$ converges to $c^\star(K)$ as $t\to0$; and $c\mapsto\E_0[\{\phi_{g_c}-\Psi_0(g_c)\}\tilde s]$ is continuous by Lemma~\ref{lem:donsker}(c).  Sandwiching now yields the envelope (Danskin) expansion: from above, $W_{P_t}(K)\le\Psi_{P_t}(g_{c^\star(K)})=W(K)+t\,a_K+O(t^{1+\am\wedge1})$; from below, $W_{P_t}(K)=\Psi_{P_t}(g_{c_t})\ge\Psi_0(g_{c_t})+t\E_0[\{\phi_{g_{c_t}}-\Psi_0(g_{c_t})\}\tilde s]-O(t^{1+\am\wedge1})\ge W(K)+t\E_0[\{\phi_{g_{c_t}}-\Psi_0(g_{c_t})\}\tilde s]-O(t^{1+\am\wedge1})$, and the drift coefficient at $c_t$ converges to $a_K$ by continuity in $c$ and $c_t\to c^\star(K)$.  Hence
\begin{equation*}
W_{P_t}(K)=W(K)+t\,a_K+o(t),
\qquad
a_K:=\E_0\bigl[\{\phi_{g_{c^\star(K)}}(O;\eta_0)-W(K)\}\,\tilde s\bigr],
\end{equation*}
for each $K\le\overline K$.  The quotient rule then gives
$\rho_{P_t}(K_0)=\rho_0(K_0)+t\,\E_0[\mathrm{IF}\,\tilde s]+o(t)=\gamma+bt+o(t)$,
and $\rho_{P_t}(K)=\rho_0(K)+O(t)$ for every other $K$.  Set $t_n^{\pm}=\pm h/(b\sqrt n)$ and $P_n^{\pm}=P_{t_n^{\pm}}$, so that $\rho_{P_n^{\pm}}(K_0)=\gamma\pm h/\sqrt n+o(n^{-1/2})$ while the flanking strict inequalities $\rho(K)<\gamma$ for $K\le K_0-1$ and $\rho(K_0+1)>\gamma$ persist for all large $n$.  Under $P_n^{+}$, $\rho(K_0)>\gamma$ and no smaller order reaches $\gamma$, so $K^\star(\gamma)=K_0$; under $P_n^{-}$, $\rho(K_0)<\gamma<\rho(K_0+1)$, so $K^\star(\gamma)=K_0+1$.  This is part (i), and the density factor is $1+t_n^{\pm}\tilde s=1+O(n^{-1/2})$ in supremum norm as claimed.

\textbf{Step 3 (contiguity).}
Differentiability in quadratic mean with score $\tilde s$, $\E_0\tilde s^2=1$, gives the local asymptotic normality expansion
$\log\,d(P_{u/\sqrt n})^{\otimes n}/dP_0^{\otimes n}=(u/\sqrt n)\sum_i\tilde s(O_i)-u^2/2+\opro(1)$ under $P_0^{\otimes n}$, for each fixed $u$ \citep[Theorem~7.2]{vandervaart1998}; mutual contiguity of $(P_{u/\sqrt n})^{\otimes n}$ and $P_0^{\otimes n}$ follows from Le Cam's first lemma \citep[Example~6.5]{vandervaart1998}, and mutual contiguity of the two perturbed sequences follows by transitivity through $P_0^{\otimes n}$.  This is part (ii).

\textbf{Step 4 (the testing bound).}
A selector $\widetilde K_n$ induces the test ``declare $+$'' iff $\widetilde K_n=K_0$, whose error sum is dominated by the selector's:
$\Pp^{+}_n(\widetilde K_n\ne K_0)$ is its first error exactly, and its second error satisfies $\Pp^{-}_n(\widetilde K_n=K_0)\le\Pp^{-}_n(\widetilde K_n\ne K_0+1)$.  The minimal sum of error probabilities over all tests of $(P^{-}_n)^{\otimes n}$ against $(P^{+}_n)^{\otimes n}$ is $\int\min\{d(P^{+}_n)^{\otimes n},d(P^{-}_n)^{\otimes n}\}=\E_{-}\bigl[\min(e^{\Lambda_n},1)\bigr]$, where $\Lambda_n$ is the log-likelihood ratio of $+$ over $-$ and $\E_{-}$ is expectation under $(P^{-}_n)^{\otimes n}$.  Differencing the two LAN expansions at $u_\pm=\pm h/b$ (the quadratic terms cancel),
$\Lambda_n=(2h/b)\,n^{-1/2}\sum_i\tilde s(O_i)+\opro(1)$ under $P_0^{\otimes n}$, hence also under $(P^{-}_n)^{\otimes n}$ by contiguity; and by Le Cam's third lemma \citep[Example~6.7]{vandervaart1998}, $n^{-1/2}\sum_i\tilde s(O_i)\dto N(-h/b,\,1)$ under $(P^{-}_n)^{\otimes n}$.  Therefore $\Lambda_n\dto Z\sim N(-2\theta^2,\,4\theta^2)$ under $(P^{-}_n)^{\otimes n}$, with $\theta=h/b$.  Since $z\mapsto\min(e^z,1)$ is bounded and continuous,
\begin{align*}
\E_{-}\bigl[\min(e^{\Lambda_n},1)\bigr]
\ \longrightarrow\
\E\bigl[\min(e^{Z},1)\bigr]
&=\Pp(Z\ge0)+\E\bigl[e^{Z}\ind{Z<0}\bigr]\\
&=\Phi(-\theta)+e^{-2\theta^2+2\theta^2}\,\Phi\Bigl(\frac{2\theta^2-4\theta^2}{2\theta}\Bigr)
=2\,\Phi(-\theta),
\end{align*}
using the lognormal identity $\E[e^{Z}\ind{Z<x}]=e^{m+v/2}\Phi\{(x-m-v)/\sqrt v\}$ for $Z\sim N(m,v)$.  Since $b\ge\sigma_{K_0}-\eps$, $2\Phi(-h/b)\ge2\Phi\{-h/(\sigma_{K_0}-\eps)\}$, which is part (iii). \hfill$\square$

\subsection{Proof of Theorem~\ref{thm:honest}}\label{app:honest}

Throughout this subsection, constants $C$ depend only on the declared constants and on $(\bar S,B_u)$. We first record a transfer lemma used repeatedly.

\begin{lemma}[Contiguity transfer]\label{lem:transfer}
Let $Q_n$ and $P_0^{\otimes n}$ be probability laws on the sample space of $O_{1:n}$ such that $Q_n$ and $P_0^{\otimes n}$ are mutually absolutely continuous for each $n$ and $(Q_n)$ is contiguous with respect to $(P_0^{\otimes n})$. If random elements $X_n$ satisfy $X_n\to0$ in outer $P_0^{\otimes n}$-probability, then $X_n\to0$ in outer $Q_n$-probability.
\end{lemma}

\begin{proof}
Fix $\eps>0$ and let $A_n^\ast$ be a measurable cover of the event $\{|X_n|>\eps\}$ under $P_0^{\otimes n}$, so that $P_0^{\otimes n}(A_n^\ast)\to0$. Mutual absolute continuity implies that measurable covers under the two laws agree up to null sets, and contiguity gives $Q_n(A_n^\ast)\to0$, hence $Q_n^\ast(|X_n|>\eps)\to0$. \hfill$\square$
\end{proof}

\textbf{Step 0 (reduction to sequences).} Suppose the display of the theorem fails. Then there are $\eps>0$, a subsequence $(n_j)$, and pairs $(u_j,s_j)\in[-B_u,B_u]\times\calS$ with coverage probability below $1-\alpha-\eps$ for every $j$. Since $[-B_u,B_u]$ is compact and $\calS$ is totally bounded in $L_2(P_0)$, we may pass to a further subsequence along which $u_j\to u^\ast\in[-B_u,B_u]$ and $s_j\to s^\ast$ in $L_2(P_0)$. The limit $s^\ast$ inherits $\E_0s^\ast=0$, and $\norm{s^\ast}_\infty\le\bar S$ after passing to an almost surely convergent further subsequence. It therefore suffices to derive a contradiction by proving the following claim along the extracted subsequence, which we relabel as the full sequence for notational ease,
\begin{equation}\label{eq:seqclaim}
\liminf_{n\to\infty}\ \Pp_{n,u_n,s_n}\Bigl(K^\star_{P_n}(\gamma)\in\widehat C(\gamma)\ \text{ for every }\gamma\in\bigl(0,\rho_{P_n}(\overline K)\bigr)\Bigr)\ \ge\ 1-\alpha
\end{equation}
along every sequence with $|u_n|\le B_u$, $u_n\to u^\ast$, $s_n\in\calS$, and $s_n\to s^\ast$ in $L_2(P_0)$, where we abbreviate $t_n=u_n/\sqrt n$, $P_n=P_{t_n,s_n}$, and $Q_n=P_n^{\otimes n}$.

\textbf{Step 1 (uniform regularity of the perturbed laws).} For all $n$ with $|t_n|\bar S\le1/2$, the density factor $1+t_ns_n$ lies in $[1/2,3/2]$. As in Step~1 of the proof of Theorem~\ref{thm:impossibility}, whose bounds depend on the score only through its supremum norm, $P_n$ preserves $|Y|\le B_Y$ almost surely, satisfies positivity with constant $\eps_\pi/3$, obeys $\sup_{a,x}|\mu^{P_n}_a(x)-\mu_a(x)|\le C|t_n|$, hence $\sup_x\norm{U_{P_n}(x)-U(x)}\le C|t_n|$, and the buffered feature-support condition places $U_{P_n}(X)$ in $\calC$ for all large $n$. The full-data embedding of that proof applies verbatim, so each $P_n$ satisfies the identification conditions.

\textbf{Step 2 (triangular local asymptotic normality and contiguity).} Let $\Lambda_n=\log dQ_n/dP_0^{\otimes n}=\sum_{i\le n}\log\{1+t_ns_n(O_i)\}$. Since $|t_ns_n|\le1/2$ pointwise, a third-order Taylor expansion of $x\mapsto\log(1+x)$ gives
\begin{equation*}
\Lambda_n\;=\;t_n\sum_{i\le n}s_n(O_i)\;-\;\frac{t_n^2}2\sum_{i\le n}s_n(O_i)^2\;+\;O\bigl(nt_n^3\bar S^3\bigr).
\end{equation*}
Because $\E_0s_n=0$, the linear term equals $u_n\G_n(s_n)$, and $\Var_0\{\G_n(s_n-s^\ast)\}\le\norm{s_n-s^\ast}_{P_0,2}^2\to0$ gives $\G_n(s_n)=\G_n(s^\ast)+\opro(1)$. For the quadratic term, Chebyshev's inequality with the bound $\bar S^4/n$ on the variance of the empirical second moment, together with $\E_0s_n^2\to\E_0(s^\ast)^2$, gives $t_n^2\sum_is_n(O_i)^2/2=(u^\ast)^2\E_0(s^\ast)^2/2+\opro(1)$. Hence
\begin{equation}\label{eq:lanhonest}
\Lambda_n\;=\;u^\ast\,\G_n(s^\ast)\;-\;\tfrac12(u^\ast)^2\,\E_0(s^\ast)^2\;+\;\opro(1)
\qquad\text{under }P_0^{\otimes n},
\end{equation}
so $\Lambda_n\dto N(-\sigma_\ast^2/2,\ \sigma_\ast^2)$ under $P_0^{\otimes n}$ with $\sigma_\ast^2=(u^\ast)^2\E_0(s^\ast)^2$, possibly degenerate. Le Cam's first lemma \citep[Example~6.5]{vandervaart1998} yields mutual contiguity of $(Q_n)$ and $(P_0^{\otimes n})$, and mutual absolute continuity at each fixed $n$ holds because the density factor is bounded away from zero and infinity.

\textbf{Step 3 (drift of the population path).} Display \eqref{eq:driftclaim} in the proof of Theorem~\ref{thm:impossibility} was established for an arbitrary bounded mean-zero score, with constants depending on the score only through its supremum bound. Applied at $t=t_n$ with score $s_n$, it gives, uniformly in $c\in\calC^K$ and $K\le\overline K$,
\begin{equation*}
\Psi_{P_n}(g_c)\;=\;\Psi_0(g_c)+t_n\,\E_0\bigl[\{\phi_{g_c}(O;\eta_0)-\Psi_0(g_c)\}\,s_n\bigr]+O\bigl(|t_n|^{1+\am\wedge1}\bigr).
\end{equation*}
By Cauchy--Schwarz and the constant envelope of the score class, the drift coefficient converges to its value at $s^\ast$ uniformly in $c$ as $\norm{s_n-s^\ast}_{P_0,2}\to0$. The envelope sandwich of Step~2 of that proof, which uses only the uniqueness of the optimal codebooks in Assumption~\ref{ass:quant}\textup{(i)}, the continuity of $c\mapsto\phi_{g_c}$ into $L_2(P_0)$ from Lemma~\ref{lem:donsker}(c), and the uniform convergence just displayed, then gives, for every $K\le\overline K$,
\begin{equation}\label{eq:Wdrifthonest}
\sqrt n\,\{W_{P_n}(K)-W(K)\}\;\longrightarrow\;u^\ast a_K(s^\ast),
\qquad
a_K(s)=\E_0\bigl[\{\phi_{g_{c^\star(K)}}(O;\eta_0)-W(K)\}\,s\bigr].
\end{equation}
Since $W(1)>0$, the quotient rule applied to $\rho_{P_n}(K)=1-W_{P_n}(K)/W_{P_n}(1)$ gives
\begin{equation}\label{eq:rhodrifthonest}
\sqrt n\,\{\rho_{P_n}(K)-\rho_0(K)\}\;\longrightarrow\;u^\ast b_K(s^\ast),
\qquad
b_K(s)=\E_0\{\mathrm{IF}_K\,s\},
\end{equation}
where $\mathrm{IF}_K$ denotes the influence function of $\rho(K)$ displayed in Step~1 of the proof of Theorem~\ref{thm:impossibility} with $K_0$ replaced by $K$. In particular $\rho_{P_n}(K)\to\rho_0(K)$ for every $K\le\overline K$.

\textbf{Step 4 (local regularity of the corrected point path).} By Theorem~\ref{thm:bvm}(i), $\sqrt n\,(\hPsi-\Psi_0)\dto\G_0$ in $\ell^\infty(\calF)$, and under Assumption~\ref{ass:quant}\textup{(i)} Lemma~\ref{lem:envelope} makes the map $\nu\mapsto\{\iota_K(\nu)\}_{K\le\overline K}$ Hadamard differentiable at $\Psi_0$ with the continuous linear derivative $\zeta\mapsto\{\zeta(g_{c^\star(K)})\}_K$. The delta method for maps with continuous linear derivatives \citep[Theorem~3.9.4]{vdvwellner1996} gives $\sqrt n\{\widehat W(K)-W(K)\}=\sqrt n(\hPsi-\Psi_0)(g_{c^\star(K)})+\opro(1)$, and the uniform asymptotic linearity \eqref{eq:linearity} replaces $\sqrt n(\hPsi-\Psi_0)(g_{c^\star(K)})$ by $\G_n\{\phi_{g_{c^\star(K)}}(\cdot;\eta_0)\}$ at the cost of another $\opro(1)$ term. The smooth map taking the vector $\{W(K)\}_K$ to $\{\rho(K)\}_K$ then yields the joint expansion
\begin{equation*}
\sqrt n\,\{\hat\rho(K)-\rho_0(K)\}\;=\;\G_n(\mathrm{IF}_K)+\opro(1)
\qquad\text{under }P_0^{\otimes n},\ \text{jointly over }K\le\overline K.
\end{equation*}
The remainder is $o(1)$ in outer $Q_n$-probability by Step~2 and Lemma~\ref{lem:transfer}. The vector $\{\G_n(\mathrm{IF}_K)\}_{K\le\overline K}$ and $\Lambda_n$ are jointly asymptotically normal under $P_0^{\otimes n}$ by \eqref{eq:lanhonest} and the multivariate central limit theorem, with asymptotic covariance $\Cov_0\{\mathrm{IF}_K,\,u^\ast s^\ast\}=u^\ast b_K(s^\ast)$ between the $K$th coordinate and $\Lambda_n$. Le Cam's third lemma \citep[Example~6.7]{vandervaart1998} gives, under $Q_n$,
\begin{equation*}
\{\G_n(\mathrm{IF}_K)\}_{K\le\overline K}\ \dto\ N\bigl(\{u^\ast b_K(s^\ast)\}_{K\le\overline K},\ \Sigma_{\mathbb H}\bigr),
\end{equation*}
where $\Sigma_{\mathbb H}$ is the covariance matrix of $\mathbb H$. Combining with \eqref{eq:rhodrifthonest},
\begin{equation}\label{eq:regularityhonest}
\bigl[\sqrt n\,\{\hat\rho(K)-\rho_{P_n}(K)\}\bigr]_{K\le\overline K}\ \dto\ \mathbb H
\qquad\text{under }Q_n .
\end{equation}
The deterministic drift of the moving target and the stochastic shift produced by Le Cam's third lemma cancel exactly. This is the local regularity of the corrected path estimator.

\textbf{Step 5 (posterior scale and quantile under the perturbation).} The proof of Theorem~\ref{thm:profile}(iii) established $\sqrt n\,\hat\sigma(K)\to\sigma_K$ for each $K$ and $q_{1-\alpha}\to m_{1-\alpha}$, both in outer $P_0^{\otimes n}$-probability. These are convergences to constants, so Lemma~\ref{lem:transfer} and Step~2 give the same limits in outer $Q_n$-probability. The posterior weights are drawn independently of the data, so their law is unchanged under $Q_n$.

\textbf{Step 6 (coverage and inversion).} Let $T_n=\max_{2\le K\le\overline K}|\hat\rho(K)-\rho_{P_n}(K)|/\hat\sigma(K)$, with the $K=1$ coordinate contributing zero because $L_\rho(1)=U_\rho(1)=0=\rho_{P_n}(1)$. By \eqref{eq:regularityhonest}, Step~5, and the continuous mapping theorem, $T_n\dto M=\max_{2\le K\le\overline K}|\mathbb H(K)|/\sigma_K$ under $Q_n$. The distribution function $F_M$ is continuous and strictly increasing on $(0,\infty)$, as shown in the proof of Theorem~\ref{thm:profile}(iii), so Slutsky's lemma gives
\begin{equation*}
Q_n\bigl(T_n\le q_{1-\alpha}\bigr)\ \longrightarrow\ F_M(m_{1-\alpha})\;=\;1-\alpha .
\end{equation*}
On the event $\{T_n\le q_{1-\alpha}\}$ the band covers $\rho_{P_n}(K)$ for every $K\le\overline K$. The band-to-profile inversion argument in the proof of Theorem~\ref{thm:profile}(iii) is deterministic and uses only the monotonicity of $K\mapsto\rho_{P_n}(K)$. Applying it with $\rho_0$ replaced by $\rho_{P_n}$ shows that on this event $K^\star_{P_n}(\gamma)\in\widehat C(\gamma)$ for every $\gamma\in(0,\rho_{P_n}(\overline K))$, and monotonized bounds are handled exactly as there. Hence \eqref{eq:seqclaim} holds, completing Step~0 and the proof. \hfill$\square$

\begin{remark}[Scope of the perturbation class]\label{rem:honestscope}
The two sequences $P_n^{\pm}$ of Theorem~\ref{thm:impossibility} are of the form $P_{n,u,s}$ with $s=\tilde s$ and $u=\pm h/b$, so the class of the theorem contains them once $B_u\ge h/(\sigma_{K_0}-\eps)$. Total boundedness of $\calS$ in $L_2(P_0)$ holds, for example, for any finite collection of scores, for parametric families indexed by compact sets and continuous in the index, and for uniformly bounded classes with finite bracketing entropy. The boundedness of the scores matches the least favorable construction in the impossibility proof, which is itself built from a truncated influence function, so no power is lost at the root-$n$ scale by the restriction.
\end{remark}

\begin{remark}[Cardinality along the drifting sequences]\label{rem:cardinality}
Steps 3 to 6 of the proof also bound the size of the report along the same sequences. Fix $\gamma=\rho_0(K_0)$ as in Theorem~\ref{thm:impossibility}. Under $Q_n$, every $K<K_0$ has $\hat\rho(K)\pto\rho_0(K)<\gamma$ by Step~4 while the band radius $q_{1-\alpha}\hat\sigma(K)$ is of order $n^{-1/2}$ by Step~5, so $U_\rho(K)<\gamma$ with probability tending to one and every such $K$ leaves the report. Likewise $L_\rho(K_0+1)\pto\rho_0(K_0+1)>\gamma$, so the defining condition of \eqref{eq:profileset} eventually excludes every $K\ge K_0+2$. Hence $Q_n\{\widehat C(\gamma)\subseteq\{K_0,K_0+1\}\}\to1$, and the honest report costs at most the two knot-adjacent counts, matching the mean cardinality near two observed in Study~4.
\end{remark}

\subsection{Proofs for the penalized profile and atomic recovery}\label{app:penalized}

Throughout this subsection, recall $W(K)=\inf_{c\in\calC^K}\Psi_0(g_c)$ and note that $W$ is automatically nonincreasing: any $K$-codebook is a $(K+1)$-codebook with a repeated center. Hence every slope ratio appearing in Proposition~\ref{prop:envelope}(ii) is nonnegative.

\subsubsection{Proof of Proposition~\ref{prop:envelope}}
Each map $\tau\mapsto W(K)+\tau K$ is affine and increasing, so $Q$, a minimum of $\overline K$ such maps, is concave, nondecreasing, and piecewise linear with at most $\overline K$ pieces. Two elementary facts organize the rest. \textbf{Supporting slopes.} If $K\in\calS_0(\tau)$, then $Q(\tau')\le W(K)+\tau'K$ for all $\tau'$ with equality at $\tau$, so the affine map with slope $K$ supports the concave $Q$ at $\tau$ and $K\in\partial Q(\tau)$, the superdifferential. \textbf{Monotone selection.} If $\tau<\tau'$, $K\in\calS_0(\tau)$, $K'\in\calS_0(\tau')$, adding the two optimality inequalities $W(K)+\tau K\le W(K')+\tau K'$ and $W(K')+\tau'K'\le W(K)+\tau'K$ gives $(\tau'-\tau)(K-K')\ge0$, so $K\ge K'$: selections are nonincreasing in $\tau$, uniformly over the choice of selection.

(i) At a differentiability point, $\partial Q(\tau)=\{Q'(\tau)\}$, so every member of the nonempty $\calS_0(\tau)$ equals $Q'(\tau)$: $\calS_0(\tau)=\{Q'(\tau)\}$. Conversely, if $\calS_0(\tau_0)=\{K\}$, the gap $\min_{L\ne K}\{W(L)+\tau_0L\}-\{W(K)+\tau_0K\}$ is positive and continuous in $\tau_0$, so $\calS_0=\{K\}$ on a neighborhood, on which $Q(\cdot)=W(K)+(\cdot)K$: every active order is a slope of $Q$, and the slope of each open linear piece is the unique selection there. The set $\{\tau:\calS_0(\tau)=\{K\}\}$ is open and is an interval: if it contains $\tau_1<\tau_2$ and $\tau\in(\tau_1,\tau_2)$, monotone selection against $\tau_1$ and against $\tau_2$ squeezes every member of $\calS_0(\tau)$ to equal $K$. At a kink $\tau_{(j)}$ the adjacent slopes are $K_{(j)}$ (right piece) and $K_{(j+1)}$ (left piece), so $\partial Q(\tau_{(j)})=[K_{(j)},K_{(j+1)}]$ and $\calS_0(\tau_{(j)})\subseteq[K_{(j)},K_{(j+1)}]$ by the supporting-slope fact; both endpoints lie in $\calS_0(\tau_{(j)})$ because $W(K_{(j)})+\tau K_{(j)}=Q(\tau)$ on the right piece and both sides are continuous at $\tau_{(j)}$, and symmetrically from the left.

(ii) $K$ is active iff there is $\tau>0$ with $W(K)+\tau K<W(L)+\tau L$ for every $L\ne K$, i.e.\ $\tau>\{W(K)-W(L)\}/(L-K)$ for every $L>K$ and $\tau<\{W(L)-W(K)\}/(K-L)$ for every $L<K$; that is, iff the open interval from $\max_{L>K}$ to $\min_{L<K}$ of the displayed ratios (conventions $\max_\emptyset=0$, $\min_\emptyset=\infty$) meets $(0,\infty)$. Since $W$ is nonincreasing all ratios are nonnegative, so the lower endpoint is $\ge0$ and the interval meets $(0,\infty)$ exactly when it is nonempty, which is the displayed strict inequality, saying that every chord arriving at $(K,W(K))$ from the left is steeper than every chord leaving to the right: a strict vertex of the greatest convex minorant. For $K=1$ the right side is $+\infty$ and the left side finite, so $1$ is always active; if $W(K)=W(K-1)$ the chord from $K-1$ contributes $0$ to the minimum, so the strict inequality fails and $K$ is inactive.

(iii) Each active order's persistence set is a nonempty open interval by (i), and by monotone selection the interval of the larger of two active orders lies entirely to the left. At any $\tau$ that is not a kink, $Q$ is differentiable, so by (i) $\tau$ belongs to the persistence interval of an active order; hence $(0,\infty)$ minus the finitely many kinks is the disjoint union of the persistence intervals, ordered with $K$ decreasing in $\tau$. Consecutive active orders $K_{(j)}<K_{(j+1)}$ therefore share a single boundary kink $\tau^\ast$, at which, by continuity, both attain $Q$: $W(K_{(j)})+\tau^\ast K_{(j)}=W(K_{(j+1)})+\tau^\ast K_{(j+1)}$, i.e.\ $\tau^\ast=\tau_{(j)}$ of Definition~\ref{def:merge}; it is strictly positive because activity of $K_{(j+1)}$ requires, by (ii) with $L=K_{(j)}$, that $\{W(K_{(j)})-W(K_{(j+1)})\}/(K_{(j+1)}-K_{(j)})$ strictly exceed a nonnegative quantity. The persistence interval of $K_{(j)}$ is thus exactly $(\tau_{(j)},\tau_{(j-1)})$, nonempty, which is the strict decrease $\tau_{(j)}<\tau_{(j-1)}$; the remaining clauses are immediate.

(iv) For all $\tau$ and $K$, $Q(\tau)\le W(K)+\tau K$, so $\sup_{\tau>0}\{Q(\tau)-\tau K\}\le W(K)$, with equality for $K\in\calS$ at any $\tau$ in its persistence interval. And restricting the minimum to $\calS$ changes nothing: off kinks the minimum is attained at an active order by (i), and at kinks by the two adjacent active orders. \hfill$\square$

\subsubsection{Proof of Proposition~\ref{prop:atomic}}
(i) The atoms lie in $\calC$, which contains the support of $P_U$. For $K\ge K_0$, the codebook consisting of the atoms (with repeats if $K>K_0$) gives $\Psi_0(g_c)=0$, so $W(K)=0$; conversely $\Psi_0(g_c)=0$ forces $\min_l\norm{u_h-c_l}=0$ for every atom, so the codebook contains all $K_0$ atoms, and for $K=K_0$ equals them up to labels. For $K<K_0$ and any $c\in\calC^K$, nearest-center assignment maps $K_0$ atoms to $K$ centers and is not injective, so some center $c_l$ is nearest to two atoms $u_h\ne u_{h'}$; then
\begin{equation*}
\begin{split}
\Psi_0(g_c)\ &\ge\ \omega_h\norm{u_h-c_l}^2+\omega_{h'}\norm{u_{h'}-c_l}^2
\ \ge\ \min_{z\in\R^q}\bigl\{\omega_h\norm{u_h-z}^2+\omega_{h'}\norm{u_{h'}-z}^2\bigr\}\\
&=\ \frac{\omega_h\omega_{h'}}{\omega_h+\omega_{h'}}\,\norm{u_h-u_{h'}}^2
\ \ge\ \frac{\omega_{\min}}{2}\,\Delta^2 ,
\end{split}
\end{equation*}
the identity by minimizing the quadratic at the weighted midpoint, and the last step because $ab/(a+b)$ is increasing in each argument, hence at least $\omega_{\min}^2/(2\omega_{\min})$. So $W(K)\ge\omega_{\min}\Delta^2/2=\bar\tau K_0$ for every $K<K_0$. Fix $\tau\in(0,\bar\tau)$: the penalized value at $K_0$ is $\tau K_0$; for $K>K_0$ it is $\tau K>\tau K_0$; for $K<K_0$ it is at least $\bar\tau K_0+\tau K>\tau K_0$. Hence $\calS_0(\tau)=\{K_0\}$ with the atom codebook, uniquely up to labels; $K_0$ is active with persistence interval containing $(0,\bar\tau)$; every $K>K_0$ has $W(K)=W(K_0)$ and is inactive by Proposition~\ref{prop:envelope}(ii); and $\calS\subseteq\{1,\dots,K_0\}$ with largest element $K_0$. Nothing used $\overline K$ beyond $\overline K\ge K_0$. In fact the merge scale below $K_0$ is even larger than $\bar\tau$: writing $K'<K_0$ for the next active order, $\tau_{(J-1)}=W(K')/(K_0-K')\ge(\omega_{\min}\Delta^2/2)/(K_0-1)=\bar\tau K_0/(K_0-1)$.

(ii) Applying the pairing bound with $K=1$ (possible since $K_0\ge2$) gives $W(1)\ge\omega_{\min}\Delta^2/2>0$, so $\rho$ is well defined; $\rho(K_0)=1$, and $\rho(K)<1$ for $K<K_0$ since $W(K)>0$. As $\rho$ is nondecreasing, for $\gamma\in(\rho(K_0-1),1)$ every $K<K_0$ has $\rho(K)\le\rho(K_0-1)<\gamma$ while $\rho(K_0)=1\ge\gamma$, so $K^\star(\gamma)=K_0$.

(iii) Write $\Delta_i(g_c)=\phi_{g_c}(O_i;\hateta^{(-b(i))})-\phi_{g_c}(O_i;\eta_0)$ and decompose, for every $K\le\overline K$, $c\in\calC^K$,
\begin{equation*}
\hPsi(g_c)-\Psi_0(g_c)=(P_n-P_0)\phi_{g_c}(\cdot;\eta_0)+\frac1n\sum_{i=1}^n\Delta_i(g_c).
\end{equation*}
\textbf{Oracle term.} Parts (a)--(b) of Lemma~\ref{lem:donsker} invoke only Assumptions~\ref{ass:bound}--\ref{ass:class}, because the entropy bound rests on the half-space geometry of Voronoi cells together with the constant envelope, so $\sup_{K,c}|(P_n-P_0)\phi_{g_c}(\cdot;\eta_0)|=\Op(n^{-1/2})$. \textbf{Increment term, stochastic part.} Conditionally on the training folds, for each fold the class $\{\Delta_\cdot(g_c)\}$ is the difference of two classes that are VC-type uniformly in $\eta$ (Lemma~\ref{lem:donsker}(a)) with envelope $2\bar F$; the per-fold conditional maximal inequality of Supplementary Section~\ref{app:onestep}, applied with the constant radius $2\bar F$ in place of a shrinking one, gives $\max_b\sup_{K,c}|(P_{n,b}-P_0)\Delta_\cdot(g_c)|=\Op(n^{-1/2})$. \textbf{Increment term, bias part.} For each fold, $P_0\Delta_\cdot(g_c)$ is the conditional bias of $\phi_{g_c}(\cdot;\hateta^{(-b)})$, to which the exact decomposition \eqref{eq:biasdecomp} applies with $f=g_c$ and $\nabla_mg_c=0$. Its second term is bounded by $C\hat r_\mu\hat r_\pi$ with $\hat r_\mu,\hat r_\pi$ the realized $L_2(P_0)$ nuisance errors, exactly as in Supplementary Section~\ref{app:eif}; for the first, Lemma~\ref{lem:geometry} gives the pointwise bound $|\mathrm{rem}|\le\norm{\Delta}^2+2\norm{\Delta}\,\diam\calC$, where the crossing indicators are bounded by one so the margin condition is never invoked, whence $\sup_{K,c}|P_0\Delta_\cdot(g_c)|\le C(\hat r_\mu^2+\hat r_\mu+\hat r_\mu\hat r_\pi)=\opro(1)$ under $r_\mu\vee r_\pi=\opro(1)$. Collecting the three bounds yields the first display. For the weighted version, decompose $\Psi^{(s)}(g_c)-\hPsi(g_c)=n^{-1}\sum_i(w_i-1)\phi_{g_c}(O_i;\eta_0)+n^{-1}\sum_i(w_i-1)\Delta_i(g_c)$ and reduce the Dirichlet weights to exponentials as in Supplementary Section~\ref{app:bvm} ($w_i=nE_i/S_n$, the factor $n/S_n$ handled on the usual event). Both index arrays are bounded with empirical VC-type entropy for every realization (Lemma~\ref{lem:donsker}(a)), so Lemma~\ref{lem:multmax} with $\hat\sigma$ bounded by the constant envelope shows that the multiplier processes scaled by $n^{-1/2}$ have conditional expectations of order one, hence the two averaged suprema are of conditional order $n^{-1/2}$, and conditional Markov yields the second display. \textbf{Recovery.} Since $|\inf_c\hPsi(g_c)-\inf_c\Psi_0(g_c)|\le\sup_c|\hPsi(g_c)-\Psi_0(g_c)|$, the first display gives $\max_{K\le\overline K}|\widehat W(K)-W(K)|=\opro(1)$, and conditionally for $W^{(s)}$. Fix $\tau\in(0,\bar\tau)$: by the computation in (i), the population penalized objective at $K_0$ beats every $K>K_0$ by at least $\tau$ and every $K<K_0$ by at least $K_0(\bar\tau-\tau)$; on the event that $2\max_K|\widehat W(K)-W(K)|$ is below $\min\{\tau,K_0(\bar\tau-\tau)\}$, whose probability tends to one, the empirical ordering is the same and $\widehat K^\dagger(\tau)=K_0$. For the threshold profile, $\widehat W(1)$ is bounded away from zero with probability tending to one, so $\max_K|\hat\rho(K)-\rho(K)|=\opro(1)$; on the event that this maximum is below $\min\{\gamma-\rho(K_0-1),\,1-\gamma\}$, $\hat\rho(K_0)\ge\gamma$ while $\hat\rho(K)<\gamma$ for all $K<K_0$, so $\widehat K^\star(\gamma)=K_0$. The conditional statements are identical with $W^{(s)}$ in place of $\widehat W$. \hfill$\square$

\subsubsection{Proofs of Corollaries~\ref{cor:penalized} and~\ref{cor:penalizedvalue}}
Write $Z_n(K)=\sqrt n\{\widehat W(K)-W(K)\}$ and $Z^{(s)}_n(K)=\sqrt n\{W^{(s)}(K)-\widehat W(K)\}$; Corollary~\ref{cor:path} gives $(Z_n(K))_{K\le\overline K}\dto Z:=(\G_0(g_{c^\star(K)}))_{K\le\overline K}$ and $(Z^{(s)}_n(K))_K\dtow Z$ in probability, so in particular $\max_K|\widehat W(K)-W(K)|=\Op(n^{-1/2})$ and $\max_K|W^{(s)}(K)-\widehat W(K)|=\opro^{\mathrm w}(1)$.

(a) Let $a(K)$ and $b(K)$ denote the max and min sides of the strict-vertex test in Proposition~\ref{prop:envelope}(ii). For active $K$, $b(K)-a(K)>0$. For inactive $K$, nondegeneracy gives $\eps_K=W(K)-\mathrm{GCM}(K)>0$; with $K_1<K<K_2$ the consecutive minorant vertices bracketing $K$, where vertices are active orders at which $W$ and the minorant agree, and with $s$ the common chord slope $\{W(K_1)-W(K_2)\}/(K_2-K_1)$, direct computation gives $a(K)\ge\{W(K)-W(K_2)\}/(K_2-K)=s+\eps_K/(K_2-K)$ and $b(K)\le\{W(K_1)-W(K)\}/(K-K_1)=s-\eps_K/(K-K_1)$, so $a(K)-b(K)\ge2\eps_K/\overline K>0$. Hence every order passes or fails the test with margin at least some $m_0>0$. Each of $a,b$ is a maximum or minimum of finitely many ratios with integer denominators $\ge1$, hence Lipschitz in $(W(K))_K$ for the supremum norm with constant $2$; on the event $\{8\max_K|\widehat W(K)-W(K)|<m_0\}$, of probability tending to one, every empirical test reproduces the population verdict, so $\widehat\calS=\calS$ and the empirical merge scales are the same linear functionals of $\widehat W$, converging to the $\tau_{(j)}$. Since neither endpoint of the reporting range is a merge scale, the finitely many strict inequalities determining which persistence intervals meet the range are reproduced as well, so $\Pp(\widehat\calS_{\mathrm{rep}}=\calS_{\mathrm{rep}})\to1$; the conditional statement is identical using $\max_K|W^{(s)}(K)-W(K)|=\opro^{\mathrm w}(1)$. On $\{\widehat\calS=\calS\}$, $\sqrt n(\hat\tau_{(j)}-\tau_{(j)})_j=D\,Z_n$ exactly, where row $j$ of $D$ carries $\pm(K_{(j+1)}-K_{(j)})^{-1}$ in the coordinates $K_{(j)},K_{(j+1)}$; multiplying by the indicator of this event and applying the continuous mapping theorem gives the unconditional limit, and on $\{\calS^{(s)}=\widehat\calS=\calS\}$, $\sqrt n(\tau^{(s)}_{(j)}-\hat\tau_{(j)})_j=D\,Z^{(s)}_n$, so the conditional limit follows from Corollary~\ref{cor:path} and the linear case of Theorem~\ref{thm:delta}.

For Corollary~\ref{cor:penalizedvalue}, ties occur only at merge scales: at any other $\tau$, $\calS_0(\tau)$ is a singleton by Proposition~\ref{prop:envelope}(i). Define $v(\tau)=\min_{K\ne K^\dagger(\tau)}\{W(K)+\tau K\}-Q(\tau)>0$ for $\tau\in T$. Since $T$ is compact and avoids the kinks, it splits into finitely many compact pieces on each of which $K^\dagger(\cdot)$ is constant and $v$ is continuous and positive; hence $v_T=\inf_Tv>0$. On the event $\{2\max_K|\widehat W(K)-W(K)|<v_T\}$, of probability tending to one, the empirical minimizer equals $K^\dagger(\tau)$ simultaneously for all $\tau\in T$, so $\sqrt n\{\widehat Q(\tau)-Q(\tau)\}=Z_n(K^\dagger(\tau))$ on $T$. The map $(z_K)_K\mapsto(z_{K^\dagger(\tau)})_{\tau\in T}$ is linear and bounded into $\ell^\infty(T)$, so the continuous mapping theorem gives the first display. At a merge scale, every order outside the tie set has a positive gap; on the corresponding event, $\sqrt n\{\widehat Q(\tau_{(j)})-Q(\tau_{(j)})\}=\min_{K\in\calS_0(\tau_{(j)})}Z_n(K)$, since $W(K)+\tau_{(j)}K=Q(\tau_{(j)})$ throughout the tie set; all coordinates being continuous maps of the single vector $Z_n$ on one event of probability tending to one, the convergences hold jointly. Under the nondegeneracy of (a) no inactive order lies on the minorant, so the tie set is the merging pair. The conditional statement on $T$ is identical: on the intersection of the two selection events, $\sqrt n\{Q^{(s)}(\tau)-\widehat Q(\tau)\}=Z^{(s)}_n(K^\dagger(\tau))$ on $T$, and the same bounded linear map applies. At a merge scale the relevant functional of the path, $z\mapsto\min_{K\in\calS_0(\tau_{(j)})}z_K$ at a tie, is directionally but not fully Hadamard differentiable, and the weighted posterior is in general inconsistent for such limits \citep{fang2019}.

(b) The scale, posterior-quantile, and band-coverage steps are verbatim those of Supplementary Section~\ref{app:profile}, with $(\hat\rho,\hat\sigma,\mathbb H)$ replaced by $(\widehat W,\hat\sigma_W,\{\G_0(g_{c^\star(K)})\}_K)$, the maximum running over $1\le K\le\overline K$, and $\sigma_{W,K}>0$ in place of $\sigma_K>0$: they yield $\liminf_n\Pp(B^W_n)\ge1-\alpha$ for the band event $B^W_n=\{L_W(K)\le W(K)\le U_W(K)\ \forall K\le\overline K\}$. On $B^W_n$, fix any $\tau>0$ and any $K_0\in\calS_0(\tau)$: for every $K'\le\overline K$,
\begin{equation*}
L_W(K_0)+\tau K_0\ \le\ W(K_0)+\tau K_0\ =\ Q(\tau)\ \le\ W(K')+\tau K'\ \le\ U_W(K')+\tau K' ,
\end{equation*}
so $K_0\in\widehat C^\dagger(\tau)$ by \eqref{eq:penset}. This holds for every $\tau$ and every member of $\calS_0(\tau)$ on the single event $B^W_n$, which proves the display. The two consequences recorded after the corollary also hold on $B^W_n$: for consecutive active orders, $\tau_{(j)}=\{W(K_{(j)})-W(K_{(j+1)})\}/\Delta K$ lies between $\{L_W(K_{(j)})-U_W(K_{(j+1)})\}_+/\Delta K$ and $\{U_W(K_{(j)})-L_W(K_{(j+1)})\}/\Delta K$, jointly over $j$, since each $W$ value is bracketed by its band; and replacing the two band statistics by their maximum, with a single posterior $(1-\alpha)$ quantile, produces an event contained in the intersection of a $\rho$-band and a $W$-band, on which the containment arguments for $\widehat C(\gamma)$ and $\widehat C^\dagger(\tau)$ run simultaneously. \hfill$\square$

\subsubsection{Proof of Proposition~\ref{prop:nontight}}
Fix $\gamma>0$, suppose $\Var\{\G_0(g_{c^\star(K_{(j)})})-\G_0(g_{c^\star(K_{(j+1)})})\}>0$, and set $\tau_n=\tau_{(j)}+\gamma n^{-1/2}$, $\Delta K=K_{(j+1)}-K_{(j)}$. On a fixed compact neighborhood of $\tau_{(j)}$ containing no other merge scale, all orders outside the merging pair are eliminated uniformly, with probability tending to one, by the gap argument in the proof of Corollary~\ref{cor:penalizedvalue} above; using $W(K_{(j+1)})+\tau_{(j)}K_{(j+1)}=W(K_{(j)})+\tau_{(j)}K_{(j)}$,
\begin{equation*}
\sqrt n\bigl\{\widehat Q(\tau_n)-Q(\tau_n)\bigr\}
=\min\bigl\{Z_n(K_{(j)}),\ Z_n(K_{(j+1)})+\gamma\,\Delta K\bigr\}+\opro(1)
\ \dto\ \min\bigl\{Z_{K_{(j)}},\ Z_{K_{(j+1)}}+\gamma\Delta K\bigr\},
\end{equation*}
which proves the local-drift display and shows explicitly that the pointwise limit changes with the $n^{-1/2}$ drift $\gamma$.

For non-tightness, keep a $\gamma>0$ for which
\[
p_1:=\Pp\{Z_{K_{(j)}}-Z_{K_{(j+1)}}>\gamma\Delta K+2\eps\}>0
\]
for some $\eps>0$, possible by the nondegenerate normality of the displayed difference. Let $X_n(\tau)=\sqrt n\{\widehat Q(\tau)-Q(\tau)\}$ and $\tau_m=\tau_{(j)}+\gamma m^{-1/2}$.  For every fixed $m$, as $N\to\infty$ the price $\tau_m$ is fixed to the right of the merge scale, so the competitor is eliminated and
\[
X_N(\tau_m)=Z_N(K_{(j)})+\opro(1),
\qquad
X_N(\tau_N)=\min\{Z_N(K_{(j)}),Z_N(K_{(j+1)})+\gamma\Delta K\}+\opro(1).
\]
Therefore
\[
\liminf_{N\to\infty}\Pp\{|X_N(\tau_N)-X_N(\tau_m)|>\eps\}\ge p_1 .
\]
Suppose, to the contrary, that $X_n$ were asymptotically tight in $\ell^\infty(T')$ for a set $T'$ containing a right neighborhood of $\tau_{(j)}$.  The asymptotic-tightness criterion in \cite[Theorem~1.5.7]{vdvwellner1996} supplies a semimetric $\varsigma'$ making $T'$ totally bounded and making $X_n$ asymptotically uniformly $\varsigma'$-equicontinuous.  Applying that criterion with $(\eps,p_1/2)$ gives a radius $r>0$ such that pairs with $\varsigma'$-distance below $r$ have increments larger than $\eps$ with limiting probability below $p_1/2$. The preceding display then implies that, for every fixed $m$, $\varsigma'(\tau_N,\tau_m)\ge r$ for all sufficiently large $N$.  Inductively choose $N_1<N_2<\cdots$ so that each new $\tau_{N_k}$ is at $\varsigma'$-distance at least $r$ from all earlier selected points.  This constructs an infinite $r$-separated subset of $T'$, contradicting total boundedness.  Thus no tight weak limit exists on a right neighborhood.  The same argument on the left of the merge scale, with the two active orders interchanged, proves the one-sided statement in general. \hfill$\square$

\subsection{Proofs of Theorems~\ref{thm:effects} and \ref{thm:projection}}\label{app:effects}

Fix $K$ and abbreviate $\beta^\star=\beta^\star(K)$, $V=V_\beta$, $\calB=\calB_K$. By Assumption~\ref{ass:class}(iii) and dominated differentiation, the population gradient map $G_0(\beta)=\Psi_0(\nabla_\beta\ell_\beta)$ is $C^1$ on the interior of $\calB$ with $G_0'(\beta^\star)=V$ nonsingular and $G_0(\beta^\star)=0$; fix $r_0>0$ such that $\bar B_0:=\{\beta:\norm{\beta-\beta^\star}\le r_0\}$ lies in the interior, $\beta^\star$ is the unique zero of $G_0$ in $\bar B_0$, and $G_0(\beta)=V(\beta-\beta^\star)+o(\norm{\beta-\beta^\star})$ as $\beta\to\beta^\star$. Write $\calG_\nabla=\{\nabla_\beta\ell_\beta^{(j)}:\beta\in\calB,\,j\le d_\beta\}\subset\calF_{\mathrm{str}}\subset\calF_{\mathrm{eff}}$ by Assumption~\ref{ass:class}(iii) and note that $\beta\mapsto\phi_{\nabla_\beta\ell_\beta}(\cdot;\eta_0)$ is Lipschitz into $L_2(P_0)$, hence into the covariance semimetric $\varsigma$.

\begin{lemma}[$Z$-functional differentiability]\label{lem:zfun}
For $\nu\in\ell^\infty(\calF_{\mathrm{eff}})$ define $T_\beta(\nu)$ as any point of $\bar B_0$ with
$\norm{\nu(\nabla_\beta\ell_\beta)|_{\beta=T_\beta(\nu)}}\le\inf_{\beta\in\bar B_0}\norm{\nu(\nabla_\beta\ell_\beta)}+\norm{\nu-\Psi_0}_\infty^2$.
Then $T_\beta$ is Hadamard differentiable at $\Psi_0$ tangentially to the set of $\zeta\in\ell^\infty(\calF_{\mathrm{eff}})$ that are $\varsigma$-continuous on $\calG_\nabla$, with derivative $T_\beta'(\zeta)=-V^{-1}\zeta(\nabla_\beta\ell_{\beta^\star})$, for every admissible selection.
\end{lemma}

\begin{proof}
Let $t_m\downarrow0$, $\norm{\zeta_m-\zeta}_\infty\to0$ with $\zeta$ as stated, $\nu_m=\Psi_0+t_m\zeta_m$, and $\beta_m=T_\beta(\nu_m)$; write $G_m(\beta)=G_0(\beta)+t_m\zeta_m(\nabla_\beta\ell_\beta)$, the map whose point of nearly minimal norm $\beta_m$ is, and $\Xi=\sup_m\norm{\zeta_m}_\infty<\infty$.

\textbf{The infimum is $o(t_m)$.} Take $\tilde\beta_m=\beta^\star-t_mV^{-1}\zeta(\nabla_\beta\ell_{\beta^\star})\in\bar B_0$ eventually. Then
$G_m(\tilde\beta_m)=V(\tilde\beta_m-\beta^\star)+o(t_m)+t_m\zeta_m(\nabla_\beta\ell_{\tilde\beta_m})
=t_m\{\zeta_m(\nabla_\beta\ell_{\tilde\beta_m})-\zeta(\nabla_\beta\ell_{\beta^\star})\}+o(t_m)$,
and $|\zeta_m(\nabla_\beta\ell_{\tilde\beta_m})-\zeta(\nabla_\beta\ell_{\beta^\star})|\le\norm{\zeta_m-\zeta}_\infty+|\zeta(\nabla_\beta\ell_{\tilde\beta_m})-\zeta(\nabla_\beta\ell_{\beta^\star})|\to0$ by $\varsigma$-continuity of $\zeta$ and Lipschitz continuity of the index map. Hence $\inf_{\bar B_0}\norm{G_m}\le\norm{G_m(\tilde\beta_m)}=o(t_m)$, and since the tolerance is $t_m^2\norm{\zeta_m}^2_\infty=O(t_m^2)$, also $\norm{G_m(\beta_m)}=o(t_m)$.

\textbf{Consistency.} $\norm{G_0(\beta_m)}\le\norm{G_m(\beta_m)}+t_m\Xi\to0$; since $G_0$ is continuous on the compact $\bar B_0$ with unique zero $\beta^\star$, $\inf\{\norm{G_0(\beta)}:\beta\in\bar B_0,\norm{\beta-\beta^\star}\ge\eps\}>0$ for every $\eps>0$, so $\beta_m\to\beta^\star$.

\textbf{Expansion.} Using differentiability of $G_0$ at $\beta^\star$ and the same continuity argument as above (now along $\beta_m\to\beta^\star$),
\begin{equation*}
o(t_m)=G_m(\beta_m)=V(\beta_m-\beta^\star)+o(\norm{\beta_m-\beta^\star})+t_m\bigl\{\zeta(\nabla_\beta\ell_{\beta^\star})+o(1)\bigr\}.
\end{equation*}
Nonsingularity of $V$ first gives $\norm{\beta_m-\beta^\star}=O(t_m)$, then $(\beta_m-\beta^\star)/t_m\to-V^{-1}\zeta(\nabla_\beta\ell_{\beta^\star})$.
\end{proof}

\begin{lemma}[Evaluation chain rule]\label{lem:chain}
Let $\{f_\beta:\beta\in\calB\}\subset\calF_{\mathrm{eff}}$ satisfy Assumption~\ref{ass:class}\textup{(iii)} (in particular $\beta\mapsto f_\beta$ Lipschitz into $(\calF_{\mathrm{eff}},\varsigma)$ and $\beta\mapsto\Psi_0(f_\beta)$ continuously differentiable near $\beta^\star$ by dominated differentiation). Then $E(\nu,\beta)=\nu(f_\beta)$ is Hadamard differentiable at $(\Psi_0,\beta^\star)$ on $\ell^\infty(\calF_{\mathrm{eff}})\times\R^{d_\beta}$, tangentially to $\{\zeta\ \varsigma\text{-continuous}\}\times\R^{d_\beta}$, with derivative $(\zeta,b)\mapsto\zeta(f_{\beta^\star})+\partial_\beta\Psi_0(f_\beta)\big|_{\beta^\star}^\top b$.
\end{lemma}

\begin{proof}
Let $t_m\downarrow0$, $\zeta_m\to\zeta$ uniformly with $\zeta$ $\varsigma$-continuous, $b_m\to b$. Then
\begin{equation*}
\frac{E(\Psi_0+t_m\zeta_m,\beta^\star+t_mb_m)-E(\Psi_0,\beta^\star)}{t_m}
=\zeta_m(f_{\beta^\star+t_mb_m})+\frac{\Psi_0(f_{\beta^\star+t_mb_m})-\Psi_0(f_{\beta^\star})}{t_m}.
\end{equation*}
The first term tends to $\zeta(f_{\beta^\star})$, since $|\zeta_m(f_{\beta^\star+t_mb_m})-\zeta(f_{\beta^\star})|\le\norm{\zeta_m-\zeta}_\infty+|\zeta(f_{\beta^\star+t_mb_m})-\zeta(f_{\beta^\star})|\to0$, by $\varsigma$-continuity of $\zeta$ and $\varsigma(f_{\beta^\star+t_mb_m},f_{\beta^\star})\le L\,t_m\norm{b_m}\to0$. The second tends to $\partial_\beta\Psi_0(f_\beta)|_{\beta^\star}^\top b$ by differentiability along the converging directions $b_m\to b$.
\end{proof}

\subsubsection{Proof of Theorem~\ref{thm:projection}}
\textbf{Step 1 (localization and exact stationarity).} The subclass $\{\ell_\beta\}\cup\calG_\nabla$ is contained in $\calF_{\mathrm{eff}}$, so the restricted version of Theorem~\ref{thm:bvm} gives $\sup_\beta|\hPsi(\ell_\beta)-\Psi_0(\ell_\beta)|\le n^{-1/2}\sup_{f\in\calF_{\mathrm{eff}}}|\sqrt n(\hPsi_f-\Psi_{0,f})|=\opro(1)$, and, for every $\eps>0$, $\Pp_w\{\sup_\beta|\Psi^{(s)}(\ell_\beta)-\hPsi(\ell_\beta)|>\eps\}\le\Pp_w\{\sup_{f\in\calF_{\mathrm{eff}}}|\sqrt n(\Psi^{(s)}_f-\hPsi_f)|>\eps\sqrt n\}\pto0$ by conditional tightness. Since $\beta\mapsto\Psi_0(\ell_\beta)$ is continuous on the compact $\calB$ with, after label alignment, unique minimizer $\beta^\star$ (Assumption~\ref{ass:proj}), the standard argmin argument gives $\hat\beta\pto\beta^\star$ and $\Pp_w(\norm{\beta^{(s)}-\beta^\star}>\eps)\pto0$ for every $\eps$. Pointwise in the data and weights, $\partial_\beta\hphi_{\ell_\beta,i}=\hphi_{\nabla_\beta\ell_\beta,i}$ because mixed partials of $\ell$ commute under Assumption~\ref{ass:class}(iii); hence $\beta\mapsto\Psi^{(s)}(\ell_\beta)$ and $\beta\mapsto\hPsi(\ell_\beta)$ are continuously differentiable with gradients $\Psi^{(s)}(\nabla_\beta\ell_\beta)$, $\hPsi(\nabla_\beta\ell_\beta)$, and on the events $\{\hat\beta\in\mathrm{int}\,\bar B_0\}$ and $\{\beta^{(s)}\in\mathrm{int}\,\bar B_0\}$, whose conditional probabilities tend to one, the minimizers are exact zeros of these gradients and therefore admissible values of $T_\beta(\hPsi)$, $T_\beta(\Psi^{(s)})$ in the sense of Lemma~\ref{lem:zfun}.

\textbf{Step 2 (delta method).} $\G_0$ has $\varsigma$-continuous paths, so Lemma~\ref{lem:zfun} and Theorem~\ref{thm:delta} (applied to the functional $T_\beta$, modified arbitrarily off the events of probability tending to one from Step 1, which does not affect weak limits) yield both displays. For the conditional statement, the interior first-order condition holds for the weighted minimizer with conditional probability tending to one in probability, that is $\Pp_w\{\Psi^{(s)}(\nabla_\beta\ell_{\beta^{(s)}})=0,\ \beta^{(s)}\in\mathrm{int}\,\calB_K\}\pto1$, which follows from the conditional consistency in Step 1 and interiority of $\beta^\star(K)$ under Assumption~\ref{ass:proj}. On those events the minimizers are exact zeros, so the delta method operates on the exact-zero functional and the tolerance in the definition of $T_\beta$ matters only with vanishing probability. The common limit is $-V^{-1}\G_0(\nabla_\beta\ell_{\beta^\star})\sim N(0,V^{-1}\Sigma_\beta V^{-1})$ with $\Sigma_\beta=\Var\{\phi_{\nabla_\beta\ell_{\beta^\star}}(O;\eta_0)\}$. Membership surfaces: $\beta\mapsto r_h(\cdot;\beta)\in\ell^\infty(\calC)$ is Hadamard (indeed Fr\'echet) differentiable by the uniformly bounded $\beta$-derivatives of Assumption~\ref{ass:class}(iii), so a second application of the chain rule transfers the conclusion; joint statements over finite sets of $K$ are coordinatewise. \hfill$\square$

\subsubsection{Proof of Theorem~\ref{thm:effects}}
Write $\Theta_{h,a}=\varrho\circ E_{h,a}\circ\Gamma$ with $\Gamma(\nu)=(\nu,T_\beta(\nu))$, $E_{h,a}(\nu,\beta)=\bigl(\nu(f^N_{h,a;\beta}),\nu(f^D_{h;\beta})\bigr)$, and $\varrho(x,y)=x/y$. $\Gamma$ is Hadamard differentiable at $\Psi_0$ tangentially to $\varsigma$-continuous directions with derivative $\zeta\mapsto(\zeta,-V^{-1}\zeta(\nabla_\beta\ell_{\beta^\star}))$ (identity component trivially; second component by Lemma~\ref{lem:zfun}). $E_{h,a}$ is Hadamard differentiable at $(\Psi_0,\beta^\star)$ by Lemma~\ref{lem:chain} applied componentwise to the families $\{f^N_{h,a;\beta}\}$, $\{f^D_{h;\beta}\}$, both within Assumption~\ref{ass:class}(iii). $\varrho$ is $C^1$ at $(N_0,D_h)$ with $D_h>0$ (Assumption~\ref{ass:proj}). The chain rule for Hadamard-differentiable maps \citep[Lemma~3.9.3]{vdvwellner1996} composes the three, giving differentiability of $\Theta_{h,a}$ at $\Psi_0$ tangentially to $C_\varsigma(\calF_{\mathrm{eff}})$ with derivative
\begin{equation*}
\Theta_{h,a}'(\zeta)
=\frac1{D_h}\Bigl[\zeta(f^N_{h,a;\beta^\star})-\psi_{h,a}\,\zeta(f^D_{h;\beta^\star})\Bigr]
+\Bigl[\frac{a_N}{D_h}-\frac{N_0\,a_D}{D_h^2}\Bigr]^\top\bigl\{-V^{-1}\zeta(\nabla_\beta\ell_{\beta^\star})\bigr\}
\end{equation*}
where $a_N=\partial_\beta\Psi_0(f^N_{h,a;\beta})|_{\beta^\star}$, $a_D=\partial_\beta\Psi_0(f^D_{h;\beta})|_{\beta^\star}$; the bracket is exactly $\partial_\beta\psi_{h,a}(\beta^\star)$ by the quotient rule. Theorem~\ref{thm:delta} applies on the smooth/structured subclass, jointly over the finite index set $(h,a)$; evaluating the derivative at $\zeta=\G_0$ and using $\Cov\{\G_0(f),\G_0(f')\}=\Cov\{\phi_f,\phi_{f'}\}$ identifies the limit as $N(0,\Var\,\Phi)$ with $\Phi$ as in \eqref{eq:compositeIF} (additive centering constants in $\Phi$ do not change the variance). Finally, on the events of probability tending to one from Step 1 above, Algorithm~\ref{alg:main}'s draw-level quantity $\psi^{(s)}_{h,a}(K)$ equals $\Theta_{h,a}(\Psi^{(s)})$ and $\hat\psi_{h,a}(K)=\Theta_{h,a}(\hPsi)$ by construction, since the per-draw $\beta$-refit is the evaluation of $T_\beta$ at $\Psi^{(s)}$, so the conditional display is precisely the statement that the algorithm targets the composite law; freezing $\beta$ at $\hat\beta$ instead composes with the constant map $\nu\mapsto(\nu,\hat\beta)$, whose derivative lacks the second component, yielding $\Phi^{\mathrm{fix}}$ only. \hfill$\square$

\subsubsection{Proof of Corollary~\ref{cor:fixedgamma}}
Since $\gamma_0$ is not a knot, Theorem~\ref{thm:profile}\textup{(i)} gives $\widehat K^\star(\gamma_0)\pto K^\star_0$ and $\Pp_w\{K^{\star(s)}(\gamma_0)=K^\star_0\}\pto1$, so the data event $\mathcal E_n=\{\widehat K^\star(\gamma_0)=K^\star_0\}$ has $\Pp(\mathcal E_n)\to1$ and the per-draw selection equals $K^\star_0$ with conditional probability tending to one in probability.

\textbf{Sampling side.} On $\mathcal E_n$ the resolution selected from the profile equals the fixed resolution $K^\star_0$, so the point estimator $\hat\psi_{h,a}(\widehat K^\star(\gamma_0))$ coincides with its fixed-resolution counterpart $\hat\psi_{h,a}(K^\star_0)$. Estimator sequences that agree on events of probability tending to one share the same weak limit, so the sampling Gaussian limit of Theorem~\ref{thm:effects} at $K=K^\star_0$ transfers verbatim.

\textbf{Conditional side.} Write $Z^{(s)}_{\mathrm{sel}}=\sqrt n\{\psi^{(s)}_{h,a}(K^{\star(s)}(\gamma_0))-\hat\psi_{h,a}(\widehat K^\star(\gamma_0))\}$ and $Z^{(s)}_0=\sqrt n\{\psi^{(s)}_{h,a}(K^\star_0)-\hat\psi_{h,a}(K^\star_0)\}$. On $\mathcal E_n$ the centering agrees, and on the per-draw event $\{K^{\star(s)}(\gamma_0)=K^\star_0\}$ the numerator agrees, so $Z^{(s)}_{\mathrm{sel}}=Z^{(s)}_0$ there. Hence for every $\varphi\in\BL$,
\begin{equation*}
\bigl|\E_w\varphi(Z^{(s)}_{\mathrm{sel}})-\E_w\varphi(Z^{(s)}_0)\bigr|
\ \le\ 2\,\Pp_w\{K^{\star(s)}(\gamma_0)\ne K^\star_0\}+2\,\ind{\mathcal E_n^c}
\ \longrightarrow\ 0
\end{equation*}
in probability, using $|\varphi|\le1$, $\Pp_w\{K^{\star(s)}(\gamma_0)\ne K^\star_0\}\pto0$, and $\Pp(\mathcal E_n^c)\to0$. The right side is free of $\varphi$, so the same bound holds for the supremum over $\varphi\in\BL$, and both terms are data-measurable and tend to zero in probability, hence in outer probability, matching the mode of Definition~\ref{def:condweak}. By Theorem~\ref{thm:effects}, $Z^{(s)}_0\dtow N(0,\Sigma_\psi)$, so the triangle inequality in the bounded-Lipschitz metric of Definition~\ref{def:condweak} gives $Z^{(s)}_{\mathrm{sel}}\dtow N(0,\Sigma_\psi)$, jointly over the reported $(h,a)$.

The argument is pointwise in $P_0$. It does not extend uniformly over the root-$n$ neighborhoods of Theorem~\ref{thm:impossibility}, on which $\Pp(\mathcal E_n)\to1$ fails, which is why the recommended report when $\widehat C(\gamma_0)$ is not a singleton is to give the subgroup effects at every supported $K$ rather than at a single selected count. \hfill$\square$

\subsection{Supporting empirical-process lemmas}\label{app:lemmas}

\begin{lemma}[Entropy, Donsker property, and index continuity]\label{lem:donsker}
Let $\Phi_\eta=\{\phi_f(\cdot;\eta):f\in\calF\}$ for $\eta$ in the truncated range, with constant envelope $\bar F$.
\begin{enumerate}[label=(\alph*),leftmargin=2em]
\item There exist $A<\infty$ and $v<\infty$, depending only on the constants declared at the start of Supplementary Section~\ref{app:proofs}, such that
$\sup_Q N\bigl(\eps\bar F,\Phi_\eta,L_2(Q)\bigr)\le(A/\eps)^v$ for all $\eps\in(0,1]$, all finitely supported $Q$, and all $\eta$ in the truncated range.
\item $\Phi_{\eta_0}$ is $P_0$-Donsker.
\item For each $K$ and all ordered tuples $c,c'\in\calC^K$ with $d_\infty(c,c')=\max_{h\le K}\norm{c_h-c'_h}\le\delta\le1$,
$\norm{\phi_{g_c}(\cdot;\eta_0)-\phi_{g_{c'}}(\cdot;\eta_0)}_{P_0,2}\le C(\delta+\delta^{\am/2})$;
in particular $c\mapsto\phi_{g_c}$ is uniformly continuous into the covariance semimetric $\varsigma$. The statement is unaffected by repeated centers; if two ordered labels have the same center, their boundary contributes no jump.
\end{enumerate}
\end{lemma}

\begin{proof}
(a) Treat the subclasses separately; finite unions multiply covering numbers by constants. For $f=f_\theta\in\calF_{\mathrm{sm}}\cup\calF_{\mathrm{str}}$, Assumption~\ref{ass:class}(iii) makes $\theta\mapsto(f_\theta,\nabla f_\theta)$ Lipschitz in supremum norm over compact finite-dimensional index sets; since the residual is bounded in the truncated range, $|\phi_{f_\theta}(o;\eta)-\phi_{f_{\theta'}}(o;\eta)|\le C\norm{\theta-\theta'}$ pointwise, so covering the index set by $(3\,\mathrm{diam}/\tau)^{d_\theta}$ balls gives $N(C\tau,\cdot,L_2(Q))\le(C'/\tau)^{d_\theta}$ for every $Q$ and $\eta$. For $f=g_c\in\calF_{\mathrm{qt}}$, write
$\phi_{g_c}(o;\eta)=\sum_{h\le K}\ind{h_c(U_\eta(x))=h}\,m_{c_h}(o;\eta)$,
$m_{c_h}=\norm{U_\eta-c_h}^2+2(U_\eta-c_h)^\top HR_\eta$,
with ties broken by lowest index. The functions $x\mapsto s_{c,hj}(U_\eta(x))$, where $s_{c,hj}(u)=\norm{u-c_h}^2-\norm{u-c_j}^2$ compares cells $h$ and $j$ as in the proof of Lemma~\ref{lem:geometry}, are affine in $U_\eta(x)$, hence range over a vector space of dimension at most $q+1$; the sets $\{s\le0\}$ therefore form a VC class of index at most $q+3$ \citep[Lemma~2.6.15 and Lemma~2.6.18]{vdvwellner1996}, cell indicators are intersections of at most $\overline K-1$ such sets and remain VC with index bounded by a constant \citep[Lemma~2.6.17]{vdvwellner1996}, so they admit uniform entropy $(A_1/\eps)^{v_1}$ \citep[Theorem~2.6.7]{vdvwellner1996}. The factors $m_{c_h}$ are uniformly bounded and pointwise Lipschitz in $c_h\in\calC$, hence of uniform entropy $(A_2/\eps)^{q}$ as above. For uniformly bounded classes, $\norm{fg-f'g'}_{Q,2}\le\norm{g}_\infty\norm{f-f'}_{Q,2}+\norm{f}_\infty\norm{g-g'}_{Q,2}$, so products multiply covering numbers and sums over $h\le\overline K$ and the union over $K\le\overline K$ keep the polynomial form. All Lipschitz constants and envelopes are uniform over the truncated $\eta$-range, giving (a).

(b) The classes are indexed by compact metric spaces with $(o,\text{index})\mapsto\phi$ jointly measurable, hence image-admissible Suslin, and standard measurability for suprema applies \citep[Section~8.2]{kosorok2008}. With (a) at $\eta_0$ and the bounded envelope, the uniform-entropy Donsker theorem \citep[Theorem~2.5.2]{vdvwellner1996} gives (b).

(c) The loss parts differ by at most $4\diam\calC\cdot\delta$ pointwise. The gradient parts differ by $2(c'_{h'}-c_h)^\top HR$ with $h=h_c(U)$, $h'=h_{c'}(U)$; decompose $c'_{h'}-c_h=(c'_{h'}-c_{h'})+(c_{h'}-c_h)$. The first term is at most $\delta$ in norm. For the second term, if $h'=h$ there is no contribution. If $h'\ne h$ but $c_{h'}=c_h$ as points, the jump is again zero; repeated centers therefore cause no difficulty. Otherwise put $d=\norm{c_{h'}-c_h}>0$. On the flip event, optimality under both codebooks gives
$0\le\norm{U-c_{h'}}^2-\norm{U-c_h}^2\le8\diam\calC\cdot\delta$,
so $\dist(U,B_{c,h'h})\le4\diam\calC\,\delta/d$. The jump magnitude is $d$, and Assumption~\ref{ass:margin}, with the trivial bound used when $4\diam\calC\,\delta/d>t_0$, gives
\begin{equation*}
\E\Bigl[\norm{c_{h'}-c_h}^2\,\norm{HR}^2\,\ind{h\ne h'}\Bigr]
\le C\sum_{(h,h')} d_{hh'}^2\min\{1,(\delta/d_{hh'})^\am \}
\le C\overline K^2(\diam\calC)^{2-\am}\delta^\am,
\end{equation*}
where pairs with $d_{hh'}=0$ are interpreted as contributing zero, and the last inequality uses $d^2\le\delta^\am$ for $d\le\delta\le1$ and $d^2(\delta/d)^\am\le(\diam\calC)^{2-\am}\delta^\am$ for $d>\delta$. Collecting, $\norm{\phi_{g_c}-\phi_{g_{c'}}}_{P_0,2}^2\le C(\delta^2+\delta^\am)$, which is (c); the $\varsigma$-statement follows since $\varsigma^2\le\norm{\cdot}_{P_0,2}^2$ for centered differences.
\end{proof}

\begin{lemma}[Estimated-score increment classes]\label{lem:incrementclass}
Let $\Delta_\eta(f)=\phi_f(\cdot;\eta)-\phi_f(\cdot;\eta_0)$ for a nuisance value $\eta$ in the truncated range, and let $\Delta_i(f)=\phi_f(O_i;\hateta^{(-b(i))})-\phi_f(O_i;\eta_0)$ be the cross-fitted increment array.
\begin{enumerate}[label=(\alph*),leftmargin=2em]
\item For each fixed $\eta$, the class $\calD_\eta=\{\Delta_\eta(f):f\in\calF\}$ has envelope $2\bar F$ and satisfies
\[
\sup_Q N\bigl(\eps\,2\bar F,\calD_\eta,L_2(Q)\bigr)\le(A/\eps)^v,
\qquad 0<\eps\le1,
\]
with constants independent of $\eta$.
\item Conditional on the training folds, the cross-fitted vector class
$\calD_n=\{(\Delta_1(f),\ldots,\Delta_n(f)):f\in\calF\}$ has envelope $2\bar F$ and
\[
\log N\bigl(\eps\,2\bar F,\calD_n,\norm{\cdot}_{P_n,2}\bigr)\le Bv\log(A/\eps),
\qquad 0<\eps\le1.
\]
The same conclusion holds with $P_n$ replaced by the empirical norm on any fixed union of folds.
\item The squared classes $\calD_\eta^2=\{d^2:d\in\calD_\eta\}$ and $\calD_n^2=\{(\Delta_1(f)^2,\ldots,\Delta_n(f)^2):f\in\calF\}$ have envelopes $4\bar F^2$ and polynomial entropy with the same type of constants, uniformly in $\eta$ and conditionally on the training folds.
\item With $\delta_n=\max_b\sup_f\norm{\Delta^{(b)}(f)}_{P_0,2}$ and $\delta'_n=\sup_f\{P_n\Delta_\cdot(f)^2\}^{1/2}$,
\[
(\delta'_n)^2\le\delta_n^2+\Op(n^{-1/2}),
\qquad
\delta'_n\le\delta_n+\Op(n^{-1/4}).
\]
\end{enumerate}
\end{lemma}

\begin{proof}
Part (a) follows from Lemma~\ref{lem:donsker}\textup{(a)}: $\calD_\eta$ is contained in the difference of two uniformly VC-type score classes with the same bounded envelope, and products of the two covering nets cover the difference class. For (b), cover the restriction of $\calD_n$ to each fold using part (a), and take the product of the $B$ foldwise nets. Since $B$ is fixed, only the entropy exponent changes by the factor $B$.

For (c), if $|a|,|b|\le 2\bar F$, then $|a^2-b^2|\le4\bar F|a-b|$. Thus an $L_2(Q)$ net of radius $\eps\bar F$ for $\calD_\eta$ induces an $L_2(Q)$ net of radius $4\eps\bar F^2$ for $\calD_\eta^2$; the cross-fitted array is handled fold by fold as in (b). This gives the squared-class entropy directly from the Lipschitz bound.

For (d), decompose
\[
P_n\Delta_\cdot(f)^2
=\sum_{b=1}^B\frac{n_b}{n}P^{(b)}_{n_b}\{\Delta^{(b)}(f)^2\}
\le \delta_n^2+\sum_{b=1}^B\frac{n_b}{n}\sup_f\bigl|(P^{(b)}_{n_b}-P_0)\Delta^{(b)}(f)^2\bigr|.
\]
Conditional on the training data, the squared fold classes have bounded envelopes and polynomial entropy by (c). The same bounded VC maximal inequality used in \eqref{eq:epbound}, now with fixed radius $O(1)$, gives each supremum $\Op(n_b^{-1/2})$; $B$ is fixed and $n_b\asymp n$, hence the sum is $\Op(n^{-1/2})$. The square-root bound follows from $\sqrt{x+y}\le\sqrt x+\sqrt y$. \hfill$\square$
\end{proof}

\begin{lemma}[Conditional multiplier maximal inequality]\label{lem:multmax}
Let $\xi_1,\dots,\xi_n$ be i.i.d.\ mean zero with $\norm{\xi_1}_{\psi_1}\le4$, independent of the fixed array $\{z_i(f):i\le n,\,f\in\calF'\}$ with $|z_i(f)|\le\bar B$ and
$\log N\bigl(\eps\bar B,\calF',\norm{\cdot}_{P_n,2}\bigr)\le v\log(A/\eps)$ for $\eps\in(0,1]$. Then, with $\hat\sigma^2=\sup_fP_nz(f)^2$,
\begin{equation*}
\E_\xi\sup_{f\in\calF'}\Bigl|\frac1{\sqrt n}\sum_{i=1}^n\xi_i z_i(f)\Bigr|
\;\le\;C_v\Bigl\{\hat\sigma\sqrt{\log(A\bar B/\hat\sigma)}+\bar B\,\frac{\log n}{\sqrt n}\Bigr\}.
\end{equation*}
\end{lemma}

\begin{proof}
For fixed $f,f'$ and $a_i=z_i(f)-z_i(f')$, Bernstein's inequality for sums of independent sub-exponential variables gives
\begin{equation*}
\Pp_\xi\Bigl\{\bigl|n^{-1/2}\textstyle\sum_i\xi_ia_i\bigr|>t\Bigr\}\le2\exp\bigl[-c\min\bigl\{t^2/\norm{a}_{P_n,2}^2,\ \sqrt n\,t/\max_i|a_i|\bigr\}\bigr],
\end{equation*}
i.e.\ the process has mixed sub-Gaussian/sub-exponential increments for the pair $d_2(f,f')=\norm{z(f)-z(f')}_{P_n,2}$ and $d_\infty(f,f')=n^{-1/2}\max_i|z_i(f)-z_i(f')|$. Chaining for processes with mixed sub-Gaussian and sub-exponential increments (\citealp{dirksen2015}; see also \citealp[Section~2.2]{vdvwellner1996}) yields
\begin{equation*}
\E_\xi\sup_{f}\Bigl|\frac1{\sqrt n}\sum_i\xi_i\{z_i(f)-z_i(f_0)\}\Bigr|
\le C\Bigl\{\int_0^{\hat\sigma'}\!\!\sqrt{\log N(\eps,d_2)}\,d\eps
+\int_0^{D_\infty}\!\!\log N(\eps,d_\infty)\,d\eps\Bigr\},
\end{equation*}
with $\hat\sigma'\le2\hat\sigma$, $D_\infty\le2\bar B/\sqrt n$. The first integral is at most $C\hat\sigma\sqrt{v\log(A\bar B/\hat\sigma)}$ by the entropy hypothesis. For the second, $d_\infty\le\norm{z(f)-z(f')}_{P_n,2}$ pointwise on an $n$-point sample only after rescaling: $\max_i|a_i|\le\sqrt n\norm{a}_{P_n,2}$, so $N(\eps,d_\infty)\le N(\eps,d_2)\le(A\bar B/\eps)^v$, whence the second integral is at most
$\int_0^{2\bar B/\sqrt n}v\log(A\bar B/\eps)\,d\eps\le Cv\,\bar Bn^{-1/2}\log(A\sqrt n)\le C_v\bar Bn^{-1/2}\log n$.
Adding the single-point bound $\E_\xi|n^{-1/2}\sum\xi_iz_i(f_0)|\le C\hat\sigma$ completes the proof.
\end{proof}

\begin{lemma}[Dirichlet weights are admissible multipliers]\label{lem:pwcheck}
The weights $w^{(s)}_i=nE_i/S_n$, $E_i$ i.i.d.\ standard exponential, satisfy the conditions of the exchangeable-bootstrap central limit theorem \citep[Theorem~2.2]{praestgaard1993} and \citep[Theorem~3.6.13]{vdvwellner1996}: they are nonnegative and exchangeable with $\sum_iw_i=n$; $\sup_n\norm{w_1}_{2,1}<\infty$; $n^{-1}\sum_i(w_i-1)^2\pto1$; and $n^{-1/2}\,\E_w\max_i|w_i-1|\to0$.
\end{lemma}

\begin{proof}
Exchangeability and the sum constraint are immediate. For the $L_{2,1}$ bound, $w_1\le n$ and
$\Pp(w_1>t)\le\Pp(E_1>t/2)+\Pp(S_n<n/2)\le e^{-t/2}+e^{-cn}$,
so $\norm{w_1}_{2,1}=\int_0^n\sqrt{\Pp(w_1>t)}\,dt\le\int_0^\infty e^{-t/4}dt+ne^{-cn/2}\le4+o(1)$.
By the strong law, $n^{-1}\sum E_i^2\to2$ and $\bar E_n\to1$ almost surely, so
$n^{-1}\sum(w_i-1)^2=(n/S_n)^2\{n^{-1}\sum E_i^2-\bar E_n^2\}\to1$.
Finally, $\E\max_iE_i\le1+\log n$ for i.i.d.\ standard exponentials, and $\E\max_iE_i^2\le C\log^2n$ by integrating the union tail bound $\Pp(\max_iE_i>t)\le ne^{-t}$. Since $S_n\sim\mathrm{Gamma}(n,1)$ gives $\E(n/S_n)^2=n^2/\{(n-1)(n-2)\}\le C$ for $n\ge3$, Cauchy--Schwarz yields $\E_w\max_iw_i=\E\{(n/S_n)\max_iE_i\}\le\{\E(n/S_n)^2\}^{1/2}\{\E\max_iE_i^2\}^{1/2}\le C\log n$, so $n^{-1/2}\E_w\max_i|w_i-1|\le n^{-1/2}(C\log n+1)\to0$.
\end{proof}

\begin{lemma}[Automatic increment rate on the quantization class]\label{lem:deltaqt}
Under Assumptions~\ref{ass:ident}--\ref{ass:margin}, for each fold $b$,
\begin{equation*}
\sup_{K\le\overline K}\sup_{c\in\calC^K}\norm{\phi_{g_c}(\cdot;\hateta^{(-b)})-\phi_{g_c}(\cdot;\eta_0)}_{P_0,2}
\;=\;\Op\bigl(r_\mu^{\am/(2+\am)}+r_\mu+r_\pi\bigr).
\end{equation*}
\end{lemma}

\begin{proof}
Write $\hU=U_{\hateta}(X)$, $\hat\Delta=\hU-U$, so $\norm{\hat\Delta}\le\norm{H}\norm{\hat\bmu-\bmu}$ pointwise and $\norm{\hat\Delta}_{P_0,2}\le Cr_\mu$. The loss parts differ by $|g_c(\hU)-g_c(U)|\le4\diam\calC\,\norm{\hat\Delta}$, of $L_2$ norm $\le Cr_\mu$. The gradient parts differ by
$2(\hU-c_{h_c(\hU)})^\top H\hR-2(U-c_{h_c(U)})^\top HR
=2(\hU-c_{h_c(\hU)})^\top H(\hR-R)+2\{\hat\Delta+(c_{h_c(U)}-c_{h_c(\hU)})\}^\top HR$.
The first term is bounded by $C\norm{\hR-R}$, of $L_2$ norm $\le C(r_\mu+r_\pi)$ (insert and subtract $\mu$, $\pi$ in $\hR$ and use the truncation). In the second, $\norm{\hat\Delta}^\top$-part has $L_2$ norm $\le Cr_\mu$, and the center-difference is nonzero only on the crossing event between $U$ and $\hU$ under $c$, on which, by Lemma~\ref{lem:geometry}(b), $\dist(U,B)\le2\norm{\hat\Delta}$ for some bisector hyperplane $B$ of $c$. For any $t>0$, by Assumption~\ref{ass:margin} and Chebyshev,
$P_0(\text{crossing})\le\overline K^2\bigl\{C_M(2t)^\am +t^{-2}\E\norm{\hat\Delta}^2\bigr\}$;
optimizing $t=(\E\norm{\hat\Delta}^2)^{1/(2+\am)}$ gives $P_0(\text{crossing})\le C(\E\norm{\hat\Delta}^2)^{\am/(2+\am)}$, so the center-difference term, bounded by $\diam\calC\cdot\norm{HR}\ind{\text{crossing}}$, has $L_2$ norm $\Op(r_\mu^{\am/(2+\am)})$ conditionally on the training fold, by Markov over the training randomness. Collecting terms proves the claim; constants are uniform over $(c,K)$.
\end{proof}

\end{document}